\documentclass[reqno,10pt,a4paper,titlepage]{amsbook}	

\usepackage{amsmath, amssymb, amsthm, amstext}		
\usepackage{txfonts}								
\usepackage[mathletters]{ucs}
\usepackage[utf8x]{inputenc}
\DeclareFontFamily{U}{mathb}{\hyphenchar\font45}
\DeclareFontShape{U}{mathb}{m}{n}{
      <5> <6> <7> <8> <9> <10> gen * mathb
      <10.95> mathb10 <12> <14.4> <17.28> <20.74> <24.88> mathb12
      }{}
\DeclareSymbolFont{mathb}{U}{mathb}{m}{n}
\DeclareFontSubstitution{U}{mathb}{m}{n}
\DeclareMathSymbol{\boxvoid}      {2}{mathb}{"6C}

\usepackage[english]{babel}
\usepackage[T1]{fontenc}
\usepackage{ae,aecompl}
\usepackage{enumerate}
\usepackage[numbers, square]{natbib}
\usepackage{appendix}
\usepackage{lmodern}
\usepackage{epsfig,color}
\usepackage[arrow, matrix, curve]{xy} 				
\usepackage{tensor} 								

\usepackage{color}								
\usepackage{mathrsfs}
\usepackage{bbm}
\usepackage{comsym}
\usepackage{graphicx}

\numberwithin{equation}{chapter}

\newtheorem{lem}{Lemma}[chapter]
\newtheorem{prop}[lem]{Proposition} 
\newtheorem{thm}[lem]{Theorem}
\newtheorem{cor}[lem]{Corollary}
\newtheorem{conj}[lem]{Conjecture} 
\newtheorem{summ}[lem]{Summary}

\theoremstyle{definition}
\newtheorem{Def}[lem]{Definition} 
\newtheorem{rem}[lem]{Remark}
\newtheorem{ex}[lem]{Example}

\newtheorem*{thm*}{Theorem}
\renewcommand{\thesection}{\thechapter.\arabic{section}}

\usepackage[textwidth=5in,textheight=8.25in,paperheight=11in]{geometry} 
\usepackage{setspace}
\allowdisplaybreaks[1]  
\input{headings.tex}

\DeclareRobustCommand{\gobblefive}[5]{}
\newcommand*{\SkipTocEntry}{\addtocontents{toc}{\gobblefive}} 

\AtBeginDocument{} 

\usepackage{hyperref}
\hypersetup{%
  pdftitle = {Twistor Theory of Higher-Dimensional Black Holes},
  pdfauthor = {Norman Metzner},
  bookmarksnumbered=true,
  pdfdisplaydoctitle=true,
  linkbordercolor = {0 0 1}
}

\begin{document}


\addtolength{\jot}{0.1cm}			

\pagenumbering{roman}

\thispagestyle{empty}
\begin{center}
\Huge{\textbf{Twistor Theory of Higher-Dimensional Black Holes}}\\[0.5cm]
\huge{Norman Metzner}\\[1cm]
\vspace{5cm}
\begin{figure}[htbp]
\begin{center}
\hspace{0.15cm}
\includegraphics[height=5.3cm]{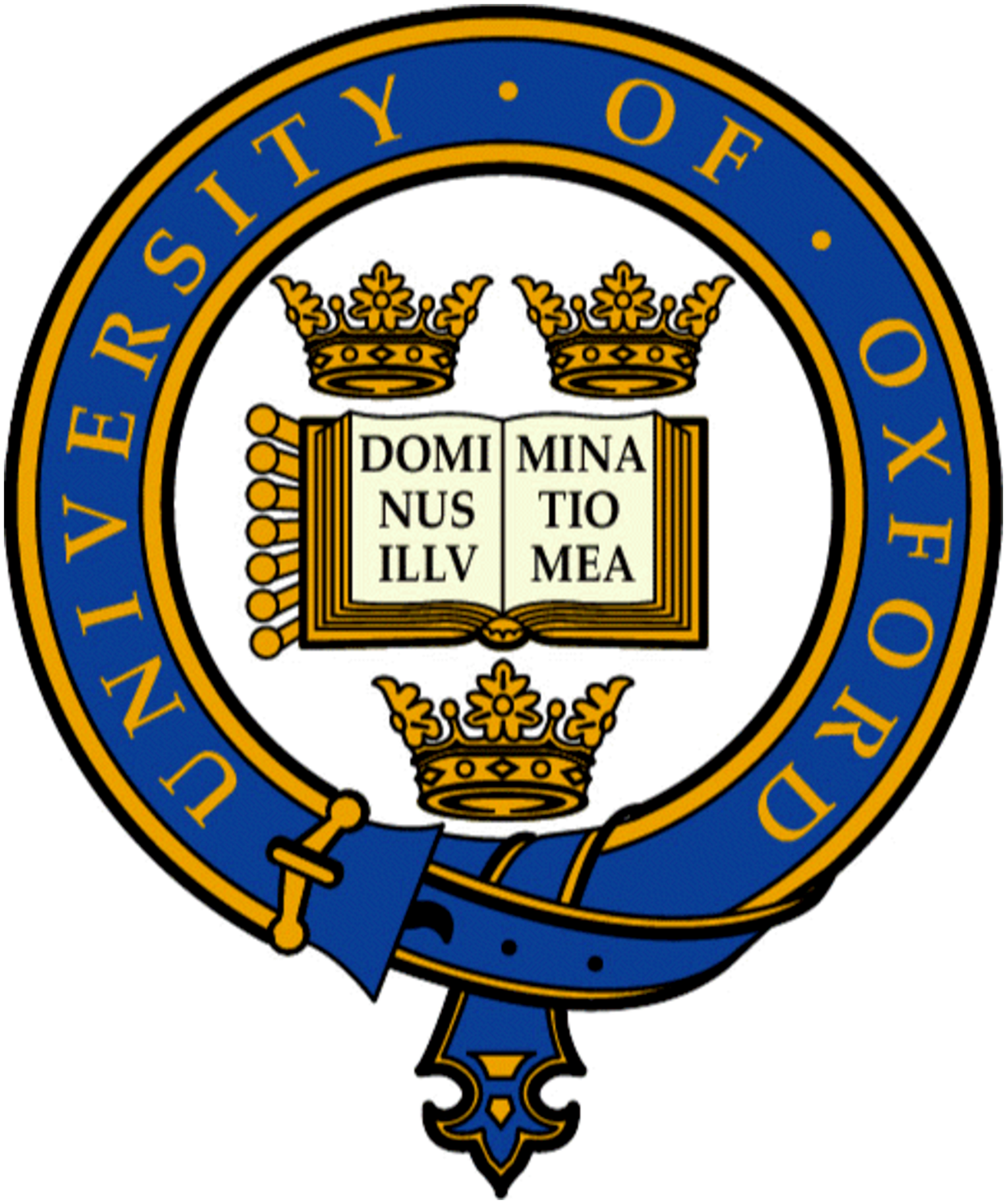}
\hspace{2cm}
\includegraphics[height=4.7cm]{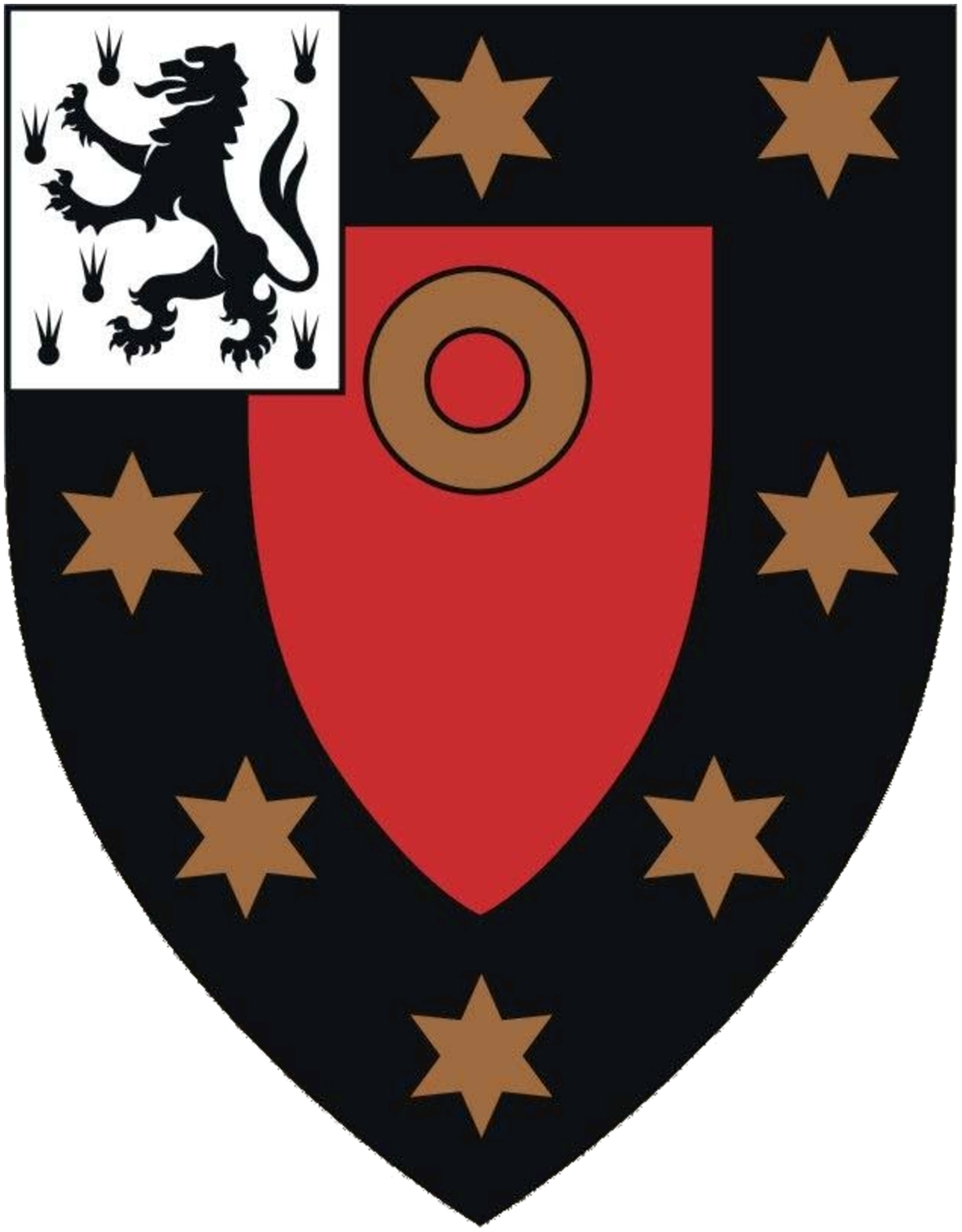}
\end{center}
\end{figure}
\vspace{-0.5cm}
\begin{large}
\begin{tabular}{ccc}
Mathematical Institute &\hphantom{fillingfillingf}& St John's College\\
University of Oxford && Oxford
\end{tabular}
\end{large}

\vspace{1.5cm}
\begin{LARGE}
Thesis submitted for the degree of \\[0.2cm]
\textit{Doctor of Philosophy in Mathematics}\\[0.2cm]
Trinity Term 2012
\end{LARGE}
\end{center}

\cleardoublepage

\vspace*{1.5cm}
{\Large \bfseries Acknowledgements}
\vspace{0.5cm}
\begin{quote}
I would like to thank my supervisors Paul Tod and Lionel Mason for the fruitful discussions and the guidance throughout all the stages towards this thesis. Furthermore, I am grateful to Piotr Chru\'{s}ciel and Nicholas Woodhouse for sharing their thoughts on some ideas, to the German National Academic Foundation (Studienstiftung des deutschen Volkes), St\,John's College Oxford and the Engineering and Physical Sciences Research Council for financial support, and to Mark Wilkinson, Parth Soneji and Christopher Hopper for making the days in the office so enjoyable. Special mention to my wonderful wife Yulia backing me up and bearing with me at all times.
\end{quote}
\SkipTocEntry\chapter*{Abstract}
\thispagestyle{plain}
{The correspondence of stationary, axisymmetric, asymptotically flat space-times and bundles over a reduced twistor space has been established in four dimensions. The main impediment for an application of this correspondence to examples in higher dimensions is the lack of a higher-dimensional equivalent of the Ernst potential. This thesis will propose such a generalized Ernst potential, point out where the rod structure of the space-time can be found in the twistor picture and thereby provide a procedure for generating solutions to the Einstein field equations in higher dimensions from the rod structure, other asymptotic data, and the requirement of a regular axis. Examples in five dimensions are studied and necessary tools are developed, in particular rules for the transition between different adaptations of the patching matrix and rules for the elimination of conical singularities.}

\cleardoublepage
       
\tableofcontents
\thispagestyle{plain}


\pagenumbering{arabic}

\chapter{Introduction} \label{ch:intro}

Although initially proposed in four dimensions, the vacuum Einstein Field Equations can be studied more generally in $n$ dimensions. Higher-dimensional general relativity plays a role for example in string theory \cite{Strominger:1996sh}, the AdS/CFT correspondence, which relates the dynamics of an $n$-dimensional black hole with those of quantum field theory in $n-1$ dimensions \cite{Maldacena:1997re, Aharony:1999ti}, and scenarios involving large extra dimensions and high-energy scattering are discussed in which higher-dimensional black holes might be produced in particle colliders \cite{Kanti:2004nr}. Furthermore, black hole space-times are Ricci-flat Lorentzian manifolds and as such central objects of study in differential geometry. This list of motivations was taken from \cite[Sec.~1]{Emparan:2008aa}.

In four dimensions general relativity has led to various striking results about black holes, for example concerning their horizon topology, their classification (see references in Section~\ref{sec:bhfacts}) or the laws of black hole mechanics \cite{Bardeen:1973aa}. Are these features exclusive to four-dimensional space-times or do some of them carry over to higher dimensions? Answers to these questions would provide valuable insights as well as a better understanding of general relativity and its objects like black holes in a broader context \cite[Sec.~1]{Emparan:2008aa}. 

In four dimensions stationary and asymptotically-flat vacuum black holes can be uniquely classified by their mass and angular momentum, and form a single family, the Kerr solutions, see \cite{Chrusciel:2008aa}, references therein and Theorem~\ref{thm:carter}. However, a generalization of this statement to five dimensions, which would be a classficiation of five-dimensional, stationary, axisymmetric and asymptotically flat black holes by their mass and two angular momenta, does not hold as the space-times found by \citet{Myers:1986aa} and \citet{Emparan:2002aa} show. These examples of black hole space-times in five dimensions have topologically different horizons, so there cannot exist a continuous parameter to link them. Now the task is: Can we determine, or at least characterize, all stationary and asymptotically flat black hole solutions of the higher-dimensional vacuum Einstein field equations \cite[Sec.~8]{Emparan:2008aa}?

Since mass and angular momenta are not enough anymore to classify the solutions, an extra piece of information is needed. This extra piece was proposed to be the so-called rod structure \cite{Emparan:2002dn,Harmark:2004rm,Hollands:2008fp} and \citet{Hollands:2008fp} were able to show that two stationary, axisymmetric and asymptotically flat black hole solutions with connected horizon must be isometric, if their mass, angular momenta and rod structures coincide. Thus the remaining problem is to prove that the only rod structures giving rise to regular black hole solutions are those associated with the known solutions or to find new examples. 

Various strategies have been employed in order to address this question, among which were direct approaches like the ones leading to the black ring, B{\"a}cklund transformations \cite[Sec.~6.6]{Mason:1996hl} or other hidden symmetries \cite{Giusto:2007fx}. Most often applied, however, was the method developed by \citet{Belinskiui:aa}, which uses the fact that the Einstein field equations for a stationary and axisymmetric space-time are integrable, see also \cite[Sec.~5.2.2.2]{Emparan:2008aa}. Via this ansatz they devised a purely algebraic procedure for generating new solutions from known seeds that lead for example to the discovery of a doubly-spinning black ring by \citet{Pomeransky:2006bd} and a candidate for a black hole with a Lens-space horizon \cite{Chen:2008fa}.

In this thesis we are going to establish another way of constructing such solutions from the rod structure and other asymptotic data. The method is based on a twistor construction by which a holomorphic rank-$(n-2)$ vector bundle over a one-dimensional complex manifold is assigned to every stationary, axisymmetric solution of the vacuum Einstein field equations in $n$ dimensions. Under not very restrictive technical assumptions the bundle is fully characterized by only one transition matrix, the so-called patching matrix, associated with each rod, and with a simple transformation from rod to rod. It was shown by \citet{Ward:1983yg} that the correspondence can be better understood and made applicable for practical examples with the help of a B{\"a}cklund transformation. This was studied extensively in \cite{Woodhouse:1988ek,Fletcher:1990aa,Fletcher:1990db}, but an application in more than four dimensions is not feasible without a generalization of the Ernst potential. 

As a first result we will therefore show that a modified version of a matrix already given in \cite{Maison:1979aa} satisfies the desired requirements. 
\begin{thm}
Let $J$ be the matrix of inner products of Killing vectors for a stationary and axisymmetric space-time in $n$ dimensions. Then for any given rod $(a_{i},a_{i+1})$ the matrix
\begin{equation*}
\renewcommand{\arraystretch}{1.5}
J'=\frac{1}{\det \skew{7}{\tilde}{A}} \left(\begin{array}{cc}\hphantom{-}1 & -χ^{\mathrm{t}} \\-χ & \det \skew{7}{\tilde}{A} \cdot \skew{7}{\tilde}{A} + χχ^{\mathrm{t}}\end{array}\right),
\end{equation*}
where $\skew{7}{\tilde}{A}$ is adapted to $(a_{i},a_{i+1})$ and $χ$ is the vector of twist potentials for $(a_{i},a_{i+1})$, is called \textit{higher-dimensional Ernst potential adapted to $(a_{i},a_{i+1})$} and obtained from $J$ by a B\"acklund transformation.
\end{thm}
Following that, we quickly check that important results in four dimensions hold in the same way in higher dimensions. Among these results is the important fact that the patching matrix is the analytic continuation of the Ernst potential $J'(0,z)=P(z)$. 

As the second step we calculate the patching matrices $P$ for the major examples in five dimensions, which essentially requires the computation of various twist potentials on the axis $r=0$. 

More importantly, the twistor correspondence provides a procedure for generating solutions from a given rod structure together with the asymptotic quantities, that is based on the fact that $P$ has simple poles at the nuts of the rod structure and the known fall-off towards infinity. However, this ansatz contains numerous free parameters and the aim is to fix these in terms of the given data by the use of boundary conditions. For this it is inevitable to find out how the patching matrices with adaptations to the different axis segments are related. 

First this is done for the outer sections of the axis.
\begin{thm}
Assume that we are given a rod structure with nuts at $\{a_{i} | a_{i}∈ℝ\}_{1≤i≤N}$. If $P_{+}$ is the patching matrix adapted to $(a_{\scriptscriptstyle N}, ∞)$, then $P_{-}^{\vphantom{1}}=MP_{+}^{-1}M$ with $M=\left(\begin{smallmatrix}0&0&1\\0&1&0\\1&0&0\end{smallmatrix}\right)$ is the patching matrix adapted to $(-∞,a_{1})$.
\end{thm} 
With the help of this connection we are able to reconstruct the known space-times flat space, Myers-Perry and the black ring from the rod structures with up to three nuts. However, we are not able to fix the conicality of the black ring, which needs further investigation of the behaviour of $J$, $P$ and the conformal factor (appearing in the $σ$-model form of the metric) in a neighborhood of a nut. A definition of the $(u,v)$-coordinates can be found in Section~\ref{sec:Jaroundnut}.
\begin{thm}
For a space-time regular on the axis the generic form of $J$ in $(u,v)$-coordinates around a nut at $u=v=0$, where two spacelike rods meet, is
\begin{equation} \label{eq:Jaroundnutintro}
\renewcommand{\arraystretch}{1.4}
J= \left(\begin{array}{ccr}
X_{0} & u^{2}Y_{0} & v^{2}Z_{0} \\
\cdot  & u^{2}U_{0} & u^{2}v^{2} V_{0} \\
\cdot  & \cdot  & v^{2}W_{0}\end{array}\right),
\end{equation}
and, furthermore, one needs
\begin{itemize}
\item $\dfrac{U_{0}^{\hphantom{1}}}{v^{2}_{\hphantom{1}}\erm^{2\nu}}=1$ as a function of $v$ on $u=0$,
\item $\dfrac{W_{0}^{\hphantom{1}}}{u^{2}_{\hphantom{1}}\erm^{2\nu}}=1$ as a function of $u$ on $v=0$.
\end{itemize}

If one of the rods is the horizon instead of a spacelike rod corresponding statements hold.
\end{thm}
Building up on this, we obtain the following results about the conformal factor.
\begin{prop}
On a segment of the axis where $u=0$ we have $\frac{U_0}{v^2\erm^{2\nu}}=\text{constant}$, and similarly where $v=0$. Therefore, the factor $\left(u^{2}+v^{2}\right)\erm^{2\nu}$ is continuous at the nut $u=v=0$ and with the conventions leading to \eqref{eq:Jaroundnutintro}, the absence of conical singularities requires
\begin{equation*}
\lim_{v\rightarrow 0}U_{0}=\lim_{u\rightarrow 0}W_0.
\end{equation*}
\end{prop}
Last in the set of tools we show how $P$ switches when going past a nut.
\begin{thm} 
Let at $z=α$ be a nut where two spacelike rods meet and assume that we have chosen a gauge where the twist potentials vanish when approaching the nut. Then 
\begin{equation*}
\renewcommand{\arraystretch}{1.5}
P_{-}^{\vphantom{\frac{1}{2}}}=\left(\begin{array}{ccc}0 & 0 & \dfrac{1}{2(z-{\alpha})} \\0 & 1 & 0 \\ 2(z-{\alpha}) & 0 & 0\end{array}\right)P_{+}^{\vphantom{\frac{1}{2}}}\left(\begin{array}{ccc}0 & 0 & 2(z-{\alpha}) \\0 & 1 & 0 \\\dfrac{1}{2(z-{\alpha})} & 0 & 0\end{array}\right),
\end{equation*}
where $P_{+}$ is adapted to $u=0$ and $P_{-}$ is adapted to $v=0$.
\end{thm}
By this we can impose more boundary conditions on the free parameters in $P$ that are possibly left and we have enough at hand to solve the conicality problem as exemplified for the black ring in Section~\ref{sec:applbr}.

This thesis is structured as follows. In Chapters~\ref{ch:backgrd}, \ref{ch:twistorspace}, \ref{ch:pwtrf}, \ref{ch:yang}, \ref{ch:bundles} we give a detailed description of the background material in order to make the thesis more self-contained, to familiarize the reader with the established construction and to identify the points which impede a generalization to higher dimensions. The material for this purpose is mainly taken from \cite{Mason:1996hl,Fletcher:1990db}. In Chapter~\ref{ch:bh} we revisit some relevant facts about black holes. The aforementioned generalization of the Ernst potential follows in Chapter~\ref{ch:fivedim}, succeeded by the patching matrices for the examples in Chapter~\ref{ch:Pexamples}. Chapter~\ref{ch:converse} contains the reconstruction of $P$ from the rod structure, results about the relation of different adaptations of $P$ across nuts and the removal of conical singularities.
\chapter{Mathematical Background} \label{ch:backgrd}

This chapter is mainly based on \citet[Ch.~2]{Mason:1996hl} and gives a brief summary of some mathematical tools which are needed later on. It is far from exhaustive, and some parts are only aimed to make the reader familiar with the notation used. In particular, we assume basic knowledge about principal bundles and Yang-Mills theory as can be found for example in \cite{Kobayashi:1996aa,Kobayashi:1996ab}. 

The starting point for twistor theory is complex Minkowski space $\mathbb{C}\Mbb$. For our purposes \textit{double null coordinates} are often convenient, that is coordinates on $\mathbb{C}\Mbb$ in which the metric takes the form
\begin{equation*}
\drm s^{2}=2(\drm z \, \drm \tilde z - \drm w \, \drm \tilde w), 
\end{equation*}
and the volume element is
\begin{equation*}
{\nu}= \drm w \wedge  \drm \tilde w \wedge  \drm z \wedge  \drm \tilde z. 
\end{equation*}
Then the vector fields $\partial _{w}$, $\partial _{\tilde w}$, $\partial _{z}$, $\partial _{\tilde z}$ form a \textit{null tetrad} at each point of $\mathbb{C}\Mbb$. In general, a basis of 4-vectors $\{W,\tilde W, Z, \tilde Z\}$ is called a null tetrad if 
\begin{equation*}
{\eta}(Z,\tilde Z)=-{\eta}(W,\tilde W)=1, \quad 24\, {\nu}(W, \tilde W, Z, \tilde Z)=1
\end{equation*}
where ${\eta}$ is the metric tensor on $\mathbb{C}\Mbb$, and all other inner products vanish.

Within this setting we can recover various real spaces (``real slices'') by imposing reality conditions on $w$, $\tilde w$, $z$, $\tilde z$. 
\begin{itemize}
\item Euclidean real slice $\Ebb$: We identify real Cartesian coordinates $x^{0}$, $x^{1}$, $x^{2}$, $x^{3}$ with $w$, $\tilde w$, $z$, $\tilde z$ via
\begin{equation*}
\left(\begin{array}{cc}
\tilde z & w \\ \tilde w & z
\end{array}\right)
=
\frac{1}{\sqrt{2}}\left(\begin{array}{cc}
x^{0}+\irm x^{1} & -x^{2}+\irm x^{3} \\x^{2}+\irm x^{3} & \hphantom{-}x^{0}-\irm x^{1}
\end{array}\right). 
\end{equation*}
That is, we get $\Ebb$ by imposing the reality conditions $\bar w = - \tilde w$ and $\bar z = \tilde z$.
\item Minkowski real slice $\Mbb$: The real coordinates $x^{0}$, $x^{1}$, $x^{2}$, $x^{3}$ on $\Mbb$ are identified with $w$, $\tilde w$, $z$, $\tilde z$ via
\begin{equation*}
\left(\begin{array}{cc}
\tilde z & w \\ \tilde w & z
\end{array}\right)
=
\frac{1}{\sqrt{2}}\left(\begin{array}{cc}
x^{0}+x^{1} & x^{2}-\irm x^{3} \\x^{2}+\irm x^{3} & x^{0}-x^{1}
\end{array}\right), 
\end{equation*}
that is we pick out the real space by the condition that $z$ and $\tilde z$ should be real, and that $\bar w= \tilde w$. 
\end{itemize}

For the definition of self-duality and anti-self-duality we need the Hodge star operation. To clarify notation and conventions first a short reminder. Let $(M,g)$ be an $n$-dimensional Riemannian or pseudo-Riemannian manifold. Given a $p$-form ${\beta}$ on $M$ with (skew-symmetric) components ${\beta}_{ab\dotsc c}$ its exterior derivative $\drm {\beta}$ has components $\partial _{[a}{\beta}_{bc\dotsc d]}$ where the square brackets stand for antisymmetrization.\footnote{For a $\binom{0}{p}$ tensor
\begin{equation*}
T_{[ab\dotsc c]}=\frac{1}{p!} \sum _{{\sigma}\in S_{p}} \sgn({\sigma}) T_{{\sigma}(a){\sigma}(b)\dotsc {\sigma}(c)}
\end{equation*}
where $S_{p}$ is the group of permutations of $p$ elements. The symmetrization $T_{(ab\dotsc c)}$ is defined in the same way but without the signum.
} The exterior or wedge product ${\beta}\wedge {\gamma}$ with a $q$-form ${\gamma}$ has components
\begin{equation*}
{\beta}_{[ab\dotsc c}{\gamma}_{de\dotsc f]}.
\end{equation*}
Now let ${\varepsilon}$ be the $n$-dimensional alternating symbol, that is ${\varepsilon}_{[a\dotsc b]}={\varepsilon}_{a\dotsc b}$ and ${\varepsilon}_{0\dotsc n}=1$, and ${\Delta}=\sqrt{|\det (g_{ab})|}$. We then define the (\textit{Hodge}) \textit{dual} of ${\beta}$ to be the $(n-p)$-form $*{\beta}$ with components
\begin{equation} \label{eq:defhodge}
*{\beta}^{\vphantom{f}}_{ab\dotsc c}=\frac{1}{(n-p)!} {\Delta} \tensor{{\varepsilon}}{_{ab\dotsc c}^{de\dotsc f}}{\beta}_{de\dotsc f}^{\vphantom{-1}},
\end{equation}
where indices are raised and lowered with the metric $g_{ab}$ or its inverse $g^{ab}$. To see that the definition is actually independent of the basis consider $p$ vector fields $X_{1},\dotsc ,X_{p}$ with their covariant images ${\theta}_{1},\dotsc ,{\theta}_{p}$ (according to the metric) and fix a volume form $\drm\, \mathrm{vol}$. Then $*{\beta}$ is defined as the unique $(n-p)$-form satisfying
\begin{equation*}
{\beta}(X_{1},\dotsc ,X_{p}) \,\drm \mathrm{vol}= (*{\beta})\wedge {\theta}_{1}\wedge \dotsc \wedge {\theta}_{p}.
\end{equation*}
This is a basis independent definition of $*{\beta}$ and coincides with \eqref{eq:defhodge}, since in coordinates
\begin{equation*}
\drm\mathrm{vol} = {\Delta} {\varepsilon}_{a\dotsc d}.
\end{equation*}
Thus one can think of $*{\beta}$ as the complement of ${\beta}$ with respect to the volume form (and the appropriate prefactor). For example in terms of an oriented orthonormal basis ${e_{1}, \dotsc , e_{n}}$ (of a vector space) the Hodge star operation is defined completely by
\begin{equation*}
*(e_{i_{1}}\wedge \dotsc \wedge e_{i_{k}})=e_{i_{k+1}}\wedge \dotsc \wedge e_{i_{n}},
\end{equation*}
where $\{i_{1},\dotsc ,i_{k},i_{k+1},\dotsc ,i_{n}\}$ is an even permutation of $\{1,\dotsc ,n\}$. Of particular interest for us are 2-forms on $(\mathbb{C}\Mbb, {\eta})$ where we have
\begin{equation*}
*{\beta}^{\vphantom{f}}_{ab}=\frac{1}{2} {\Delta} {\varepsilon}^{\vphantom{-1}}_{abcd}{\eta}^{ce}{\eta}^{df}{\beta}_{ef}^{\vphantom{-1}}. 
\end{equation*}
Using the properties of the alternating symbol, it follows that the Hodge star operation is idempotent, that is $*^{2}=1$, thus has eigenvalues $\pm 1$.\footnote{In general, on a Riemannian manifold $(M,g)$ for a $p$-form ${\beta}$ it is $*\negthinspace*\negthinspace{\beta}=(-1)^{p(n-p)}s{\beta}$, where $s$ is the signature of $g$. For complex manifolds the concept of signature is void, thus we can set $s=1$ and then $*^{2}=\id$ for 2-forms on a four-dimensional manifold.} 
\begin{Def}
A 2-form is called \textit{self-dual} (SD) if $*{\beta}={\beta}$, and \textit{anti-self-dual} (ASD) if $*{\beta}=-{\beta}$. 
\end{Def}
The space of 2-forms then decomposes into the direct sum of eigenspaces, because in double null coordinates we have 
\begin{equation} \label{eq:sdforms}
{\alpha}=\drm w \wedge  \drm z, \quad \tilde {\alpha} = \drm \tilde w \wedge  \drm \tilde z, \quad {\omega}=\drm w \wedge  \drm \tilde w - \drm z \wedge  \drm \tilde z
\end{equation}
as a basis for SD 2-forms, and 
\begin{equation*}
\drm w \wedge  \drm \tilde z, \quad \drm \tilde w \wedge  \drm z, \quad \drm w \wedge  \drm \tilde w + \drm z \wedge  \drm \tilde z
\end{equation*}
as a basis for ASD 2-forms. \\

Let $E$ be a rank-$n$ vector bundle with a \textit{connection} $\Drm$, that is a differential operator that maps sections $s$ of $E$ to 1-forms with values in $E$ and in a local trivialization it is given by
\begin{equation*}
\Drm s=\drm s+{\Phi}s 
\end{equation*}
where ${\Phi}$ is matrix-valued 1-form called \textit{gauge potential}, \textit{potential} or sometimes also \textit{connection}. In a general gauge theory, that is a principal bundle $P(M,G)$ over a (pseudo-) Riemannian manifold $M$ with with gauge group $G\subseteq \GL (n)$, the \textit{curvature} of $\Drm$ is the matrix-valued 2-form $F=F_{ab}\, \drm x^{a}\wedge \drm x^{b}$ with
\begin{equation*}
F_{ab}=\partial _{a}{\Phi}_{b}-\partial _{b}{\Phi}_{a}+[{\Phi}_{a},{\Phi}_{b}]. 
\end{equation*}
The Yang-Mills equations are 
\begin{equation*}
\Drm F=0, \quad \Drm *F=0
\end{equation*}
where the first one is called the \textit{Bianchi identity} and the second one is the Euler-Lagrange equation of the Lagrangian density $\frac{1}{2}\tr(F\wedge *F)=\frac{1}{4}\tr(F_{ab}F^{ab})$. We then see that if $F$ is SD or ASD then $\Drm *F=\pm \Drm F=0$, that is the second set of Yang-Mills equations follow from the Bianchi identity. \\

Important objects for twistor theory is null 2-planes. 
\begin{Def}
Let ${\Pi}$ be an affine 2-plane in $\mathbb{C}\Mbb$ and ${\Pi}^{\perp}$ the normal bundle of ${\Pi}$. The 2-plane ${\Pi}$ is called (\textit{partially}) \textit{null} if $\{0\}\subsetneq {\mathrm{T}\Pi^{\vphantom{1}}_{p}}\cap {\Pi}_{p}^{\perp} \subsetneq {\mathrm{T}\Pi^{\vphantom{1}}_{p}}$ for all $p∈\Pi$, and it is called (\textit{totally}) \textit{null} if ${\mathrm{T}\Pi}\cap {\Pi}^{\perp} = {\mathrm{T}\Pi}$.
\end{Def}
Hence, totally null means that ${\eta}(A,B)=0$ for all tangent vectors $A$, $B$ of ${\Pi}$. From now on when we speak about null 2-planes we will always mean totally null 2-planes (unless mentioned differently). With each null 2-plane ${\Pi}$ we associate a \textit{tangent bivector} ${\pi}=A\wedge B$ where $A$, $B$ are independent tangent vectors of ${\Pi}$. The tangent bivector has components ${\pi}^{ab}=A^{[a}B^{b]}$. 

\begin{lem} \label{lem:null2plane}
If ${\Pi}$ is a null 2-plane, then ${\pi}_{ab}{\pi}^{ab}=0$, and ${\pi}_{ab}\, \drm x^{a}\wedge \drm x^{b}$ is either SD or ASD.
\end{lem}
\begin{proof}
First of all we observe that ${\pi}$ is determined up to scale by the tangent space of ${\Pi}$, because another choice of independent tangent vectors can be written as a linear combination of $A$ and $B$ and thus gives the same ${\pi}$ up to a nonzero scalar factor. Conversely, the tangent space of ${\Pi}$ is given by all $P^{b}$ with ${\pi}_{ab}P^{b}=0$, since $A$, $B$ are null and orthogonal. 
But ${\pi}$ can also be characterized up to a nonzero scalar factor by the condition that $*{\pi}_{ab}P^{b}=0$ for all $P^{b}$ tangent vectors of ${\Pi}$. This follows from ${\pi}_{ab}=A_{[a}B_{b]}$, and $*{\pi}_{ab}={\varepsilon}_{abcd}^{\vphantom{1}}A^{c}B^{d}$ which implies that $*{\pi}_{ab}P^{b}=0$ if and only if $P$ is a linear combination of $A$ and $B$. Hence, ${\pi}={\mu}*{\pi}$ for some ${\mu}\neq 0$. However, the eigenvalues of $*$ are $\pm 1$, therefore ${\pi}_{ab}^{\vphantom{1}}=*{\pi}_{ab}$ or ${\pi}_{ab}^{\vphantom{1}}=-*{\pi}_{ab}$, that is ${\pi}$ SD or ASD. Again from ${\pi}_{ab}=A_{[a}B_{b]}$ and the fact that $A$, and $B$ are null and orthogonal it is obvious that ${\pi}_{ab}{\pi}^{ab}=0$. 
\end{proof}

\begin{Def}
An affine null 2-plane ${\Pi}$ is called an \textit{${\alpha}$-plane} if ${\pi}$ is SD and a \textit{${\beta}$-plane} if ${\pi}$ is ASD. 
\end{Def}
In double null coordinates the surfaces of constant $w$, $z$ and $\tilde w$, $\tilde z$, respectively, have tangent bivectors $\drm w \wedge  \drm z$ and $\drm \tilde w \wedge  \drm \tilde z$, respectively, and are therefore ${\alpha}$-planes. 

If ${\pi}$ is the tangent bivector of an ${\alpha}$-plane through the origin, then it must be a linear combination of the 2-forms in \eqref{eq:sdforms}. It is only determined up to a scalar factor, that is we can set the coefficient of $\drm w\wedge \drm z$ to 1 without loss of generality. Then we have
\begin{equation*}
{\pi}=\drm w\wedge \drm z-{\zeta}(\drm w\wedge \drm \tilde w-\drm z\wedge \drm \tilde z)+{\mu}\, \drm \tilde w \wedge  \drm \tilde z. 
\end{equation*}
So, the nonzero coefficients are ${\pi}_{wz}=1$, ${\pi}_{w\tilde w}=-{\pi}_{z\tilde z}=-{\zeta}$ and ${\pi}_{\tilde w \tilde z}={\mu}$. Then the requirement from Lemma~\ref{lem:null2plane} turns into
\begin{equation*}
0={\pi}_{ab}{\pi}^{ab}={\pi}_{ab}{\pi}_{cd}{\eta}^{ac}{\eta}^{bd}=-{\mu}-{\zeta}^{2}-{\zeta}^{2}-{\mu} \quad \Leftrightarrow   \quad {\mu}=-{\zeta}^{2}.
\end{equation*}
This implies
\begin{equation*}
{\pi}^{ab}=L^{[a}M^{b]}
\end{equation*}
where
\begin{equation*}
L=\partial _{w}-{\zeta}\partial _{\tilde z}, \quad M=\partial _{z}-{\zeta}\partial _{\tilde w}
\end{equation*}
for some ${\zeta}\in \mathbb{C}$. The case where the coefficient of $\drm w\wedge \drm z$ vanishes yields $ζ=0$, so that ${\pi}=\drm \tilde w \wedge  \drm \tilde z$ (up to a constant) and $L=\partial _{w}$, $M=\partial _{z}$. There is no point in the argument where we needed that our null tetrad is induced by coordinates, that is the above statement is also true if we set $L=W-{\zeta}\tilde Z$ and $M=Z-{\zeta}\tilde W$ for some null tetrad $W$, $\tilde W$, $Z$, $\tilde Z$ not necessarily induced by coordinates. Conversely, $L$ and $M$ span an ${\alpha}$-plane through the origin for every ${\zeta}\in \mathbb{C}$. Including the point ${\zeta}=\infty $ by mapping it to the ${\alpha}$-plane spanned by $\partial _{\tilde w}$ and $\partial _{\tilde z}$ then yields a one-to-one correspondence between ${\alpha}$-planes through the origin and points of the Riemann sphere, ${\Pi}_{{\zeta}}\leftrightarrow {\zeta}$. This correspondence will be important for the characterization of twistor spaces. \\

As a last point in this chapter a reminder about equivalence and reconstruction of holomorphic vector bundles (these statements can be found in standard textbooks such as \cite{Grauert:1994aa, Fritzsche:2002aa}). Let $π:E→X$ be a rank-$r$ holomorphic vector bundle with an open cover $(U_{i})_{i∈I}$ of $X$ and trivialization functions (biholomorphic maps) $h_{i}: π^{-1}(U_{i})→U_{i}×\Cbb^{r}$.

The (holomorphic) transition functions $g_{ij}: U_{i}∩U_{j}→\GL(r,\Cbb)$ for all $i,j∈I$ are obtained from the biholomorphic maps
\begin{equation*}
\skew{3}{\tilde}{g}_{ij}^{\vphantom{1}} \coloneqq h_{i}^{\vphantom{1}} \circ h_{j}^{-1} : U_{i}^{\vphantom{1}}∩U_{j}^{\vphantom{1}} × \Cbb^{r} → U_{i}^{\vphantom{1}}∩U_{j}^{\vphantom{1}} × \Cbb^{r}, 
\end{equation*}
which are of the form $\skew{3}{\tilde}{g}_{ij}(x,v)=(x, g_{ij}(x)v)$. By construction the transition functions obviously satisfy the relations
\begin{equation}\label{eq:transfct}
g_{ij}g_{jk}=g_{ik}, \quad g_{ii} = \id.
\end{equation}

Two holomorphic vector bundles $E→X$ and $F→X$ are called \textit{isomorphic} if there exists a bijective map $f:E→F$ such that $f$ and its inverse are vector bundle homorphisms, that is $f$ is a fibre-preserving holomorphic map $π_{E}=π_{F}\circ f$ such that for any $x∈X$ a linear map $f_{x}: E_{x}→F_{x}$ is induced and analogously for the $f^{-1}$.

Given two systems of transition functions, $(\hat{g}_{ij})_{i,j∈I}$ and $(\check{g}_{i'j'})_{i',j'∈I}$ with covers $(U_{i})_{i∈I}$ and $(V_{i'})_{i'∈I'}$, then they are called \textit{equivalent} if for a common refinement of the open covers $(W_{k})_{k∈K}$ there are holomorphic maps $f_{k}: W_{k}→\GL(r,\Cbb)$ with
\begin{equation*}
f_{k}\hat{g}_{kl} = \check{g}_{kl} f_{l}.
\end{equation*}
The transition functions with respect to the refined cover can be taken as the restrictions of the initial transition functions and a common refinement is for example obtained by taking all possible intersections of $U_{i}$ and $V_{i'}$.
\begin{prop} \label{prop:bundle}
For holomorphic  vector bundles the following statements hold. \hphantom{a long word}
\begin{enumerate}
\item Let $\{U_{i}\}_{i∈I}$ be an open covering of $X$, and let $g_{ij}∈ \GL_{r}(\Ocal_{X}(U_{i}∩U_{j}))$ satisfy the cocycle relation \eqref{eq:transfct}. Then there exists as holomorphic vector bundle $E$ of rank $r$ with these transition functions. \cite{Grauert:1994aa}
\item Isomorphic bundles can be represented by equivalent systems of transition functions. \cite{Fritzsche:2002aa}
\end{enumerate}
\end{prop}
\chapter{Twistor Space} \label{ch:twistorspace}

In this chapter the twistor space for complexified Minkowski space is introduced using \citet[Sec.~9.2]{Mason:1996hl}.

\section{Definition of Twistor Space}

In the previous chapter we have seen that ${\alpha}$-planes through the origin in $\mathbb{C}\Mbb$ are spanned by $L=\partial _{w}-{\zeta}\partial _{\tilde z}$ and $M=\partial _{z}-{\zeta}\partial _{\tilde w}$ or by $\partial _{\tilde w}$ and $\partial _{\tilde z}$
in the limiting case $({\zeta}=\infty )$. As before, the coordinates $w$, $z$, $\tilde w$, $\tilde z$ will always be double-null coordinates. Thus, a general ${\alpha}$-plane, not necessarily passing through the origin, is labelled by three complex coordinates: the parameter ${\zeta}$ which determines the tangent space, together with the parameters 
\begin{equation} \label{eq:aplanepar}
{\lambda}={\zeta}w+\tilde z \quad \text{and} \quad  {\mu}={\zeta}z+\tilde w
\end{equation}
that are constant over the ${\alpha}$-plane.

\begin{Def}
The \textit{twistor space} of $\mathbb{C}\Mbb$ is the three-dimensional complex manifold consisting of all (affine) ${\alpha}$-planes in $\mathbb{C}\Mbb$. 
\end{Def}
A way to determine the global geometry is by writing the equations of an ${\alpha}$-plane in homogeneous form
\begin{equation} \label{eq:aplanehom}
\tilde z Z^{2}+wZ^{3}=Z^{0}, \quad \tilde w Z^{2}+zZ^{3}=Z^{1}
\end{equation}
with complex constants $Z^{{\alpha}}$, ${\alpha}=0,1,2,3$. The order of these constants is a convention in twistor theory. If $Z^{2}\neq 0$, then \eqref{eq:aplanehom} is equivalent to \eqref{eq:aplanepar} for
\begin{equation*}
{\lambda}=\frac{Z^{0}}{Z^{2}}, \ {\mu}=\frac{Z^{1}}{Z^{2}}, \ {\zeta}=\frac{Z^{3}}{Z^{2}}. 
\end{equation*}
In the case $Z^{2}=0$, $Z^{3}\neq 0$, the parameter ${\zeta}$ must be infinite and the tangent space is spanned by $\partial _{\tilde w}$ and $\partial _{\tilde z}$. Consequently, we can identify the twistor space of $\mathbb{C}\Mbb$ with an open subset of $\mathbb{C}\Pbb^{3}$ if we include this ${\alpha}$-plane of constant $\tilde w$, $\tilde z$ and regard the $Z^{{\alpha}}$s as homogeneous coordinates. 

The excluded points of $\mathbb{C}\Pbb^{3}$ are
\begin{equation*}
I=\Big\{\left[Z^{0}:Z^{1}:Z^{2}:Z^{3}\right]\in \mathbb{C}\Pbb^{3}\,:\, Z^{2}=Z^{3}=0\Big\}.
\end{equation*}
There are 2 homogeneous coordinates left which parameterize $I$, hence it is a $\mathbb{C}\Pbb^{1}$ and the twistor space of $\mathbb{C}\Mbb$ is, as a complex manifold, $\mathbb{C}\Pbb^{3}-\mathbb{C}\Pbb^{1}$. We can cover the twistor space of $\mathbb{C}\Mbb$ by two coordinate patches $V$ and $\tilde V$ with $V$ being the complement of the plane $Z^{2}=0$ (that is the plane ${\zeta}=\infty $) and $\tilde V$ being the complement of the plane $Z^{3}=0$ (that is the plane ${\zeta}=0$). The parameters ${\lambda}$, ${\mu}$, ${\zeta}$ are coordinates on $V$, and on $\tilde V$ we can use coordinates $\tilde {\lambda}$, $\tilde {\mu}$, $\tilde {\zeta}$ with
\begin{equation*}
\tilde {\lambda}=\frac{Z^{0}}{Z^{3}}, \ \tilde {\mu}=\frac{Z^{1}}{Z^{3}}, \ \tilde {\zeta}=\frac{Z^{2}}{Z^{3}},
\end{equation*}
which on the overlap $V\cap \tilde V$ gives the relations
\begin{equation*}
\tilde {\lambda}=\frac{{\lambda}}{{\zeta}}, \ \tilde {\mu}=\frac{{\mu}}{{\zeta}}, \ \tilde {\zeta}=\frac{1}{{\zeta}}.
\end{equation*}
For the twistor space we denote by $\Tbb$ the copy of $\mathbb{C}^{4}$ on which $(Z^{0},Z^{1},Z^{2},Z^{3})$ are linear coordinates, and by $\Pbb\Tbb$ the corresponding projective space $\mathbb{C}\Pbb^{3}$.

Now let $U\subset \mathbb{C}\Mbb$ and assume that its intersection with each ${\alpha}$-plane is connected (but possibly empty).\footnote{The connectivity assumption is not necessary for the considerations in this chapter, but will later make the Penrose-Ward transform work in a natural way.} 
\begin{Def}
The \textit{twistor space} of $U$ is the subset
\begin{equation*}
\Pcal = \{Z \in  \Pbb\Tbb\, : \, Z\cap U \neq  \varnothing\}
\end{equation*}
of $\Pbb\Tbb$. 
\end{Def}
If $U$ is open, then $\Pcal$ is open, and if $U=\mathbb{C}\Mbb$, then $\Pcal$ is the complement of $I$. 

\begin{rem}\mbox{}
\begin{enumerate}
\item For many aspects of twistor theory it turns out that spinor calculus is a suitable tool. However, in our context we do not gain anything by using spinors and thus they are not going to be introduced in this work.\\ \mbox{} \hfill $\blacksquare$
\item The \textit{Klein correspondence} shows that compactified $\mathbb{C}\Mbb$ together with the complex conformal group can be identified with a $\Cbb\Pbb^{3}⊂\Cbb\Pbb^{5}$, the \textit{Klein quadric}, together with the so-called projective general linear group $\GL(4,\mathbb{C})/\mathbb{C}^{\times }$. \cite[Sec.~2.4 and 9.2]{Mason:1996hl}

\end{enumerate}
\end{rem}

\section{Lines in $\Pbb\Tbb$} 

The equations~\eqref{eq:aplanehom} allow further conclusions: if we hold $w$, $z$, $\tilde w$, $\tilde z$ fixed, and vary $Z^{{\alpha}}$, then the equations determine a two-dimensional subspace of $\Tbb$, that is a projective line in $\Pbb\Tbb$. This is the Riemann sphere of ${\alpha}$-planes through the point in $\mathbb{C}\Mbb$ with coordinates $w$, $z$, $\tilde w$, $\tilde z$ which we denote by $\hat x$. It corresponds to the twistor space of $\{x\}\subset \mathbb{C}\Mbb$. 

Two points $x,y\in \mathbb{C}\Mbb$ are null separated if and only if they lie on an ${\alpha}$-plane, hence if and only if $\hat x \cap  \hat y \neq  \varnothing$. In other words, two lines in the twistor space intersect if and only if the corresponding points are separated by a null vector. A conformal metric (that is a class of conformally equivalent metrics) is completely determined by saying when two vectors are null separated (Appendix~\ref{app:confmetric}). Hence, the conformal geometry of $\mathbb{C}\Mbb$ is encoded in the linear geometry of $\Pbb\Tbb$. 

\section{The Correspondence Space}

Let $U$ be a subset of $\mathbb{C}\Mbb$. The \textit{correspondence space} $\Fcal$ is the set of pairs $(x,Z)$ with $x\in U$ and $Z$ an ${\alpha}$-plane through $x$. It is fibred over $U$ and $\Pcal$ by projections
\begin{equation*}
\begin{xy}
  \xymatrix@C=0.7cm@R=0.7cm{
      & \Fcal \ar[ld]_q \ar[rd]^p		&  \\
      U	 	&	&	\Pcal     
  }
\end{xy}
\end{equation*}
which map $(x,Z)$ to $x$ and $Z$, respectively, see Figure~\ref{fig:corrspace} (the lines through $x$ respectively $y$ in $\Fcal$ visualize the projective lines of ${\alpha}$-planes through $x$ respectively $y$). 
\begin{figure}[htbp]
\begin{center}
     \scalebox{0.4}{\input{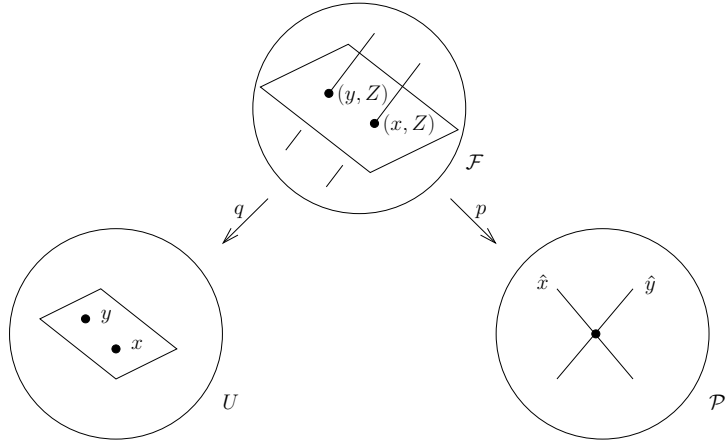}}
     \caption{The correspondence space between $U$ and its twistor space $\Pcal$. \cite{Mason:1996hl}} 
     \label{fig:corrspace}
\end{center}
\end{figure}
The projection maps $p$ and $q$ are surjective. Points in $\Fcal$ are labelled by $(w,z,\tilde w,\tilde z, {\zeta})$ (including ${\zeta}=\infty $), and the coordinate expressions for the projections are
\begin{align*}
p: (w,z,\tilde w,\tilde z, {\zeta}) & \mapsto  ({\lambda},{\mu},{\zeta})=({\zeta}w+\tilde z, {\zeta}z+\tilde w, {\zeta}),\\
q: (w,z,\tilde w,\tilde z, {\zeta}) & \mapsto  (w,z,\tilde w,\tilde z).
\end{align*}
The tangent spaces to the leaves of the fibration of $p$ are spanned at each point by vector fields $l=\partial _{w}-{\zeta}\partial _{\tilde z}$ and $m=\partial _{z}-{\zeta}\partial _{\tilde w}$ on $\Fcal$. 

A function on $\Pcal$ is a function of the twistor coordinates $({\lambda},{\mu},{\zeta})$, and pulling it back by $p:\Fcal\to \Pcal$ therefore yields a function on the correspondence space $\Fcal$ that is constant along $l$ and $m$. 

\section{Reality Structures}

In the previous chapter we saw that the real slices are characterized as fixed point sets of an antiholomorphic involution ${\sigma}:\mathbb{C}\Mbb\to \mathbb{C}\Mbb$. Using double-null coordinates ${\sigma}$ was defined for the two cases as
\begin{enumerate}
\item [$(\Ebb)$] ${\sigma}(w,z,\tilde w,\tilde z)=(- \overline{\tilde w},\overline{\tilde z},- \overline w,\overline z)$,
\item [$(\Mbb)$] ${\sigma}(w,z,\tilde w,\tilde z)=(\overline{\tilde w},\overline z, \overline w,\overline{\tilde z})$.
\end{enumerate}
For $\Mbb$ this picks out a real hypersurface $\Pbb\mathbb{N}\subset \Pbb\Tbb$ by $Z\cap {\sigma}(Z)\neq \varnothing$. The ${\alpha}$-plane $Z$ has complex dimension 2, and by imposing $Z\cap {\sigma}(Z)\neq \varnothing$ in the $\Mbb$-case we add three real conditions which leaves one real degree of freedom. Hence, if $Z\in \Pbb\mathbb{N}-I$, then $Z\cap {\sigma}(Z)$ is a real line that is null (${\alpha}$-planes are totally null), that is a real null geodesic. In turn, given a real null geodesic we can pick two points $x$, $y$ on it. They are null separated, and determine an ${\alpha}$-plane $Z$ as described above. Furthermore, they are real points, thus $Z\cap {\sigma}(Z)\neq \varnothing$ which means $Z\in \Pbb\mathbb{N}-I$. Therefore, we have shown that $\Pbb\mathbb{N}-I$ is the space of real null geodesics. Together with what we have seen in the last chapter this implies that the set of real null geodesics through a real point $x\in \Mbb$ is a Riemann sphere. Figure~\ref{fig:realcorr} depicts the correspondence in this real case. 
\begin{figure}[htbp]
\begin{center}
     \scalebox{0.4}{\input{twistorcorrespondence2.tex}}
     \caption{The correspondence between $\Mbb$ and $\Pbb\mathbb{N}-I$. \cite{Penrose:2007aa}} 
     \label{fig:realcorr}
\end{center}
\end{figure}
It also indicates how a point in $\Pbb\mathbb{N}-I$ corresponds to a locus in $\Mbb$ and a point $x\in \Mbb$ represented by its light cone corresponds to a locus (Riemann sphere) in $\Pbb\mathbb{N}-I$. Thus, the relationship between $U$ and its twistor space $\Pcal$ is a non-local correspondence in general. 

\chapter{The Penrose-Ward Transform} \label{ch:pwtrf}

This chapter is about the Penrose-Ward transform which associates a solution to the anti-self-dual Yang-Mills (ASDYM) equation on a domain $U$ in $\mathbb{C}\Mbb$ to a holomorphic vector bundle on the twistor space $\Pcal$ of $U$, using \citet[Sec.~10.1 and 10.2]{Mason:1996hl}. 

That raises the question: What is the special importance of the ASDYM equation for our considerations? To answer this, we have to say a few words about integrable systems. In classical mechanics the notion of integrability is a precise concept. However, for systems obeying partial differential equations in which there are infinitely many degrees of freedom, we do not have a clear-cut characterization of integrability. A lot of theories were developed and it is easy to give examples of `integrability', yet there is no single effective characterization that covers all cases \cite[Ch.~1]{Mason:1996hl}. 

But whatever the definition of integrability is, a reduction of an integrable system, which is obtained by imposing a symmetry or specifying certain parameters, always yields another integrable system. Thus we have a partial ordering, in the sense that we say system $A$ is less than system $B$ if $A$ is a reduction of $B$. This ordering motivates the search for a `maximal element', that is an integrable system from which all others can be derived. Such a system has not yet been found, but it turns out that almost all known systems in dimension 1 and 2, and a lot of important systems in dimension 3 arise as reductions of the ASDYM equation \cite[Ch.~4.1]{Mason:1996hl}. 

The example of special significance for us will be the Ernst equation for stationary axisymmetric gravitational fields. 

\section{Lax Pairs and Fundamental Solutions}

Let $D$ be a connection on a complex rank-$n$ vector bundle $E$ over some region $U$ in $\mathbb{C}\Mbb$, and $F$ its curvature 2-form. If $D=\drm + {\Phi}$, then $F=F_{ab} \,\drm x^{a}\wedge \drm x^{b}$, where in a coordinate induced local trivialization
\begin{equation*}
F_{ab}=\partial _{a}{\Phi}_{b}-\partial _{b}{\Phi}_{a}+[{\Phi}_{a},{\Phi}_{b}]. 
\end{equation*}
The ASD conditions then take in double-null coordinates the form
\begin{equation} \label{eq:asdcond1}
\begin{split}
F_{zw} & = \partial _{z}{\Phi}_{w}-\partial _{w}{\Phi}_{z}+[{\Phi}_{z},{\Phi}_{w}]=0, \\
F_{\tilde z\tilde w} & = \partial _{\tilde z}{\Phi}_{\tilde w}-\partial _{\tilde w}{\Phi}_{\tilde z}+[{\Phi}_{\tilde z},{\Phi}_{\tilde w}]=0, \\
F_{z\tilde z}-F_{w \tilde w} & = \partial _{z}{\Phi}_{\tilde z}-\partial _{\tilde z}{\Phi}_{z}-\partial _{w}{\Phi}_{\tilde w}+\partial _{\tilde w}{\Phi}_{w}\\
& \hspace{0.4cm}+[{\Phi}_{z},{\Phi}_{\tilde z}]-[{\Phi}_{w},{\Phi}_{\tilde w}]=0.
\end{split}
\end{equation}
With 
\begin{equation*}
\Drm_{w}=\partial _{w}+{\Phi}_{w},\ \Drm_{z}=\partial _{z}+{\Phi}_{z},\ \Drm_{\tilde w}=\partial _{\tilde w}+{\Phi}_{\tilde w},\ \Drm_{\tilde z}=\partial _{\tilde z}+{\Phi}_{\tilde z}
\end{equation*}
the ASD condition can be written as
\begin{equation} \label{eq:asdcond2}
[\Drm_{z}, \Drm_{w}]=0,\ [\Drm_{\tilde z}, \Drm_{\tilde w}]=0,\ [\Drm_{z}, \Drm_{\tilde z}]-[\Drm_{w}, \Drm _{\tilde w}]=0,
\end{equation}
so that the ASDYM equation corresponds to the vanishing of the curvature on every ${\alpha}$-plane. Equivalently, we can require that the \textit{Lax pair} of operators
\begin{equation*}
L=\Drm_{w}-{\zeta}\Drm_{\tilde z}, \ M=\Drm_{z}-{\zeta}\Drm_{\tilde w}
\end{equation*}
should commute for every value of the complex `spectral parameter' ${\zeta}$, where $L$ and $M$ act on vector-valued functions on $\mathbb{C}\Mbb$. 

This last compatibility condition says that for a section $s$ of $E$, represented by a column vector of length $n$, the linear system
\begin{equation*}
Ls=0,\ Ms=0
\end{equation*}
can be integrated for each fixed value of ${\zeta}$. Therefore, putting $n$ independent solutions together, we obtain an $n\times n$ matrix fundamental solution $f$ (dependent on ${\zeta}$) such that the columns of $f$ form a frame field for $E$ consisting of sections that are covariantly constant on the ${\alpha}$-planes tangent to $\partial _{w}-{\zeta}\partial _{\tilde z}$ and $\partial _{z}-{\zeta}\partial _{\tilde w}$. Henceforth, we suppose that $U$ is open and that each ${\alpha}$-plane that meets $U$ intersects it in a connected and simply connected set, for example $U$ an open ball. The second condition about simply connectedness ensures that the sections are single valued.

The matrix $f$ satisfies
\begin{equation} \label{eq:reconpot1}
\begin{split}
(\partial _{w}+{\Phi}_{w})f-{\zeta}(\partial _{\tilde z}+{\Phi}_{\tilde z})f & = 0,\\
(\partial _{z}+{\Phi}_{z})f-{\zeta}(\partial _{\tilde w}+{\Phi}_{\tilde w})f & = 0, 
\end{split}
\end{equation}
and depends holomorphically on ${\zeta}$ (varied over the complex plane), and the coordinates $w$, $z$, $\tilde w$, $\tilde z$. However, if $f$ were regular, by which we mean holomorphic with non-vanishing determinant, on the entire ${\zeta}$-Riemann sphere, then Liouville's theorem would imply that $f$ is independent of ${\zeta}$. Consequently,
\begin{equation*}
\Drm_{w}f=\Drm_{z}f=\Drm_{\tilde w}f=\Drm_{\tilde z}f=0, 
\end{equation*}
in other words $f$ were covariantly constant and the connection flat.

Given a choice of gauge, $f$ is unique up to $f\mapsto fH$ where $H$ is a non-singular matrix-valued function of ${\zeta}$, $w$, $z$, $\tilde w$, $\tilde z$ such that
\begin{equation} \label{eq:free1}
\partial _{w}H-{\zeta}\partial _{\tilde z}H=0,\ \partial _{z}H-{\zeta}\partial _{\tilde w}H=0. 
\end{equation}
So, $H$ is essentially a function of ${\lambda}={\zeta}w+\tilde z$, ${\mu}={\zeta}z+\tilde w$ and ${\zeta}$. 

If $\Pcal$ denotes the twistor space of $U$, and $V$, $\tilde V$ is a two-set Stein open cover of $\Pcal$\footnote{For Stein manifolds see for example \citet[Def.~4.2.7]{Field:1982aa} and how the two-set cover can be chosen see for example \citet[Sec.~3.3]{Popov:1999aa}.} such that $V$ is contained in the complement of ${\zeta}=\infty $, and $\tilde V$ is contained in the complement of ${\zeta}=0$, then $f$ can be regarded as a function on the correspondence space $\Fcal$ and $H$ as the pull-back of a holomorphic function on $V$, as by \eqref{eq:free1} $H$ is constant along the leaves of $p:\Fcal\to \Pcal$. 

If $D$ is not flat, $f$ cannot be chosen so that it is regular for all finite values of ${\zeta}$ and at ${\zeta}=\infty $. However, with $\tilde {\zeta}=\frac{1}{{\zeta}}$ we get a solution $\tilde f$ for the linear system 
\begin{equation} \label{eq:reconpot2}
\tilde {\zeta} \Drm_{w}\tilde f-\Drm_{\tilde z}\tilde f=0,\ \tilde {\zeta} \Drm_{z}\tilde f-\Drm_{\tilde w}\tilde f=0,
\end{equation}
which is holomorphic on the whole $\tilde {\zeta}$-Riemann sphere except for $\tilde {\zeta}=0$. This solution is unique up to $\tilde f\mapsto \tilde f \tilde H$, where $\tilde H$ is holomorphic on $\tilde V$ and satisfies the equation corresponding to \eqref{eq:free1}.

\section{The Patching Matrix}

On the overlap of the domains of $f$ and $\tilde f$ in $\Fcal$ we must have
\begin{equation*}
f=\tilde f P,
\end{equation*}
where $P$ satisfies \eqref{eq:free1}, and thus it is the pull-back by $p$ of a holomorphic function on $V\cap \tilde V$. We call $P$ \textit{patching matrix} associated to $\Drm$. It is determined by $\Drm$ up to equivalence $P \sim \tilde H^{-1}PH$ where $H$ is regular on $V$,  $\tilde H$ is regular on $\tilde V$ and both satisfy \eqref{eq:free1}. The matrices in the equivalence class of $P$ are called \textit{patching data} of $\Drm$. If $P$ is in the equivalence class of the identity solution, then $P=\tilde H^{-1}H$, hence
\begin{equation} \label{eq:free2}
fH=\tilde f \tilde H. 
\end{equation}
The left-hand side of \eqref{eq:free2} is regular in $V$, and the right-hand side is regular in $\tilde V$, so we have a global solution in ${\zeta}$, and therefore vanishing curvature. If there is no such solutions, the curvature has to be nonzero. 

The transformation of ${\Phi}$ under a gauge transformation is 
\begin{equation*}
{\Phi} \mapsto  {\Phi}'=g^{-1}{\Phi}g+g^{-1}\,\drm g, 
\end{equation*}
where $g$ is a function of $w$, $z$, $\tilde w$, $\tilde z$ with values in the gauge group. A solution for the new potential can be attained by replacing $f\mapsto g^{-1}f$ and $\tilde f\mapsto g^{-1}\tilde f$, which leaves the patching matrix unchanged. 

We have obtained a map which assigns patching data to every ASDYM field. It is called \textit{forward Penrose-Ward transform}. Indeed, the converse is true as well, that is a patching matrix encodes an ASDYM field.

\section{The Reverse Transform}

For the following arguments we need the so-called Birkhoff's factorization theorem. We do not state it in full precision, but only to the extend that is necessary here (for more details see \citet[Sec.~9.3]{Mason:1996hl}). 

Suppose $P({\lambda},{\mu},{\zeta})$ is holomorphic matrix-valued function on $V\cap \tilde V$ with non-vanishing determinant. Birkhoff's factorization theorem says that for fixed values of $w$, $z$, $\tilde w$, $\tilde z$ we can factorize $P$ in the form
\begin{equation*}
P({\zeta}w+\tilde z, {\zeta}z+\tilde w,{\zeta})=\tilde f^{-1}{\Delta}f,
\end{equation*}
where $f(w,z,\tilde w,\tilde z,{\zeta})$ is regular for $|{\zeta}|\leq 1$, $\tilde f(w,z,\tilde w,\tilde z,{\zeta})$ is regular for $|{\zeta}|\geq 1$ (including ${\zeta}=\infty $), and ${\Delta}=\diag ({\zeta}^{k_{1}},\dotsc ,{\zeta}^{k_{n}})$ for some integers $k_{1},\dotsc ,k_{n}$ which may depend on the point in $\mathbb{C}\Mbb$. For functions $P$ for which ${\Delta}=1$, this factorization is unique up to $f\mapsto cf$, $\tilde f\mapsto c\tilde f$ for some constant $c\in \GL(n,\mathbb{C})$. Furthermore, given a $P$ such that ${\Delta}=1$ at some point of $\mathbb{C}\Mbb$, then ${\Delta}=1$ in an open set of $\mathbb{C}\Mbb$.\footnote{\label{fn:trivass}This statement is consequence of Birkhoff's factorization theorem. If $P(w,{\zeta})$ depends smoothly on additional parameters $w=(w_{1},w_{2},\dotsc )$ and ${\Delta}=1$ at some point $w$, then ${\Delta}=1$ in an open neighbourhood of $w$, and $f$, $\tilde f$ can be chosen such that they depend smoothly on the parameters. The statement also holds if we replace `smooth' by `holomorphic' in the case that $P$ depends holomorphically on ${\zeta}$ (in a neighbourhood of the unit circle) and on the complex parameters $w$. Attempts to extend the factorization to the entire parameter space typically fail on a submanifold of codimension 1, where ${\Delta}$ `jumps' to another value than the identity \cite[Prop.~9.3.4]{Mason:1996hl}.
} 

Given a Birkhoff factorization for fixed $(w,z,\tilde w,\tilde z)$ such that
\begin{equation*}
P({\zeta}w+\tilde z,{\zeta}z+\tilde w,{\zeta})=\tilde f^{-1}f,
\end{equation*}
we can recover ${\Phi}$ in terms of $f$ or $\tilde f$ by 
\begin{equation} \label{eq:pot1}
{\Phi}_{w}-{\zeta}{\Phi}_{\tilde z}=(-\partial _{w}f+{\zeta}\partial _{\tilde z}f)f^{-1}=(-\partial _{w}\tilde f+{\zeta}\partial _{\tilde z}\tilde f)\tilde f^{-1}
\end{equation}
with a similar equation for the other two components. By the uniqueness statement any other factorization is given by $f'=gf$, $\tilde f'=g\tilde f$, where $g$ is independent of ${\zeta}$. The new potential ${\Phi}'$, which is obtained from $f'$, $\tilde f'$, is related to the previous one via
\begin{equation*}
{\Phi}=g^{-1}{\Phi}'g+g^{-1}\,\drm g. 
\end{equation*}
This is obvious regarding \eqref{eq:pot1}. Hence, $P$ determines ${\Phi}$ up to gauge transformations. 

It remains to show that our patching matrix $P$ is the same as the one that is associated to the ASDYM field we have just constructed. In other words, we must ensure that applying first the reverse and then the forward Penrose-Ward transform to the function $P$ gets us back to the starting point. 

Suppose $P$ is chosen such that ${\Delta}=1$ at some point of $\mathbb{C}\Mbb$, then ${\Delta}=1$ in an open set $U$ of $\mathbb{C}\Mbb$. The constancy of $P$ along $\partial _{w}-{\zeta}\partial _{\tilde z}$ implies
\begin{equation*}
0=(\partial _{w}-{\zeta}\partial _{\tilde z})(\tilde f^{-1}f)=-\tilde f^{-1}((\partial _{w}-{\zeta}\partial _{\tilde z})\tilde f)\tilde f^{-1}f+\tilde f^{-1}(\partial _{w}-{\zeta}\partial _{\tilde z})f, 
\end{equation*}
or equivalently
\begin{equation} \label{eq:const1}
(\partial _{w}f-{\zeta}\partial _{\tilde z}f)f^{-1}=(\partial _{w}\tilde f-{\zeta}\partial _{\tilde z}\tilde f)\tilde f^{-1},
\end{equation}
at every point in $U$ and for all ${\zeta}$ in some neighbourhood of the unit circle. The left-hand side of \eqref{eq:const1} is holomorphic for $|{\zeta}|<1$, and the right-hand side is holomorphic for $|{\zeta}|>1$ except for a simple pole at infinity. Therefore, by an extension of Liouville's theorem both sides have to be of the form $-{\Phi}_{w}+{\zeta}{\Phi}_{\tilde z}$, where ${\Phi}_{w}$ and ${\Phi}_{\tilde z}$ are independent of ${\zeta}$. We take them to be the two components of ${\Phi}$. The same procedure with $\partial _{z}-{\zeta}\partial _{\tilde w}$ instead of $\partial _{w}-{\zeta}\partial _{\tilde z}$ defines ${\Phi}_{z}$ and ${\Phi}_{\tilde w}$. Then, by construction
\begin{equation*}
\Drm_{w}f-{\zeta}\Drm_{\tilde z}f=0,\ \Drm_{z}-{\zeta}\Drm_{\tilde w}f=0,
\end{equation*}
where $\Drm=\drm+{\Phi}$ acts on the columns of $f$. Thus, the linear system associated to $\Drm$ is integrable, and $\Drm=\drm+{\Phi}$ is ASD.

The above constructions gives an explicit formula of the potential ${\Phi}$ using $f$ and $\tilde f$.

\begin{lem} \label{lem:recon}
The gauge potential ${\Phi}$ is given in terms of $f$ and $\tilde f$ by
\begin{equation*}
{\Phi}=h\partial h^{-1}+\tilde h \tilde \partial  \tilde h^{-1},
\end{equation*}
where $h=\left.f\right|_{{\zeta}=0}$ and $\tilde h=\left.\tilde f\right|_{{\zeta}=\infty }$.
\end{lem}
\begin{proof}
Setting ${\zeta}=0$ in \eqref{eq:reconpot1} gives the first two components ${\Phi}_{w}$, ${\Phi}_{z}$; setting $\tilde {\zeta}=0$ in \eqref{eq:reconpot2} gives the second pair of components ${\Phi}_{\tilde w}$, ${\Phi}_{\tilde z}$.
\end{proof}

So, we have shown that $\Drm$ can be recovered from the patching matrix, and indeed even more, namely that any patching matrix such that ${\Delta}=1$ at some point of $\mathbb{C}\Mbb$ generates a solution of the ASDYM equation in an open set of $\mathbb{C}\Mbb$.

\section{The Abstract Form of the Transform}

The construction as described so far is very explicit and suggests that it depends on the cover $V$, $\tilde V$ and the chosen coordinates. However, it does not show, as stated in its abstract form, that the transform is in fact between ASDYM fields on $U$ and holomorphic vector bundles on $\Pcal$. The above choices were only a way to represent the bundle in a concrete way.

If $\Drm$ is an ASD connection on a rank-$n$ vector bundle $B\to U$, then the patching matrix $P$ on $U$ determines a rank-$n$ vector bundle $B'\to \Pcal$, where $P$ can be regarded as the transition matrix between the holomorphic trivializations (of $B'$) over $V$ and $\tilde V$. The fibre of $B'$ attached to a point $Z\in \Pcal$ is 
\begin{equation*}
B'_{Z}=\{s\in {\Gamma}(Z\cap U,B): \left. \Drm s\right|_{Z\cap U}=0\},
\end{equation*}
where ${\Gamma}(Z\cap U,B)$ is the space of sections of $B$ over $Z\cap U$. That is the fibre over an ${\alpha}$-plane $Z\in \Pcal$ is the vector space of covariantly constant sections of $B$ over $Z\cap U$ with respect to $\Drm$. Because $P$ has the Birkhoff factorization with ${\Delta}=1$ at each point of $U$, the restriction of $B'$ to each line in $\Pcal$ corresponding to a point in $U$ is holomorphically trivial. The holomorphic bundle is uniquely determined up to equivalence by $\Drm$, as the freedom in the construction of $P$ from $\Drm$ is precisely the freedom in choice of two local holomorphic trivializations. Conversely, given $B'\to \Pcal$, the patching matrix $P$ (and hence $\Drm$) can be recovered in a direct geometrical way. 

\begin{thm}[\citet{Ward:1977ab}] \label{thm:twicorr}
Let $U\subset \mathbb{C}\Mbb$ be an open set such that the intersection of $U$ with every ${\alpha}$-plane that meets $U$ is connected and simply connected. Then there is a one-to-one correspondence between solutions to the ASDYM equation on $U$ with gauge group $\GL(n,\mathbb{C})$ and holomorphic vector bundles $B'\to \Pcal$ such that $\left.B'\right|_{\hat x}$ is trivial for every $x\in U$.
\end{thm}
\begin{proof}
See \citet[Thm.~10.2.1]{Mason:1996hl}.
\end{proof}
\chapter[Yang's Equation]{Yang's Equation,  \texorpdfstring{${\sigma}$}{Sigma}-Model and Ernst Potential} \label{ch:yang}

We will see how the ASDYM can be written in a form called Yang's equation using \citet[Sec.~3.3]{Mason:1996hl}. Then it is shown how Yang's equation turns by a symmetry reduction into the Ernst equation which describes stationary axisymmetric solutions to the Einstein equations \cite[Sec.~6.6]{Mason:1996hl}.

\section{Yang's Equation and the $J$-Matrix}
The ASD condition on the curvature 2-form $*F=-F$ is coordinate-independent and invariant under gauge transformations as well as under conformal isometries of $\mathbb{C}\Mbb$. It can be written in other forms that are more tractable for certain aspects, even though some of the symmetries are broken. 

The first two equations of \eqref{eq:asdcond1}, equivalent to
\begin{equation*}
[\Drm_{z}, \Drm_{w}]=0,\ [\Drm_{\tilde z}, \Drm_{\tilde w}]=0,
\end{equation*}
are an integrability condition for the existence of matrix-valued functions $h$ and $\tilde h$ on $\mathbb{C}\Mbb$ such that
\begin{align*}
\partial _{w}h+{\Phi}_{w}h=0,\quad \partial _{z}h+{\Phi}_{z}h=0, \\
\partial _{\tilde w}\tilde h+{\Phi}_{\tilde w}\tilde h=0,\quad \partial _{\tilde z}\tilde h+{\Phi}_{\tilde z}\tilde h=0.
\end{align*}
The potential ${\Phi}$ determines $h$ and $\tilde h$ uniquely up to $h\mapsto h B$, $\tilde h \mapsto  \tilde h A$, where $B$ and $A$ are matrices depending only on $\tilde w$, $\tilde z$ and $w$, $z$, respectively. For a gauge transformed potential ${\Phi}\mapsto g{\Phi}g^{-1}+g^{-1}\,\drm g$ we can replace $h\mapsto g^{-1}h$ and $\tilde h \mapsto  g^{-1} \tilde h$, which leaves the expression $\tilde h^{-1}h$ invariant. 

We define \textit{Yang's matrix} \cite{Yang:1977aa} as $J \coloneqq \tilde h^{-1}h$. It is determined by the connection $\Drm$ up to $J\mapsto A^{-1}JB$. Note that $h$ and $\tilde h$ satisfy the same differential equations as those in Lemma~\ref{lem:recon}. Thus, the definition for Yang's matrix is equivalent to $J=f({\zeta}=0)\tilde f^{-1}({\zeta}=\infty )$ with $f$ and $\tilde f$ as in Chapter~\ref{ch:pwtrf}. To consider the converse direction let $J=\tilde h^{-1}h$ be a Yang's matrix. The connection $\Drm$ is determined by $J$, since we have
\begin{equation} \label{eq:yang1}
\begin{split}
J^{-1}\tilde \partial  J & = J^{-1}\partial _{\tilde w} J \,\drm \tilde w + J^{-1}\partial _{\tilde z} J \,\drm \tilde z\\
			& = h^{-1}\tilde h\, \partial _{\tilde w}(\tilde h^{-1}h)\,\drm \tilde w + \tilde w \leftrightarrow  \tilde z\\
			& = \left(h^{-1}\tilde h (\partial _{\tilde w}\tilde h^{-1})h+h^{-1}\partial _{\tilde w}h\right)\drm \tilde w + \tilde w \leftrightarrow  \tilde z\\
			& = \left(-h^{-1}(\partial _{\tilde w}\tilde h)\tilde h^{-1}h+h^{-1}\partial _{\tilde w}h\right)\drm \tilde w + \tilde w \leftrightarrow  \tilde z\\
			& = \left(h^{-1} {\Phi}_{\tilde w}h+h^{-1}\partial _{\tilde w}h\right)\drm \tilde w + \tilde w \leftrightarrow  \tilde z, 
\end{split}
\end{equation}
so that ${\Phi}$ is equivalent to $J^{-1}\tilde \partial  J$ by the gauge transformation ${\Phi}\mapsto h{\Phi}h^{-1}+h^{-1}\,\drm h$.\footnote{Note that after the gauge transformation the potential $h{\Phi}h^{-1}+h^{-1}\,\drm h$ has vanishing $w$ and $z$ component.
} 
Given the matrix $J$ we see that the first two equations of \eqref{eq:asdcond1} are implied by \eqref{eq:yang1},
because 
\begin{equation*}
{\Phi}_{w}=0,\ {\Phi}_{z}=0, \ {\Phi}_{\tilde w}=J^{-1}\partial _{\tilde w}J, \ {\Phi}_{\tilde z}=J^{-1}\partial _{\tilde z}J
\end{equation*}
yields the first equation in \eqref{eq:asdcond1} trivially, and the second one follows from
\begin{align*}
\partial _{\tilde z}{\Phi}_{\tilde w}-&\partial _{\tilde w}{\Phi}_{\tilde z}+[{\Phi}_{\tilde z},{\Phi}_{\tilde w}] \\
& = \partial _{\tilde z}(J^{-1}\partial _{\tilde w}J)-\partial _{\tilde w}(J^{-1}\partial _{\tilde z}J)+[J^{-1}\partial _{\tilde z}J,J^{-1}\partial _{\tilde w}J]\\
& = 0.
\end{align*}
A calculation similar to \eqref{eq:yang1} shows that
\begin{align*}
\partial _{w}(J^{-1}\partial _{\tilde w}J)-\partial _{z}(J^{-1}\partial _{\tilde z}J) = h^{-1} (F_{w\tilde w}-F_{z\tilde z})h,
\end{align*}
which is equivalent to $F_{w\tilde w}-F_{z\tilde z}$ under the above gauge transformation, and hence the ASD equations are equivalent to \textit{Yang's equation}
\begin{equation} \label{eq:yang2}
\partial _{w}(J^{-1}\partial _{\tilde w}J)-\partial _{z}(J^{-1}\partial _{\tilde z}J)=0.
\end{equation}
However, they are no longer covariant under conformal transformations as this would change the 2-planes spanned by $\partial _{w}$, $\partial _{z}$ and $\partial _{\tilde w}$, $\partial _{\tilde z}$.

Geometrically the construction of $J$ can be interpreted as follows. Starting from the original gauge we make a transformation by $g=h$ or $g=\tilde h$, respectively. This yields an equivalent gauge with vanishing ${\Phi}_{w}$, ${\Phi}_{z}$ in the first case and vanishing ${\Phi}_{\tilde w}$, ${\Phi}_{\tilde z}$ in the second case. If we have frame fields $\{e_{1},\dotsc ,e_{n}\}$ and $\{\tilde e_{1},\dotsc ,\tilde e_{n}\}$, respectively, corresponding to the new gauge potentials, then
\begin{align}
\Drm_{w}e_{i}=0,\quad \Drm_{z}e_{i}=0, \label{eq:yangframe1} \\
\Drm_{\tilde w}\tilde e_{i}=0,\quad \Drm_{\tilde z}\tilde e_{i}=0, \label{eq:yangframe2}
\end{align}
for $i=1,\dotsc ,n$. Furthermore, by definition of Yang's matrix $e_{j}=\tilde e_{i}J_{ij}$. Therefore, $J$ is the linear transformation from a frame field satisfying \eqref{eq:yangframe1} to a frame field satisfying \eqref{eq:yangframe2}.
The connection potentials in the frames $e_{i}$ and $\tilde e_{i}$ are, respectively, 
\begin{equation*}
J^{-1}\tilde \partial  J \quad \text{and} \quad J\partial J^{-1}.
\end{equation*}
The freedom in the construction of $J$ from $\Drm$ is the freedom to transform the first frame by $B$ and the second frame by $A$.

\section{Reduction of Yang's Equation}

In this section we show how Yang's equation can be reduced by an additional symmetry so that it provides the linkage to Einstein equations. 

Such a reduction can in general be approached as follows. Suppose we are given a two-dimensional subgroup of the conformal group, generated by two conformal Killing vectors $X$ and $Y$, which span the tangent space of the orbit of our two-dimensional symmetry group at each point. To impose the symmetry on a solution of the ASDYM equation, we require that the Lie derivative of the gauge potential ${\Phi}$ along $X$ and $Y$ vanishes.\footnote{This requirement can be justified in a more systematic way if invariant connections, Lie derivates on sections of vector bundles etc. are introduced \cite{Mason:1996hl}. Here the Lie derivative of ${\Phi}$ along the Killing fields is the ordinary Lie derivative operator on differential form.
}
If the symmetry group is Abelian, $[X,Y]=0$, then there exists a coordinate system with $X$, $Y$ as the first two coordinate vector fields and by the symmetry the components of ${\Phi}$ depend only upon the second pair of coordinates. The reduced version of equation~\eqref{eq:asdcond1} is obtained by changing the coordinates and discarding the derivatives with respect to the ignorable coordinates. 

Another way is to impose the symmetry on Yang's equation which is not always straightforward and we will see how it works in our case.

We are interested in the reduction by the commuting Killing vectors
\begin{equation*}
X=w\partial _{w}-\tilde w \partial _{\tilde w},\ Y=\partial _{\tilde z}+\partial _{z}. 
\end{equation*}
The coordinates can be adapted by the transformation
\begin{equation*}
w=r\erm^{\irm{\theta}},\ \tilde w=r\erm^{-\irm{\theta}},\ z=t-x,\ \tilde z=t+x.
\end{equation*}
Then the metric is easily calculated to have the form
\begin{equation*}
\drm s^{2}=\drm t^{2}-\drm x^{2}-\drm r^{2}-r^{2}\,\drm {\theta}^{2},
\end{equation*}
and the Killing vectors are $X=-2\irm \partial _{{\theta}}$, $Y=2\partial _{t}$ so that the symmetries are a rotation ${\theta}\mapsto {\theta}+{\theta}_{0}$ and a time translation $t\mapsto t+t_{0}$. 

In the Minkowski real slice the coordinates are real and the spatial metric is of cylindrical polar form. Thus, a reduction by $X$ and $Y$ refers to stationary axisymmetric solutions of the ASDYM equation and their continuations to $\mathbb{C}\Mbb$. The crucial point for our considerations is the coincidence that apart from their role in Yang-Mills theory the reduced equation turns out to be equivalent to the Ernst equation for stationary axisymmetric gravitational fields in general relativity. 

We want to construct Yang's matrix for the invariant potential in this stationary axisymmetric reduction. First, we note that it is still possible to choose the invariant gauge such that ${\Phi}_{w}={\Phi}_{z}=0$ and
\begin{equation} \label{eq:redpot1}
{\Phi}=-P\frac{\drm \tilde w}{\tilde w}+Q \,\drm \tilde z, 
\end{equation}
where\footnote{This $P$ should not be confused with the Patching matrix.} $P$ and $Q$ depend only on $x$ and $r$. This can be seen as follows. Above we have shown that for $h$ being a solution of
\begin{equation} \label{eq:gauge1}
\partial _{w}h+{\Phi}_{w}h=0,\quad \partial _{z}h+{\Phi}_{z}h=0,
\end{equation}
the gauge transformation with $g=h$ yields a potential with ${\Phi}_{w}={\Phi}_{z}=0$, because \eqref{eq:gauge1} is equivalent to
\begin{equation*}
h^{-1}{\Phi}_{w}h+h^{-1}\partial _{w}h=0,\quad h^{-1}{\Phi}_{z}h+h^{-1}\partial _{z}h=0.
\end{equation*}
The integrability of \eqref{eq:gauge1} and hence the existence of $h$ was ensured by the ASD condition \eqref{eq:asdcond2}.

Now with the additional symmetry the question arises whether this is still possible but with $h$ and $\Phi$ depending on $x$ and $r$ only. Under a gauge transformation we then have
\begin{equation*}
{\Phi}_{w}\mapsto h^{-1}{\Phi}_{w}h+\frac{1}{2}\erm^{-\irm {\theta}} h^{-1}h_{r}, \quad {\Phi}_{z}\mapsto h^{-1}{\Phi}_{z}h-\frac{1}{2} h^{-1}h_{x},
\end{equation*}
which is obtained by just substituting the coordinates and discarding the dependency on $t$ and ${\theta}$. In the same way the first equation of \eqref{eq:asdcond1} becomes
\begin{equation} \label{eq:redasdcond1}
-\frac{1}{2}\partial _{x}{\Phi}_{w}-\frac{1}{2}\erm^{-\irm {\theta}} \partial _{r}{\Phi}_{z}+[{\Phi}_{z},{\Phi}_{w}]=\left[\Phi_{z}-\frac{1}{2}∂_{x}, \Phi_{w}+\frac{1}{2}\erm^{-\irm {\theta}}∂_{r}\right] =0.
\end{equation}
Then, analogous to the general case, \eqref{eq:asdcond1} in form of \eqref{eq:redasdcond1} is exactly the integrability condition for the existence of an invariant gauge in which ${\Phi}_{w}={\Phi}_{z}=0$ and where $\Phi$ and $h$ depend only on $x$ and $r$, that is
\begin{equation*}
{\Phi}_{w}h+\frac{1}{2}\erm^{-\irm {\theta}}\partial _{r}h=0, \quad {\Phi}_{z}h-\frac{1}{2}\partial _{x}h=0.
\end{equation*}

Using the notation from \eqref{eq:redpot1} and by a similar calculation like above the second equation in \eqref{eq:asdcond1}, $F_{\tilde z \tilde w}=0$, takes the form
\begin{equation} \label{eq:redasdcond2}
P_{x}+rQ_{r}+2[Q,P]=0,
\end{equation}
and the third equation, $F_{z\tilde z}-F_{w\tilde w}=0$, becomes
\begin{equation*}
0=\partial _{z}{\Phi}_{\tilde z}-\partial _{w}{\Phi}_{\tilde w}=-\frac{1}{2}\partial _{x}Q+\frac{1}{2}\erm^{-\irm {\theta}}\partial _{r}\left(\frac{P}{\tilde w}\right),
\end{equation*}
or equivalently
\begin{equation} \label{eq:redasdcond3}
P_{r}-rQ_{x}=0.
\end{equation}
Condition~\eqref{eq:redasdcond2} guarantees the existence of Yang's matrix $J(x,r)$ such that
\begin{equation*}
{\Phi}_{\tilde w}=-\frac{P}{\tilde w} = J^{-1} \partial _{\tilde w} J = \frac{1}{2} \erm^{\irm {\theta}}J^{-1}\partial _{r}J,\ {\Phi}_{\tilde z}=J^{-1}\partial _{\tilde z}J,
\end{equation*}
equivalent to
\begin{equation} \label{eq:redpot2}
2P=-rJ^{-1}\partial _{r}J,\quad 2Q=J^{-1}\partial _{x}J.
\end{equation}
It can again be interpreted as a change of gauge, but this time $h$ and $\tilde h$ depend only on $x$ and $r$, hence so does $J$. 

Starting with a matrix $J(x,r)$, and defining the gauge potential by \eqref{eq:redpot2}, it is an easy calculation that \eqref{eq:redasdcond2} is satisfied. Equation~\eqref{eq:redasdcond3} becomes
\begin{equation} \label{eq:redyang}
r\partial _{x}(J^{-1}\partial _{x}J)+\partial _{r}(rJ^{-1}\partial _{r}J)=0.
\end{equation}
So, every solution to \eqref{eq:redyang} determines a stationary axisymmetric ASDYM field and every stationary axisymmetric ASDYM field can be obtained in that way. The Yang's matrix $J$ determines the connection up to $J\mapsto A^{-1}JB$ with constant matrices $A$ and $B$. Note that in the general case above $A$, $B$ depended on $w$, $z$ and $\tilde w$, $\tilde z$, respectively. But here we restricted to gauge transformations such that the condition ${\Phi}_{w}={\Phi}_{z}=0$ is preserved, hence they have to be constant (formula for gauge transformation of ${\Phi}$ involves derivatives of $h$, respectively $\tilde h$).

\section{Reduction of Einstein Equations}

The next step will be to show how reduced Yang's equation~\eqref{eq:redyang} emanates from a reduction of the Einstein equations. This was originally discovered by \citet{Witten:1979aa, Ward:1983yg}.

Let $g_{ab}$ be a metric tensor in $n$ dimensions (real or complex), and $X_{i}^{a}$, $i=0,\dotsc ,n-s-1$, be $n-s$ commuting Killing vectors that generate an orthogonally transitive isometry group with non-null $(n-s)$-dimensional orbits. This means the distribution of $s$-plane elements orthogonal to the orbits of $X_{i}$ is integrable, in other words $[U,V]$ is orthogonal to all $X_{i}$ whenever $U$ and $V$ are orthogonal to all $X_{i}$. 

Define $J=(J_{ij})\coloneqq (g_{ab}^{\vphantom{1}}X_{i}^{a}X_{j}^{b})$, and denote by $\nabla $ the Levi-Civita connection. We have
\begin{equation*}
X_{k}J_{ij}=(\Lcal_{X_{k}}g)(X_{i},X_{j})+g(\Lcal_{X_{k}}X_{i},X_{j})+g(X_{i},\Lcal_{X_{k}}X_{j})=0,
\end{equation*}
as the first term vanishes due to the fact that $X_{k}$ is a Killing vector, and the last two terms since the Killing vectors commute. Thus $J$ is constant along the orbits of the Killing vectors. The Killing equation is
\begin{equation} \label{eq:killequ}
0=\Lcal_{X_{i}}g_{ab}=\nabla _{a}X_{ib}+\nabla _{b}X_{ia},
\end{equation}
and since
\begin{equation} \label{eq:comm}
[X_{i},X_{j}]=0 \quad \forall \, i,j
\end{equation}
we get
\begin{equation} \label{eq:delJredeinst}
\begin{split}
\partial _{b}^{\vphantom{b}}J_{ij}^{\vphantom{b}} & = \nabla _{b}^{\vphantom{b}}(X_{i}^{a\vphantom{b}}X_{ja}^{\vphantom{b}})=X_{i}^{a\vphantom{b}}\nabla _{b}^{\vphantom{b}}X_{ja}^{\vphantom{b}}+X_{j}^{a\vphantom{b}}\nabla _{b}^{\vphantom{b}}X_{ia}^{\vphantom{b}} \stackrel{\eqref{eq:killequ}}{=} X_{i}^{a\vphantom{b}}\nabla _{b}^{\vphantom{b}}X_{ja}^{\vphantom{b}}-X_{j}^{a\vphantom{b}}\nabla _{a}^{\vphantom{b}}X_{ib}^{\vphantom{b}} \\
& \stackrel{\eqref{eq:comm}}{=} X_{i}^{a\vphantom{b}}\nabla _{b}^{\vphantom{b}}X_{ja}^{\vphantom{b}}-X_{i}^{a\vphantom{b}}\nabla _{a}^{\vphantom{b}}X_{jb}^{\vphantom{b}}\stackrel{\eqref{eq:killequ}}{=} 2X_{i}^{a\vphantom{b}}\nabla _{b}^{\vphantom{b}}X_{ja}^{\vphantom{b}}.
\end{split}
\end{equation}
This yields
\begin{equation*}
X_{j}^{a\vphantom{b}}\nabla _{a}^{\vphantom{b}}X_{ib}^{\vphantom{b}}=X_{i}^{a\vphantom{b}}\nabla _{a}^{\vphantom{b}}X_{jb}^{\vphantom{b}}=-X_{i}^{a\vphantom{b}}\nabla _{b}^{\vphantom{b}}X_{ja}^{\vphantom{b}}=-\frac{1}{2}\partial _{b}^{\vphantom{b}}J_{ij}^{\vphantom{b}}.
\end{equation*}
Let $U$ and $V$ be vector fields orthogonal to the orbits. From
\begin{align*}
0 & = \nabla _{a}(U^{a}\underbrace{V^{b}X_{ib}}_{=0}) -\nabla _{a}(V^{a}\underbrace{U^{b}X_{ib}}_{=0}) \\
& = (\nabla _{a}U^{a})\underbrace{V^{b}X_{ib}}_{=0}+U^{a}(\nabla _{a}V^{b})X_{ib}+U^{a}V^{b}(\nabla _{a}X_{ib}) - U \leftrightarrow V
\end{align*}
we obtain
\begin{align*}
U^{a}V^{b}\nabla _{a}X_{ib}-V^{a}U^{b}\nabla _{a}X_{ib} & = -X_{ib}U^{a}\nabla _{a}V^{b}+X_{ib}V^{a}\nabla _{a}U^{b} \\
											& = -X_{ib}(U^{a}\nabla _{a}V^{b}-V^{a}\nabla _{a}U^{b})\\
											& = -X_{ib}[U,V]^{b}\\
											& = 0,
\end{align*}
where the last step follows from the orthogonal transitivity. This result together with the Killing equation $\nabla _{a\vphantom{b]}}X_{ib\vphantom{b]}}=\nabla _{[a}X_{|i|b]}$ leads to
\begin{align*}
U^{a}V^{b}\nabla _{a}X_{ib} & = U^{[a}V^{b]}\nabla _{a}X_{ib} = 0, \\
U^{a\vphantom{b}}X_{j}^{b}\nabla _{a}^{\vphantom{b}}X_{ib}^{\vphantom{b}} & = \frac{1}{2} U^{a\vphantom{b}} \partial _{a}^{\vphantom{b}}J_{ij}^{\vphantom{b}}, \\
X_{j}^{a \vphantom{b}}X_{k}^{b}\nabla _{a}^{\vphantom{b}}X_{ib}^{\vphantom{b}} & = \frac{1}{2} X_{j}^{a \vphantom{b}} \partial _{a}^{\vphantom{b}}J_{ki}^{\vphantom{b}} = 0.
\end{align*}
However, 
\begin{equation*}
\frac{1}{2}J^{jk}\Big((\partial _{a}J_{ki})X_{jb}-(\partial _{b}J_{ki})X_{ja}\Big),
\end{equation*}
where $J^{ij}J_{jk}={\delta}^{i}_{k}$,\footnote{Here we need that the orbits of the isometry group are non-null, otherwise $J$ would be degenerate.} gives the same expressions when contracted with combinations of Killing and orthogonal vectors. Hence, they must have the same components, 
\begin{equation} \label{eq:cov1}
\nabla _{a}X_{ib}=\frac{1}{2}J^{jk}\Big((\partial _{a}J_{ki})X_{jb}-(\partial _{b}J_{ki})X_{ja}\Big).
\end{equation}
Moreover, for a Killing vector $X$ we can use the Ricci identity
\begin{equation*}
\nabla _{b}\nabla _{c\vphantom{b}}X_{d}=R_{abcd}X^{a}
\end{equation*}
with $R_{abcd}$ the Riemann tensor, to take the second derivative of \eqref{eq:cov1}. After some computations (see Appendix~\ref{app:redeinst}), we get
\begin{equation*}
R_{ab}^{\vphantom{b}}X_{i}^{a\vphantom{b}}X_{j}^{b}=-\frac{1}{2}J_{ik}^{\vphantom{b}} \sqrt{g}^{-1}\partial _{a}^{\vphantom{b}}(\sqrt{g}g^{ab}J^{kl}\partial _{b}^{\vphantom{b}}J_{lj}^{\vphantom{b}}),
\end{equation*}
where $g=\det (g_{ab})$ and $R_{ab}^{\vphantom{c}}=\tensor{R}{^{c}_{acb}}$ is the Ricci tensor. If the Einstein vacuum equations, $R_{ab}=0$, hold, then
\begin{equation} \label{eq:redeinst1}
\partial _{a}(\sqrt{g}g^{ab}J^{kl}\partial _{b}J_{lj})=0.
\end{equation}
Since $J$ is constant along the orbits \eqref{eq:redeinst1} is essentially an equation on $S$, where $S$ is the quotient space by the Killing vectors identified with any of the $s$-surfaces orthogonal to the orbits. Denote by $h_{ab}$ the metric on $S$ and by $D$ the corresponding Levi-Civita connection so that we get for the determinant
\begin{equation*}
g=\det{g_{ab}}=-r^{2}\det(h_{ab}),
\end{equation*}
where $-r^{2}=\det J$.\footnote{The sign might vary depending on the signature of the metric and the Killing vectors (spacelike or timelike). In some of the literature the condition is $r^{2}=|\det J|$. Here it is adapted to the Lorentzian case with one timelike and one spacelike Killing vector.} We know that for functions $u$ on $S$ covariant and partial derivative are equal $\partial _{a}u=D_{a}u$. Considered this together with the expression for the Laplace-Beltrami operator
\begin{equation*}
D^{2}u=D_{a}D^{a}u= \frac{1}{\sqrt{|h|}}\partial _{a}\left(\sqrt{|h|}h^{ab}\partial _{b}u\right)
\end{equation*}
equation~\eqref{eq:redeinst1} becomes
\begin{equation} \label{eq:redeinst2}
D_{a}(rJ^{-1}D^{a}J)=0,
\end{equation}
where the indices now run over $1,\dotsc ,s$ and are lowered and raised with $h_{ab}$ and its inverse. Now remembering that 
\begin{equation*}
\drm \det J = \det J \tr (J^{-1} \,\drm J) \quad \Rightarrow  \quad \drm (\log \det J)=\tr (J^{-1}\,\drm J)
\end{equation*}
we can take the trace of \eqref{eq:redeinst2} to get furthermore
\begin{align*}
0 & = D_{a}(r \tr(J^{-1}D^{a}J))=D_{a}(r D^{a} (\log \det J)) \\
& = D_{a}(rD^{a}(\log -r^{2}))= D_{a}(2r \cdot  \frac{1}{r}D^{a}r)\\
& = 2D^{2}r,
\end{align*}
hence $r$ is harmonic on $S$. From now on we assume that the gradient of $r$ is not null.

In the case $s=2$ isothermal coordinates always exist (see Appendix~\ref{app:isothermal}), that is we can write the metric on $S$ in the form
\begin{equation*}
{\erm}^{2ν}(\drm r^{2}+\drm x^{2})
\end{equation*}
where $x$ is the harmonic conjugate to $r$.\footnote{The function $x$ is said to be \textit{harmonic conjugate} to $r$, if $x$ and $r$ satisfy the Cauchy-Riemann equations.} As the Killing vectors commute, there exist coordinates $(y_{0},\dotsc ,y_{n-3})$, where the $X_{i}$ are the first $n-2$ coordinate vector fields. Taking furthermore the isothermal coordinates for the last two components, the full the metric then has the form
\begin{equation} \label{eq:sigmametric}
\drm s^{2} = \sum _{i,j=0}^{n-3}J_{ij}\,\drm y^{i} \drm y^{j} + {\erm}^{2ν} (\drm r^{2}+\drm x^{2}),
\end{equation}
which is known as the (\textit{Einstein}) \textit{${\sigma}$-model}. In this case \eqref{eq:redeinst2} reduces to \eqref{eq:redyang} and we obtain the following proposition.

\begin{prop}[Proposition 6.6.1 in \citet{Mason:1996hl}]
Let $g_{ab}$ be a solution to Einstein's vacuum equation in $n$ dimensions. Suppose that it admits $n-2$ independent commuting Killing vectors generating an orthogonally transitive isometry group with non-null orbits, and that the gradient of $r$ is non-null. Then $J(x,r)$ is the Yang's matrix of a stationary axisymmetric solution to the ASDYM equation with gauge group $\GL(n-2,\mathbb{C})$.
\end{prop}

The partial converse of the Proposition yields a technique for solving Einstein's vacuum equations as follows. Any real solution $J(x,r)$ to reduced Yang's equation~\eqref{eq:redyang} such that
\begin{enumerate}[(a)]
\item $\det J=-r^{2}$,
\item $J$ is symmetric
\end{enumerate}
determines a solution to the Einstein vacuum equations, because we can reconstruct the metric from given $J$ and ${\erm}^{2ν}$ via \eqref{eq:sigmametric}, and then \eqref{eq:redyang} is equivalent to the vanishing of the components of $R_{ab}$ along the Killing vectors (as we have shown above). The remaining components of the vacuum equations can be written as
\begin{equation} \label{eq:redeinst3}
2 \irm \partial _{{\xi}}\left(\log\left(r{\erm}^{2ν}\right)\right)=r\tr \left(\partial_{\xi}\left(J^{-1}\right) \partial_{\xi} J\right),
\end{equation}
with ${\xi}=x+\irm r$, together with the complex conjugate equation (if $x$ and $r$ are real), where ${\xi}$ is replaced by $\skew{1}{\bar}{{\xi}}= x-\irm r$ and $\irm$ by $-\irm$. These equations can be obtained by a direct calculation of the Christoffel symbols, curvature tensors and so on \cite[App.~D, Eq.~(D9)]{Harmark:2004rm}. They are automatically integrable if \eqref{eq:redyang} is satisfied and under the constraint $\det J=-r^{2}$ (see Appendix~\ref{app:intefe}), and they determine ${\erm}^{ν}$ up to a multiplicative constant. The constraint, however, is not significant for the following reason. We know that in polar coordinates $u=\log r$ is a solution to the (axisymmetric) Laplace equation
\begin{equation} \label{eq:axilap}
\partial_{r}^{\vphantom{2}}(r\partial_{r}^{\vphantom{2}}u)+r\partial _{x}^{2}u=0,
\end{equation}
so is $u=c \log r + \log d$ for constants $c$ and $d$. Now suppose $J$ is a solution to \eqref{eq:redyang}, and consider $\erm^{u}J=dr^{c}J$. Plugging this new matrix in \eqref{eq:redyang}, we see that it is again a solution of reduced Yang's equation if \eqref{eq:axilap} holds. The determinant constraint can thus be satisfied by an appropriate choice of the constants, since we have
\begin{equation*}
\det(\erm^{u}J)=\erm^{(n-2)u}\det J=d^{n-2} r^{(n-2)c}\det J.
\end{equation*}
The condition $J=J^{\mathrm{t}}$ is a further $\mathbb{Z}_{2}$ symmetry of the ASD connection. 

This coincidence between Einstein equations and Yang's equation is remarkable, since, although we started from a curved-space problem, by the correspondence it can essentially be regarded as a problem on Minkowski space reduced by a time translation and a rotation, hence on flat space.

Considering the case $n-2=s=2$, Yang's matrix $J$ can be written as
\begin{equation} \label{eq:backldec4d}
J=\left(\begin{array}{cc}
f{\alpha}^2-r^2f^{-1} & -f{\alpha} \\
-f{\alpha} & f
\end{array}\right),
\end{equation}
where $f$ and ${\alpha}$ are functions of $x$ and $r$. It can be read off that the metric takes the form
\begin{equation*}
\drm s^{2} = f(\drm t-{\alpha} \,\drm {\theta})^{2}-f^{-1}r^{2}\,\drm{\theta}^{2}-{\erm}^{2ν}(\drm r^{2}+\drm x^{2}).
\end{equation*}
If $f$ and ${\alpha}$ are real for real $x$ and $r$ this is known as the stationary axisymmetric gravitational field written in canonical Weyl coordinates.

\section[Solution Generation]{Solution Generation by the Symmetries of Yang's Equation}

A tedious but straightforward calculation shows that in canonical Weyl coordinates reduced Yang's equation~\eqref{eq:redyang} becomes
\begin{equation} \label{eq:weylequ}
\begin{split}
r^{2}\nabla ^{2}\log f+(f\partial _{r}{\alpha})^{2}+(f\partial _{x}{\alpha})^{2} & = 0,\\
\partial _{r}(r^{-1}f^{2}\partial _{r}{\alpha})+\partial _{x}(r^{-1}f^{2}\partial _{x}{\alpha}) & = 0,
\end{split}
\end{equation}
where $\nabla ^{2}=r^{-1}\partial _{r}(r\partial _{r})+\partial _{x}^{2}$ for real $r$ is the axisymmetric form of the three-dimensional Laplacian in cylindrical polar coordinates. The second equation is an integrability condition for ${\psi}$ with 
\begin{equation*}
\partial _{x}{\psi}=-r^{-1}f^{2}\partial _{r}{\alpha},\quad \partial _{r}{\psi}=-r^{-1}f^{2}\partial _{x}{\alpha},
\end{equation*}
or equivalently
\begin{equation*}
r\partial _{x}{\psi}+f^{2}\partial _{r}{\alpha}=0, \quad r \partial _{r}{\psi}-f^{2}\partial _{x}{\alpha}=0. 
\end{equation*}
If we consider the matrix
\begin{equation} \label{eq:ernstpot4d}
J'=\frac{1}{f} \left(\begin{array}{cc}
{\psi}^2+f^2 & {\psi} \\
{\psi} & 1
\end{array}\right)
\end{equation}
instead of $J$,\footnote{Since ${\psi}$ is only determined up to constant $J'$ is only determined up to
\begin{equation*}
J'\mapsto \left(\begin{array}{cc}
1 & 0 \\
{\gamma} & 1
\end{array}\right)
J'
\left(\begin{array}{cc}
1 & {\gamma} \\
0 & 1
\end{array}\right)
\end{equation*}
for a constant ${\gamma}$.} we find that \eqref{eq:redyang} for $J'$ again comes down to \eqref{eq:weylequ}, but with ${\alpha}$ replaced by ${\psi}$ as one of the variables. Solutions to Einstein's vacuum equations can now be obtained by solving \eqref{eq:redyang} for $J'$ subject to the conditions $\det J'=1$, $J'=J'^{\mathrm{t}}$. In this context \eqref{eq:redyang} is called \textit{Ernst equation} and the complex function $\Ecal=f+\irm {\psi}$ is the \textit{Ernst potential} \cite{Ernst:1968aa}, which is often taken as the basic variable in the analysis of stationary axisymmetric fields. We will also refer to $J'$ as the Ernst potential.

Note that \eqref{eq:redyang} has the obvious symmetry $J\mapsto A^{\mathrm{t}}JA$, where $A\in \SL(2,\mathbb{C})$ is constant.\footnote{The requirement $\det A=1$ is necessary to preserve the constraint $\det J'=1$.} This corresponds to the linear transformation 
{\arraycolsep=2pt
\begin{equation} \label{eq:killtrf}
(\begin{array}{cc}X & Y\end{array}) \mapsto  (\begin{array}{cc}X & Y\end{array}) A
\end{equation}
}of Killing vectors in the original space-time. Yet, the construction of $J'$ is not covariant with respect to general linear transformations in the space of Killing vectors, that is for \eqref{eq:killtrf} we do not have $J'\mapsto A^{\mathrm{t}}J'A$.
Given a solution of the Einstein vacuum equations in form of $J'$, this leads to a method of generating new solutions. First, recover the corresponding $J$ by solving for ${\alpha}$ in terms of ${\psi}$ and $f$. Then replace $J$ by $C^{\mathrm{t}}JC$, $C\in \SL(2,\mathbb{C})$, and construct $J'$ from the new $J$. Again replace $J'$ by $D^{\mathrm{t}}J'D$, $D\in \SL(2,\mathbb{C})$, and so on. This produces an infinite-parameter family of solutions to the Einstein vacuum equations starting from one original seed (if $C$ and $D$ are real, the transformations preserve the reality, stationarity and axisymmetry of the solutions as well). 
\chapter{Black Holes and Rod Structure} \label{ch:bh}

A remarkable consequence of Einstein's theory of gravitation is that under certain circumstances an astronomical object cannot exist in an equilibrium state and hence must undergo a gravitational collapse. The result is a space-time in which there is a ``region of no escape'' --- a black hole. Black hole space-times are of interest not only in four but also in higher dimensions. 

This chapter will provide basic knowledge in general relativity that is needed later on. Yet, proofs or further details are omitted at many points as this would be beyond the scope of our discussion. Moreover, we will see what questions arise in higher-dimensional black hole space-times and in the next chapter we will describe how twistor theory can be useful in tackling these problems. 

In the following a space-time will be a real time-orientable and space-orientable\footnote{In fact, one can argue that time-orientability implies space-orientability \cite[Sec.~6.1]{Hawking:1973aa}.} Lorentzian manifold, where not stated differently. 

\section{Relevant Facts on Black Holes} \label{sec:bhfacts}

The following definitions are taken from \citet{Wald:1984rz} and \citet[Sec.~2]{Chrusciel:2008aa}. Denote by ${\Jcal}^{\pm}(m)$ the causal future or past of a space-time point $m$. An important concept is asymptotic flatness which roughly speaking means that the gravitational field and matter fields (if present) become negligible in magnitude at large distance from the origin. More precisely we say an $n$-dimensional space-time $(M,g)$ is \textit{asymptotically flat and stationary} if $M$ contains a spacelike hypersurface $S_{\mathrm{ext}}$ diffeomorphic to $\mathbb{R}^{n-1}\backslash B(R)$, where $B(R)$ is an open coordinate ball of radius $R$, that is contained in a hypersurface satisfying the requirements of the positive energy theorem and with the following properties. There exists a complete Killing vector field ${\xi}$ which is timelike on $S_{\mathrm{ext}}$ (stationarity) and there exists a constant ${\alpha} > 0$ such that, in local coordinates on $S_{\mathrm{ext}}$ obtained from $\mathbb{R}^{n-1}\backslash B(R)$, the metric ${\gamma}$ induced by $g$ on $S_{\mathrm{ext}}$, and the extrinsic curvature tensor $K_{ij}$ of $S_{\mathrm{ext}}$, satisfy the fall-off conditions
\begin{equation*}
{\gamma}_{ij} - {\delta}_{ij} = O_{k}(r^{-{\alpha}}),\quad K_{ij} = O_{k-1}(r^{-1-{\alpha}})
\end{equation*}
for some $k >1$, where we write $f = O_{k}(r^{{\alpha}})$ if $f$ satisfies 
\begin{equation*}
\partial _{k_{1}} \cdots \partial _{k_{l}}f=O(r^{{\alpha}-l}),\quad	0\leq l\leq k.
\end{equation*}
In \cite[Sec.~2.1]{Chrusciel:2008aa} it is argued that these assumptions together with the vacuum field equations imply that the full metric asymptotes that of Minkowski space. 

The Killing vector ${\xi}$ models a ``time translation symmetry''. The one-parameter group of diffeomorphisms generated by ${\xi}$ is denoted by $\phi_{t}:M\to M$. Let $M_{\mathrm{ext}}=\bigcup_{t}\phi_{t}(S_{\mathrm{ext}})$ be the \textit{exterior region}, then the \textit{domain of outer communications} is defined as
\begin{equation*}
\langle\langle M_{\mathrm{ext}} \rangle \rangle = {\Jcal}^{+}(M_{\mathrm{ext}})\cap {\Jcal}^{-}(M_{\mathrm{ext}}).
\end{equation*}
We call $B= M \backslash {\Jcal}^{-}(M_{\mathrm{ext}})$ the \textit{black hole region} and its boundary $\Hcal^{+}=\partial B$ the \textit{black hole event horizon}.\footnote{This definition does in fact not depend on the choice of $S_{\mathrm{ext}}$ \cite[Sec.~2.2]{Chrusciel:2008aa}.} Analogously, the \textit{white hole} and the \textit{white hole event horizon} are $W= M \backslash {\Jcal}^{+}(M_{\mathrm{ext}})$ and $\Hcal^{-}=\partial W$. So, if the space time contains a black hole, that means it is not (entirely) contained in the causal past of its exterior region. $\Hcal=\Hcal^{+}\cup \Hcal^{-}$ is a null hypersurface generated by (inextendible) null geodesics \cite[p.~312]{Hawking:1973aa}. 

The second type of symmetry we will impose models a rotational symmetry. We say a space-time admits an \textit{axisymmetry} if there exists a complete spacelike Killing vector field with periodic orbits. This is a $\U(1)$-symmetry so we imagine it as a rotation around a codimension-2 hypersurface. As shown in \citet[Sec.~3.1]{Myers:1986aa} and \citet[Sec.~1.1]{Emparan:2008aa}, note that for $\dim M>4$ there is the possibility to rotate around multiple independent planes (if $M$ is asymptotically flat). For spatial dimension $n-1$ we can group the coordinates in pairs $(x_{1},x_{2}), (x_{3},x_{4}), \ldots$ where each pair defines a plane for which polar coordinates $(r_{1},\varphi_{1}), (r_{2},\varphi_{2}), \ldots$ can be chosen. Thus there are $N=\lfloor \frac{n-1}{2}\rfloor$ independent (commuting) rotations each associated to an angular momentum. An $n$-dimensional space-time $M$ will be called \textit{stationary and axisymmetric} if it admits $n-3$ of the above $\U(1)$ axisymmetries in addition to the timelike symmetry. However, note that this yields an important limitation. For globally asymptotically flat space-times we have by definition an asymptotic factor of $S^{n-2}$ in the spatial geometry, and $S^{n-2}$ has isometry group $\operatorname{O}(n-1)$. The orthogonal group $\operatorname{O}(n-1)$ in turn has Cartan subgroup $\U(1)^{N}$ with $N=\lfloor \frac{n-1}{2}\rfloor$, that is there cannot be more than $N$ commuting rotations. But each of our rotational symmetries must asymptotically approach an element of $\operatorname{O}(n-1)$ so that $\U(1)^{n-3}\subseteq\U(1)^{N}$, and hence
\begin{equation*}
n-3\leq N=\left\lfloor \frac{n-1}{2}\right\rfloor,
\end{equation*}
which is only possible for $n=4,5$. Therefore, stationary and axisymmetric solutions in our sense can only have a globally asymptotically flat end in dimension four and five. However, much of the theory, for example the considerations in Chapter~\ref{ch:yang}, is applicable in any dimension greater than four so that henceforth we still consider stationary and axisymmetric space-times and it will be explicitly mentioned if extra care is necessary.

Note for example that we have to change the above definitions of a black hole and stationarity in dimensions greater or equal than 6 because they require asymptotical flatness. Instead of asking $S_{\mathrm{ext}}$ to be diffeomorphic to $\mathbb{R}^{n}\backslash B(R)$ we set the condition that $S_{\mathrm{ext}}$ is diffeomorphic to $\mathbb{R}^{n}\backslash B(R)\times N$, where $B(R)$ is again an open coordinate ball of radius $R$ and $N$ is a compact manifold with the relevant dimension. This is also called asymptotic Kaluza-Klein behaviour. The fall-off conditions then have to hold for the product metric on $\mathbb{R}^{n}\backslash B(R)\times N$. The definition of stationarity is verbatim the same only with other asymptotic behaviour of $S_{\mathrm{ext}}$.

The time translation isometry must leave the event horizon of the black hole region invariant as it is completely determined by the metric which is invariant under the action of ${\xi}$. Hence ${\xi}$ must lie tangent to the horizon, but a tangent vector to a hypersurface with null normal vector must be null or spacelike. So, ${\xi}$ must be null or spacelike everywhere on the horizon. This motivates the following definition. A null hypersurface $\Ncal$ is called a \textit{Killing horizon} of a Killing vector ${\xi}$ if (on $\Ncal$) ${\xi}$ is normal to $\Ncal$, that is ${\xi}$ is in particular null on $\Ncal$. The integral curves of the normal vector of a null hypersurface are in fact geodesics and these geodesics generate the null hypersurface.\footnote{This and some of the following statements can be found in standard textbooks or lecture notes, for example \cite{Townsend:1997aa}.} If $l$ is normal to $\Ncal$ such that $\nabla _{l}l=0$ (affinely parameterized), then ${\xi}=f\cdot l$ and
\begin{equation*}
\nabla _{{\xi}}{\xi}=f\nabla _{l}(fl)={\xi}(f)\cdot l+f^{2}\nabla _{l}l=f^{-1}{\xi}(f)\cdot {\xi}={\kappa}\cdot {\xi},
\end{equation*}
where ${\kappa}={\xi}(\ln|f|)$ is the so-called \textit{surface gravity}. It can be shown that
\begin{equation*}
{\kappa}^{2}=-\frac{1}{2} \left.(\nabla ^{{\mu}}{\xi}^{{\nu}})(\nabla _{{\mu}}{\xi}_{{\nu}}) \right|_{\Ncal}\ ,
\end{equation*}
and that ${\kappa}$ is constant on the horizon. Thus, the following definition makes sense: A Killing horizon is called \textit{non-degenerate} if ${\kappa}\neq 0$, and \textit{degenerate} otherwise. Henceforth, we will restrict our attention to non-degenerate horizons, since it is a necessary assumption in the upcoming theorems and degenerate horizons may have a somewhat pathological behaviour.

The following theorem links the two notions of event and Killing horizon.

\begin{thm}[Theorem~1 in \citet{Hollands:2007aa}] \label{thm:eviskill}
Let $(M,g_{ab})$ be an analytic, asymptotically flat $n$-dimensional solution of the vacuum Einstein equations containing a black hole and possessing a Killing field $ξ$ with complete orbits which are timelike near infinity. Assume that a connected component, $\Hcal$, of the event horizon of the black hole is analytic and is topologically $R×Σ$, with $Σ$ compact, and that $κ≠0$. Then there exists a Killing field $K$, defined in a region that covers $\Hcal$ and the entire domain of outer communication, such that $K$ is normal to the horizon and $K$ commutes with $ξ$.
\end{thm}

This theorem essentially states that the event horizon of a black hole that has settled down is a Killing horizon (not necessarily for $∂_{t}$). In four dimensions this theorem has already been proven earlier \cite{Hawking:1973aa}. In four dimensions Theorem~\ref{thm:eviskill} is an essential tool in the proof of the following uniqueness theorem, which is the result of a string of papers \cite{Israel:1968aa, Carter:1971aa, Hawking:1972aa, Robinson:1975aa}. 

\begin{thm} \label{thm:carter}
Let $(M,g)$ be a non-degenerate, connected, analytic, asymptotically flat, stationary and axisymmetric four-dimensional space-time that is non-singular on and outside an event horizon, then $(M,g)$ is a member of the two-parameter Kerr family. The parameters are mass $M$ and angular momentum $L$.
\end{thm}

Since stationarity can be interpreted as having reached an equilibrium, this says that in four dimensions the final state of a gravitational collapse leading to a black hole is uniquely determined by its mass and angular momentum. The assumption of axisymmetry is actually not necessary in four dimensions, as it follows from the higher-dimensional rigidity theorem.

\begin{thm} \label{thm:rig}
Let $(M,g_{ab})$ be an analytic, asymptotically flat $n$-dimensional solution of the vacuum Einstein equations containing a black hole and possessing a Killing field $ξ$ with complete orbits which are timelike near infinity. Assume that a connected component, $\Hcal$, of the event horizon of the black hole is analytic and is topologically $R×Σ$, with $Σ$ compact, and that $κ≠0$.
\begin{enumerate}
\item If ${\xi}$ is tangent to the null generators of $\Hcal$, then the space-time must be static; in this case it is actually unique (for $n=4$ the Schwarzschild solution). \cite{Sudarsky:1992aa}
\item \label{item:rig} If ${\xi}$ is not tangent to the null generators of $\Hcal$, then there exist $N$, $N\geq 1$, additional linear independent Killing vectors  that commute mutually and with ${\xi}$. These Killing vector fields generate periodic commuting flows, and there exists a linear combination
\begin{equation*}
K={\xi}+{\Omega}_{1}X_{1}+\dotsc +{\Omega}_{N}X_{N}, \quad {\Omega}_{i}\in \mathbb{R},
\end{equation*}
so that the Killing field $K$ is tangent and normal to the null generators of the horizon $\Hcal$, and $g(K,X_{i})=0$ on $\Hcal$. \cite[Thm.~2]{Hollands:2007aa}
\end{enumerate}
Thus, in case \eqref{item:rig} the space-time is axisymmetric with isometry group $\mathbb{R}\times \mathrm{U}(1)^{N}$.
\end{thm}

\begin{rem}
In dimension four the last theorem has already been shown in \cite{Hawking:1972aa, Hawking:1973aa}.\\ \mbox{} \hfill $\blacksquare$
\end{rem}

Interestingly, the higher-dimensional analogue of Theorem~\ref{thm:carter} is not true, that is there exist different stationary and axisymmetric vacuum solutions with the same mass and angular momenta. Examples for five dimensions are given in \citet{Hollands:2008fp}. They have topologically different horizons so there cannot exist a continuous parameter to link them. A useful tool for the study of solutions in five dimensions is the so-called rod structure. Before defining it let us set the basic assumptions about our space-time.

Henceforth we are going to assume that, if not mentioned differently, we are given a \textit{vacuum (non-degenerate black hole) space-time $(M,g)$ which is five-dimensional, globally hyperbolic, asymptotically flat, stationary and axisymmetric, and that is analytic up to and including the boundary $r=0$.\footnote{The coordinate $r$ is defined in Chapter~\ref{ch:yang}. Further statements on that are below in the following section.} We are not considering space-times where there are points with a discrete isotropy group.} Note that the stationarity and axisymmetry implies orthogonal transitivity, which was necessary for the construction in Chapter~\ref{ch:yang} (see Appendix~\ref{app:orthtrans}). The assumption about analyticity might seem unsatisfactory, but in this thesis we are going to focus on concepts concerning the uniqueness of five-dimensional black holes rather than regularity. 

\section{Rod Structure} \label{sec:rodstr}

Remember the ${\sigma}$-model construction in Chapter~\ref{ch:yang}. Using the $(r,x)$-coordinates from there we define as in \citet[Sec.~III.B.1]{Harmark:2004rm}.
\begin{Def} \label{def:rodstr}
A \textit{rod structure} is a subdivision of the $x$-axis into a finite number of intervals where to each interval a constant three-vector (up to a non-zero multiplicative factor) is assigned. The intervals are referred to as \textit{rods}, the vectors as \textit{rod vectors} and the finite number of points defining the subdivision as \textit{nuts}.
\end{Def}
In order to assign a rod structure to a given space-time we quote the following proposition. 
\begin{prop}[Proposition 3 in \citet{Hollands:2008fp}]
Let $(M, g_{ab})$ be the exterior of a stationary, asymptotically flat, analytic, five-dimensional vacuum black hole space-time with connected horizon and isometry group $G=\U(1)^{2}\times \mathbb{R}$. Then the orbit space $\hat M = M / G$ is a simply connected 2-manifold with boundaries and corners. If $\skew{7}{\tilde}{A}$ denotes the matrix of inner products of the spatial (periodic) Killing vectors then furthermore, in the interior, on the one-dimensional boundary segments (except the segment corresponding to the horizon), and at the corners $\skew{7}{\tilde}{A}$ has rank 2, 1 or 0, respectively.
\end{prop}
Furthermore, since $\det \skew{7}{\tilde}{A} \neq  0$ in the interior of $\hat M$, the metric on the quotient space must be Riemannian. Then $\hat M$ is an (orientable) simply connected two-dimensional analytic manifold with boundaries and corners. The Riemann mapping theorem thus provides a map of $\hat M$ to the complex upper half plane where some further arguments show that the complex coordinate can be written as ${\zeta}=x+ \irm r$. So, starting with a space-time $(M,g)$ the line segments of the boundary $\partial  \hat M$ give a subdivision of the $x$-axis 
\begin{equation} \label{eq:sptrodsdef}
(-\infty , a_{1}),\ (a_{1},a_{2}), \dotsc  , (a_{N-1},a_{N}),\ (a_{N},\infty )
\end{equation}
as the boundary of the complex upper half plane. This subdivision is moreover unique up to translation $x\mapsto x+\text{const.}$ which can be concluded from the asymptotic behaviour \cite[Sec.~4]{Hollands:2008fp}. This subdivision is now our first ingredient for the rod structure assigned to $(M,g)$. For reasons which will become clear in Sections~\ref{sec:twistconstr} and \ref{sec:adapt} we take the set of nuts to be $\{a_{0}=∞,a_{1},…,a_{N}\}$.

As the remaining ingredient we need the rod vectors. The imposed constraint $-r^{2}=\det J(r,x)$ implies $\det J(0,x)=0$, and therefore
\begin{equation*}
\dim \ker J(0,x)\geq 1.
\end{equation*}
We will refer to the set $\{r=0\}$ as the \textit{axis}. Taking the subdivision \eqref{eq:sptrodsdef} we define the rod vector for a rod $(a_{i},a_{i+1})$ as the vector that spans $\ker J(0,x)$ for $x\in (a_{i},a_{i+1})$ (we will not distinguish between the vector and its $\mathbb{R}$-span). A few comments on that. 

First consider the horizon. From Theorem~\ref{thm:rig} we learn that $K={\xi}+{\Omega}_{1}X_{1}+{\Omega}_{2}X_{2}$~(I) is null on the horizon and $\left.g(K,X_{i})\right|_{\Hcal}=0$~(II).\footnote{It is not hard to see that if $\tilde X_{1}$, $\tilde X_{2}$ is a second pair of commuting Killing vectors which generate an action of $\U(1)^{2}$, then $\tilde X_{i}$ is related to $X_{i}$ by a constant matrix \cite[Eq. (9) and Sec.~4]{Hollands:2008fp}. Hence, we can without loss of generality assume that our periodic Killing vectors are the ones from Theorem~\ref{thm:rig}.} These conditions are equivalent to
\begin{align}
g_{ti} + \sum _{j}{\Omega}_{j}g_{ij} & = 0 \quad \text{on }\Hcal, \tag{II} \\
g_{tt} + 2 \sum _{i}{\Omega}_{i}g_{ti} +\sum _{i,j}{\Omega}_{i}{\Omega}_{j}g_{ij} & = 0 \quad \text{on }\Hcal, \tag{I}\\
\stackrel{\text{(II)}}{\Longleftrightarrow} g_{tt} + \sum _{i} {\Omega}_{i} g_{ti} & = 0 \quad \text{on }\Hcal. \notag
\end{align}
Hence
\begin{equation*}
\renewcommand{\arraystretch}{1.5}
J\tilde K = \left(\begin{array}{c}g_{tt} + \sum _{i} {\Omega}_{i} g_{ti} \\g_{ti} + \sum _{j}{\Omega}_{j}g_{ij}\end{array}\right) = 0  \quad \text{on }\Hcal,\ \text{where}\ \tilde K = \left(\begin{array}{c}1 \\{\Omega}_{1} \\{\Omega}_{2}\end{array}\right).
\end{equation*}
In other words $\tilde K$ is an eigenvector of $J$ on $\Hcal$. So, by the change of basis ${\xi}\mapsto K$, $X_{i}\mapsto X_{i}$ the first row and column of $J$ diagonalizes with vanishing eigenvalue towards $\Hcal$. On the other hand away from any of the rotational axes the axial symmetries $X_{1}$, $X_{2}$ are independent and non-zero, thus the rank of $J$ drops on the horizon precisely by one and the kernel is spanned by $\tilde K$. Note that if the horizon is connected precisely one rod in \eqref{eq:sptrodsdef} will correspond to $\Hcal$. 

Second, consider the rods which do not correspond to the horizon (assuming that $\Hcal$ is connected). Proposition 1 and the argument leading to Proposition 3 in \citet{Hollands:2008fp} show that on those rods the rotational Killing vectors are linearly dependent and the rank of $J$ again drops precisely by one. Whence, on each rod $(a_{i}, a_{i+1})$ that is not the horizon, there is a vanishing linear combination $aX_{1}+bX_{2}$. Therefore the vector $\left(\begin{array}{ccc}0 & a & b\end{array}\right)^{\mathrm{t}}$ spans the $\ker J(0,x)$, $x\in (a_{i}, a_{i+1})$. By \citet[Prop.~1]{Hollands:2008fp} $a$ and $b$ are constant so that we take $aX_{1}+bX_{2}$ as the rod vector on $(a_{i}, a_{i+1})$.

\begin{rem}
The fact that $a$ and $b$ are constant is not explicitly shown in the proof of \cite[Prop.~1]{Hollands:2008fp}, but follows quickly from \cite[Eq.~(11)]{Hollands:2008fp}. For $x$ being a point where $X_{1}$, $X_{2}$ are linearly dependent or where one (but not both) of them vanishes, and for $O_{x}$ being the orbit of $x$ under the action which is generated by $X_{1}$, $X_{2}$ we have
\begin{equation*}
0=a(\tilde x)X_{1}+b(\tilde x)X_{2}, \quad \tilde x ∈ O_{x}.
\end{equation*}
Since $X_{1}$ and $X_{2}$ commute it is
\begin{equation*}
0 = \Lcal_{X_{1}} (0) = \Lcal_{X_{1}}(a)X_{1}+\Lcal_{X_{1}}(b)X_{2}.
\end{equation*}
On the other hand $\Lcal_{X_{1}}(a)=\skew{2}{\dot}{a}$ is the derivative along the orbit (the orbit is one-dimensional), hence
\begin{equation*}
0 = b X_{2} \left(-\frac{\skew{2}{\dot}{a}}{a}+\frac{\dot b}{b}\right). 
\end{equation*}
Assuming that $b$ and $X_{2}$ do not vanish, this can be integrated and implies $\frac{a}{b}=\text{const.}$ so that without loss of generality both factors can be taken as constants. If $b$ or $X_{2}$ vanishes on the orbit, then one obtains immediately $\skew{2}{\dot}{a} = 0$. Note that if one of the Killing vectors vanishes somewhere on $O_{x}$ it vanishes everywhere on $O_{x}$, otherwise one could follow the integral curve where the Killing vector is non-zero up to the first point where it vanishes and there it stops, which is a contradiction to the periodicity. This shows that the two cases above are disjoint.\\ \mbox{} \hfill $\blacksquare$
\end{rem}

Note that the nuts of the rod structure are the points which correspond to the corners of $\hat M$ and that is where the rank of $J$ drops precisely by two. So, at those points $\dim \ker J = 2$. 

\begin{ex}[Rod Structure of Four-Dimensional Schwarzschild Space-Time, taken from Section~3.1 in \citet{Fletcher:1990aa}] \label{ex:schwarz}
The Schwarzschild solution in four dimensions has in usual coordinates the form
\begin{equation*}
\drm s^{2} = \left(1-\frac{2m}{R}\right) \drm T^{2}-\left(1-\frac{2m}{R}\right)^{-1} \drm R^{2} - R^{2}(\drm {\Theta}^{2}+ \sin^{2}{\Theta} \,\drm {\Phi}^{2}),
\end{equation*}
and is obtained in Weyl coordinates $(t,r,{\varphi},x)$ by replacing
\begin{equation*}
x=(R-m)\cos {\Theta}, \quad r=(R^{2}-2mR)^{\frac{1}{2}}\sin {\Theta},\quad t=T, \quad {\varphi}= {\Phi}.
\end{equation*}
If the symmetry group is generated by $X=\partial _{{\varphi}}$ and $Y=\partial _{t}$, then one can calculate the matrix of inner products of the Killing vectors (see Appendix~\ref{app:schwarz}) as
\begin{equation*}
\renewcommand\arraystretch{1.8}
J=\left(\begin{array}{cc}-\dfrac{r^{2}}{f} & 0 \\ \hphantom{-}0 & f\end{array}\right)
\end{equation*}
where
\begin{equation*}
f=\frac{r_{+}+r_{-}-2m}{r_{+}+r_{-}+2m}\quad  \text{ with }\quad  r^{2}_{\pm}=r^{2}+(x\pm m)^{2}.
\end{equation*}
Note that $r_{+}=|x+m|$, $r_{-}=|x-m|$ for $r=0$ so that for $-m\leq x\leq m$ and $r=0$ we have $r_{+}=x+m$, $r_{-}=m-x$. Hence, $f$ vanishes for $r=0$, $-m\leq x\leq m$. Yet, applying l'H\^opital's rule twice shows that $\frac{r^{2}}{f}$ does not vanish for $r=0$, $-m<x<m$. So, the rod structure can be read off. It consists of the subdivision of the $x$-axis into $(-\infty ,-m)$, $(-m,+m)$ and $(+m,+\infty )$ and the rod vectors as in Figure~\ref{fig:SchWrodstr}. The semi-infinite rods correspond to the rotation axis and the finite one to the horizon. At $\{r=0, x=\pm m\}$ the entry $\frac{r^{2}}{f^{\hphantom{2}}}$ blows up. Furthermore, we see that the boundary values of the rods are related to the mass of the black hole.\\ \mbox{} \hfill $\blacksquare$
\begin{figure}[htbp]
\begin{center}
     \scalebox{0.8}{\input{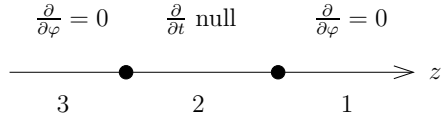}}
     \caption{Rod structure of the four-dimensional Schwarzschild solution. The numbers are only for the ease of reference to the parts of the axis later on.} 
     \label{fig:SchWrodstr}
\end{center}
\end{figure}
\end{ex}

There is a better way of visualizing the topology associated with the rod structure in five dimensions (from private communication with Piotr Chru\'{s}ciel). First consider five-dimensional Minkowski space. We leave the time coordinate and only focus on the spatial part. It is Riemannian and has dimension four, thus we can write it in double polar coordinates $(r_{1},{\varphi}_{1},r_{2},{\varphi}_{2})$. Then the first quadrant in Figure ~\ref{fig:Minkrodstr}, that is $\{r_{1}\geq 0,r_{2}\geq 0\}$, corresponds to the space-time.
\begin{figure}[htbp]
\begin{center}
     \scalebox{0.6}{\input{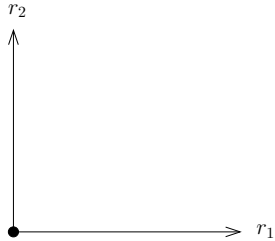}}
     \caption{Rod Structure for five-dimensional Minkowski space.} 
     \label{fig:Minkrodstr}
\end{center}
\end{figure}
The diagram suppresses the angles, so that each point in $\{r_{1}\geq 0,r_{2}\geq 0\}$ represents $S^{1}\times S^{1}$ where the radius of the corresponding circle is $r_{i}$. On the axes it thus degenerates to $\{\text{pt}\}\times S^{1}$. The boundary of our space-time, $r=0$, is in these polar coordinates $\{r_{1}=0\}\cup \{r_{2}=0\}$, and the nut is at the origin $r_{1}=r_{2}=0$ (see also Section~\ref{sec:exMinkrodstr}).

Since our interest lies in asymptotically flat space-times, the rod structures for other space-times will be obtained from this one by modifying its interior and leaving the asymptotes unchanged. For example we can cut out a quarter of the unit disc as in Figure~\ref{fig:MProdstr}.
\begin{figure}[htbp]
\begin{center}
     \scalebox{0.6}{\input{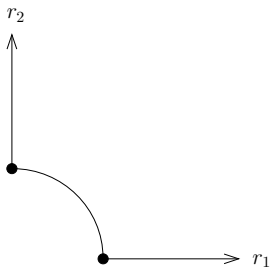}}
     \caption{Rod Structure with horizon topology $S^{3}$.} 
     \label{fig:MProdstr}
\end{center}
\end{figure}
But cutting out the quarter of the unit disc is nothing else than cutting out $r_{1}^{2}+r_{2}^{2}\leq 1$ (obviously taking the radius not to be one does not make any difference for the topology). Therefore the middle rod is the boundary of a region with topology $S^{3}$. So, if this is the horizon of a black hole, then the black hole has horizon topology $S^{3}$. Finally look at Figure~\ref{fig:BRrodstr1}.
\begin{figure}[htbp]
\begin{center}
     \scalebox{0.6}{\input{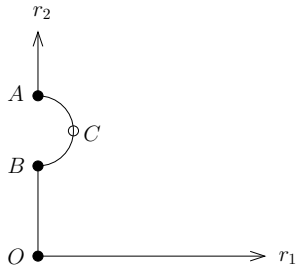}}
     \caption{Rod Structure with horizon topology $S^{2}\times S^{1}$. The nuts are at $A$, $B$ and $O$, where the rod between $A$ and $B$ corresponds to the horizon.} 
     \label{fig:BRrodstr1}
\end{center}
\end{figure}
This rod structure has three nuts:\;at $A$, at $B$ and at the origin $O$, that is at $r_{1}=r_{2}=0$. If the rod limited by $A$ and $B$ represents the horizon then the horizon topology is $S^{2}\times S^{1}$, which can be seen by rotating Figure~\ref{fig:BRrodstr1} first about the vertical and then about the horizontal axis. Another visualization is depicted in Figure~\ref{fig:BRrodstr2}, 
\begin{figure}[htbp]
\begin{center}
     \scalebox{0.6}{\input{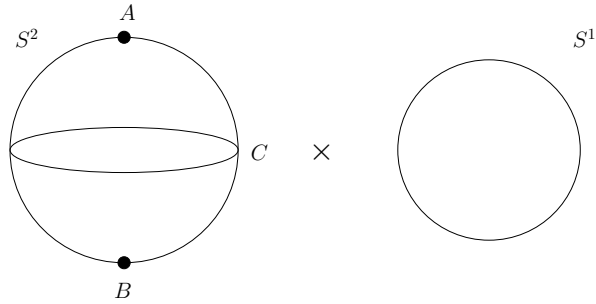}}
     \caption{Visualization of $S^{2}\times S^{1}$ topology.} 
     \label{fig:BRrodstr2}
\end{center}
\end{figure}
where the labelled points correspond to
\begin{align*}
A&:\ \{\text{pt}\}\times S^{1} \in  \mathbb{R}^{2}\times \mathbb{R}^{2},\ B:\ \{\text{pt}\}\times S^{1} \in  \mathbb{R}^{2}\times \mathbb{R}^{2}, \ O:\ \{\text{pt}\}\times \{\text{pt}\} \in  \mathbb{R}^{2}\times \mathbb{R}^{2},\\ \ C&:\ \hphantom{\{}S^{1} \hspace{0.1cm}\times S^{1} \in  \mathbb{R}^{2}\times \mathbb{R}^{2}.
\end{align*}
Interpolating the transition between those points explains the topology as well.

The relation between those more geometrical diagrams and the above definition of rod structure, that is only the $z$-axis with the nuts, can be made by the Riemann mapping theorem \cite[Sec.~4]{Hollands:2008fp}.

The rod structure is essential for the characterization of stationary axisymmetric black hole solutions. As mentioned above, such a solution is in five dimensions no longer uniquely given by its mass and angular momenta. But two solutions with connected horizon are isometric if their mass, angular momenta and rod structures coincide, if the exterior of the space-time contains no points with discrete isotropy group \cite{Hollands:2008fp} (see also Theorem~\ref{thm:holluniqueness}). An important question for the classification of five-dimensional black holes is which rod structures are admissible. Are there any restrictions on the possible configurations?
\chapter{Bundles over Reduced Twistor Space} \label{ch:bundles}

We have seen earlier that by the Penrose-Ward transform an ASDYM field corresponds to a holomorphic vector bundle over twistor space. However, we are not interested in the full ASDYM field, but in its reduction through symmetries. This chapter will provide a construction of bundles representing solutions which are invariant under a subgroup of the conformal group, and explain how we obtain in this way a twistor characterization of solutions of the Einstein equations using \citet[Ch.~11]{Mason:1996hl} and \citet{Fletcher:1990db}. 

\section{Reduced Twistor Space}

Proper conformal transformations of compactified $\Cbb\Mbb$ map ${\alpha}$-planes to ${\alpha}$-planes, and induce holomorphic motions of the twistor space which coincide with those of the natural action of $\GL(4,\mathbb{C})$ on $\mathbb{C}\Pbb^{3}$, see Chapter~\ref{ch:twistorspace} and \cite[Sec.~9.2]{Mason:1996hl}. Given an ASDYM field that is invariant under a group of conformal symmetries, by the Penrose-Ward transform the bundle over the twistor space belonging to the ASDYM field is invariant under the subgroup of $\GL(4,\mathbb{C})$ that corresponds to the group of conformal symmetries. 

This can be made more precise as follows. Let
\begin{equation*}
X = X^{a}\partial _{a}=a\,\partial _{w \vphantom{\tilde w}}+b\,\partial _{z \vphantom{\tilde w}}+\tilde a\,\partial _{\tilde w}+\tilde b\,\partial _{\tilde z}
\end{equation*}
be a conformal Killing vector on $U\subset \mathbb{C}\Mbb$ open, $\Pcal$ the twistor space and $\Fcal$ the correspondence space associated to $U$,
\begin{equation*}
\begin{xy}
  \xymatrix@C=0.7cm@R=0.7cm{
      & \Fcal \ar[ld]_q \ar[rd]^p		&  \\
      U	 	&	&	\Pcal     
  }
\end{xy}
\end{equation*}
The idea is to lift $X$ from $U$ to $\Fcal$ and project into $\Pcal$. Since $X$ is a conformal Killing vector, its flow preserves the metric up to scaling. Hence, $\partial _{(c}X_{d)}$ is proportional to the metric tensor, that is in the usual double-null coordinates
\begin{equation*}
(g_{ab})=\left(\begin{array}{cccc}\hphantom{-}0 & 0 & -1 & 0 \\\hphantom{-}0 & 0 & 0 & 1 \\-1 & 0 & 0 & 0 \\\hphantom{-}0 & 1 & 0 & 0\end{array}\right)
\end{equation*}
the conditions
\begin{equation} \label{eq:condconfkil}
\begin{split}
\partial _{(w \vphantom{\tilde w}}X_{\tilde w \vphantom{\tilde w})}=-\partial _{(z \vphantom{\tilde w}}X_{\tilde z)} \quad  & \Leftrightarrow  \quad \partial _{w \vphantom{\tilde w}}a+\partial _{\tilde w} \tilde a=\partial _{z \vphantom{\tilde w}}b+\partial _{\tilde z} \tilde b,\\
\partial _{(\tilde w}X_{z \vphantom{\tilde w})} = 0 \quad  & \Leftrightarrow  \quad \partial _{z \vphantom{\tilde w}}a=\partial _{\tilde w} \tilde b,\\
\partial _{(\tilde w}X_{\tilde w)}=\partial _{(\tilde z}X_{\tilde z)}=0 \quad  & \Leftrightarrow  \quad \partial _{\tilde w}a=\partial _{\tilde z}b=0,\\
\partial _{(\tilde w}X_{\tilde z)}=0 \quad  & \Leftrightarrow  \quad \partial _{\tilde z}a=\partial _{\tilde w}b,
\end{split}
\end{equation}
together with the same equations but tilded and untilded variables interchanged.

Now we are looking for a vector field $X'$ that acts on ${\alpha}$-planes and is induced by the vector field $X$ on $U$. So, we have to define a vector field $X''$ on $\Fcal$ such that
\begin{equation} \label{eq:projkil}
q_{*}X''=X,
\end{equation}
and that preserves the distribution spanned by $l$, $m$, that is
\begin{equation*}
[X'',l]=0,\quad [X'',m]=0
\end{equation*}
modulo combinations of $l$ and $m$. From \eqref{eq:projkil} we see that $X''$ has to be of the form
\begin{equation} \label{eq:liftedkil}
X''=a\partial _{w \vphantom{\tilde w}}+b\partial _{z \vphantom{\tilde w}}+\tilde a\partial _{\tilde w}+\tilde b\partial _{\tilde z}+Q\partial _{{\zeta} \vphantom{\tilde w}}.
\end{equation}
For fixed ${\zeta}$ the conditions~\eqref{eq:condconfkil} imply
\begin{align*}
[X,l] & =-lX={\zeta}(a_{\tilde z}\partial _{w \vphantom{\tilde w}}+{b_{\tilde z}\partial _{z \vphantom{\tilde w}}}+{\tilde a}_{\tilde z}\partial _{\tilde w}+{\tilde b}_{\tilde z}\partial _{\tilde z})\\
& \hspace{0.4cm}-(a_{w \vphantom{\tilde w}}\partial _{w \vphantom{\tilde w}}+b_{w \vphantom{\tilde w}}\partial _{z \vphantom{\tilde w}}+{{\tilde a}_{w \vphantom{\tilde w}}\partial _{\tilde w}}+{\tilde b}_{w \vphantom{\tilde w}}\partial _{\tilde z})\\
& = {\zeta}^{2}a_{\tilde z}\partial _{\tilde z}+{\zeta}a_{\tilde z}l-b_{w \vphantom{\tilde w}}m+{\zeta} {\tilde b}_{\tilde z}\partial _{\tilde z}-{\zeta}a_{w \vphantom{\tilde w}}\partial _{\tilde z}-a_{w \vphantom{\tilde w}}l-{\tilde b}_{w \vphantom{\tilde w}}\partial _{\tilde z}\\
& = Q \partial _{\tilde z} \mod (l,m)
\end{align*}
where 
\begin{equation} \label{eq:liftq}
Q={\zeta}^{2}a_{\tilde z}+{\zeta}({\tilde b}_{\tilde z}-a_{w \vphantom{\tilde w}})-{\tilde b}_{w}.
\end{equation}
Similarly,
\begin{equation*}
[X,m] = Q \partial _{\tilde w} \mod(l,m).
\end{equation*}
An easy calculation using \eqref{eq:condconfkil} again shows also $lQ=mQ=0$, hence $[X'',l]=[X'',m]=0$. Therefore, with $Q$ as in \eqref{eq:liftq} we define $X''$ to be the \textit{lift} of $X$ from $U$ to $\Fcal$. Its projection $X'=p_{*}X''$ is a well-defined holomorphic vector field on the twistor space, and the flow of $X'$ is the action of conformal motions generated by $X$ on ${\alpha}$-planes. 

Remember that ${\lambda}={\zeta}w+\tilde z$, ${\mu}={\zeta}z+\tilde w$ and ${\zeta}$ defined coordinates on the twistor space and using $w$, $z$, $\tilde w$, $\tilde z$, ${\zeta}$ as coordinates on $\Fcal$ the projection was given by
\begin{equation*}
p:(w,z,\tilde w,\tilde z,{\zeta}) \mapsto  ({\lambda},{\mu},{\zeta})=({\lambda}={\zeta}w+\tilde z,{\mu}={\zeta}z+\tilde w,{\zeta}).
\end{equation*}
Thus, 
\begin{equation*}
X'=({\zeta}a+\tilde b+wQ)\partial _{{\lambda}}+({\zeta}b+\tilde a+zQ)\partial _{{\mu}}+Q\partial _{{\zeta}}
\end{equation*}
where the components are constant on ${\alpha}$-planes, that is functions of ${\lambda}$, ${\mu}$, ${\zeta}$.

Let $U⊂\Cbb\Mbb$ be an open set satisfying the condition as in Theorem~\ref{thm:twicorr}, and $H$ be a subgroup of the conformal group. We have shown that the Lie algebra $\hfr$ gives rise to Killing vectors $X$ on $U$ and to holomorphic vector fields $X'$ on $\Pcal$. Now assume that $\hfr$ acts freely on $\Pcal$, then we define the \textit{reduced twistor space} $\Rcal$ as the quotient of $\Pcal$ over $\hfr$.\footnote{A technically more detailed definition can be found in \cite[Sec.~11.3]{Mason:1996hl}.} A vector bundle $B'\to \Pcal$, that is the Penrose-Ward transform of an ASDYM connection $\Drm$ on a vector bundle $B\to U$, is said to be \textit{invariant} under $\hfr$ if it is the pullback of an unconstrained vector bundle $E\to \Rcal$.

In our example of stationary and axisymmetric solutions the reduction was generated by the two commuting Killing vectors
\begin{equation*}
X=w\,\partial _{w \vphantom{\tilde w}}-\tilde w\,\partial _{\tilde w},\quad Y=\partial _{\tilde z}+\partial _{z \vphantom{\tilde w}}. 
\end{equation*}
The above construction yields
\begin{equation*}
Q_{X}=-{\zeta}a_{w \vphantom{\tilde w}}=-{\zeta}, \quad Q_{Y}=0.
\end{equation*}
This gives us
\begin{equation*}
\begin{array}{ll}
X''=w\partial _{w \vphantom{\tilde w}}-\tilde w\partial _{\tilde w}-{\zeta}\partial _{{\zeta}}, & Y''=\partial _{\tilde z}+\partial _{z \vphantom{\tilde w}},\\
X'=-{\mu}\partial _{{\mu}}-{\zeta}\partial _{{\zeta}} & Y'=\partial _{{\lambda}}+{\zeta}\partial _{{\mu}}.
\end{array}
\end{equation*}
With coordinates $t$, ${\theta}$, $x$, $r$ as in Chapter~\ref{ch:yang} and the gauge such that \eqref{eq:redpot1} holds,
\begin{equation*}
{\Phi}=-P\frac{\drm \tilde w}{\tilde w}+Q \,\drm \tilde z.\footnote{This is a different $Q$ as in the lift of the conformal Killing vector field.}
\end{equation*}
The pullback of local invariant sections of a bundle $E\to \Pcal$ to $\Fcal$ by $p$ are simultaneous solutions to
\begin{equation} \label{eq:locinvsec1}
\begin{split}
0 & = \Drm_{l}s=\partial _{l}s+{\Phi}(l)s=\partial _{w}s-{\zeta}(\partial _{\tilde z}+Q)s,\\
0 & = \Drm_{m}s=\partial _{m}s+{\Phi}(m)s=\partial _{z}s-{\zeta}(\partial _{\tilde w}-\tilde w^{-1}P)s,\\
0 & = X''(s), \quad 0=Y''(s)
\end{split}
\end{equation}
where $s$ is a function of ${\zeta}$ and the space-time coordinates. The first pair of equations is required by the constancy of $s$ on ${\alpha}$-planes, and the second pair is the symmetry condition.

Introducing the \textit{invariant spectral parameter} ${\sigma}={\zeta}\erm^{\irm {\theta}}$ the symmetry conditions can be stated as $s=s(x,r,{\sigma})$. A substitution of coordinates gives for the first pair of equations in \eqref{eq:locinvsec1} the form
\begin{equation} \label{eq:locinvsec2}
\begin{split}
(\partial _{r}-{\sigma}\partial _{x}+r^{-1}{\sigma}\partial _{{\sigma}})s-{\sigma}(J^{-1}J_{x})s & = 0,\\
(\partial _{x}+{\sigma}\partial _{r}-r^{-1}{\sigma}^{2}\partial _{{\sigma}})s+{\sigma}(J^{-1}J_{r})s & = 0,
\end{split}
\end{equation}
where $J(x,r)$ is defined by
\begin{equation*}
P=-rJ^{-1}J_{r}, \quad Q=J^{-1}J_{x}.\footnote{Change of scale compared to \eqref{eq:redpot1}.}
\end{equation*}
Equations~\eqref{eq:locinvsec2} is a linear system for the reduced form of Yang's equation that is integrable if and only if \eqref{eq:redyang} holds. 

Another useful parameter for this example is
\begin{equation} \label{eq:deftau}
{\tau}=x+\frac{1}{2}r({\sigma}-{\sigma}^{-1})=\frac{1}{2}({\lambda}-{\zeta}^{-1}{\mu}),
\end{equation}
which is constant along $l$ and $m$ as it depends only on ${\lambda}$, ${\mu}$, ${\zeta}$. Hence, ${\tau}$ is a function on the twistor space. 

The reduced twistor space has dimension $3-2=1$, thus there exists one invariant coordinate. Since
\begin{equation*}
X'{\tau} = 0, \quad Y'{\tau}=0,
\end{equation*}
${\tau}$ is constant along the orbits of $X'$ and $Y'$, so we can take it to be this coordinate. The pair of planes with ${\sigma}=0$ and ${\sigma}=\infty $ corresponds to ${\tau}=\infty $.

The section $s$ can now also be written in terms of $x$, $r$, ${\tau}$, and the reduced system becomes
\begin{align*}
(\partial _{r}-{\sigma}\partial _{x})s-{\sigma}(J^{-1}J_{x})s & = 0,\\
(\partial _{r}+{\sigma}\partial _{x})s+{\sigma}(J^{-1}J_{r})s & = 0,
\end{align*}
where ${\sigma}$ is a function of $x$, $r$, ${\tau}$ by \eqref{eq:deftau}, and $J$ satisfies the reduced form of Yang's equation if and only if this reduced system is integrable for every ${\tau}$. 

A point of $\Rcal$ is a leaf of the foliation of $\Pcal$ spanned by $X'$, $Y'$. This is a two-parameter family of ${\alpha}$-planes, each member being the orbit of one ${\alpha}$-plane under the flow of $X'$ and $Y'$. Points of $\Rcal$ are labelled by ${\tau}$. 

The reduced twistor space $\Rcal$ has a non-Hausdorff topology, which can be seen as follows. Let $U\subset \mathbb{C}\Mbb$ as in Theorem~\ref{thm:twicorr}, and $\Pcal$ the twistor space of $U$. Consider a leaf of $\Pcal$, that is a point of $\Rcal$, with constant ${\tau}$. This is the family of ${\alpha}$-planes two of which pass through a general point of $U$ with coordinates $t$, ${\theta}$, $x$, $r$, corresponding to the two roots of the quadratic equation
\begin{equation} \label{eq:quadleaf}
r{\sigma}^{2}+2(x-{\tau}){\sigma}-r=0
\end{equation}
for ${\sigma}$. There are two cases. If we can continuously change one root into the other by moving the point of $U$ but keeping ${\tau}$ fixed, then there is only one leaf in the foliation for this value of ${\tau}$, hence ${\tau}$ labels a single point in $\Rcal$. Otherwise ${\tau}$ labels two points. The discriminant of \eqref{eq:quadleaf} is $\left(\frac{x-{\tau}}{r}\right)^{2}+1$, so that there is only one solution to \eqref{eq:quadleaf} if ${\tau}=x\pm \irm r$. Thus, ${\tau}$ labels only one point of $U$ if ${\tau}=x \pm \irm r$ for some point of $U$, and two otherwise. For ${\tau}=\infty $ we will always get two points, the leaves on which ${\sigma}=0$ and ${\sigma}=\infty $. If $\Pcal$ is compact, then $\Rcal$ as a quotient of $\Pcal$ is compact, and from the range of the coordinate ${\tau}$ we see that it covers $\mathbb{C}\Pbb^{1}$. So, $\Rcal$ is a compact Riemannian surface covering $\mathbb{C}\Pbb^{1}$, but not Hausdorff, as for ${\tau}=x \pm \irm r$ with $x$, $r$ belonging to a point of the boundary of $U$ there are two points in $\Rcal$. They cannot be separated in the quotient topology since every neighbourhood contains a ${\tau}=x \pm \irm r$ for some point in $U$. 

Despite our complex manifold not being Hausdorff it still makes sense in our case to consider holomorphic vector bundles over it \cite[App.~1]{Woodhouse:1988ek}.

\section{The Twistor Construction} \label{sec:twistconstr}

The considerations below are based on \citet{Fletcher:1990db} which in turn go back to \citet{Woodhouse:1988ek}. An alternative reference is \citet[App.~B]{Klein:2005aa}.

To obtain the formulae in the following in accordance with most of the literature we change the notation such that \emph{the coordinate $-x$ is now named $z$, the parameter $-{\tau}$ will be $w$} as of now. Due to the clash with the standard notation for double-null coordinates this would have been confusing in the previous paragraphs. Furthermore, we will take ${\zeta}$ to be the spectral parameter instead of ${\sigma}$.

Another way of defining the reduced twistor space is to take the quotient of $U$ and $\Fcal$ by $\hfr$ so that we have ${\Sigma}=U/\hfr$ and $\Fcal_{\mathrm{r}}={\Sigma}\times \mathbb{C}\Pbb^{1}$. Then we obtain the reduced version of the double fibration
\begin{equation*}
\begin{xy}
  \xymatrix@C=0.7cm@R=0.7cm{
      & \Fcal_{\mathrm{r}} \ar[ld]_q \ar[rd]^p		&  \\
      {\Sigma}	 	&	&	\Rcal     
  }
\end{xy}
\end{equation*}
which can be used for a slightly more abstract approach. Henceforth, instead of beginning with an ASD vacuum space-time, we take a two-dimensional complex conformal manifold ${\Sigma}$ on which we are given a holomorphic solution $r$ of the Laplace equation. By $z$ we denote the harmonic conjugate of $r$. We define $\Fcal_{\mathrm{r}}={\Sigma}\times \Xcal$, where $\Xcal$ is a ${\zeta}$ Riemann sphere. The \textit{reduced twistor space} $\Rcal$ associated to ${\Sigma}$ and $r$ is constructed from $\Fcal_{\mathrm{r}}$ by identifying $(σ,{\zeta})$ and $(σ',{\zeta}')$ if they lie on the same connected component of one of the surfaces given by
\begin{equation} \label{eq:quadw}
r{\zeta}^{2}+2(w-z){\zeta}-r=0
\end{equation}
for some value of $w$ and where $z=z(σ)$, $r=r(σ)$.

We can use $w$ as a local holomorphic coordinate on $\Rcal$, which is a non-Hausdorff Riemann surface. Like above, $w$ corresponds to one point of $\Rcal$ if one can continuously change the roots of \eqref{eq:quadw} into each other by going on a path in ${\Sigma}$ and keeping $w$ fixed; and two points otherwise.\footnote{Note that the condition for $w$ to correspond only to one point is an open condition in $\Rcal$ as $z$ and $r$ are smooth functions on ${\Sigma}$.} Let $S$ be the $w$ Riemann sphere, and $V$ be the set of values for $w$ which correspond only to one point in $\Rcal$. Then $V\subset S$ is open, and if ${\Sigma}$ is simply connected, then
\begin{equation} \label{eq:eqV}
V=\{z({\sigma})+\irm r({\sigma}): {\sigma}\in {\Sigma}\}. 
\end{equation}
In general, $V$ is not connected. 

However, for axis-regular solutions (for the definition see below) $V$ can be enlarged so that it becomes a simply connected open set $V'\subset S$. This situation is depicted in Figure~\ref{fig:redtwis}.
\begin{figure}[htbp]
\begin{center}
     \scalebox{0.7}{\input{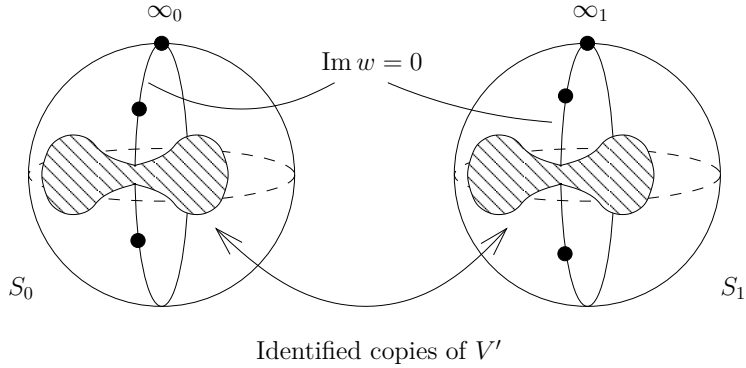}}
     \caption{Non-Hausdorff reduced twistor space with real poles (bullet points) as in the relevant examples.} 
     \label{fig:redtwis}
\end{center}
\end{figure}
For $w=\infty $ we will always have two points, corresponding to ${\zeta}=0$ and ${\zeta}=\infty $, whatever $x$ and $r$ are. Therefore, $w=\infty $ is never in $V$. We denote the points in $\Rcal$ corresponding to $w=\infty $ by $\infty _{0}$ (for ${\zeta}=0$) and $\infty _{1}$ (for ${\zeta}=\infty $).

Now we need to construct the reduced form of the Penrose-Ward transform. First the forward direction, where we are given $J$ as a solution of \eqref{eq:redyang}. In the unconstrained case the bundle over twistor space was defined by specifying what the fibres are. For an ${\alpha}$-plane $Z$, that is a point in the twistor space $\Pcal$, it was the space of covariantly constant sections over $Z$ in the given ASD bundle $B\to U$ (see end of Chapter~\ref{ch:pwtrf}). With our symmetry assumptions these are the sections defined by \eqref{eq:locinvsec1}. We have seen above that with our additional symmetry \eqref{eq:locinvsec1} comes down to \eqref{eq:locinvsec2}. Hence, the holomorphic bundle $E\to \Rcal$ is constructed by taking the fibre over a point of $\Rcal$ to be the space of solutions of \eqref{eq:locinvsec2} on the corresponding connected surface in $\Fcal_{\mathrm{r}}$. The integrability condition is precisely \eqref{eq:redyang}. 

Conversely, let $E\to \Rcal$ be a holomorphic rank-$n$ vector bundle together with a choice of frame in the fibres. For a fixed ${\sigma}\in {\Sigma}$ let ${\pi}:\Xcal\to \Rcal$ be the restricted projection of $\Fcal_{\mathrm{r}}\to \Rcal$ to $\{{\sigma}\}\times \Xcal$, that is the identification $({\sigma},{\zeta}) \sim ({\sigma}',{\zeta}')$ as above, and denote by ${\pi}^{*}(E)$ the pullback bundle of $E$ to $\Fcal_{\mathrm{r}}$. We have to assume that ${\pi}^{*}(E)$ is a trivial holomorphic bundle over $\Xcal$.\footnote{This is less restrictive than it seems since if it is satisfied at one point ${\sigma}$ then it holds in a neighbourhood of ${\sigma}$, compare comment on page~\pageref{fn:trivass}.} 

The matrix $J$ can again be recovered within $J\mapsto AJB$, where $A$ and $B$ are constant, by the splitting procedure as in Chapter~\ref{ch:pwtrf}, but adapted to the symmetry constraint. Suppose $E$ is  given by patching matrices $\{P_{{\alpha}{\beta}}(w)\}$ according to an open cover $\{\Rcal_{{\alpha}}\}$ of $\Rcal$ such that $\infty _{0}\in \Rcal_{0}$ and $\infty _{1}\in \Rcal_{1}$. Then ${\pi}^{*}(E)$ is given by patching matrices
\begin{equation} \label{eq:patmat1}
P_{{\alpha}{\beta}}(w({\sigma}))=P_{{\alpha}{\beta}}\left(\frac{1}{2}r({\sigma})({\zeta}^{-1}-{\zeta})+z({\sigma})\right)
\end{equation}
according to the open cover $\{{\pi}^{-1}(\Rcal_{{\alpha}})\}$ of $\Xcal$. The triviality assumption implies that there exist splitting matrices $f_{{\alpha}}({\zeta})$ such that
\begin{equation}\label{eq:patmat2}
P_{{\alpha}{\beta}}\left(\frac{1}{2}r({\sigma})({\zeta}^{-1}-{\zeta})+z({\sigma})\right)=f_{{\alpha}}^{\vphantom{-1}}({\zeta})f_{{\beta}}^{-1}({\zeta}).\footnote{See Proposition~\ref{prop:bundle}.}
\end{equation}
We define $J\coloneqq f_{0}(0)f_{1}(\infty )^{-1}$. Another splitting would be of the form $f_{{\alpha}}C$ for an invertible matrix $C$, which has to be holomorphic on the entire ${\zeta}$ Riemann sphere, thus $C$ is constant. But this leaves $J$ invariant and the definition is independent of the choice of splitting. The splitting matrices $f_{{\alpha}}$ depend smoothly on $r$, $z$ as ${\sigma}$ varies, so $J$ does. Although $J$ might have singularities where the triviality condition does not hold. 

Next we have to show that the so defined $J$ satisfies \eqref{eq:redyang}. Therefore, we exploit the freedom to choose the splitting matrices such that $f_{1}(\infty )=1$, thus $J=f_{0}(0)$. Let $Z_{1}=\partial _{r}+{\zeta}\partial _{z}+r^{-1}{\zeta}\partial _{{\zeta}}$ and $Z_{2}=-\partial _{z}+{\zeta}\partial _{r}-r^{-1}{\zeta}^{2}\partial _{{\zeta}}$ be the vector fields in \eqref{eq:locinvsec2}. An easy calculation shows $Z_{i}(w)=0$, $i=1,2$. Using this and acting with $Z_{i}$, $i=1,2$, on \eqref{eq:patmat2} yields for the left-hand side zero and hence for the right-hand side
\begin{equation*}
0=Z_{i}^{\vphantom{-1}}(f_{{\alpha}}^{\vphantom{-1}})f_{{\beta}}^{-1}+f_{{\alpha}}^{\vphantom{-1}}Z_{i}^{\vphantom{-1}}(f_{{\beta}}^{-1})=Z_{i}^{\vphantom{-1}}(f_{{\alpha}}^{\vphantom{-1}})f_{{\beta}}^{-1}-f_{{\alpha}}^{\vphantom{-1}}f_{{\beta}}^{-1}Z_{i}^{\vphantom{-1}}(f_{{\beta}}^{\vphantom{-1}})f_{{\beta}}^{-1},
\end{equation*}
which can be rearranged as
\begin{equation} \label{eq:Jsolproof1}
f_{{\alpha}}^{-1}Z_{i}^{\vphantom{-1}}(f_{{\alpha}}^{\vphantom{-1}})=f_{{\beta}}^{-1}Z_{i}^{\vphantom{-1}}(f_{{\beta}}^{\vphantom{-1}}).
\end{equation}
In \eqref{eq:Jsolproof1} the left-hand side is holomorphic on $\Rcal_{{\alpha}}$ and the right-hand side on $\Rcal_{{\beta}}$, so we obtain a function that is holomorphic on the entire ${\zeta}$ Riemann sphere, and therefore by Liouville's theorem both sides are independent of ${\zeta}$. Evaluate \eqref{eq:Jsolproof1} at ${\zeta}=0$, that is on $\Rcal_{0}$. This can be written as
\begin{equation*}
f_{0}^{-1}(0)Z_{i}^{\vphantom{-1}}(0)(f_{0}^{\vphantom{-1}}(0)) = 
\begin{cases}
J^{-1}\partial _{r}J,	& i=1\\
J^{-1}\partial _{x}J,	& i=2.
\end{cases}
\end{equation*}
But we can also write it as
\begin{equation*}
f_{0}^{\vphantom{-1}}(0)^{-1}Z_{i}^{\vphantom{-1}}(0)(f_{0}^{\vphantom{-1}}(0)) = -Z_{i}^{\vphantom{-1}}(f_{0}^{-1}(0))f_{0}^{\vphantom{-1}}(0).
\end{equation*}
Combining both and rearranging yields
\begin{equation*}
Z_{i}^{\vphantom{-1}}(f_{0}^{-1})+\left.(f_{0}^{-1}Z_{i}^{\vphantom{-1}}(f_{0}^{\vphantom{-1}}))f_{0}^{-1}\right|_{{\zeta}=0}=0,\quad i=1,2.
\end{equation*}
Therefore, $f_{0}^{-1}$ is a solution for \eqref{eq:locinvsec2} at ${\zeta}=0$ and then for all ${\zeta}$ since we have already seen that it is independent of ${\zeta}$. The existence of a simultaneous solution to both equations means they are integrable, thus $J$ satisfies \eqref{eq:redyang}.

Another freedom in the construction of $J$ from $E$ is the change of local trivialization $P_{{\alpha}{\beta}}\mapsto H_{{\alpha}}^{\vphantom{-1}}P_{{\alpha}{\beta}}^{\vphantom{-1}}H_{{\beta}}^{-1}$, where each $H_{{\alpha}}: \Rcal_{{\alpha}}\to \GL(n,\mathbb{C})$ is holomorphic. Then $J$ becomes $AJB$ with $A=H_{0}(\infty _{0})$ and $B=H_{1}(\infty _{1})$. This freedom of multiplying $J$ by constant matrices is the freedom of making linear transformations in the fibres of $E$ over $\infty _{0}$ and $\infty _{1}$. By construction $J$ is a linear map $E_{\infty _{1}}\to E_{\infty _{0}}$, but in the context of general relativity it was also the matrix of inner products of Killing vectors. Hence, $E_{\infty _{1}}$ has to be interpreted as the space of Killing vectors in space-time, and $E_{\infty _{0}}$ as its dual.

Furthermore, from the context of general relativity we have the requirement for $J$ to be real and symmetric, and the constraint $\det J=-r^{2}$. These conditions lead to further constraints on the bundle $E\to \Rcal$ which are in general still rather complicated \cite[Sec.~5]{Woodhouse:1988ek}. To obtain a simplification, we consider only axis-regular solutions.
\begin{Def}
An Ernst potential $J'$ is called \textit{axis-regular} if the corresponding bundle $E' \to  \Rcal_{V}$ satisfies $E' = η^{*}(\hat E)$ where $\hat E$ is a bundle over $\Rcal'=\Rcal_{V'}$ such that $\hat E|_{S_{0}}$ and $\hat E|_{S_{1}}$ are trivial.

Here $\Rcal'$ is a double cover of the $w$-Riemann sphere identified over the two copies of the set $V'$ where $V'$ is open, simply connected, invariant under $w \mapsto  \bar w$ and $V \subseteq  V'$ ($V$ as in \eqref{eq:eqV}). The map $η: \Rcal_{V}\to \Rcal_{V'}$ is the projection.

We shall also say that a metric $J$ is axis-regular if the corresponding Ernst potential $J'$ is.
\end{Def}
Roughly speaking this definition just says that $J'$ is axis-regular if we can enlarge the region where two spheres are identified to a simply connected patch such that this identification also extends to the fibres of $\hat E$. The exact shape of $V'$ is not important. Despite the further identifications we still have the projections $\Fcal_{\mathrm{r}}\to \Rcal'$ and can therefore construct $J$ from a holomorphic bundle $E'\to \Rcal'$ that satisfies the triviality condition. Choose the copies of the two Riemann spheres in $\Rcal'$ such that $\infty _{0}\in S_{0}$ and $\infty _{1}\in S_{1}$. 

If the bundle was not axis-regular, it meant that there are more than only isolated points where the two spheres cannot be identified. Thus, in the light of later results (Proposition~\ref{prop:AnalyCont}, Corollary~\ref{cor:singofP}) and \cite[App.~F]{Harmark:2004rm} axis-regularity is necessary for the the space-time not to have curvature singularities at $r=0$.

This completes our construction of a correspondence between general relativity and twistor theory. The following diagram depicts the established relation in a nutshell. 
\begin{equation*}
\hspace{-0.15cm}
\begin{xy}
  \xymatrix@!R@C=1.7cm{
      & \parbox{3.2cm}{\footnotesize{ASDYM connection for rank-$(n-2)$ vector bundle $B\to U$}} \ar@{<->}[r]^(0.615){\text{\tiny PW trf}}_(0.62){\text{\tiny Ch.~\ref{ch:pwtrf}}} \ar[d]^(0.54){\txt{\tiny Ch.~\ref{ch:bundles}}}_(0.545){\txt{\tiny symmetry \\ \tiny reduction}} &  \parbox[c][0.9cm][c]{1.2cm}{\footnotesize{$B'\to \Pcal$}} \ar[d]^(0.54){\txt{\tiny Ch.~\ref{ch:bundles}}}_(0.545){\txt{\tiny symmetry \\ \tiny reduction}} \\
      \parbox{2.8cm}{\footnotesize{stationary axi\-sym\-me\-tric space-time of di\-men\-sion $n$}} \ar@{<->}[r]^(0.48){\text{\tiny coincidence}}_(0.47){\text{\tiny Ch.~\ref{ch:yang}}}	 	& \makebox[3.4cm][c]{\footnotesize{reduced ASDYM}} \ar@{<->}[r]_(0.62){\text{\tiny Ch.~\ref{ch:bundles}}}	& \parbox[c][0.3cm][c]{1.4cm}{\footnotesize{\hspace{0.1cm} $E\to \Rcal$}}     
  }
\end{xy}
\end{equation*}

\section{Review of the Four-Dimensional Case}

For the rest of this chapter we review results for the case $n-2=2$, that is $J$ is a $2\times 2$-matrix and $E$ a rank-2 vector bundle. To characterize the bundle $E\to \Rcal'$ in terms of patching matrices we choose a four-set open cover $\{U_{0},\dotsc ,U_{3}\}$ of $\Rcal'$ such that $U_{0}\cup U_{2}\supset S_{0}$ with $V'\subset U_{2}$ and $\infty _{0}\in U_{0}$, and $U_{1}\cup U_{3}\supset S_{1}$ with $V'\subset U_{3}$ and $\infty _{1}\in U_{1}$. Now we use the following theorem.
\begin{thm}[Grothendieck]
Let $E\to \mathbb{C}\Pbb^{1}$ be a rank-$a$ vector bundle. Then
\begin{equation*}
E=L^{k_{1}}\oplus \dotsc \oplus L^{k_{a}}=\Ocal(-k_{1})\oplus \dotsc \oplus \Ocal(-k_{a})
\end{equation*}
for some integers $k_{1},\dotsc ,k_{a}$ unique up to permutation. Here, $L^{k_{i}}=L^{\otimes k_{i}}$ with $L$ the tautological bundle.\footnote{$\Ocal(1)$ is the \textit{tautological bundle}: If $z^{0}$ and $z^{1}$ are linear coordinates on $\mathbb{C}^{2}$, then ${\zeta}=z^{1}/z^{0}$ is an affine (stereographic) coordinate on $\mathbb{C}\Pbb^{1}$. Each value of ${\zeta}$, including ${\zeta}=\infty $, defines a one-dimensional subspace $L_{{\zeta}}\subset \mathbb{C}^{2}$. So, as ${\zeta}$ varies, these subspaces form a line bundle $L\to \mathbb{C}\Pbb^{1}$, called tautological bundle $\Ocal(1)$. Covering $\mathbb{C}\Pbb^{1}$ by $V=\{z^{0}\neq 0\}$ and $\tilde V=\{z^{1}\neq 0\}$, we have a trivialization given by
\begin{equation*}
t ([z], {\lambda}(z^{0}, z^{1})) = ([z], {\lambda}z^{0}) \quad \text{and} \quad \tilde t ([z], {\lambda}(z^{0}, z^{1})) = ([z], {\lambda}z^{1})
\end{equation*}
so that the transition function is $z^{1}/z^{0}={\zeta}$.} 
\end{thm}
Hence we can choose a trivialization such that $\left.E\right|_{S_{0}}=L^{p}\oplus L^{q}$ and $\left.E\right|_{S_{1}}=L^{p'}\oplus L^{q'}$, that is 
\begin{equation*}
\renewcommand{\arraystretch}{1.5}
P_{02}=\left(\begin{array}{cc}(2w)^{p} & 0 \\0 & (2w)^{q}\end{array}\right),\quad 
P_{13}=\left(\begin{array}{cc}(2w)^{p'} & 0 \\0 & (2w)^{q'}\end{array}\right), 
\end{equation*}
where we assume that without loss of generality $\{w=0\}\subset V'$ which can be achieved by adding a real constant to $w$. The above form of the patching matrices can also be concluded from the Birkhoff factorization. 

Now the triviality assumption and the symmetry imply that $p=-p'$ and $q=-q'$ which can be seen as follows. As a generalization of the winding number the determinant of a patching matrix has to be topologically invariant which implies that $p+q$ is topologically invariant since $\det P_{02} \sim w^{p+q}$. The triviality of the pullback of $E$ to $\Fcal_{\mathrm{r}}$ then implies that $p+q=p'+q'$. We will see later that the symmetry of $J$ requires that $P_{02}^{\vphantom{-1\mathrm{t}}}=P_{13}^{-1\mathrm{t}}$, thus $p=-p'$ and $q=-q'$. 

That reduces the patching data to two integers $p$, $q$ and a single holomorphic patching matrix $P(w)=P_{23}(w)$ defined for $w\in V'$. The remaining patching matrices are obtained by concatenation. 

We can then also reduce the reconstruction of $J$ from $E'\to \Rcal'$ given $p$, $q$ and $P$. For fixed $r$, $z$ we get a map ${\pi}:\Xcal\to \Rcal'$ by
\begin{equation} \label{eq:quadw2}
{\zeta} \mapsto  w=\frac{1}{2} r({\zeta}^{-1}-{\zeta})+z,
\end{equation}
where $\Xcal$ is again the ${\zeta}$ Riemann sphere, namely $\{(r,z)\}\times \Xcal \subset  \Fcal_{\mathrm{r}}$. The map is not yet well-defined as outside $V'$ a value of $w$ corresponds to two points of $\Rcal'$. So, we have to add a rule which of the two points is to be ${\pi}({\zeta})$. For the construction of $J(z,r)$ we need $z \pm \irm r$ to be in $V'$, otherwise ${\pi}$ is only a double cover of one of the spheres in $\Rcal'$ (${\zeta}=0$ and ${\zeta}=\infty $ lie on different spheres) for the following reason. First note that ${\zeta}$ and $-{\zeta}^{-1}$ are mapped to the same $w$, so if $w\notin V'$ it corresponds to two points, each on one sphere, and we have to decide to which sphere we map ${\zeta}$ and $-{\zeta}^{-1}$. Note further, as soon as we can find a path from ${\zeta}$ to $-{\zeta}^{-1}$ whose image does not go through $V'$, ${\pi}$ covers only one sphere, because we can only change the sphere by going through $V'$. Now the image of every path from ${\zeta}$ to $-{\zeta}^{-1}$ has to go through one of the points $w=z + \irm r$ as ${\zeta}$ and $-{\zeta}^{-1}$ lie on different branches of the root function of quadratic equation~\eqref{eq:quadw2}, hence the path has to go through a branch point. Any other value of $w$ can be avoided by picking a suitable path. Thus, if $w=z + \irm r\notin V'$, then we can find a path connecting ${\zeta}$ and $-{\zeta}^{-1}$ whose image does not go through $V'$ and ${\pi}$ must be a double cover of one of the spheres. However, as we have seen above, if $r=r({\sigma})$ and $z=z({\sigma})$ for ${\sigma}\in {\Sigma}$ this condition is satisfied automatically (which will be assumed further on).

If $V'$ is simply connected, ${\pi}$ can be fixed by the condition that ${\zeta}=0$ is mapped to $\infty _{0}\in S_{0}$ and ${\zeta}=\infty $ to $\infty _{1}\in S_{1}$. Yet, even if the space-time is regular on $r=0$, that is there exist smooth non-singular coordinates on a neighbourhood of $\{r=0\}$, the Ernst Potential may have poles on $r=0$ (see Example~\ref{ex:schwarz}). In most of the relevant examples $V'$ is the complement of a finite set --- infinity and a finite set of isolated singularities, $V'=\mathbb{C}\Pbb^{1} \backslash \{\infty ,w_{1},\dotsc ,w_{n}\}$, where the $w_{i}$ lie on the real axis, thus, suppose $V'$ is not simply connected. The $w_{i}$ correspond to two points of $\Xcal$, the roots of
\begin{equation} \label{eq:rootswi}
r{\zeta}^{2}+2(w_{i}-z){\zeta}-r=0,
\end{equation}
where $r$ and $z$ are still fixed. Here we need to assign the roots to $S_{0}$, $S_{1}$, say ${\pi}({\zeta}_{i}^{0})\in S_{0}$ and ${\pi}({\zeta}_{i}^{1})\in S_{1}$. Choose a cover $V_{0}$, $V_{1}$ of $\Xcal$ such that $\{0,{\zeta}_{1}^{0},\dotsc ,{\zeta}_{n}^{0}\}\subset V_{0}$ and $\{\infty ,{\zeta}_{1}^{1},\dotsc ,{\zeta}_{n}^{1}\}\subset V_{1}$. To define $J(z,r)$ we use what is called the Ward ansatz in \citet[Sec.~5.5]{Woodhouse:1988ek}. With given patching data $p$, $q$ and $P$ we have patching matrices
\begin{equation*}
\renewcommand{\arraystretch}{1.5}
P_{02}^{\vphantom{-1}}=P_{13}^{-1}=\left(\begin{array}{cc}(2w)^{p} & 0 \\0 & (2w)^{q}\end{array}\right), \quad P_{23}=P.
\end{equation*}
Now instead of looking for splitting matrices for $P_{01}$, we first do a transformation to our patching matrices. Let $m_{0}=\frac{2{\zeta}w}{r}$ and $m_{1}=-\frac{2w}{{\zeta}r}$, and define
\begin{equation*}
\renewcommand{\arraystretch}{1.5}
M_{0}=\left(\begin{array}{cc}m_{0}^{-p} & 0 \\0 & m_{0}^{-q}\end{array}\right), \  
M_{1}=\left(\begin{array}{cc}m_{1}^{p} & 0 \\0 & m_{1}^{q}\end{array}\right), \  
M_{2}=\left(\begin{array}{cc}r^{p}{\zeta}^{-p} & 0 \\0 & r^{q}{\zeta}^{-q}\end{array}\right). 
\end{equation*}
Then $M_{{\alpha}}({\zeta})$, ${\alpha}=0,1$, is holomorphic and invertible on $U_{{\alpha}}$ and $M_{2}$ is holomorphic and invertible on $U_{2}$ and $U_{3}$. Furthermore, since 
\begin{equation*}
m_{0}(0)=\left.\frac{2{\zeta}}{r}\left(\frac{1}{2}r({\zeta}^{-1}-{\zeta})+z\right)\right|_{{\zeta}=0}=1,
\end{equation*}
and similarly $m_{1}(\infty )=1$, we get $M_{0}(0)=M_{1}(\infty )=\id$ so that $J(z,r)$ is unchanged by using $\tilde P_{{\alpha}{\beta}}^{\vphantom{1}}=M_{{\alpha}}^{\vphantom{1}}P_{{\alpha}{\beta}}^{\vphantom{1}}M_{{\beta}}^{-1}$ instead of $P_{{\alpha}{\beta}}^{\vphantom{1}}$. Yet, $\tilde P_{02}^{\vphantom{1}}=1$, which means we can reduce our open cover to $U'_{0}=U_{0}^{\vphantom{1}}\cup U_{2}^{\vphantom{1}}$ and $U'_{{\alpha}}=U_{{\alpha}}^{\vphantom{1}}$, ${\alpha}=1,3$. For the evaluation of $J(z,r)$ we seek a splitting $Q_{0}({\zeta})$ and $Q_{1}({\zeta})$ of the new patching matrix $P'_{01}$ which is
\begin{equation} \label{eq:regsplit}
\renewcommand{\arraystretch}{1.5}
\begin{split}
P'_{01} & = M_{0}^{\vphantom{-1}} P_{01}^{\vphantom{-1}} M_{1}^{-1} = M_{0}^{\vphantom{-1}} P_{02}^{\vphantom{-1}}P_{23}^{\vphantom{-1}}P_{31}^{\vphantom{-1}} M_{1}^{-1} \\
	& = \left(\begin{array}{cc}r^{p}{\zeta}^{-p} & 0 \\0 & r^{q}{\zeta}^{-q}\end{array}\right)
P\left(\frac{1}{2}r({\zeta}^{-1}-{\zeta})+z\right)
\left(\begin{array}{cc}(-r{\zeta})^{p} & 0 \\0 & (-r{\zeta})^{q}\end{array}\right)\\
 & = Q_{0}^{\vphantom{-1}}Q_{1}^{-1}.
\end{split}
\end{equation}
Here $Q_{0}$, $Q_{1}$ are holomorphic in ${\zeta}$ and non-singular in $V_{0}$ and $V_{1}$, respectively. As before, we set $J \coloneqq Q_{0}^{\vphantom{-1}}(0)Q_{1}^{-1}(\infty )$, and obtain a solution of \eqref{eq:redyang}. A different labeling of the roots of \eqref{eq:rootswi} yields a different solution $J$, yet these different solutions are analytic continuations of each other and one can show that they are different parts of the Penrose diagram of the maximal analytic extension of the metric (this will be considered again in Chapter~\ref{ch:fivedim} from a slightly different point of view).

From the interpretation of $J$ as the matrix of inner products of Killing vectors in general relativity, we require $J$ to be real and symmetric. Therefore, the questions arises which conditions we have to impose on our bundle to obtain solutions with the desired properties. The following can be immediately read off from the splitting procedure in any dimension (restriction to $n=2$ not necessary).

Let $i$ be the map that interchanges the spheres, that is it is the identity on $V'$ and otherwise the two points of $\Rcal'$ that correspond to the same $w$ are interchanged. Then $J^{-1}$ is obtained from the pullback bundle $i^{*}(E)$ and furthermore $J^{-1\mathrm{t}}$ is obtained from the dual bundle $E^{*}$. Hence, 
\begin{equation*}
J=J^{\mathrm{t}} \Leftrightarrow  J^{-1}=(J^{-1})^{\mathrm{t}} \Leftrightarrow  i^{*}(E)=E^{*} \Leftrightarrow  P^{-1}=(P^{-1})^{\mathrm{t}} \Leftrightarrow  P=P^{\mathrm{t}},
\end{equation*}
or in other words $J$ is symmetric, which means $J=J^{\mathrm{t}}$, if and only if $P=P^{\mathrm{t}}$, that is $P$ is symmetric.

Remember in general for a holomorphic function ${\varphi}(a)$ we have that ${\varphi}$ is real on the real numbers if and only if ${\varphi}(\skew{2}{\bar}{a})=\overline{{\varphi}(a)}$. So, if $P$ is real in the sense $\overline{P(w)}=P(\skew{1}{\bar}{w})$, then by \eqref{eq:regsplit} we have
\begin{equation*}
\overline{P(w)}=P(\skew{1}{\bar}{w}) \Rightarrow  
\begin{cases}
\overline{Q_{0}({\zeta})}\cdot  \overline{Q_{1}({\zeta})}^{-1}=Q_{0}(\bar {\zeta})Q_{1}(\bar {\zeta})^{-1},\ & r,z\in \mathbb{R};\\[0.2cm]
\overline{Q_{0}({\zeta})}\cdot  \overline{Q_{1}({\zeta})}^{-1}=Q_{0}(-\bar {\zeta})Q_{1}(-\bar {\zeta})^{-1},\ & z\in \mathbb{R}, r\in \irm \mathbb{R}.
\end{cases}
\end{equation*}
So, $J$ must be real as well, provided $r,z\in \mathbb{R}$ or $z\in \mathbb{R}$, $r\in \irm \mathbb{R}$. In terms of the bundle this is the condition $\bar E=j^{*}(E)$ where $\bar E$ is the complex conjugate bundle of $E$ and $j^{*}(E)$ is the pullback with the complex conjugation on the spheres
\begin{equation*}
\renewcommand{\arraystretch}{1.2}
j: \begin{array}{c}
S_{0} \to  S_{0},\ w\mapsto  \skew{1}{\bar}{w}; \\
S_{1} \to  S_{1},\ w\mapsto  \skew{1}{\bar}{w}.
\end{array}
\end{equation*}
The converse, that is a real $J$ implies a real $P$, is then also obvious.

Moreover, for $n=2$ it can be shown that (see for example \cite{Woodhouse:1988ek}, but as cited here it is taken from \cite{Fletcher:1990db}):
\begin{itemize}
\item If $\det P=1$, then $\det J=\left(-r^{2}\right)^{p+q}$. 
\item If $J$ is obtained from an axis-regular space-time and if the definition of ${\pi}$ is such that ${\zeta}_{i}^{0}\to 0$ and ${\zeta}_{i}^{1}\to \infty $ for $r\to \infty $ and all $i$, then $p=1$, $q=0$ and $P(z)=J'(z,0)$ on the rotational axis or on the horizon. Here, $J'$ is the Ernst potential. Thus, $P$ is the analytic continuation of the boundary values of the Ernst potential.

As an argument in \citet{Fletcher:1990aa} shows, this can be seen from the determinant condition above and the fact that $J$ is bounded on the axis. The determinant condition implies $p+q=1$ and an asymptotic relation between $J$ and $P$ reveals that if $J$ is bounded as $r\to 0$ then $p$ and $q$ must be non-negative. So, either $p=1$, $q=0$ or $p=0$, $q=1$. We can assume that the trivialization over $S_{0}$ and $S_{1}$ can be chosen such that the first holds. The same argument implies that the twistor data obtained from $J'$ has $p=q=0$, since $\det J'=1$. This goes along with the interesting effect that the gauge groups for $J$ and $J'$  (in their Yang-Mills interpretation) are different \cite[Sec.~6]{Woodhouse:1988ek}.
\item If $J$ comes from an asymptotically flat space-time in the sense that its Ernst potential has the same asymptotic form as the Ernst potential of Minkowski space with rotation and translation as Killing vectors, then $P(\infty )=1$, and conversely.
\end{itemize}

These results can be used to look at some examples.
\begin{ex}
\mbox{}
\begin{enumerate}
\item \textit{Minkowski Space with Translation and Rotation} \\
Written in cylindrical polar coordinates
\begin{equation*}
\drm s^{2}=\drm t^{2}- r^{2} \,\drm {\theta}^{2} - \drm r^{2} - \drm z^{2}
\end{equation*}
we can take the Killing vectors to be $X=\partial _{{\theta}}$, $Y=\partial _{t}$. The matrix of inner products and the Ernst potential are then
\begin{equation*}
J=\left(\begin{array}{cc}-r^{2} & 0 \\\hphantom{-}0 & 1\end{array}\right),\quad  J'=\left(\begin{array}{cc}1 & 0 \\0 & 1\end{array}\right).
\end{equation*}
Since it is axis-regular we also have $p=1$ and $q=0$ with $V'$ the entire complex plane. 
\item \textit{The Kerr solution} \\
The patching data for the Kerr solution is (without proof)
\begin{equation*}
\renewcommand{\arraystretch}{1.5}
P(w)=\frac{1}{w^{2}-{\sigma}^{2}}\left(\begin{array}{cc}
(w+m)^{2}+a^{2} & 2am \\2am & (w-m)^{2}+a^{2}
\end{array}\right)
\end{equation*}
where ${\sigma}=\sqrt{m^{2}-a^{2}}$ for $a<m$. Axis-regularity implies again $p=1$ and $q=0$. The open set $V'$ is the complement of $\{\infty ,w_{1},w_{2}\}$ where $w_{1}={\sigma}$ and $w_{2}=-{\sigma}$. 
\end{enumerate}
\end{ex}
\chapter{Twistor Approach in Five Dimensions} \label{ch:fivedim}

Most of what we have seen about the twistor construction at the end of Chapter~\ref{ch:bundles} generalizes without any effort to five and higher dimensions. Only, instead of two, the rank of the bundle will be three so that we have three integers instead of only $p$ and $q$ in our twistor data (respectively $n-2$ in higher dimensions). The splitting procedure itself is not affected by increasing the rank. However, since the splitting itself is complicated, we have seen that the Ernst potential $J'$ is a crucial tool for any practical application of the twistor characterization of stationary axisymmetric solutions of Einstein's field equations, and the way we obtained $J'$ in Chapter~\ref{ch:yang} seemed to be tailored to four dimensions with two Killing vectors. So, in order to pursue this strategy we have to define an Ernst potential in five dimensions. First, we are going to say a few words about B\"acklund transformations, because we will see that we secretly used them to obtain $J'$. Since some of the following extends immediately to higher dimensions as well, we will present most of it in $n$ dimensions. 

In the last part of this section we generalize results from \citet[Sec.~2.4]{Fletcher:1990aa} in order to conclude the important fact that the integers in the twistor data are non-negative as in four dimensions.

\section{B\"acklund Transformations}

In Chapter~\ref{ch:yang} we derived Yang's equation as one way of writing the ASDYM equations with gauge group $\GL(n,\mathbb{C})$. It has a number of `hidden' symmetries one of which is the B\"acklund transformation. 

As in \citet[Sec.~4.6]{Mason:1996hl} we can decompose a generic $J$-matrix in the following way\footnote{Using $\left(\begin{array}{cc}1 & \skew{4}{\tilde}{B} \hphantom{^{-1}} \\0 & \skew{7}{\tilde}{A}^{-1}\end{array}\right)^{-1}
=\left(\begin{array}{cc}1 & -\skew{4}{\tilde}{B} \skew{7}{\tilde}{A} \\0 & \skew{7}{\tilde}{A}\end{array}\right)$.}
\begin{equation} \label{eq:bdecomp}
\renewcommand{\arraystretch}{1.5}
J=
\left(\begin{array}{cc}A^{-1}-\skew{4}{\tilde}{B} \skew{7}{\tilde}{A} B & -\skew{4}{\tilde}{B} \skew{7}{\tilde}{A} \\ \skew{7}{\tilde}{A} B & \skew{7}{\tilde}{A}\end{array}\right)
=
\left(\begin{array}{cc}1 & \skew{4}{\tilde}{B} \hphantom{^{-1}} \\0 & \skew{7}{\tilde}{A}^{-1}\end{array}\right)^{-1}
\left(\begin{array}{cc}A^{-1} & 0 \\B \hphantom{^{-1}} & 1\end{array}\right),
\end{equation}
where $A$ is a $k\times k$ non-singular matrix ($k<n$), $\skew{7}{\tilde}{A}$ is $\skew{2}{\tilde}{k} \times  \skew{2}{\tilde}{k}$ non-singular matrix with $k+\skew{2}{\tilde}{k}=n$. Then, $B$ is a $\skew{2}{\tilde}{k}\times k$ and $\skew{4}{\tilde}{B}$ a $k\times \skew{2}{\tilde}{k}$ matrix. The term `generic' rules out for example cases where $\skew{7}{\tilde}{A}$ is not invertible. Substituting this in Yang's equation~\eqref{eq:yang2} we get the coupled system of equations
\begin{equation} \label{eq:btrf}
\begin{split}
& ∂ _{\tilde z}(\skew{7}{\tilde}{A} B_{z \vphantom{\tilde z} }A)-∂ _{\tilde w}(\skew{7}{\tilde}{A} B_{w \vphantom{\tilde z} }A)=0,\\
& ∂ _{z \vphantom{\tilde z} }(A \skew{4}{\tilde}{B}_{\tilde z}\skew{7}{\tilde}{A})-∂ _{w \vphantom{\tilde z} }(A \skew{4}{\tilde}{B}_{\tilde w}\skew{7}{\tilde}{A})=0,\\
& ∂ _{z \vphantom{\tilde z} }(\skew{7}{\tilde}{A}^{-1}\skew{7}{\tilde}{A}_{\tilde z})\skew{7}{\tilde}{A}^{-1}-∂ _{w \vphantom{\tilde z} }(\skew{7}{\tilde}{A}^{-1}\skew{7}{\tilde}{A}_{\tilde w})\skew{7}{\tilde}{A}^{-1}+B_{z \vphantom{\tilde z} }A\skew{4}{\tilde}{B}_{\tilde z}-B_{w \vphantom{\tilde z} }A\skew{4}{\tilde}{B}_{\tilde w}=0,\\
& A^{-1}∂ _{z \vphantom{\tilde z} }(A_{\tilde z} A^{-1})-A^{-1}∂ _{w \vphantom{\tilde z} }(A_{\tilde w} A^{-1})+\skew{4}{\tilde}{B}_{\tilde z}\skew{7}{\tilde}{A} B_{z \vphantom{\tilde z} }-\skew{4}{\tilde}{B}_{\tilde w}\skew{7}{\tilde}{A} B_{w \vphantom{\tilde z} }=0, 
\end{split}
\end{equation}
where an index denotes a partial derivative (see Appendix~\ref{app:btrfcalc} for detailed calculation). The first two equations are integrability conditions and they imply the existence of matrices $B'$ and $\skew{4}{\tilde}{B}'$ such that
\begin{align*}
∂ _{\tilde z} \skew{4}{\tilde}{B}' & = \skew{7}{\tilde}{A} B_{w \vphantom{\tilde z} }A, & ∂ _{\tilde w} \skew{4}{\tilde}{B}' & = \skew{7}{\tilde}{A} B_{z \vphantom{\tilde z} }A,\\
∂ _{z \vphantom{\tilde z} } B' & = A\skew{4}{\tilde}{B}_{\tilde w} \skew{7}{\tilde}{A}, & ∂ _{w \vphantom{\tilde z} } B' & = A\skew{4}{\tilde}{B}_{\tilde z} \skew{7}{\tilde}{A}.
\end{align*}
\begin{Def}
Together with $B'$ and $\skew{4}{\tilde}{B}'$ we define the other primed quantities as 
\begin{equation*}
\left(A,\skew{7}{\tilde}{A},B,\skew{4}{\tilde}{B},k,\skew{2}{\tilde}{k}\right) \mapsto \left(A'=\skew{7}{\tilde}{A}^{-1},\skew{7}{\tilde}{A}'=A^{-1},B',\skew{4}{\tilde}{B}',k'=\skew{2}{\tilde}{k},\skew{2}{\tilde}{k}'=k\right).
\end{equation*}
We call the matrix $J'$, that is obtained from $J$ by \eqref{eq:bdecomp} with the primed blocks instead of the unprimed, the \textit{B\"acklund transform}.
\end{Def}
\begin{prop}[Section~4.6 in \citet{Mason:1996hl}] \label{prop:btsolyang}\mbox{}
\begin{enumerate}
\item $J'$ is again a solution of Yang's equation.
\item $(J')'=J$
\end{enumerate}
\end{prop}
\begin{proof}\mbox{}
\begin{enumerate}
\item Substitute the unprimed by the primed versions in \eqref{eq:btrf}.
\item Noting that $B_{w}^{\vphantom{1}}=\skew{7}{\tilde}{A}^{-1}\skew{4}{\tilde}{B}'_{\tilde z}A^{-1}=A'\skew{4}{\tilde}{B}'_{\tilde z}\skew{7}{\tilde}{A}$ and similar for the other integrability equations, this statement is obvious.
\end{enumerate}
\end{proof}

However, the connection between ASDYM and Einstein's field equations in the stationary and axisymmetric case was given by the reduced Yang's equation 
\begin{equation*}
r ∂ _{x}(J^{-1}∂ _{x}J)+∂ _{r}(rJ^{-1}∂ _{r}J)=0,
\end{equation*}
see \eqref{eq:redyang}. The reduction was induced by the coordination transformation
\begin{equation*}
z=t+x,\ \tilde z=t-x,\ w= r \erm^{\irm {\theta}},\ \tilde w=r \erm^{-\irm {\theta}}.
\end{equation*}
Now we have to be careful, because the $r$-derivatives also act on the $r$-factors appearing. But an analogous calculation as above (see Appendix~\ref{app:redbtrfcalc}) shows that the reduced Yang's equation is equivalent to a similar set of equations
\begin{equation} \label{eq:redbtrf}
\begin{split}
& r ∂ _{x}(\skew{7}{\tilde}{A} B_{x}A)+∂ _{r}(r\skew{7}{\tilde}{A} B_{r}A)=0,\\
& r∂ _{x}(A \skew{4}{\tilde}{B}_{x}\skew{7}{\tilde}{A})+∂ _{r}(rA \skew{4}{\tilde}{B}_{r}\skew{7}{\tilde}{A})=0,\\
& r∂ _{x}(\skew{7}{\tilde}{A}^{-1}\skew{7}{\tilde}{A}_{x})\skew{7}{\tilde}{A}^{-1}-∂ _{r}(r\skew{7}{\tilde}{A}^{-1}\skew{7}{\tilde}{A}_{r})\skew{7}{\tilde}{A}^{-1}+rB_{r}A\skew{4}{\tilde}{B}_{r}-rB_{r}A\skew{4}{\tilde}{B}_{r}=0,\\
& rA^{-1}∂ _{x}(A_{x} A^{-1})-A^{-1}∂ _{r}(rA_{r} A^{-1})+r\skew{4}{\tilde}{B}_{r}\skew{7}{\tilde}{A} B_{r}-r\skew{4}{\tilde}{B}_{r}\skew{7}{\tilde}{A} B_{r}=0, 
\end{split}
\end{equation}
where all matrices are functions of $x$ and $r$. The integrability conditions are now
\begin{align*}
∂ _{r} \skew{4}{\tilde}{B}' & = r \skew{7}{\tilde}{A} B_{x}A, & ∂ _{x} \skew{4}{\tilde}{B}' & = -r\skew{7}{\tilde}{A} B_{r}A,\\
∂ _{r} B' & = rA\skew{4}{\tilde}{B}_{x} \skew{7}{\tilde}{A}, & ∂ _{x} B' & = -rA\skew{4}{\tilde}{B}_{r} \skew{7}{\tilde}{A}.
\end{align*}
So, we have to adapt the transform, namely we define $\skew{7}{\tilde}{A}'=r^{-2}A^{-1}$, $A'=\skew{7}{\tilde}{A}^{-1}$ and the rest as above\footnote{Note that this modification is not mentioned in \cite[Sec.~4.6]{Mason:1996hl}, but necessary to obtain the correct result in Example~\ref{ex:BTfourdim}.}. Again by substitution it can be checked that this gives a solution of the reduced Yang's equation. However, one should note that this modified B\"acklund transformation is not an involution anymore since the definition for $\skew{7}{\tilde}{A}'$, $A'$ is not involutive. The definition for $\skew{4}{\tilde}{B}'$, $B'$ is still an involution so that the inverse B\"acklund transformation is obtained by $B=B''$, $\skew{4}{\tilde}{B}=\skew{4}{\tilde}{B}''$ and $A=r^{-2}(\skew{7}{\tilde}{A}')^{-1}$, $\skew{7}{\tilde}{A}=(A')^{-1}$.

\begin{prop}\label{prop:bdet}
\begin{equation*}
\det J' = (-r)^{2(1-k)}
\end{equation*}
\end{prop}
\begin{proof} 
Note that for the decomposition of a general matrix in block matrices we know from basic linear algebra 
\begin{equation*}
\det \left(\begin{array}{cc}P & Q \\R & S\end{array}\right)
= \det(S) \det(P-QS^{-1}R),
\end{equation*}
if $S$ is invertible. Applied to our decomposition \eqref{eq:bdecomp} this yields
\begin{equation} \label{eq:detj}
\det J = \det(\skew{7}{\tilde}{A}) \det(A^{-1}-\skew{4}{\tilde}{B} \skew{7}{\tilde}{A} B + \skew{4}{\tilde}{B} \skew{7}{\tilde}{A} \skew{7}{\tilde}{A}^{-1}\skew{7}{\tilde}{A} B)=\det(\skew{7}{\tilde}{A})\det(A^{-1}).
\end{equation}
Now using the fact that $\det J =- r^{2}$ we obtain
\begin{align*}
\det J' & = \det (\skew{7}{\tilde}{A}') \det ((A')^{-1}) \\
	& = (-r)^{-2k} \det (A^{-1}) \det (\skew{7}{\tilde}{A})\\
	& = (-r)^{-2k} \det (J)\\
	& = (-r)^{2(1-k)}
\end{align*}
\end{proof}

\begin{ex}[Bäcklund Transformation in Four Dimensions] \label{ex:BTfourdim}
As in \cite[Sec.~6.6]{Mason:1996hl}, we consider the four-dimensional case with two Killing vectors. As in \eqref{eq:backldec4d} we set $A=-r^{-2} f$, $\skew{7}{\tilde}{A} = f$, $B=-\skew{4}{\tilde}{B}={α}$. Then the integrability equations take the form
\begin{equation*}
∂ _{x}(r^{-1}f^{2}∂ _{x}{α})+∂ _{r}(r^{-1}f^{2} ∂ _{r}{α})=0,
\end{equation*}
like in \eqref{eq:weylequ}. The blocks for the B\"acklund transform are $B'=-\skew{4}{\tilde}{B}'={ψ}$, $A'=\skew{7}{\tilde}{A}^{-1}=f^{-1}$, $\skew{7}{\tilde}{A}'=-r^{-2}A^{-1}=f^{-1}$ defining
\begin{equation*}
J'=\frac{1}{f}\left(\begin{array}{cc}f^{2}+{ψ}^{2} & {ψ} \\{ψ} & 1\end{array}\right),
\end{equation*}
where the existence of ${ψ}$ is ensured by the integrability condition as above. Hence, our Ernst potential is obtained by a B\"acklund transformation.\\ \mbox{} \hfill $\blacksquare$
\end{ex}

\section{Higher-Dimensional Ernst Potential} \label{sec:highErnst}

Let us first recall the definition of twist 1-forms, twist potentials and some of their properties.
\begin{Def}
Consider an $n$-dimensional (asymptotically flat) space-time $M$ with $X_{0}$ a stationary and $X_{1},\dotsc ,X_{n-3}$ axial Killing vectors, all mutually commuting. The \textit{twist $1$-forms} are defined as
\begin{align*}
ω_{1 a}^{\vphantom{1}} & = Δ\,{ε}_{ab\dotsc cde}^{\vphantom{1}}X_{1}^{b}\cdots X_{n-3}^{c}∇ ^{d}X_{1}^{e}, \\
& \hspace{0.2cm} \vdots \\
ω_{n-3, a}^{\vphantom{1}} & =Δ\,{ε}_{ab\dotsc cde}^{\vphantom{1}}X_{1}^{b}\cdots X_{n-3}^{c}∇ ^{d}X_{n-3}^{e},
\end{align*}
where $Δ=\sqrt{-g}=r{\erm}^{2ν}$, according to \eqref{eq:sigmametric}.\footnote{Note that in \cite{Hollands:2008fp} the notation is taken from \cite{Wald:1984rz} where ${ε}$ is already the volume element.
}
\end{Def}
Adopting a vector notation
\begin{equation*}
ω=\left(\begin{array}{c}ω_1 \\ \vdots \\ω_{n-3}\end{array}\right), \quad X=\left(\begin{array}{c}X_1 \\ \vdots \\X_{n-3}\end{array}\right)
\end{equation*}
this can be written as
\begin{equation*}
ω_{a}^{\vphantom{1}} = Δ\,{ε}_{ab\dotsc cde}^{\vphantom{1}}X_{1}^{b}\cdots X_{n-3}^{c}∇ ^{d}X^{e}. 
\end{equation*}
Let ${\theta}_{I}^{\vphantom{j}}$, $I\in \{1,\dotsc ,n-3\}$, be the dual to $X_{I}^{\vphantom{j}}$, that is ${\theta}_{Ik}^{\vphantom{j}}=g_{kj}^{\vphantom{j}}X_{I}^{j}$. 
\begin{prop} \mbox{}
\begin{enumerate}
\item The twist 1-forms can be written as
\begin{equation} \label{eq:twisthodge}
ω_{I}=* ({\theta}_{1}∧ \dotsc ∧ {\theta}_{n-3}∧ \drm {\theta}_{I}), \quad I\in \{1,\dotsc ,n-3\}.
\end{equation}
\item $ω$ is closed.
\item $ω$ annihilates the Killing vector fields $X_{0},\dotsc ,X_{n-3}$.
\end{enumerate}
\end{prop}
\begin{proof} \mbox{}
\begin{enumerate}
\item Suppressing the $I$-index for ${\theta}$, we note that 
\begin{equation*}
∇ _{[d}^{\vphantom{f}}{\theta}_{e]}^{\vphantom{f}}=∂ _{[d}^{\vphantom{f}}{\theta}_{e]}^{\vphantom{f}}-\tensor*{Γ}{_{[de]}^{f}}{\theta}_{f\vphantom{f]}}^{\vphantom{f}}=∂ _{[d}^{\vphantom{f}}{\theta}_{e]}^{\vphantom{f}}
\end{equation*}
as we use the Levi-Civita connection for which $\tensor*{Γ}{_{de \vphantom{[d]}}^{f}}=\tensor*{Γ}{_{(de)}^{f}}$. But $∂ _{[d}^{\vphantom{1}}{\theta}_{e]}^{\vphantom{1}}$ are the components of $\drm {\theta}$ so that from the definition of the Hodge dual we see that \eqref{eq:twisthodge} is an equivalent way of writing the twist 1-forms.
\item Analogously to the proof of \cite[Thm.~7.1.1]{Wald:1984rz} implied by the vacuum field equations.
\item We denote by $x_{K}$ the coordinate belonging to $X_{K}$, $K=0,\dotsc ,n-3$, and for $M/\Gcal$ we use again $r$ and $z$ as coordinates. It is sufficient to show that $ω=ω_{r}\,\drm r + ω_{z}\,\drm z$ with all other components vanishing. We will do the following calculations a bit more in detail as the results will proof useful later on. Using
\begin{equation*}
{\theta}_{Ik}=g_{ka}^{\vphantom{1}}X^{a}_{I}=J_{kI}^{\vphantom{1}}\quad \text{for } I=1,\dotsc ,n-3,
\end{equation*}
we first calculate
\begin{align*}
\drm {\theta}_{I} & = ∂ _{r}J_{0I} \, \drm r ∧  \drm x^{0} + \dotsc  + ∂ _{r}J_{n-3,I} \, \drm r ∧  \drm x^{n-3} \\
	& \hspace{0.4cm} + ∂ _{z}J_{0I} \, \drm z ∧  \drm x^{0} + \dotsc  + ∂ _{z}J_{n-3,I} \, \drm z ∧  \drm x^{n-3}\\
	& = - (∂ _{r}J_{0I} \, \drm x^{0} + \dotsc  + ∂ _{r}J_{n-3,I} \, \drm x^{n-3}) ∧  \drm r \\
	& \hspace{0.4cm} - (∂ _{z}J_{0I} \, \drm x^{0} + \dotsc  + ∂ _{z}J_{n-3,I} \, \drm x^{n-3}) ∧  \drm z.
\end{align*}
This gives for $I\in \{1,\dotsc ,n-3\}$ the following
\begin{align*}
{\theta}_{1}∧ \dotsc ∧ {\theta}_{n-3}∧ \drm {\theta}_{I} & = - (J_{K1}\, \drm x^{K})∧ \dotsc ∧ (J_{M,n-3}\, \drm x^{M})\\
	& \hspace{0.4cm}∧ \left((∂ _{r}J_{RI}\, \drm x^{R})∧ \drm r+(∂ _{z}J_{RI}\, \drm x^{R})∧ \drm z \right),
\end{align*}
where the summation over capital latin indices runs from 0 to $n-3$. At this point we can already see that only $ω_{r}$, $ω_{z}$ are non-zero, but let us determine them somewhat more precisely. Our volume form is
\begin{equation*}
\drm \mathrm{vol} = r \, \drm x^{0}∧ \dotsc ∧ \drm x^{n-3}∧ (ω \, \drm r)∧ (ω \, \drm z),
\end{equation*}
so that
\begin{equation*}
*\left(r \, \drm x^{0}∧ \dotsc ∧  \drm x^{n-3}∧ (ω \, \drm r)\right)= (-1)^{n-2} ω \, \drm z,
\end{equation*}
since the left-hand side contains $n-2$ spacelike and one timelike component. Thus
\begin{equation*}
*\left( \drm x^{0}∧ \dotsc ∧  \drm x^{n-3}∧  \drm r\right)= (-1)^{n-2}\frac{1}{r} \, \drm z;
\end{equation*}
similar for $\drm r\leftrightarrow  \drm z$ but with the opposite sign. Whence we obtain
\begin{align} \notag
ω_{Ir} & =(-1)^{n-1}\frac{1}{r} ∑ _{{\sigma}\in S_{n-2}} \sgn ({\sigma}) J_{1,{\sigma}(0)} \cdots J_{n-3,{\sigma}(n-4)} ∂ _{z}J_{I,{\sigma}(n-3)}\\
	& = (-1)^{n-1}\frac{1}{r} \det
\renewcommand{\arraystretch}{2}	
	\left(\begin{array}{cccc}
	J_{01} &  \cdots & J_{0,n-3} & ∂ _{z}J_{0I} \\ 
	\vdots &  \ddots &  \vdots &  \vdots \\
	J_{n-3,1} &  \cdots & J_{n-3,n-3} & ∂ _{z} J_{n-3,I}
	\end{array}\right), \label{eq:compsomega}
\end{align}
where $S_{n-2}$ are the permutations of $n-2$ elements. An analogous equation holds for $ω_{Iz}$.\qedhere
\end{enumerate}
\end{proof}
The last statement in the above proposition implies that $ω$ can also be regarded as 1-forms on the interior of $M / \Gcal$. Due to the form of the ${\sigma}$-model metric there should be no confusion if we denote both the form on $M$ and the one on $M / \Gcal$ by the same symbol. Being a form on $M / \Gcal$ means that $ω$ has only non-vanishing components for the $r$- and $z$-coordinate, and of course as a form on $M / \Gcal$ it will again be a closed. This gives rise to the following definition. 
\begin{Def}
Locally there exist functions on $M / \Gcal$ such that 
\begin{equation*}
∂ _{r} χ_{I}=ω_{Ir} \text{ and } ∂ _{z} χ_{I}=ω_{Iz}\quad  \text{for } I=1,\dotsc ,n-3,
\end{equation*}
or equivalently in vector notation
\begin{equation*}
\drm χ = ω.
\end{equation*}
These functions $χ_{I}$ are called \textit{twist potentials}. 
\end{Def}

The construction of the Ernst Potential in Example~\ref{ex:BTfourdim} is tailor-made for dimension four, and it is not immediately obvious how to generalize it to higher dimensions. Nevertheless, there is an ansatz in \citet{Maison:1979aa}, where it is noted that the full metric on space-time, that is essentially $J$, can be reconstructed from knowing the two twist potentials (in five dimensions) $χ$, the $2×2$-matrix $\skew{7}{\tilde}{A}=\left(X_{I}^{a}X_{K}^{b}g_{ab}\right)_{I,K=1,2}$ and its non-vanishing determinant $\det \skew{7}{\tilde}{A}$ on the factor space $M / \Gcal$. The matrix in \cite[Eq.~(16)]{Maison:1979aa} will then be our candidate for the higher-dimensional Ernst Potential. Note, however, that the condition $\det \skew{7}{\tilde}{A}≠0$ needs further investigation. 

Closely connected to the non-vanishing determinant is the concept of adaptations to certain parts of the axis $r=0$. Recall that the assumption of axis-regularity was an important one. It said that the region where the two spheres of the reduced twistor space are identified can be enlarged to a simply connected patch such that this identification also extends to the fibres of the bundle. The exact shape of $V'$ is not important, however, there is still an ambiguity if we have a nut on $r=0$.\footnote{Remember that we assume that there is only a finite number of isolated nuts.} Figure~\ref{fig:adaptation} shows how we can choose different extensions of $V$.
\begin{figure}[htbp]
\begin{center}
     \scalebox{0.6}{\input{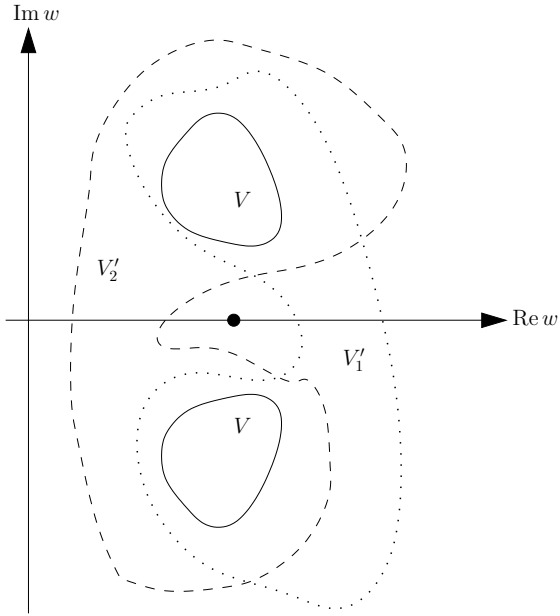}}
     \caption{Two different extensions $V'_{1}$ and $V'_{2}$ of $V$ around a pole of $J'$ (bullet on the real axis).} 
     \label{fig:adaptation}
\end{center}
\end{figure}
From Example~\ref{ex:schwarz} we learn that the choice of $V'$ matters, that is we obtain different Ernst potentials $J'$ for different extensions.
\begin{Def} 
In any dimension given an axis-regular $J'$ we shall call it \textit{adapted to the rod $(a_{i},a_{i+1})$} if $V'|_{r=0}\subseteq (a_{i},a_{i+1})$.
\end{Def}
For a rod corresponding to a rotational axis we know that along this rod a linear combination of the axial Killing vectors vanishes. By a change of basis we can always assume that without loss of generality this vector is already in the basis, say $X_{K}=0$, $K>0$. Then we make the following definition.
\begin{Def}
In dimension $n$ we call \textit{$\skew{7}{\tilde}{A}$ adapted to $(a_{i},a_{i+1})$} the $(n-3)×(n-3)$-matrix that is obtained from $J$ by 
\begin{enumerate}[(a)]
\item cancelling the $K^{\mathrm{th}}$ column and row, if $(a_{i},a_{i+1})$ is a rotational axis and $X_{K}$ the corresponding rod vector; 
\item cancelling the $0^{\mathrm{th}}$ column and row, if $(a_{i},a_{i+1})$ is the horizon.
\end{enumerate}
\end{Def}
The following lemma shows the reason for the latter definition.
\begin{lem}
$\skew{7}{\tilde}{A}$ adapted to $(a_{i},a_{i+1})$ is invertible on $(a_{i},a_{i+1})$ and becomes singular at the limiting nuts.
\end{lem}
\begin{proof}
This is a simple consequence of what we have seen in Section~\ref{sec:rodstr}. First, consider a rod corresponding to a rotational axis $(a_{i},a_{i+1})$. Here $J$ is an $(n-2)\times (n-2)$ matrix and has rank $n-3$. Since $X_{K}=0$ on $(a_{i},a_{i+1})$, we cancelled a zero column and row, thus it follows that $\skew{7}{\tilde}{A}$ has full rank and $\det \skew{7}{\tilde}{A}\neq 0$ on $(a_{i},a_{i+1})$. Of course, where the rank of $J$ drops further, that is where $\dim \ker J (0,z)\geq 2$, the matrix $\skew{7}{\tilde}{A}$ cannot have full rank anymore and $\det \skew{7}{\tilde}{A}=0$. So, this adaptation becomes singular as soon as we reach one of the nuts limiting this rod. 

Second, consider the horizon rod $(a_{h},a_{h+1})$. Here the first row and column of $J$, that is the one with the asymptotically timelike Killing vector in it, becomes zero. Then $\skew{7}{\tilde}{A}$ adapted to $(a_{h},a_{h+1})$ has full rank on $(a_{h},a_{h+1})$, and becomes singular at the nuts where the rotational axes intersect the horizon. 
\end{proof}

Note that at the beginning of this section we defined the twist 1-forms by singling out $X_{0}$. However, the definition works just as well with the set $X_{0},\dotsc ,\hat X_{K},\dotsc ,X_{n-2}$ of Killing vectors. For an adaptation to a certain rod we obtain in the same way as above $n-3$ twist potentials by omitting the rod vector for this rod. 

Now, we have collected all the components for the following definition, which is a modification of \cite[Eq.~(16)]{Maison:1979aa}.

\begin{Def} \label{def:ernstpot}
In $n$ dimensions and for a given rod $(a_{i},a_{i+1})$ we call the matrix
\begin{equation} \label{eq:highernst}
\renewcommand{\arraystretch}{1.5}
J'=\frac{1}{\det \skew{7}{\tilde}{A}} \left(\begin{array}{cc}\hphantom{-}1 & -χ^{\mathrm{t}} \\-χ & \det \skew{7}{\tilde}{A} \cdot \skew{7}{\tilde}{A} + χχ^{\mathrm{t}}\end{array}\right),
\end{equation}
where $\skew{7}{\tilde}{A}$ is adapted to $(a_{i},a_{i+1})$ and $χ$ is the vector of twist potentials for $(a_{i},a_{i+1})$, \textit{higher-dimensional Ernst potential adapted to $(a_{i},a_{i+1})$}.
\end{Def}
\begin{rem}\mbox{}
\begin{enumerate}
\item For $n=5$ this is becomes on the horizon rod
\begin{equation*}
\renewcommand{\arraystretch}{1.5}
J'= \frac{1}{\det \skew{7}{\tilde}{A}} \left(\begin{array}{ccc}
\hphantom{-}1 & -χ_{1} & -χ_{2} \\
-χ_{1}^{\vphantom{1}} & \det \skew{7}{\tilde}{A} \cdot J_{11}^{\vphantom{1}}+χ_{1}^{2} & \det \skew{7}{\tilde}{A} \cdot J_{12}^{\vphantom{1}}+χ_{1}^{\vphantom{1}}χ_{2}^{\vphantom{1}} \\
-χ_{2}^{\vphantom{1}} & \det \skew{7}{\tilde}{A} \cdot J_{21}^{\vphantom{1}}+χ_{2}^{\vphantom{1}}χ_{1}^{\vphantom{1}} & \det \skew{7}{\tilde}{A} \cdot J_{22}^{\vphantom{1}}+χ_{2}^{2}
\end{array}\right).
\end{equation*}
\item Note that $J=\big(g(X_{i},X_{j})\big)$ is a matrix of scalar quantities, hence $J$ is bounded for $r\to 0$. So, the domain of $J'$ is only restricted by $\det \skew{7}{\tilde}{A}$ and by the arguments above we see that for an adaptation to $(a_{i},a_{i+1})$ the limit $J'(0,z)$ is well-defined for $z∈(a_{i},a_{i+1})$.
\end{enumerate}
\end{rem}
\begin{lem}[\citet{Maison:1979aa}] \label{lem:highernstprop}
$J'=(J')^{\mathrm{t}}$, $\det J' =1$ and
\begin{equation*}
\renewcommand{\arraystretch}{1.5}
(J')^{-1}=\left(\begin{array}{cc}χ^{\mathrm{t}} \skew{7}{\tilde}{A}^{-1} χ + \det \skew{7}{\tilde}{A} & χ^{\mathrm{t}} \skew{7}{\tilde}{A}^{-1} \\\skew{7}{\tilde}{A}^{-1} χ & \skew{7}{\tilde}{A}^{-1}\end{array}\right). 
\end{equation*}
\end{lem}
\begin{proof}
The first part is obvious, the second part is an easy calculation and Proposition~\ref{prop:bdet} implies for $k=1$ that $\det J'=1$.
\end{proof}
Even though ingredients and the matrix $J'$ itself were known already before, the crucial new step for the twistor construction is to recognize the following.
\begin{thm}
$J'$ is obtained from $J$ by a B\"acklund transformation.
\end{thm}
\begin{proof}
Consider without loss of generality $J'$ adapted to the horizon. We decompose $J$ according to \eqref{eq:bdecomp}
\begin{equation*}
\renewcommand{\arraystretch}{1.5}
J=\left(\begin{array}{cc}A^{-1}-\skew{4}{\tilde}{B} \skew{7}{\tilde}{A} B & -\skew{4}{\tilde}{B} \skew{7}{\tilde}{A} \\ \skew{7}{\tilde}{A} B & \skew{7}{\tilde}{A}\end{array}\right)
\end{equation*}
with $\skew{7}{\tilde}{A}$ as above the spatial part of $J$, 
\begin{equation*}
B=-\skew{4}{\tilde}{B} ^{\mathrm{t}}=\skew{7}{\tilde}{A}^{-1} \left(\begin{array}{c}J_{01} \\\vdots \\J_{0,n-3}\end{array}\right),
\end{equation*}
and $k=1$, $\skew{2}{\tilde}{k}=n-3$. In order to calculate $A$ remember that the inverse of a general matrix was given
\begin{equation*}
\skew{7}{\tilde}{A}^{-1} = \frac{1}{\det \skew{7}{\tilde}{A}} (\tilde \Abb)^{\mathrm{t}},
\end{equation*}
where $\tilde \Abb = \left((-1)^{i+j} \det(\skew{7}{\tilde}{\Acal}_{ij})\right)$ with $\skew{7}{\tilde}{\Acal}_{ij}$ being the matrix obtained from $\skew{7}{\tilde}{A} = (\skew{7}{\tilde}{A}_{ij})$ by crossing out the $i^{\text{th}}$ row and $j^{\text{th}}$ column. 

As a first step towards $A$ we do a Laplace expansion of $\det J$ along the first row
\begin{equation*}
\det J = J_{00} \det \skew{7}{\tilde}{A} + ∑ _{k=1}^{n-3} (-1)^{k+2}J_{0k} \det {\Jcal}_{0k},
\end{equation*}
and by another Laplace expansion (also using that $J$ is symmetric)
\begin{equation*}
\det {\Jcal}_{0k} = ∑ _{l=1}^{n-3}(-1)^{l+1} J_{0l} \det \skew{7}{\tilde}{\Acal}_{kl}.
\end{equation*}
Substituting the latter one in the expression for $\det J$ yields
\begin{align*}
\det J & = J_{00} \det \skew{7}{\tilde}{A} - ∑ _{k,l=1}^{n-3} (-1)^{k+l}J_{0k} \det \skew{7}{\tilde}{\Acal}_{kl} J_{0l} \\
	& = J_{00} \det \skew{7}{\tilde}{A} - \det \skew{7}{\tilde}{A} 
	\left(\begin{array}{ccc}J_{01} & \cdots & J_{0,n-3}\end{array}\right)
	\skew{7}{\tilde}{A}^{-1}
	\left(\begin{array}{c}J_{01} \\ \vdots \\J_{0,n-3}\end{array}\right),
\end{align*}
which is equivalent to
\begin{equation*}
\frac{\det J}{\det \skew{7}{\tilde}{A}} = J_{00} +\skew{4}{\tilde}{B}^{\mathrm{t}}\skew{7}{\tilde}{A} B = A^{-1}.
\end{equation*}
Hence
\begin{equation*}
A = \frac{\det \skew{7}{\tilde}{A}}{\det J} = - r^{-2} \det \skew{7}{\tilde}{A},
\end{equation*}
and we can, in line with the B\"acklund transformation, define 
\begin{equation*}
A'= \skew{7}{\tilde}{A}^{-1},\quad  \skew{7}{\tilde}{A}' = -r^{-2} A^{-1} = \det \skew{7}{\tilde}{A}^{-1}. 
\end{equation*}

To complete our proof of showing that $J'$ is a B\"acklund transform we have to verify that the integrability conditions yield the right expression for the twist potentials. They take the form
\begin{align*}
∂ _{r} \skew{4}{\tilde}{B}' & = r^{-1} \det (\skew{7}{\tilde}{A}) \skew{7}{\tilde}{A}\, ∂ _{z}\left(\skew{7}{\tilde}{A}^{-1}\left(\begin{array}{c}J_{01} \\ \vdots \\J_{0,n-3}\end{array}\right) \right), \\
∂ _{z} \skew{4}{\tilde}{B}' & = -r^{-1} \det (\skew{7}{\tilde}{A}) \skew{7}{\tilde}{A}\, ∂ _{r}\left(\skew{7}{\tilde}{A}^{-1}\left(\begin{array}{c}J_{01} \\ \vdots \\J_{0,n-3}\end{array}\right) \right),
\end{align*}
and $B'=-(\skew{4}{\tilde}{B}')^{\mathrm{t}}$. We are done if we can show that $\skew{4}{\tilde}{B}'=χ$ because then our constructed $J'$ coincides with the B\"acklund transform (apart from swapping the first and last row, respectively column, which does not affect our considerations). 

So, it remains to show
\begin{equation} \label{eq:intvstwist}
\begin{split}
∂ _{r}χ = ω_{r} = r^{-1} \det (\skew{7}{\tilde}{A}) \skew{7}{\tilde}{A}\, ∂ _{z}\left(\skew{7}{\tilde}{A}^{-1}\left(\begin{array}{c}J_{01} \\ \vdots \\J_{0,n-3}\end{array}\right) \right), \\
∂ _{z}χ = ω_{z} = -r^{-1} \det (\skew{7}{\tilde}{A}) \skew{7}{\tilde}{A}\, ∂ _{r}\left(\skew{7}{\tilde}{A}^{-1}\left(\begin{array}{c}J_{01} \\ \vdots \\J_{0,n-3}\end{array}\right) \right).
\end{split}
\end{equation}

Starting from the left-hand side of \eqref{eq:intvstwist} we Laplace expand the determinant in \eqref{eq:compsomega} twice; first along the last column and then along the first row
\begin{align*}
r \cdot ω_{Ir} & = (-1)^{n-1}(-1)^{1+(n-2)} ∂ _{z}J_{0I} \cdot \det \skew{7}{\tilde}{A} \\
	& \hspace{0.4cm}+ (-1)^{n-1}∑ _{l=1}^{n-3}(-1)^{(l+1)+(n-2)} ∂ _{z}\skew{7}{\tilde}{A}_{l,I} \cdot  \det 
\renewcommand{\arraystretch}{1.5}	
	\left(\begin{array}{ccc}J_{01} &  \cdots & J_{0,n-3} \\ \vdots &  &  \vdots \\\widehat{J_{l,1}} &  \cdots & \widehat{J_{l,n-3}} \\ \vdots &  &  \vdots \\J_{n-3,1} &  \cdots & J_{n-3,n-3}\end{array}\right) \\
	& = ∂ _{z}J_{0I} \cdot \det \skew{7}{\tilde}{A}+∑ _{l=1}^{n-3}(-1)^{l} ∂ _{z}\skew{7}{\tilde}{A}_{lI} \left(∑ _{k=1}^{n-3}(-1)^{k+1}J_{0k}\cdot \det \skew{7}{\tilde}{A}_{kl}\right)\\
	& = \det \skew{7}{\tilde}{A} \left[∂ _{z} J_{0I}- 
	\left(\begin{array}{ccc}∂ _{z} \skew{7}{\tilde}{A}_{1I} &  \cdots & ∂ _{z} \skew{7}{\tilde}{A}_{n-3,I}\end{array}\right)
	\skew{7}{\tilde}{A}^{-1}
	\left(\begin{array}{c}J_{01} \\ \vdots \\J_{0,n-3}\end{array}\right)
	\right].
\end{align*}
However, the last line is the $I^{\text{th}}$ component of
\begin{align*}
r \cdot  ω_{r} & = \det \skew{7}{\tilde}{A} \left[\left(\begin{array}{c}∂ _{z} J_{01} \\ \vdots \\∂ _{z} J_{0,n-3}\end{array}\right)
	- (∂ _{z} \skew{7}{\tilde}{A})	\skew{7}{\tilde}{A}^{-1}
	\left(\begin{array}{c}J_{01} \\ \vdots \\J_{0,n-3}\end{array}\right) 
	\right]\\
	& = \det \skew{7}{\tilde}{A} \cdot \skew{7}{\tilde}{A} \, ∂ _{z} \left[\skew{7}{\tilde}{A}^{-1}
	\left(\begin{array}{c}J_{01} \\ \vdots \\J_{0,n-3}\end{array}\right) 
	\right], 
\end{align*}
which corresponds to the right-hand side of \eqref{eq:intvstwist}. An analogous computation verifies the second equation of \eqref{eq:intvstwist}.
\end{proof}
This completes the justification to call the $J'$ defined in \eqref{eq:highernst} Ernst potential --- by Proposition~\ref{prop:btsolyang} it satisfies Yang's equation and by Lemma~\ref{lem:highernstprop} it is symmetric and has unit determinant.

\section{Properties of $P$ and the Bundle} \label{sec:adapt}

Having different adaptations of $J'$ which are related to different extenstions of $V$ the question arises whether these yield equivalent bundles over the reduced twistor space. To see that is true is not very hard using results in \citet{Fletcher:1990aa}.
\begin{prop}[Proposition~3.1 in \citet{Fletcher:1990aa}] \label{prop:EquBdls}
Suppose $E\to \Rcal_{V}$, the bundle corresponding to a solution $J$ of Yang's equation on the set $V$, can be represented as the pullback of the bundles $E^{1}\to \Rcal_{V'_{1}}$ and $E^{2}\to \Rcal_{V'_{2}}$, where $V'_{1}$ and $V'_{2}$ are simply connected open sets which intersect the real axis in distinct intervals. Then $E^{1}$ and $E^{2}$ are the pullbacks of a bundle $\tilde E \to  \Rcal_{\tilde V}$ where $\tilde V=V_{1}\cup V_{2}$. Moreover, we can express $\tilde E$ in standard form in two different ways; and one of the collections of patching matrices is identical to the collection used to describe $E^{1}$ and the other collection is identical to that used to describe $E^{2}$.
\end{prop}
\begin{proof}[Sketch of Proof]
The proof in \citet{Fletcher:1990aa} only makes use of the construction for the reduced twistor space and not of the rank of the bundle. So, it carries over to higher dimensions.

Roughly speaking the argument is that for two different adaptations with reduced twistor spaces $\Rcal_{V_{1}}$ and $\Rcal_{V_{2}}$, as in Figure~\ref{fig:adaptation}, their pullbacks to $\Rcal_{V}$ have to be equivalent (in the sense of Proposition~\ref{prop:bundle}), since in both cases it yields the bundle defined by the solution $J$ for \eqref{eq:redyang}. But if they are equivalent on $\Rcal_{V}$ and are both analytical continuations, then they have to belong to the same bundle and are only two different representations of it. 
\end{proof}

The bridge between Ernst potential and patching matrix is built by analytic continuation. 
\begin{prop}[Proposition~7.2 in \citet{Woodhouse:1988ek}] \label{prop:AnalyCont}
In five as well as in four dimensions if $P$ is a patching matrix of an axis-regular Ernst potential $J'$ on $V$, then $J'$ is analytic on (a choice of) $V'$ and $J'(0,z)=P(z)$ for real $z$. 
\end{prop}
\begin{proof}
This is literally the same as for \citet{Woodhouse:1988ek}, since it does not make use of the rank of the bundle. Note that even though in the statement there the assumption of $J$ being positive definite is made, it in fact is not necessary for the proof.
\end{proof}

\begin{cor} \label{cor:singofP}
A patching matrix $P$ has real singularities, that is points $z\in \mathbb{R}$ where an entry of $P$ has a singularity, at most at the nuts of the rod structure.
\end{cor}
\begin{proof}
Suppose $P$ corresponds to the bundle $\tilde E \to  \Rcal_{\tilde V}$, where $\tilde V$ is the maximally extended region over which the spheres can be identified. Suppose further that $P$ has a real singularity $\skew{2} \tilde a \in \mathbb{R}$ which is not one of the nuts, say without loss of generality $\skew{2} \tilde a\in (a_{i},a_{i+1})$. Then from Proposition~\ref{prop:EquBdls} we know that $\tilde E$ can on $V_{i}$ be expressed in standard form. But using Proposition~\ref{prop:AnalyCont} that means that $P(z)=J'_{i}(z,0)$ for $z\in (a_{i},a_{i+1})$. On the other hand we have seen earlier already that $J'_{i}$ is regular on $(a_{i},a_{i+1})$ and only becomes singular when approaching the nut. Contradiction!

For every other bundle $E→\Rcal_{V}$ it must be $V⊆\tilde V$, hence $E$ is the pullback of $\tilde E$ and as such the patching matrix of $E$ cannot have poles where the patching matrix of $\tilde E$ has not.
\end{proof}

As in \citet{Fletcher:1990aa} we sometimes also call the singularities of $P$ \textit{double points}, because the singularities of $P$ are the points where the two Riemann spheres of the reduced twistor space cannot be identified.

\begin{prop} \label{prop:simplepoles}
The real singularities of a patching matrix $P$ are simple poles.
\end{prop}
\begin{proof}
We have seen above that on the real axis $r=0$ the singularities of $P$ are caused by the term $\det \skew{7}{\tilde}{A}$. So, consider the rod $(a_{i},a_{i+1})$ where $\skew{7}{\tilde}{A}$ has full rank. The determinant of a matrix equals the product of its eigenvalues. Furthermore, towards the nuts $a_{i}$, $a_{i+1}$ we know, also from above, that the rank of $\skew{7}{\tilde}{A}$ drops precisely by one which is the case if and only if $\det \skew{7}{\tilde}{A}$ contains the factors $z-a_{i}$ and $z-a_{i+1}$, respectively, with multiplicity one. 
\end{proof}
Note that if $P$ has a pole with $r≠0$, then it is obviously not splittable at this point, see \eqref{eq:patmat2}, hence we do not obtain $J$ at this point by the splitting procedure which means the metric might be singular at this point. We would like to exclude such situations. However, referring again to \eqref{eq:patmat2} it is evident that the splitting procedure in general does break down for $r→0$. So, one has to study this limit by means of other tools, see Section~\ref{sec:Pnearaxis}, and we have seen in Example~\ref{ex:schwarz} even for regular space-times it is possible that $P$ has real poles. But as in Chapter~\ref{ch:bundles} we assume that there are only finitely many of them.

\section[Patching Matrix near the Axis]{Patching Matrix near the Axis and Twistor Data Integers} \label{sec:Pnearaxis}

In this last part we will sketch a generalization of the asymptotic formula of $J$ for $r\to 0$. Originally, the result in four dimensions goes back to \citet{Ward:1983yg}, but here we will follow the lines of \citet[Sec.~2.4]{Fletcher:1990aa}. 

In Section~\ref{sec:highErnst} we defined
\begin{equation*}
\renewcommand{\arraystretch}{1.4}
J'=\left(\begin{array}{cc}\hphantom{-}g & -gχ^{\mathrm{t}} \\-gχ & \skew{7}{\tilde}{A} + g\cdot χχ^{\mathrm{t}}\end{array}\right),
\end{equation*}
where $g=(\det \skew{7}{\tilde}{A})^{-1}$. We have shown above that $P(w)$ is the analytic continuation of $J(r=0,z)$, hence we can assume that it also takes the form
\begin{equation*}
\renewcommand{\arraystretch}{1.4}
P(w)=\left(\begin{array}{cc}\hphantom{-}g & -gχ^{\mathrm{t}} \\-gχ & \skew{7}{\tilde}{A} + g\cdot χχ^{\mathrm{t}}\end{array}\right),
\end{equation*}
but the entries being functions of $w$. Analogously to Chapter~\ref{ch:bundles} we have
\begin{equation*}
P_{02}^{}=P_{13}^{-1}=\diag \big((2w)^{p_{0}},(2w)^{p_{1}} ,(2w)^{p_{2}}\big),\ P_{23}^{}=P.
\end{equation*}
Furthermore, we assume for the integers $p_{0}\geq p_{1}\geq p_{2}$, which is just a question of ordering. Following exactly the lines as in Chapter~\ref{ch:bundles} we obtain
\begin{equation*}
Q_{0}^{\vphantom{-1}}Q_{1}^{-1}=\left(\begin{array}{@{}c@{\hspace{-0.2mm}}c@{\hspace{-0.2mm}}c@{}}\left(\dfrac{r}{ζ}\right)^{p_{0}} &   &  \\ & \left(\dfrac{r}{ζ}\right)^{p_{1}} &  \\ &  & \left(\dfrac{r}{ζ}\right)^{p_{2}}\end{array}\right)P(w)\left(\begin{array}{@{}c@{\hspace{-0.2mm}}c@{\hspace{-0.2mm}}c@{}} \vphantom{\left(\dfrac{r}{ζ}\right)^{p_{n-3}}}(-rζ)^{p_{0}} &  &  \\ & \vphantom{\left(\dfrac{r}{ζ}\right)^{p_{2}}}\left(-rζ\right)^{p_{1}} &  \\ &  & \vphantom{\left(\dfrac{r}{ζ}\right)^{p_{2}}}(-rζ)^{p_{2}}\end{array}\right),
\end{equation*}
where the left-hand side is a function of $ζ$. With
\begin{equation*}
R_{01}=\left(\begin{array}{@{}c@{\hspace{-0.2mm}}c@{\hspace{-0.2mm}}c@{}}
ζ^{-p_{0}} &  &  \\ & ζ^{-p_{1}} &  \\ &  & ζ^{-p_{2}}\end{array}\right)
P(w)
\left(\begin{array}{@{}c@{\hspace{-0.2mm}}c@{\hspace{-0.2mm}}c@{}}(-ζ)^{p_{0}} &  &  \\ & (-ζ)^{p_{1}} &  \\ &  & (-ζ)^{p_{2}}\end{array}\right)
\end{equation*}
this can be rewritten as
\begin{equation*}
Q_{0}^{\vphantom{-1}}Q_{1}^{-1}=\left(\begin{array}{@{}c@{\hspace{-0.2mm}}c@{\hspace{-0.2mm}}c@{}}r^{p_{0}} &  &  \\ & r^{p_1} &  \\ &  & r^{p_{2}}\end{array}\right)R_{01}\left(\begin{array}{@{}c@{\hspace{-0.2mm}}c@{\hspace{-0.2mm}}c@{}}r^{p_{0}} &  &  \\ & r^{p_1}&  \\ &  & r^{p_{2}}\end{array}\right).
\end{equation*}
Suppose we can split $g=g_{0}^{\vphantom{-1}}g_{1}^{-1}$, where $g_{0}^{\vphantom{-1}}$ is holomorphic in $ζ$, that is on $S_{0}$, and $g_{1}^{-1}$ is holomorphic in $ζ^{-1}$, that is on $S_{1}$. This may be achieved by the use of a Laurent series. We claim that then $R_{01}$ can be split as follows
\begin{align*}
R_{01} & = \underbrace{\left(\begin{array}{@{}c@{\hspace{-0.2mm}}c@{\hspace{-0.2mm}}c@{}}ζ^{-p_{0}} &  &  \\ & ζ^{-p_{1}} &  \\ &  & ζ^{-p_{2}}\end{array}\right)
\left(\begin{array}{cc} \hphantom{-}g_{0}^{\vphantom{-1}} & 0 \\ \rule{0cm}{5mm} -g_{0}^{\vphantom{-1}} χ & \skew{7}{\tilde}{A}_{0}^{\vphantom{-1}}\end{array}\right)
\left(\begin{array}{@{}c@{\hspace{-0.2mm}}c@{\hspace{-0.2mm}}c@{}}ζ^{p_{0}} &  &  \\ & ζ^{p_{1}} &  \\ &  & ζ^{p_{2}}\end{array}\right)}_{\eqqcolon M_{0}^{}(ζ)} \times  \\
 & \hspace{0.4cm} \underbrace{\left(\begin{array}{@{}c@{\hspace{-0.2mm}}c@{\hspace{-0.2mm}}c@{}}ζ^{-p_{0}} &  &  \\ & ζ^{-p_1} &  \\ &  & ζ^{-p_{2}}\end{array}\right)
 \left(\begin{array}{cc}g_{1}^{-1} & -g_{1}^{-1} χ^{\mathrm{t}} \\ \rule{0cm}{5mm} 0 & \hphantom{-}\skew{7}{\tilde}{A}_{1}^{-1}\end{array}\right)
\left(\begin{array}{@{}c@{\hspace{-0.2mm}}c@{\hspace{-0.2mm}}c@{}}(-ζ)^{p_{0}} &  &  \\ & (-ζ)^{p_{1}} &  \\ &  & (-ζ)^{p_{2}}\end{array}\right)}_{\eqqcolon M_{1}^{-1}(ζ)},
\end{align*}
where $\skew{7}{\tilde}{A}_{0}^{\vphantom{-1}}$ is holomorphic in $ζ$, that is on $S_{0}^{\vphantom{-1}}$, and lower triangular, and $\skew{7}{\tilde}{A}_{1}^{-1}$ is holomorphic in $ζ^{-1}$, that is on $S_{1}^{\vphantom{-1}}$, and upper triangular. All these steps generalize straight-forwardly to more than five dimensions. In the case of $n$ dimensions the range of integers is $p_0≥…≥p_{n-3}$. In the following paragraphs it does not simplify the notation if we restrict to $n=5$, therefore we will consider the more general setting.

The matrices $M_{0}$, $M_{1}$ are not necessarily holomorphic and non-singular on $S_{0}$ or $S_{1}$, respectively. Yet, we can expand $χ(w)$ and the off-diagonal entries of $\skew{7}{\tilde}{A}_{0}^{\vphantom{-1}}(w)$ around $w=z$ which yields\footnote{Remember $w = \frac{1}{2}r(ζ^{-1}-ζ)+z$.}
\begin{equation*}
χ(w) = ∑ _{k=0}^{\infty } \frac{χ^{(k)}(z)r^{k}}{2^{k}k!}(ζ^{-1}-ζ)^{k},
\end{equation*}
and similarly for the entries of $\skew{7}{\tilde}{A}_{0}^{\vphantom{-1}}(w)$. Here $χ^{(k)}$ denotes the $k^{\mathrm{th}}$ derivative. 

Denote by $J_{kl,0}$ the $(kl)$-entry of $\skew{7}{\tilde}{A}_{0}^{\vphantom{-1}}$ ($k,l=1,\dotsc ,n-3$). Now consider the $(kl)$-entry of $M_{0}$ where $k,l=0,\dotsc ,n-3$ and without loss of generality $k>l$. This entry gets multiplied by $ζ^{p_{k}-p_{l}}$, where $p_{k}-p_{l}\geq 0$ by the chosen order of the $p_{i}$ and because $M_{0}$ is lower triagonal. Hence, all the terms up to $\Ocal(r^{p_{k}-p_{l}})$ in the Taylor series for the $(kl)$-entry are indeed holomorphic on $S_{0}$ and the terms up to $\Ocal(r^{p_{k}-p_{l}-1})$ vanish at $ζ=0$. Therefore the lowest order contribution for $r\to 0$ for this entry is
\begin{equation}\label{eq:asymptentry}
\begin{split}
- & \frac{r^{p_{0}-p_{l}}}{2^{p_{0}-p_{l}}(p_{0}-p_{l})!}χ_{l}^{(p_{0}-p_{l})}(z)\cdot g_{0}, \quad \text{or}\\
& \frac{r^{p_{k}-p_{l}}}{2^{p_{k}-p_{l}}(p_{k}-p_{l})!}J_{kl,0}^{(p_{k}-p_{l})}(z)\cdot g_{0}\hphantom{,} \quad \text{for } k\geq 1. 
\end{split}
\end{equation}
Analogous arguments work for $M_{1}^{-1}$. Thus the obtained asymptotic behaviour as $r\to 0$ is
\begin{equation*}
J→\left(\begin{array}{@{}c@{\hspace{-0.2mm}}c@{\hspace{-0.2mm}}c@{}}r^{p_0} &  &  \\ & \ddots &  \\ &  & r^{p_{n-3}}\end{array}\right)
\underbrace{\left(\begin{array}{@{}cc@{}}g & -g Υ^{\mathrm{t}} \\-g Υ \rule{0cm}{5mm} & g Υ Υ^{\mathrm{t}} - \skew{7}{\tilde}{\Acal}\end{array}\right)
\left(\begin{array}{@{}c@{\hspace{-2mm}}c@{\hspace{-2mm}}c@{}}(-1)^{p_0} &  &  \\ & \ddots &  \\ &  & (-1)^{p_{n-3}}\end{array}\right)}_{\eqqcolon B}
\left(\begin{array}{@{}c@{\hspace{-0.2mm}}c@{\hspace{-0.2mm}}c@{}}r^{p_0} &  &  \\ & \ddots &  \\ &  & r^{p_{n-3}}\end{array}\right),
\end{equation*}
where 
\begin{equation*}
Υ=\left(\begin{array}{c}Υ^{(1)} \\ \vdots \\Υ^{(n-3)}\end{array}\right) \quad  \text{with } Υ^{(l)}=\frac{r^{p_{0}-p_{l}}}{2^{p_{0}-p_{l}}(p_{0}-p_{l})!}χ^{(p_{0}-p_{l})}(z),
\end{equation*}
and $\skew{7}{\tilde}{\Acal}=\skew{7}{\tilde}{\Acal}_{0}^{\vphantom{-1}}\skew{7}{\tilde}{\Acal}_{1}^{-1}$ with the entries of $\skew{7}{\tilde}{\Acal}_{0}^{\vphantom{-1}}$ as in \eqref{eq:asymptentry} and analogously for $\skew{7}{\tilde}{\Acal}_{1}^{-1}$. 

Note that all exponents of $r$ in $B$ are greater than 0. Thus on a part of the axis $r=0$ where $A$ is not singular, that is everywhere apart from the nuts, we see that boundedness of $J$ implies $p_{i}\geq 0$. 

It remains to justify the claim about the splitting of $R_{01}$. As \cite[Sec.~2.4]{Fletcher:1990aa} shows it is not hard to find a splitting in the four-dimensional case under the assumption that $g$ splits appropriately
\begin{align*}
&\underbrace{\left(\begin{array}{@{}cc@{}}
g(-1)^{p_{0}} & -gχ_{1}ζ^{p_{1}-p_{0}}(-1)^{p_{1}} \\ \rule{0cm}{7mm}
-gχ_{1}^{\vphantom{1}}ζ^{p_{0}-p_{1}}(-1)^{p_{0}} & (J_{00}^{\vphantom{1}}+gχ_{1}^{2})(-1)^{p_{1}}
\end{array}\right)}_{D} \\[0.1cm]
& \hspace{0.2cm}=
\underbrace{\left(\begin{array}{@{}cc@{}}
g_{0}^{\vphantom{1}}(-1)^{p_{0}} & 0 \\ \rule{0cm}{7mm}
-g_{0}χ_{1}ζ^{p_{0}-p_{1}}(-1)^{p_{0}} & J_{00}(-1)^{p_{1}}
\end{array}\right)}_{D_{0}^{\vphantom{-1}}}
\underbrace{\left(\begin{array}{@{}cc@{}}
g_{1}^{\vphantom{1}} & -g_{1}^{-1}χ_{1}^{\vphantom{1}}ζ^{p_{1}-p_{0}}(-1)^{p_{1}-p_{0}} \\ \rule{0cm}{7mm}
0 & 1
\end{array}\right)}_{D_{1}^{-1}}.
\end{align*}
To obtain a splitting in the five-dimensional case we make the following inductive ansatz
\begin{align*}
\left(\begin{array}{cc} D & B \\ \rule{0cm}{5mm} B^{\mathrm{t}} & C \end{array}\right) & =
\left(\begin{array}{cc} D_{0}^{\vphantom{-1}} & 0 \\ \rule{0cm}{5mm} B_{0}^{\mathrm{t}\vphantom{-1}} & C_{0}^{\vphantom{-1}} \end{array}\right)
\left(\begin{array}{cc} D_{1}^{-1} & B_{1}^{-1} \\\rule{0cm}{5mm}  0 & C_{1}^{-1} \end{array}\right)\\
& = \left(\begin{array}{cc} \hphantom{-}g_{0}^{\vphantom{-1}} & 0 \\ \rule{0cm}{5mm} -g_{0}^{\vphantom{-1}} χ & \skew{7}{\tilde}{A}_{0}^{\vphantom{-1}}\end{array}\right)
\left(\begin{array}{cc}g_{1}^{-1} & -g_{1}^{-1} χ^{\mathrm{t}} \\ \rule{0cm}{5mm} 0 & \hphantom{-}\skew{7}{\tilde}{A}_{1}^{-1}\end{array}\right)
\end{align*}
where $D$ is the upper ($2×2$)-block whose splitting $D=D_{0}^{\vphantom{-1}}D_{1}^{-1}$ we know. The indices nought and one bear the same meaning as above. Then $B_{0}^{\vphantom{-1}}$ and $B_{1}^{-1}$ are fixed by $D$, $D_{0}^{\vphantom{-1}}$, $D_{1}^{-1}$ and $B$. It is not hard to see that $B_{0}^{\vphantom{-1}}$, $B_{1}^{-1}$, $C_{0}^{\vphantom{-1}}$, $B_{1}^{-1}$ can be found such that they are holomorphic on the appropriate sphere. One has to ensure that certain entries of $\skew{7}{\tilde}{A}$ are non-zero, which can without loss of generality be assumed if we choose the adaptation to the right part of the axis. This inductive argument can then be continued to obtain splittings higher dimensions if needed. 

For $n=5$ we get for example
\begin{equation}
\begin{array}{ll}
B_{0}^{\vphantom{-1}} = \left(\begin{array}{c}-g_{0}^{\vphantom{1}}χ_{2}^{\vphantom{1}}ζ^{p_{1}-p-{2}}(-1)^{p_{0}} \\ \rule{0cm}{7mm} J_{01}ζ^{p_{1}-p_{2}}(-1)^{p_{1}}\end{array}\right), &
B_{1}^{-1} = \left(\begin{array}{c}-g_{1}^{-1}χ_{2}^{\vphantom{1}}ζ^{p_{2}-p_{0}}(-1)^{p_{2}-p_{0}} \\ \rule{0cm}{7mm} \dfrac{J_{01}}{J_{00}} ζ^{p_{2}-p_{1}}(-1)^{p_{2}-p_{1}}\end{array}\right) \\ \rule{0cm}{9mm}
C_{0}^{\vphantom{-1}} = -\dfrac{J_{01}^{2}}{J_{00}}(-1)^{p_{2}}+J_{11}(-1)^{p_{2}}, & C_{1}^{-1} = 1.
\end{array}
\end{equation}
Note that we assumed the metric to be analytic throughout the space-time (up to and including the axis), so that it is not necessary to split $J_{kl}$ or the twist potentials.

Using the freedom in choosing the splitting matrices as in Chapter~\ref{ch:bundles} so that $Q_{1}^{-1}(\infty )= \id$ and assuming that $\det P=1$ we obtain from the above splitting formula
\begin{equation*}
\det \left(\left(\begin{array}{@{}c@{\hspace{-0.2mm}}c@{\hspace{-0.2mm}}c@{}}ζ^{-p_{0}} &  &  \\ & \ddots &  \\ &  & ζ^{-p_{n-3}}\end{array}\right)P(w)\left(\begin{array}{@{}c@{\hspace{-0.2mm}}c@{\hspace{-0.2mm}}c@{}}(-ζ)^{p_{0}} &  &  \\ & \ddots &  \\ &  & (-ζ)^{p_{n-3}}\end{array}\right)\right) = (-1)^{p_{0}+\dotsc +p_{n-3}},
\end{equation*}
and thus
\begin{equation*}
-r^{2} = \det J = \det Q_{0}(0) = \left(-r^{2}\right)^{p_{0}+\dotsc +p_{n-3}}. 
\end{equation*}
So, if all $p_{i}$ are non-negative this forces
\begin{equation*}
p_{0}=1,\ p_{1}=0, \dotsc ,\ p_{n-3}=0. 
\end{equation*}
In this case the asymptotic formula simplifies to
\begin{equation*}
J=\left(\begin{array}{@{}c@{\hspace{0.2cm}}c@{\hspace{0.2cm}}c@{\hspace{0.2cm}}c@{}}r &  &  &  \\ & 1 &  &  \\ &  &  \ddots &  \\ &  &  & 1\end{array}\right)
\left(\begin{array}{@{}cc@{}}-g & -g Υ^{\mathrm{t}} \\ \rule{0cm}{5mm} \hphantom{-}g Υ & g Υ Υ^{\mathrm{t}} - \skew{7}{\tilde}{A}\end{array}\right)
\left(\begin{array}{@{}c@{\hspace{0.2cm}}c@{\hspace{0.2cm}}c@{\hspace{0.2cm}}c@{}}r &  &  &  \\ & 1 &  &  \\ &  &  \ddots &  \\ &  &  & 1\end{array}\right)
+ \text{ higher orders in } r,
\end{equation*}
where 
\begin{equation*}
Υ=\left(\begin{array}{c}Υ^{(1)} \\ \vdots \\Υ^{(1)}\end{array}\right) \quad  \text{with } Υ^{(1)}=\frac{r}{2}\left.\frac{\drm}{\drm w}\right|_{w=z}χ(w)
\end{equation*}
and $\skew{7}{\tilde}{A}$ is the same as in the definition of $J'$. 

If $\det J =1$ as for the Ernst potential we get back our analytic continuation formula
\begin{equation*}
J(r,z) = P(z) + \text{ higher orders in } r.
\end{equation*}

\chapter{Patching Matrix for Relevant Examples} \label{ch:Pexamples}

In the previous chapter we have discussed how to extend the twistor construction to five dimensions. Let us now look at some examples. For all the space-times we first need to know the metric in the ${\sigma}$-model form, that is we have to calculate $J(r,z)$. After that the Ernst potential, respectively the patching matrix, can be computed which mainly means determining the twist potentials on the axis. The easiest to start with is flat space in five dimensions. 

\section{Five-Dimensional Minkowski Space} \label{sec:exMinkrodstr}
The real five-dimensional Minkowski space is the manifold $\mathbb{R}^{5}$ with metric given in the standard basis as
\begin{equation*}
\drm s^{2} = -\drm x_{0}^{2}+\drm x_{1}^{2}+\drm x_{2}^{2}+\drm x_{3}^{2}+\drm x_{4}^{2}.
\end{equation*}
Along the lines of \cite[Sec.~4]{Hollands:2008fp} we define the Killing vector fields $X_{1}=∂_{\varphi}$, $X_{2}=∂_{\psi}$ as generating the rotations in the (12)-plane and the (34)-plane, that is 
\begin{equation*}
{\varphi} = \arctan \left(\frac{x_{1}}{x_{2}}\right), \quad {\psi} = \arctan \left(\frac{x_{3}}{x_{4}}\right).
\end{equation*}
To obtain the ${\sigma}$-model coordinates we define
\begin{equation*}
r=R_{1}^{\vphantom{1}}R_{2}^{\vphantom{1}}, \quad z=\frac{1}{2}(R_{1}^{2}-R_{2}^{2}),
\end{equation*}
where $R_{1}^{\vphantom{1}}=\sqrt{x_{1}^{2}+x_{2}^{2}}$ and $R_{2}^{\vphantom{1}}=\sqrt{x_{3}^{2}+x_{4}^{2}}$. This yields 
\begin{equation*}
\sqrt{r^{2}+z^{2}}=\frac{1}{2}(R_{1}^{2}+R_{2}^{2}),
\end{equation*}
thus we get the relations
\begin{equation*}
R_{1}^{2}= \sqrt{r^{2}+z^{2}}+z, \quad R_{2}^{2}= \sqrt{r^{2}+z^{2}}-z. 
\end{equation*}
Now by the definition of the symmetries we have $x_{1}=R_{1}\cos {\varphi}$, $x_{2}=R_{1}\sin {\varphi}$ so that
\begin{equation*}
\drm x_{1}^{2}+\drm x_{2}^{2}= \drm R_{1}^{2} + R_{1}^{2} \,\drm {\varphi}^{2} = \drm R_{1}^{2} + \left(\sqrt{r^{2}+z^{2}}+z\right) \drm {\varphi}^{2}, 
\end{equation*}
and similarly
\begin{equation*}
\drm x_{3}^{2}+\drm x_{4}^{2}= \drm R_{2}^{2} + R_{2}^{2} \,\drm {\psi}^{2} = \drm R_{2}^{2} + \left(\sqrt{r^{2}+z^{2}}-z\right) \drm {\psi}^{2},
\end{equation*}
whence
\begin{equation*}
\drm r^{2}+\drm z^{2} = 2\left(R_{1}^{2}+R_{2}^{2}\right)\left( \drm R_{1}^{2} + \drm R_{2}^{2}\right).
\end{equation*}
Hence,
\begin{equation*}
{\erm}^{2ν} = \frac{1}{2(R_{1}^{2}+R_{2}^{2})} = \frac{1}{2\sqrt{r^{2}+z^{2}}}.
\end{equation*}
Summing up, the metric in ${\sigma}$-model form is given by
\begin{align} \label{eq:JMink}
J(r,z) & =
\renewcommand\arraystretch{1.5}
\left(\begin{array}{ccc}-1 & 0 & 0 \\
\hphantom{-}0 & z+\sqrt{r^{2}+z^{2}} & 0 \\
\hphantom{-}0 & 0 & -z+\sqrt{r^{2}+z^{2}}
\end{array}\right)\\
\text{and} \quad {\erm}^{2ν} & = \frac{1}{2\sqrt{r^{2}+z^{2}}}. 
\end{align}
Since $\dim(\ker J(0,z))>1$ only for $z=0$, we can read off that the metric admits two semi-infinite rods, namely $\opint{-\infty }{0}$ and $\opint{0}{\infty }$. Because $J$ is diagonal ${\theta}_{I}\wedge \drm {\theta}_{I}=0$ as it contains an overall factor of the form $\drm y_{i}\wedge \drm y_{i}$. It follows that the twist potentials are constant on all of Minkowski space and thereby we obtain the patching matrix as 
\begin{equation} \label{eq:PMink}
P_{\pm}(z)=\diag \left(\mp \frac{1}{2z},-1,\pm  2z\right),
\end{equation}
where the upper sign combination is for $P$ adapted to $z>0$, and the lower one for $z<0$.

\section{Twist Potentials on the Axis}

As part of the algorithm for obtaining $P(z)$ from the metric we need to calculate the twist potentials just on the axis. Explicit expressions for twist potentials have been obtained for example in \cite{Tomizawa:2004aa, Tomizawa:2009aa}, but these are given only for the black ring and not in Weyl  coordinates which we need here. Therefore it is simpler to rederive some results, not only for completeness but also for providing a way of calculating the twist potentials on the axis for other space-times where they are not yet in the literature.

First we derive general formulae. Assume that the metric takes the form
\begin{align*}
\drm s^{2} & = J_{00} \,\drm t^{2} + 2 J_{01}\, \drm t \drm {\varphi} + 2 J_{02}\, \drm t \drm {\psi} + J_{11}\, \drm {\varphi} ^{2} + 2 J_{12}\,\drm {\varphi} \drm {\psi} \\
& \hspace{0.4cm} + J_{22}\, \drm {\psi}^{2} + \erm^{2{\nu}} \left(\drm r^{2}+\drm z^{2}\right),
\end{align*}
and rewrite it as 
\begin{align*}
\drm s^{2} & = -F^{2} \left(\drm t + {\omega}_{1}\, \drm {\varphi} + {\omega}_{2}\,\drm {\psi} \right)^{2} + G^{2} \left(\drm {\psi}+ {\Omega} \, \drm {\varphi}\right)^{2} \\ 
& \hspace{0.4cm} + H^{2}\, \drm {\varphi}^{2} + \erm^{2{\nu}} \left(\drm r^{2}+\drm z^{2}\right),
\end{align*}
with
\begin{align*}
& F^{2} = - J_{00}, \quad  -F^{2} {\omega}_{1} = J_{01}, \quad -F^{2} {\omega}_{2} = J_{02}, \\
- & F^{2}{\omega}_{2}^{2}+G^{2} = J_{22}^{\vphantom{1}}, \quad  -F^{2} {\omega}_{1} {\omega}_{2}+G^{2}{\Omega} = J_{12}, \\
- & F^{2}{\omega}_{1}^{2}+G^{2}{\Omega}^{2}+H^{2} = J_{11}^{\vphantom{1}}.
\end{align*}
The latter form has been chosen to facilitate calculating $P$ adapted to part of the axis where $z\to \infty $ and $\partial_{\varphi}=0$. In terms of the orthonormal frame
\begin{align*}
{\theta}^{0} & = F \left(\drm t + {\omega}_{1}\, \drm {\varphi} + {\omega}_{2}\,\drm {\psi} \right), \quad {\theta}^{1} = G \left(\drm {\psi}+ {\Omega} \, \drm {\varphi}\right), \\
{\theta}^{2} & = H \, \drm {\varphi}, \quad {\theta}^{3} = \erm ^{{\nu}} \, \drm r, \quad {\theta}^{4} = \erm ^{{\nu}} \, \drm z,
\end{align*}
the Killing 1-forms take the form
\begin{align*}
\frac{∂}{\partial t} & \to  T = - F\,{\theta}^{0} = - F^{2} \left(\drm t + {\omega}_{1}\, \drm {\varphi} + {\omega}_{2}\,\drm {\psi} \right),\\
\frac{∂}{\partial \psi} & \to  Ψ = G\,{\theta}^{1}-F{\omega}_{2}\,{\theta}^{0}.
\end{align*}
Using $\drm {\varphi} = H^{-1}\, {\theta}^{2}$, $\drm {\psi} = G^{-1}\,{\theta}^{1}-{\Omega}H^{-1}\,{\theta}^{2}$ this yields for the first twist potential
\begin{align*}
\drm {\chi}_{1} & = * \left(T\wedge Ψ\wedge \drm T\right) \\
 & = *\left(F^{3}G\,{\theta}^{0}\wedge {\theta}^{1}\wedge (\drm {\omega}_{1}\wedge \drm {\varphi}+\drm {\omega}_{2} \wedge  \drm {\psi})\right)\\
 & = *\left(F^{3}G\,{\theta}^{0}\wedge {\theta}^{1}\wedge \left(\drm {\omega}_{1}\wedge H^{-1}\,{\theta}^{2}+\drm {\omega}_{2} \wedge  (-{\Omega}H^{-1}\,{\theta}^{2})\right)\right) \hspace{1.4cm}\\
 & = - \frac{F^{3}G}{H} *\left({\theta}^{0}\wedge {\theta}^{1}\wedge {\theta}^{2}\wedge (\drm {\omega}_{1}-{\Omega} \,\drm {\omega}_{2} )\right),
\end{align*}
and for the second
\begin{align*}
 \drm {\chi}_{2} & = * \left(T\wedge Ψ\wedge \drm Ψ\right)\\
 & = *\left(-FG\,{\theta}^{0}\wedge {\theta}^{1}\wedge \drm Ψ\right)\\
 & = -*\left(FG\,{\theta}^{0}\wedge {\theta}^{1}\wedge \left(G^{2}\, \drm {\Omega} \wedge  \drm {\varphi} - F^{2}{\omega}_{2}(\drm {\omega}_{1}\wedge \drm {\varphi}+ \drm {\omega}_{2}\wedge \drm {\psi})\right)\right)\\
 & = \frac{FG}{H} *\left({\theta}^{0}\wedge {\theta}^{1}\wedge {\theta}^{2}\wedge (G^{2}\,\drm {\Omega}-F^{2}{\omega}_{2}\, \drm {\omega}_{1}-F^{2}{\omega}_{2}{\Omega} \,\drm {\omega}_{2} )\right).
\end{align*}
Since $J=J(r,z)$ all the functions depend only on $r$, $z$, hence so do ${\chi}_{i}$ and ${\omega}_{i}$. Then the total derivatives are $\drm {\chi}_{i} = \partial _{r} {\chi}_{i}\, \drm r+ \partial _{z} {\chi}_{i}\, \drm z$ and analogous for ${\omega}_{i}$. Furthermore, noting that $\drm r = \erm ^{-{\nu}} \, {\theta}^{3}$, $\drm z = \erm ^{-{\nu}} \, {\theta}^{4}$ the above equations read
\begin{align*}
\drm {\chi}_{1} & = \partial _{r} {\chi}_{1}\, \drm r+ \partial _{z} {\chi}_{2}\, \drm z \\
& = - \frac{F^{3}G}{H} *\Bigg({\theta}^{0}\wedge {\theta}^{1}\wedge {\theta}^{2}\wedge \Big((\partial _{r} {\omega}_{1}-{\Omega} \, \partial _{r} {\omega}_{2}) \erm^{-{\nu}}\, {\theta}^{3}\\
& \hspace{0.4cm} + (\partial _{z} {\omega}_{1}-{\Omega} \, \partial _{z} {\omega}_{2}) \erm^{-{\nu}}\, {\theta}^{4} \Big)\Bigg)\\
& = - {\epsilon} \, \frac{F^{3}G}{H} \Big((\partial _{r} {\omega}_{1}-{\Omega} \, \partial _{r} {\omega}_{2})\drm z-(\partial _{z} {\omega}_{1}-{\Omega} \, \partial _{z} {\omega}_{2}) \drm r\Big)\\
\Rightarrow  \partial _{z}{\chi}_{1} & = -{\epsilon} \, \frac{F^{3}G}{H}\Big(\partial _{r} {\omega}_{1}-{\Omega} \, \partial _{r} {\omega}_{2}\Big),
\end{align*}
and
\begin{align*}
\drm {\chi}_{2} & = \partial _{r} {\chi}_{2}\, \drm r+ \partial _{z} {\chi}_{2}\, \drm z \\
& = \frac{FG}{H} *\Bigg({\theta}^{0}\wedge {\theta}^{1}\wedge {\theta}^{2}\wedge \Big(G^{2} \partial _{r} {\Omega}-F^{2}{\omega}_{2} \partial _{r} {\omega}_{1}-F^{2}{\omega}_{2}{\Omega} \partial _{r} {\omega}_{2} ) \erm^{-{\nu}}\,{\theta}^{3}\\
& \hphantom{=} + (G^{2} \partial _{z} {\Omega}-F^{2}{\omega}_{2} \partial _{z} {\omega}_{1}-F^{2}{\omega}_{2}{\Omega} \partial _{z} {\omega}_{2} \Big) \erm^{-{\nu}}\,{\theta}^{4}\Bigg)\\
& = {\epsilon}\, \frac{FG}{H} \Bigg(\Big(G^{2} \partial _{r} {\Omega}-F^{2}{\omega}_{2} \partial _{r} {\omega}_{1}-F^{2}{\omega}_{2}{\Omega} \partial _{r} {\omega}_{2}\Big) \drm z \\
& \hphantom{=} - \Big(G^{2} \partial _{r} {\Omega}-F^{2}{\omega}_{2} \partial _{r} {\omega}_{1}-F^{2}{\omega}_{2}{\Omega} \partial _{r} {\omega}_{2}\Big) \drm r \Bigg) \\
\Rightarrow  \partial _{z}{\chi}_{2} & = {\epsilon}\, \frac{FG}{H} \left(G^{2} \partial _{r} {\Omega}-F^{2}{\omega}_{2} \partial _{r} {\omega}_{1}-F^{2}{\omega}_{2}{\Omega} \partial _{r} {\omega}_{2}\right),
\end{align*}
with ${\epsilon}\in \{\pm 1\}$ only depending on the chosen orientation of our orthonormal tetrad. Here comes the point where we need to specify our metric functions in order to calculate the twist potentials. First we are going to look at the asymptotics since they will give us important information later. 

\section{Asymptotic Minkowski Space-Times} \label{sec:exasympt}
For a stationary and axisymmetric space-time in five dimensions we learn from \cite[Sec. IV.C]{Harmark:2004rm} not only that it asymptotes Minkowski space but also how. In ${\sigma}$-model form $\big($for $\sqrt{r^{2}+z^{2}}\to \infty $ and $\frac{z}{\sqrt{r^{2}+z^{2}}}$ finite$\big)$ the metric coefficients converge as follows 
\begin{equation} \label{eq:asymptmetr}
\begin{split}
J_{00} & = -1 + \frac{4M}{3{\pi}} \frac{1}{\sqrt{r^{2}+z^{2}}}+\Ocal\left((r^{2}+z^{2})^{-1}\right),\\
J_{01} & = - \frac{L_{1}}{{\pi}} \frac{\sqrt{r^{2}+z^{2}}-z}{r^{2}+z^{2}}+\Ocal\left((r^{2}+z^{2})^{-1}\right),\\
J_{02} & = - \frac{L_{2}}{{\pi}} \frac{\sqrt{r^{2}+z^{2}}+z}{r^{2}+z^{2}}+\Ocal\left((r^{2}+z^{2})^{-1}\right), \\ 
J_{11} & =  \left(\sqrt{r^{2}+z^{2}}-z\right)\left[1+\frac{2}{3{\pi}}\frac{M+{\eta}}{\sqrt{r^{2}+z^{2}}}+\Ocal\left((r^{2}+z^{2})^{-1}\right)\right],\\
J_{12} & = {\zeta} \frac{r^{2}}{\left(r^{2}+z^{2}\right)^{\frac{3}{2}}}+\Ocal\left((r^{2}+z^{2})^{-1}\right),\\
J_{22} & = \left(\sqrt{r^{2}+z^{2}}+z\right)\left[1+\frac{2}{3{\pi}}\frac{M-{\eta}}{\sqrt{r^{2}+z^{2}}}+\Ocal\left((r^{2}+z^{2})^{-1}\right)\right],\\
{\erm}^{2ν} & = \frac{1}{2\sqrt{r^{2}+z^{2}}}+\Ocal\left((r^{2}+z^{2})^{-1}\right).
\end{split}
\end{equation}
Furthermore, $M$ is the mass of the space-time configuration and $L_{1}$, $L_{2}$ are the angular momenta; ${\zeta}$ and ${\eta}$ are constant where ${\eta}$ is not gauge-invariant, that is it changes under $z\to z+\text{const.}$, unlike ${\zeta}$; the periodicity of ${\varphi}$ and ${\psi}$ is assumed to be $2{\pi}$ (the case when it is $2{\pi}{\varepsilon}$ is given in \cite[Sec. IV.C]{Harmark:2004rm} as well).

Our aim is to integrate $\partial _{z}{\chi}_{i}$ along the $z$-axis, therefore we first need the metric coefficients in the limit $z\to +\infty $ for fixed $r$. In this case $\sqrt{r^{2}+z^{2}}\to z\left(1+\frac{r^{2}}{2z^{2}}\right)$ which yields for the leading order terms in $z$
\begin{align*}
\renewcommand\arraystretch{2.5}
\begin{array}{lll}
J_{00} \sim  -1+ \dfrac{4M}{3{\pi}}\dfrac{1}{z}, & J_{01} \sim  -\dfrac{L_{1}}{2{\pi}}\dfrac{r^{2}}{z^{3}}, & J_{02} \sim  - \dfrac{2L_{2}}{{\pi}} \dfrac{1}{z}, \\
J_{11} \sim  \dfrac{r^{2}}{2z}\left(1+\dfrac{2}{3{\pi}}\dfrac{M+{\eta}}{z}\right), & J_{12} \sim  {\zeta} \dfrac{r^{2}}{z^{3}}, & J_{22} \sim  2z \left(1+\dfrac{2}{3{\pi}}\dfrac{M-{\eta}}{z}\right);
\end{array}
\end{align*}
therefore
\begin{align*}
F^{2} & \sim  1- \frac{4M}{3{\pi}}\frac{1}{z}, \quad F^{2}{\omega}_{1} \sim  -\frac{L_{1}}{2{\pi}}\frac{r^{2}}{z^{3}},\quad F^{2}{\omega}_{2} \sim  - \frac{2L_{2}}{{\pi}} \frac{1}{z}, \\
G^{2} & = J_{22}+F^{2}{\omega}_{2} \sim  2z +\frac{4}{3{\pi}}(M-{\eta}), \\ 
G^{2}{\Omega} & = J_{12}+F^{2}{\omega}_{1}{\omega}_{2} \sim  \frac{r^{2}}{z^{3}}\left({\zeta}+\frac{3L_{1}L_{2}}{{\pi}(3{\pi}z-4M)}\right), \\
H^{2} & = J_{11}^{\vphantom{1}}+F^{2}{\omega}_{1}^{2}-G^{2}{\Omega}^{2} \sim  \frac{r^{2}}{2z}\left(1+\frac{2}{3{\pi}}\frac{M+{\eta}}{z}\right).
\end{align*}
We observe that the leading order terms in $z$ of $F$, ${\omega}_{2}$, $G$ are $\Ocal(1)$ and of ${\omega}_{1}$, ${\Omega}$, $H^{2}$ are $\Ocal(r^{2})$. So, the terms containing ${\Omega}$ can for $r\to 0$ be neglected.\footnote{Taking two limits successively is not ambiguous in this case, since we remember that 
\begin{equation*}
\lim_{y\to y_{0}}\lim_{x\to x_{0}} f(x,y) = \lim_{x\to x_{0}}\lim_{y\to y_{0}} f(x,y)
\end{equation*}
if all the limits exist and if at least one of the limits converges uniformly keeping the second variable fixed. The above functions meet these conditions.} Furthermore we have
\begin{equation*}
\frac{H^{2}}{r^{2}} \xrightarrow[]{r\to 0\,} \frac{1}{2z} \left(1+\frac{2}{3{\pi}}\frac{M+{\eta}}{z}\right),\quad \frac{1}{r} \partial _{r}{\omega}_{1} = \frac{3L_{1}}{z^{2}}\frac{1}{3{\pi}z-4M}.
\end{equation*}
Then
\begin{align*}
\partial _{z}{\chi}_{1} & = -{\epsilon} F^{3}G \frac{r}{H} \frac{\partial _{r}{\omega}_{1}}{r} \\
& \xrightarrow[]{r\to 0} - {\epsilon} \left(\frac{3{\pi}z-4M}{3{\pi}z}\right)^{\frac{3}{2}} \sqrt{2z}\sqrt{2z}\,\frac{3L_{1}}{z^{2}}\frac{1}{3{\pi}z-4M} \\
& = - \frac{2{\epsilon}L_{1}}{{\pi}} \left(\frac{z-\frac{4M}{3{\pi}}}{z^{5}}\right)^{\frac{1}{2}}\\
& \sim  - \frac{2{\epsilon}L_{1}}{{\pi}} \frac{1}{z^{2}},
\end{align*}
where in the last step again only the leading order term in $z$ was kept. Integration is easy and gives
\begin{equation*}
{\chi}_{1} \sim \frac{2{\epsilon}L_{1}}{{\pi}z}.
\end{equation*}
Using
\begin{equation*}
\frac{G^{2}\partial _{r}{\Omega}}{r} \xrightarrow[]{r\to 0} \frac{2}{z^{3}}\left({\zeta}+\frac{3L_{1}L_{2}}{{\pi}(3{\pi}z-4M)}\right),
\end{equation*}
we can play the same game for ${\chi}_{2}$ and get
\begin{align*}
\partial _{z}{\chi}_{2} & = {\epsilon}FG \frac{r}{H} \frac{1}{r}\left(G^{2} \partial _{r} {\Omega}-F^{2}{\omega}_{2} \partial _{r} {\omega}_{1}-F^{2}{\omega}_{2}{\Omega} \partial _{r} {\omega}_{2}\right) \\
& \xrightarrow[]{r\to 0} {\epsilon} \left(\frac{3{\pi}z-4M}{3{\pi}z}\right)^{\frac{1}{2}}\sqrt{2z}\sqrt{2z}\Bigg(\frac{2}{z^{3}}\left({\zeta}+\frac{3L_{1}L_{2}}{{\pi}(3{\pi}z-4M)}\right) \\
& \hphantom{\xrightarrow[]{r\to 0}}-\frac{2L_{2}}{{\pi}}\frac{1}{z}\cdot \frac{3L_{1}}{z^{2}}\frac{1}{3{\pi}z-4M}\Bigg) \\
& = 4{\epsilon} {\zeta} \left(\frac{z-\frac{4M}{3{\pi}}}{z^{5}}\right)^{\frac{1}{2}} \sim  \frac{4{\epsilon}{\zeta}}{z}, 
\end{align*}
which can to leading order in $z$ be integrated to
\begin{equation*}
{\chi}_{2} \sim - \frac{4{\epsilon}{\zeta}}{z}.
\end{equation*}
Now we can go straight on to calculate $P$ and obtain again by only considering the leading order in $z$ the following
\begin{equation*}
\frac{1}{g} = \det \skew{7}{\tilde}{A} \xrightarrow[]{r\to 0} \left(-1+\frac{4M}{3{\pi}z}\right)2z\left(1+\frac{2(M-{\eta})}{3{\pi}z}\right) ~ = - \frac{6{\pi}z-4(M+{\eta})}{3{\pi}}
\end{equation*}
\begin{align*}
\Rightarrow \  & g \sim -\frac{1}{2z} - \frac{M+{\eta}}{3{\pi}z^{2}} \quad \text{for large } z \\
\Rightarrow \  & {\chi}_{1} g = \left(-\frac{1}{2z} - \frac{M+{\eta}}{3{\pi}z^{2}}\right) \frac{2{\epsilon}L_{1}}{{\pi}z} \sim -\frac{{\epsilon}L_{1}}{{\pi}z^{2}},\\
& {\chi}_{2} g = \left(-\frac{1}{2z} - \frac{M+{\eta}}{3{\pi}z^{2}}\right)\left(- \frac{4{\epsilon}{\zeta}}{z}\right) \sim \frac{2{\epsilon}{\zeta}}{z^{2}},\\
& J_{00}^{\vphantom{1}} + g {\chi}_{1}^{2} \sim J_{00}^{\vphantom{1}} = -1+ \frac{4M}{3{\pi}z}, \quad J_{02}^{\vphantom{1}}+g{\chi}_{1}^{\vphantom{1}}{\chi}_{2}^{\vphantom{1}} \sim J_{02}^{\vphantom{1}} = -\frac{2L_{2}}{{\pi}z}, \\
& J_{22}^{\vphantom{1}} + g {\chi}_{2}^{2} \sim J_{22}^{\vphantom{1}} = 2z + \frac{4(M-{\eta})}{3{\pi}}
\end{align*}
So, finally we obtain the patching matrix to leading order in $z$ beyond \eqref{eq:PMink} as
\begin{equation} \label{eq:asymptPtop}
\renewcommand\arraystretch{2.5}
P_{+} = \left(\begin{array}{ccc}-\dfrac{1}{2z} - \dfrac{M+{\eta}}{3{\pi}z^{2}} & \dfrac{{\epsilon}L_{1}}{{\pi}z^{2}} & -\dfrac{2{\epsilon}{\zeta}}{z^{2}} \\\dfrac{{\epsilon}L_{1}}{{\pi}z^{2}} & -1+ \dfrac{4M}{3{\pi}z} & -\dfrac{2L_{2}}{{\pi}z} \\-\dfrac{2{\epsilon}{\zeta}}{z^{2}} & -\dfrac{2L_{2}}{{\pi}z} & 2z + \dfrac{4(M-{\eta})}{3{\pi}}\end{array}\right).
\end{equation}
The index indicates that the patching matrix is adapted to the top asymptotic end. The adaptation $P_{-}$ to the bottom asymptotic end, that is the one which extends to $z\to -\infty $, is obtained by swapping ${\varphi}$ and ${\psi}$ in their roles. This leads to $z\mapsto -z$, $L_{1}\leftrightarrow L_{2}$. Furthermore, one has to check what happens with ${\zeta}$ and ${\eta}$ in this case. From \cite[Eq.~(5.18)]{Harmark:2004rm} we see that ${\zeta}\mapsto {\zeta}$ and ${\eta}\mapsto -{\eta}$ for the Myers-Perry solution. But all asymptotically flat space-times have the same fall off up to the order \eqref{eq:asymptmetr}, so this behaviour must be generic. For the ease of reference later on we will include $P_{-}$ again explicitly
\begin{equation}\label{eq:asymptPbot}
\renewcommand\arraystretch{2.5}
P_{-} = \left(\begin{array}{ccc}\dfrac{1}{2z} - \dfrac{M-{\eta}}{3{\pi}z^{2}} & -\dfrac{{\epsilon}L_{2}}{{\pi}z^{2}} & -\dfrac{2{\epsilon}{\zeta}}{z^{2}} \\-\dfrac{{\epsilon}L_{2}}{{\pi}z^{2}} & -1- \dfrac{4M}{3{\pi}z} & \dfrac{2L_{1}}{{\pi}z} \\-\dfrac{2{\epsilon}{\zeta}}{z^{2}} & \dfrac{2L_{1}}{{\pi}z} & -2z + \dfrac{4(M+{\eta})}{3{\pi}}\end{array}\right).
\end{equation}

The Myers-Perry solution, which we will study next, is the five-dimensional pendant of the Kerr solutions, that is it describes a five-dimensional spinning black hole. 

\section{Five-Dimensional Myers-Perry Solution} \label{sec:PMP}

It was in 1986 that the Schwarzschild solution in dimension greater than four was extended from the static to the stationary case by Myers and Perry \citet{Myers:1986aa}. The calculations in the first part of this example up to the expression for $J(r,z)$ is based on \citet{Harmark:2004rm}. Myers-Perry space-time is asymptotically flat and has horizon topology $S^{3}$. Its metric is given by
\begin{equation} \label{eq:MPmetric}
\begin{split}
\drm s^{2} & = -\drm t^{2} + \frac{{\rho}_{0}^{2}}{{\Sigma}}\left[\drm t - a_{1} \sin^{2}{\theta} \,\drm {\varphi} - a_{2} \cos^{2} {\theta} \,\drm {\psi}\right]^{2} \\
& \hspace{0.4cm} + ({\rho}^{2}+a_{1}^{2})\sin^{2}{\theta} \,\drm {\varphi}^{2} + ({\rho}^{2}+a_{2}^{2})\cos^{2}{\theta} \,\drm {\psi}^{2} \\
& \hspace{0.4cm} + \frac{{\Sigma}}{{\Delta}} \,\drm{\rho}^{2} + {\Sigma} \,\drm {\theta}^{2},
\end{split}
\end{equation}
where
\begin{equation} \label{eq:MPnot}
\begin{split}
{\Delta} & = {\rho}^{2} \left(1+\frac{a_{1}^{2}}{{\rho}^{2}}\right)\left(1+\frac{a_{2}^{2}}{{\rho}^{2}}\right) - {\rho}_{0}^{2}, \\
{\Sigma} & = {\rho}^{2} + a_{1}^{2} \cos^{2}{\theta} + a_{2}^{2} \sin^{2}{\theta},
\end{split}
\end{equation}
and the coordinate ranges are
\begin{equation*}
t\in \mathbb{R},\quad {\varphi},{\psi} \in  \left[0,2{\pi}\right), \quad {\theta} \in  [0,{\pi}].
\end{equation*}
A straight forward computation shows
\begin{equation*}
\det J = - \frac{1}{4} {\rho}^{2} {\Delta} \sin^{2} 2{\theta}, 
\end{equation*}
hence
\begin{equation*}
r = \frac{1}{2} {\rho} \sqrt{{\Delta}} \sin 2{\theta}.
\end{equation*}
The second coordinate $z$ is obtained via the relation
\begin{equation} \label{eq:zMP}
{\erm}^{2ν} \left(\drm r^{2}+ \drm z^{2}\right) = \frac{{\Sigma}}{{\Delta}} \,\drm{\rho}^{2} + {\Sigma} \,\drm {\theta}^{2}
\end{equation}
as
\begin{equation*}
z = \frac{1}{2} {\rho}^{2} \left(1+\frac{a_{1}^{2}+a_{2}^{2}-{\rho}_{0}^{2}}{2{\rho}^{2}}\right) \cos 2{\theta},
\end{equation*}
which can be verified either by direct substitution into \eqref{eq:zMP} or by a more general approach as sketched in \cite[App.~H]{Harmark:2004rm}. This determines the metric in ${\sigma}$-model coordinates. However, expressing it in prolate spherical coordinates $x$, $y$ is more convenient. They are defined by
\begin{equation*}
r={\alpha} \sqrt{(x^{2}-1)(1-y^{2})}, \quad z = {\alpha} xy,
\end{equation*}
where ${\alpha}$ is a constant and the ranges are taken to be $x\geq 1$, $-1\leq y\leq 1$. With the ansatz $x=x({\rho})$, $y=y({\theta})$ this is solved by 
\begin{equation} \label{eq:MPnutpar}
\begin{split}
x & = \frac{2 {\rho}^{2}+a_{1}^{2}+a_{2}^{2}-{\rho}_{0}^{2}}{\sqrt{\left({\rho}_{0}^{2}-a_{1}^{2}-a_{2}^{2}\right)^{2}-4a_{1}^{2}a_{2}^{2}}}, \quad y = \cos 2{\theta},\\ 
{\alpha} & = \frac{1}{4}\sqrt{\left({\rho}_{0}^{2}-a_{1}^{2}-a_{2}^{2}\right)^{2}-4a_{1}^{2}a_{2}^{2}}.
\end{split}
\end{equation}
In terms of $r$, $z$ it can be shown (see \cite[App.~G]{Harmark:2004rm}) that this is
\begin{equation} \label{eq:prolcoord}
x=\frac{R_{+}+R_{-}}{2{\alpha}}, \quad y=\frac{R_{+}-R_{-}}{2{\alpha}},
\end{equation}
where $R_{\pm}=\sqrt{r^{2}+(z \pm {\alpha})^{2}}$. In terms of prolate spherical coordinates the $J$-matrix and ${\erm}^{2ν}$ are given by
\begin{equation} \label{eq:MPjmatrix}
\renewcommand{\arraystretch}{2.5}
\begin{array}{l}
J_{00}  = -\displaystyle \frac{4{\alpha}x+(a_{1}^{2}-a_{2}^{2})y-{\rho}_{0}^{2}}{4{\alpha}x+(a_{1}^{2}-a_{2}^{2})y+{\rho}_{0}^{2}}, \quad 
J_{01}  = -\displaystyle \frac{a_{1}^{\vphantom{1}}{\rho}_{0}^{2}(1-y)}{4{\alpha}x+(a_{1}^{2}-a_{2}^{2})y+{\rho}_{0}^{2}}, \\
J_{02}  = -\displaystyle \frac{a_{2}^{\vphantom{1}}{\rho}_{0}^{2}(1+y)}{4{\alpha}x+(a_{1}^{2}-a_{2}^{2})y+{\rho}_{0}^{2}}, \quad 
J_{12}  = -\displaystyle \frac{1}{2}\frac{a_{1}^{\vphantom{1}}a_{2}^{\vphantom{1}}{\rho}_{0}^{2}(1-y)^{2}}{4{\alpha}x+(a_{1}^{2}-a_{2}^{2})y+{\rho}_{0}^{2}}, \\
J_{11}  = \displaystyle \frac{1-y}{4} \left[4{\alpha}x+{\rho}_{0}^{2}+a_{1}^{2}-a_{2}^{2}+\frac{2a_{1}^{2}{\rho}_{0}^{2}(1-y)}{4{\alpha}x+(a_{1}^{2}-a_{2}^{2})y+{\rho}_{0}^{2}}\right], \\
J_{22} = \displaystyle \frac{1+y}{4} \left[4{\alpha}x+{\rho}_{0}^{2}-a_{1}^{2}+a_{2}^{2}+\frac{2a_{2}^{2}{\rho}_{0}^{2}(1+y)}{4{\alpha}x+(a_{1}^{2}-a_{2}^{2})y+{\rho}_{0}^{2}}\right], \\
{\erm}^{2ν} = -\displaystyle \frac{4{\alpha}x+(a_{1}^{2}-a_{2}^{2})y+{\rho}_{0}^{2}}{8{\alpha}^{2}(x^{2}-y^{2})}.
\end{array}
\end{equation}
Substituting \eqref{eq:prolcoord} yields $J(r,z)$. 

The rod structure consists of three components $(-\infty ,{\alpha})$, $(-{\alpha},{\alpha})$, $({\alpha},\infty )$ as we will see from the explicit expressions for the rod vectors below. First note that $r=0$ implies $R_{\pm} = |z\pm{\alpha}|$. 
\begin{enumerate}[(1)]
\item If $z$ lies in the semi-infinite spacelike rod $({\alpha},\infty )$, we have $x=\frac{z}{{\alpha}}$, $y=1$, thus $J_{k1}=0$ for $k=0,1,2$. Therefore the rod vector is $∂_{\varphi}$.
\item If $z$ lies in the finite timelike rod $(-{\alpha},{\alpha})$, we have $x=1$, $y=\frac{z}{{\alpha}}$. Thus the kernel of $J$ is spanned by the vector
\begin{equation*}
\left(\begin{array}{ccc}1 & {\Gamma}_{1} & {\Gamma}_{2}\end{array}\right)^{\mathrm{t}},
\end{equation*}
where ${\Gamma}_{1,2}$ are the angular velocities
\begin{equation*}
{\Gamma}_{1}= \frac{{\rho}_{0}^{2}+a_{1}^{2}-a_{2}^{2}-4{\alpha}}{2a_{1}^{\vphantom{1}}{\rho}_{0}^{2}}, \quad {\Gamma}_{2}= \frac{{\rho}_{0}^{2}-a_{1}^{2}+a_{2}^{2}-4{\alpha}}{2a_{2}^{\vphantom{1}}{\rho}_{0}^{2}}.
\end{equation*}
It can be shown that this rod corresponds to an event horizon with topology $S^{3}$ (a brief reasoning can be found in \cite[proof of Prop.~2 in Sec.~3]{Hollands:2008fp}).
\item If $z$ lies in the semi-infinite spacelike rod $(-\infty ,-{\alpha})$, we have $x=-\frac{z}{{\alpha}}$, $y=-1$, thus $J_{k2}=0$ for $k=0,1,2$. Therefore the rod vector is $∂_{\psi}$.
\end{enumerate}
From the given eigenvectors it becomes clear that $\dim \ker J (0,\pm {\alpha})>1$. The conserved Komar quantities are
\begin{equation} \label{eq:MPmassmom}
M = \frac{3{\pi}}{8} {\rho}_{0}^{2}, \quad L_{1}^{\vphantom{1}} = \frac{3{\pi}}{8} a_{1}^{\vphantom{1}} {\rho}_{0}^{2}, \quad L_{2}^{\vphantom{1}} = \frac{3{\pi}}{8} a_{2}^{\vphantom{1}} {\rho}_{0}^{2}.
\end{equation}

Having done this preparation the patching matrix can be obtained in the same way as before for the asymptotic metric. The limit $r\to 0$ can here be regarded as $\sin {\theta} \to  0$, then with ${\Sigma}\to {\rho}^{2}+a_{1}^{2}$ we obtain
\begin{align*}
F^{2} & \to  \frac{{\rho}^{2}+a_{1}^{2}-{\rho}_{0}^{2}}{{\rho}^{2}+a_{1}^{2}}, \quad F^{2}{\omega}_{1} \to  -\frac{a_{1}^{\vphantom{1}}{\rho}_{0}^{2}}{{\Sigma}}\sin^{2}{\theta},\quad F^{2}{\omega}_{2} \to  - \frac{a_{2}^{\vphantom{1}}{\rho}_{0}^{2}}{{\Sigma}} \cos^{2}{\theta}, \\
G^{2} & \to  \left({\rho}^{2}+a_{2}^{2}\right)\cos^{2}{\theta}+\frac{{\rho}_{0}^{2}a_{2}^{2}}{{\Sigma}}\cos^{4}{\theta}+F^{2}{\omega}_{2}^{2}, \\ 
G^{2}{\Omega} & \to  \frac{a_{1}a_{2}}{{\Sigma}} {\rho}_{0}^{2}\sin^{2}{\theta} \cos^{2}{\theta}+F^{2}{\omega}_{1}^{\vphantom{1}}{\omega}_{2}^{\vphantom{1}}\\
H^{2} & \to  ({\rho}^{2}+a_{1}^{2})\sin^{2}{\theta}+\frac{a_{2}^{\vphantom{1}}{\rho}_{0}^{2}}{{\Sigma}}\sin^{4}{\theta}-G^{2}{\Omega}^{2}+F^{2}{\omega}_{1}^{2}.
\end{align*}
We observe that $F$, ${\omega}_{2}$, $G$ are $\Ocal(1)$ and ${\omega}_{1}$, ${\Omega}$, $H^{2}$ are $\Ocal(\sin^{2}{\theta})$. So, the terms containing ${\Omega}$ can for ${\theta}\to 0$ be neglected. Furthermore we have
\begin{equation*}
\frac{H^{2}}{\sin^{2}{\theta}} \xrightarrow[]{{\theta}\to 0\,} {\rho}^{2}+a_{1}^{2},
\end{equation*}
and, using ${\omega}_{1}=\dfrac{a_{1}^{\vphantom{1}}{\rho}_{0}^{2}\sin^{2}{\theta}}{{\Sigma}-{\rho}_{0}^{2}}$, also
\begin{equation*}
\frac{\partial _{{\theta}}\,{\omega}_{1}}{\sin {\theta}} = 2a_{1}^{\vphantom{1}} {\rho}_{0}^{2} \cos {\theta} \cdot  \frac{({\Sigma}-{\rho}_{0}^{2})-\sin^{2}{\theta}(a_{2}^{2}-a_{1}^{2})}{\left({\Sigma}-{\rho}_{0}^{2}\right)^{2}} \to  \frac{2 a_{1}^{\vphantom{1}}{\rho}_{0}^{2}}{{\rho}^{2}+a_{1}^{2}-{\rho}_{0}^{2}}.
\end{equation*}
As a last preparation we rewrite 
\begin{align*}
G^{2} & \to  {\rho}^{2}+a_{2}^{2}+\frac{{\rho}_{0}^{2}a_{2}^{2}}{{\rho}^{2}+a_{1}^{2}}+\left(\frac{{\rho}^{2}+a_{1}^{2}-{\rho}_{0}^{2}}{{\rho}^{2}+a_{1}^{2}}\right)^{-1}\cdot \frac{a_{2}^{2}{\rho}_{0}^{4}}{({\rho}^{2}+a_{1}^{2})} \\
& = {\rho}^{2}+a_{2}^{2}+\frac{{\rho}_{0}^{2}a_{2}^{2}}{{\rho}^{2}+a_{1}^{2}} + \frac{a_{2}^{2}{\rho}_{0}^{4}}{({\rho}^{2}+a_{1}^{2})({\rho}^{2}+a_{1}^{2}-{\rho}_{0}^{2})}\\
& = \frac{({\rho}^{2}+a_{2}^{2})({\rho}^{2}+a_{1}^{2}-{\rho}_{0}^{2})+{\rho}_{0}^{2}a_{2}^{2}}{{\rho}^{2}+a_{1}^{2}-{\rho}_{0}^{2}} = \frac{{\rho}^{2}{\Delta}}{{\rho}^{2}+a_{1}^{2}-{\rho}_{0}^{2}}
\end{align*}
Finally, this enables us to calculate
\begin{align*}
\partial _{{\rho}}{\chi}_{1} & \xrightarrow[]{{\theta}\to 0\,} \frac{{\epsilon} F^{3}G}{\sqrt{{\Delta}}}\cdot \frac{\sin {\theta}}{H}\cdot \frac{\partial _{{\theta}}{\omega}_{1}}{\sin {\theta}} \\
& = \frac{{\epsilon}}{\sqrt{{\Delta}}} \left(\frac{{\rho}^{2}+a_{1}^{2}-{\rho}_{0}^{2}}{{\rho}^{2}-a_{1}^{2}}\right)^{\frac{3}{2}} \left(\frac{{\rho}^{2}{\Delta}}{{\rho}^{2}+a_{1}^{2}-{\rho}_{0}^{2}}\right)^{\frac{1}{2}} ({\rho}^{2}+a_{1}^{2})^{-\frac{1}{2}} \cdot \frac{2a_{1}^{\vphantom{1}}{\rho}_{0}^{2}}{{\rho}^{2}+a_{1}^{2}-{\rho}_{0}^{2}}\\
& = \frac{2{\epsilon}{\rho}a_{1}^{\vphantom{1}}{\rho}_{0}^{2}}{({\rho}^{2}+a_{1}^{2})^{2}},
\end{align*}
which can easily be integrated to
\begin{equation*}
\left.{\chi}_{1}\right|_{{\theta}=0} = -\frac{{\epsilon}{\rho}_{0}^{2}a_{1}^{\vphantom{1}}}{{\rho}^{2}+a_{1}^{2}}.
\end{equation*}
For the second twist potential we need
\begin{equation*}
{\Omega} = \sin^{2} {\theta} \left(\frac{a_{1}^{\vphantom{1}}a_{2}^{\vphantom{1}}{\rho}_{0}^{2}\cos^{2}{\theta}}{G^{2}{\Sigma}}+\frac{a_{1}^{\vphantom{1}}{\rho}_{0}^{2}{\omega}_{2}^{\vphantom{1}}}{G^{2}{\Sigma}}\right),
\end{equation*}
which implies
\begin{equation*}
\frac{\partial _{{\theta}}{\Omega}}{\sin {\theta}} \to  2 \left(\frac{a_{1}^{\vphantom{1}}a_{2}^{\vphantom{1}}{\rho}_{0}^{2}}{G^{2}{\Sigma}}+\frac{a_{1}^{\vphantom{1}}{\rho}_{0}^{2}}{G^{2}{\Sigma}}\cdot \frac{a_{2}^{\vphantom{1}}{\rho}_{0}^{2}}{F^{2}{\Sigma}}\right),
\end{equation*}
so that
\begin{align*}
- \partial _{{\rho}} {\chi}_{2} & = {\epsilon} \frac{FG}{\sqrt{{\Delta}}} \left(\frac{G^{2}}{H}\partial _{{\theta}}{\Omega}-\frac{F^{2}{\omega}_{2}}{H}\partial _{{\theta}}{\omega}_{1}+\frac{F^{2}{\omega}_{2}}{H}{\Omega}\partial _{{\theta}}{\omega}_{2}\right) \\
&  \sim \frac{{\epsilon} FG}{\sqrt{{\Delta}}}\cdot \frac{\sin {\theta}}{H} \left(\frac{G^{2}\partial _{{\theta}}{\Omega}}{\sin {\theta}}- \frac{F^{2}{\omega}_{2}\partial _{{\theta}}{\omega}_{1}}{\sin {\theta}}\right) \\
& \xrightarrow[]{{\theta}\to 0\,} \frac{{\epsilon}}{\sqrt{{\Delta}}} \left(\frac{{\rho}^{2}+a_{1}^{2}-{\rho}_{0}^{2}}{{\rho}^{2}+a_{1}^{2}}\right)^{\frac{1}{2}} \left(\frac{{\rho}^{2}{\Delta}}{{\rho}^{2}+a_{1}^{2}-{\rho}_{0}^{2}}\right)^{\frac{1}{2}}\frac{1}{({\rho}^{2}+a_{1}^{2})^{\frac{1}{2}}} \\
& \hphantom{=} \times  \left[2G^{2}\frac{a_{1}^{\vphantom{1}}a_{2}^{\vphantom{1}}{\rho}_{0}^{2}}{{\Sigma}}\left(\frac{1}{G^{2}}+\frac{{\rho}_{0}^{2}}{G^{2}F^{2}{\Sigma}}\right)-\frac{a_{2}^{\vphantom{1}}{\rho}_{0}^{2}}{{\Sigma}}\cdot \frac{2a_{1}^{\vphantom{1}}{\rho}_{0}^{2}}{{\Sigma}-{\rho}_{0}^{2}}\right] \\
& = \frac{{\epsilon}{\rho}}{{\rho}^{2}+a_{1}^{2}}\cdot \frac{a_{1}^{\vphantom{1}}a_{2}^{\vphantom{1}}{\rho}_{0}^{2}}{{\Sigma}} \left(2+ \frac{2{\rho}_{0}^{2}}{F^{2}{\Sigma}}-\frac{2{\rho}_{0}^{2}}{{\Sigma}-{\rho}_{0}^{2}}\right) \\
& = \frac{2{\epsilon}a_{1}^{\vphantom{1}}a_{2}^{\vphantom{1}}{\rho}_{0}^{2}{\rho}}{\left({\rho}^{2}+a_{1}^{2}\right)^{2}};
\end{align*}
integration
\begin{equation*}
\left.{\chi}_{2}\right|_{{\theta}=0} = \frac{{\epsilon}a_{1}^{\vphantom{1}}a_{2}^{\vphantom{1}}{\rho}_{0}^{2}}{{\rho}^{2}+a_{1}^{2}}.
\end{equation*}
On ${\theta}=0$ it is
\begin{equation*}
{\rho}^{2} = 2z + \frac{1}{2} \left({\rho}_{0}^{2}-a_{1}^{2}-a_{2}^{2}\right),
\end{equation*}
so the notation in the following calculation of $P$ can be somewhat streamlined by introducing
\begin{equation*}
{\beta} = \frac{1}{4} \left(-{\rho}_{0}^{2}+a_{1}^{2}-a_{2}^{2}\right),\ {\gamma} = \frac{1}{4} \left({\rho}_{0}^{2}+a_{1}^{2}-a_{2}^{2}\right),
\end{equation*}
hence
\begin{align*}
F^{2} & = \frac{z+{\beta}}{z+{\gamma}} = - g_{tt}, \quad F^{2}{\omega}_{2} = \frac{{\epsilon}a_{2}^{\vphantom{1}}{\rho}_{0}^{2}}{2(z+{\gamma})} = - g_{t{\psi}} \\
{\chi}_{1} & = \frac{{\epsilon}{\rho}_{0}^{2}a_{1}^{\vphantom{1}}}{2(z+{\gamma})}, \hspace{1.2cm} {\chi}_{2} = - \frac{{\epsilon}{\rho}_{0}^{2}a_{1}^{\vphantom{1}}a_{2}^{\vphantom{1}}}{2(z+{\gamma})} \\
g_{{\psi}{\psi}} & = {\rho}^{2}+a_{2}^{2}+\frac{{\rho}_{0}^{2}a_{2}^{2}}{{\Sigma}} = 2z + \frac{1}{2}\left({\rho}_{0}^{2}-a_{1}^{2}+a_{2}^{2}\right) + \frac{{\rho}_{0}^{2}a_{2}^{2}}{2z + \frac{1}{2}\left({\rho}_{0}^{2}+a_{1}^{2}-a_{2}^{2}\right)}\\
& = 2z -2{\beta} + \frac{{\rho}_{0}^{2}a_{2}^{2}}{2(z+{\gamma})} = \frac{2}{z+{\gamma}} \left((z-{\beta})(z+{\gamma})+\frac{a_{2}^{2}{\rho}_{0}^{2}}{4}\right).
\end{align*}
This implies
\begin{align*}
\det A & = g^{-1} = g_{tt}g_{{\psi}{\psi}}-g_{t{\psi}}^{2} \\
& = - \frac{z+{\beta}}{z+{\gamma}}\cdot \frac{2}{z+{\gamma}} \left((z-{\beta})(z+{\gamma})+\frac{a_{2}^{\vphantom{1}}{\rho}_{0}^{2}}{4}\right)-\frac{a_{2}^{2}{\rho}_{0}^{4}}{4(z+{\gamma})^{2}} \\
& = - \frac{1}{4(z+{\gamma})^{2}} \left[8(z^{2}-{\beta}^{2})(z+{\gamma})+2a_{2}^{2}{\rho}_{0}^{2}\left((z+{\beta}+\frac{{\rho}_{0}^{2}}{2}\right)\right]\\
& = - \frac{1}{4(z+{\gamma})^{2}} \left[8(z^{2}-{\beta}^{2})+2a_{2}^{2}{\rho}_{0}^{2}\right]\\
& = - \frac{2(z^{2}-{\alpha}^{2})}{z+{\gamma}},
\end{align*}
where in the last step
\begin{equation*}
8{\beta}^{2}-2a_{2}^{2}{\rho}_{0}^{2} = \frac{1}{2} ({\rho}_{0}^{2}-a_{1}^{2}+a_{2}^{2})^{2}-2a_{2}^{2}{\rho}_{0}^{2} = 8 {\alpha}^{2}
\end{equation*}
with ${\alpha}$ as in \cite[Eq.\,(5.15)]{Harmark:2004rm}. We then obtain for $z\in ({\alpha},\infty )$, $r=0$ and with ${\epsilon}=1$ the following (the index refers to the numbering as above for the rod structure)
\begin{equation} \label{eq:MPpatmat}
\renewcommand\arraystretch{2.5}
P_{1} = \left(\begin{array}{ccc}-\dfrac{z+{\gamma}}{2(z^{2}-{\alpha}^{2})} & -\dfrac{{\rho}_{0}^{2}a_{1}^{\vphantom{1}}}{4(z^{2}-{\alpha}^{2})} & \dfrac{{\rho}_{0}^{2}a_{1}^{\vphantom{1}}a_{2}^{\vphantom{1}}}{4(z^{2}-{\alpha}^{2})} \\
\cdot  & -\dfrac{z^{2}+z({\beta}-{\gamma})+{\gamma}^{2}-{\beta}{\gamma}-{\alpha}^{2}}{z^{2}-{\alpha}^{2}} & -\dfrac{a_{2}^{\vphantom{1}}{\rho}_{0}^{2}(z-{\gamma})}{2(z^{2}-{\alpha}^{2})} \\
\cdot  & \cdot  & \hspace{-1.2cm} 2(z-{\beta})+\dfrac{a_{2}^{2}{\rho}_{0}^{2}(z-{\gamma})}{2(z^{2}-{\alpha}^{2})}\end{array}\right)_{\vphantom{\frac{1}{2}}},
\end{equation}
using $P_{1}=(P_{ij})$ with\footnote{Here $P_{ij}$ are only the entries of the patching matrix and not the transition matrices of the bundle.}
\begin{align*}
P_{12} & = g_{t{\psi}} + g {\chi}_{1}{\chi}_{2} = - \frac{a_{2}^{\vphantom{1}}{\rho}_{0}^{2}}{2(z+{\gamma})}+\frac{{\rho}_{0}^{4}a_{1}^{2}a_{2}^{\vphantom{1}}}{8(z+{\gamma})(z^{2}-{\alpha}^{2})} \\
& = \frac{a_{2}^{\vphantom{1}}{\rho}_{0}^{2}}{8(z+{\gamma})(z^{2}-{\alpha}^{2})} \left(-4(z^{2}-{\alpha}^{2})+a_{1}^{2}{\rho}_{0}^{2}\right) = \frac{4a_{2}^{\vphantom{1}}{\rho}_{0}^{2}(z^{2}-{\gamma}^{2})}{8(z+{\gamma})(z^{2}-{\alpha}^{2})}\\
& = - \frac{a_{2}^{\vphantom{1}}{\rho}_{0}^{2}(z-{\gamma})}{2(z^{2}-{\alpha}^{2})},\\ 
P_{22}^{\vphantom{1}} & = g_{{\psi}{\psi}}^{\vphantom{1}} + g {\chi}_{2}^{2} \\
& = \frac{2}{z+{\gamma}} \left((z-{\beta})(z+{\gamma})+\frac{a_{2}^{2}{\rho}_{0}^{2}}{4}\right)-\frac{z+{\gamma}}{2(z^{2}-{\alpha}^{2})} \cdot  \frac{{\rho}_{0}^{4}a_{1}^{2}a_{2}^{2}}{4(z+{\gamma})^{2}} \\
& = 2 (z-{\beta}) + \frac{1}{z+{\gamma}} \cdot \frac{a_{2}^{2}{\rho}_{0}^{2}}{2} \Big(\underbrace{z^{2}-{\alpha}^{2}-\frac{1}{4}a_{1}^{2}{\rho}_{0}^{2}}_{z^{2}-{\gamma}^{2}}\Big)\cdot \frac{1}{z^{2}-{\alpha}^{2}} \\
& = 2(z-{\beta})+\dfrac{a_{2}^{2}{\rho}_{0}^{2}(z-{\gamma})}{2(z^{2}-{\alpha}^{2})},
\end{align*}
and
\begin{align*}
P_{11}^{\vphantom{1}} & = g_{tt}^{\vphantom{1}}+g{\chi}_{1}^{2} = -\frac{z+{\beta}}{z+{\gamma}} - \frac{z+{\gamma}}{2(z^{2}-{\alpha}^{2})} \cdot \frac{{\rho}_{0}^{4}a_{1}^{2}}{4(z+{\gamma})^{2}}\\
& = -\frac{1}{8(z+{\gamma})(z^{2}-{\alpha}^{2})}\left(8(z+{\beta})(z^{2}-{\alpha}^{2})+{\rho}_{0}^{4}a_{1}^{2}\right) \\
& = -\frac{1}{z^{2}-{\alpha}^{2}} \left(z^{2}+z({\beta}-{\gamma})+{\gamma}^{2}-{\beta}{\gamma}-{\alpha}^{2}\right),
\end{align*}
where the last step comes from
\begin{align*}
8({\beta}-{\gamma})({\gamma}^{2}-{\alpha}^{2}) & = -\frac{8}{16} \cdot \frac{{\rho}_{0}^{2}}{2} \left(({\rho}_{0}^{2}+a_{1}^{2}-a_{2}^{2})^{2}-({\rho}_{0}^{2}-a_{1}^{2}-a_{2}^{2})^{2}+4a_{1}^{2}a_{2}^{2}\right)\\
& = -\frac{1}{4} {\rho}_{0}^{2}\left(4a_{1}^{2}({\rho}_{0}^{2}-a_{2}^{2})+4a_{1}^{2}a_{2}^{2}\right) = -a_{1}^{2}{\rho}_{0}^{4}
\end{align*}
thus
\begin{align*}
8(z+{\beta})(z^{2}-{\alpha}^{2})+{\rho}_{0}^{4}a_{1}^{2} & = 8(z+{\beta})(z^{2}-{\alpha}^{2})-8({\beta}-{\gamma})({\gamma}^{2}-{\alpha}^{2})\\
& = 8(z+{\gamma})(z^{2}+z({\beta}-{\gamma})+{\gamma}^{2}-{\beta}{\gamma}-{\alpha}^{2}).
\end{align*}

The Myers-Perry solution with vanishing angular momenta is as in four dimensions the Schwarzschild solution, sometimes also called Schwarzschild-Tangherlini solution.

\subsection*{Five-Dimensional Schwarzschild Space-Time} \mbox{}\\

In the case of $a_{1}=a_{2}=0$ the Myers-Perry metric becomes
\begin{align*}
\drm s^{2} & = \left(-1+\frac{{\rho}_{0}^{2}}{{\rho}^{2}}\right)\drm t^{2} + {\rho}^{2}\sin^{2}{\theta} \,\drm {\varphi}^{2} + {\rho}^{2}\cos^{2}{\theta} \,\drm {\psi}^{2} \\
& \hspace{0.4cm} + \left(1-\frac{{\rho}_{0}^{2}}{{\rho}^{2}}\right)^{-1} \,\drm {\rho}^{2} + {\rho}^{2} \,\drm {\theta}^{2},
\end{align*}
and from \eqref{eq:MPjmatrix} we can read off the $J$-matrix as
\begin{equation*}
J=\diag\left(-\frac{x-1}{x+1}, {\alpha}(1-y)(1+x), {\alpha}(1+y)(1+x)\right).
\end{equation*}
The twist potentials are globally constant and we set them without loss of generality to zero. The adaptations to the three different parts of the axis, then take the following form.
\begin{enumerate}[(1)]
\item Spacelike rod $z\in ({\alpha},\infty )$: Here we get with $y=1$ that
\begin{equation*}
\skew{7}{\tilde}{A} =\diag\left(-\frac{x-1}{x+1},2{\alpha}(x+1)\right), 
\end{equation*}
hence with $x=\frac{z}{{\alpha}}$ we obtain
\begin{equation*}
P_{1}(z)=\diag\left(-\frac{1}{2(z-\alpha)}, -\frac{z-\alpha}{z+\alpha},2(z+{\alpha})\right).
\end{equation*}
\item Horizon rod $z\in (-{\alpha},{\alpha})$: Using that $x=1$ on this part of the axis the above definition gives
\begin{equation*}
\skew{7}{\tilde}{A} = 2{\alpha} \cdot  \diag\left(1-y,1+y\right), 
\end{equation*}
hence with $y=\frac{z}{{\alpha}}$ we obtain
\begin{equation*}
P_{2}(z)=\diag\left(-\frac{1}{4(z^{2}-\alpha^{2})}, -2(z-{\alpha}),2(z+{\alpha})\right).
\end{equation*}
\item Spacelike rod $z\in (-\infty ,-{\alpha})$: Here we get with $y=-1$ that
\begin{equation*}
\skew{7}{\tilde}{A} =\diag\left(-\frac{x-1}{x+1},2{\alpha}(x+1)\right), 
\end{equation*}
hence $x=-\frac{z}{{\alpha}}$ yields
\begin{equation*}
P_{3}(z)=\diag\left(\frac{1}{2({z+\alpha})}, -\frac{z+\alpha}{z-{\alpha}},-2(z-{\alpha})\right).
\end{equation*}
\end{enumerate}

Last comes the example that features a major novelty of higher-dimensional relativity in comparison with four dimensions.

\section{Black Ring Solutions}

A five-dimensional black ring is a space-time with a black hole whose horizon has topology $S^{1}\times S^{2}$. This solution was originally obtained from the Kaluza-Klein $C$-metric solutions via a double Wick rotation of coordinates and analytic continuation of parameters \citet{Emparan:2002aa}. Because the ranges of mass and angular momenta of the black ring overlap with the Myers-Perry solution, it is an example for the problem in higher-dimensional general relativity that a black hole space-time is not uniquely determined only by  mass and angular momenta. This is because their different horizon topologies do not allow a smooth transition of both solutions into one another by changing the parameters smoothly. A detailed analysis of the black ring solution and its properties can be found in \cite{Elvang:2006aa,Emparan:2006aa}. For the first part of this section we take the results and notation from \cite[Sec.~VI]{Harmark:2004rm}.

One way of defining the black ring solution is the following
\begin{equation} \label{eq:BRmetric}
\begin{split}
\drm s^{2} & = -\frac{F(v)}{F(u)} \left(\drm t - C {\kappa} \frac{1+v}{F(v)}\,\drm {\varphi}\right)^{2} \\
		& \hspace{0.4cm} +\frac{2{\kappa}^{2}F(u)}{(u-v)^{2}}\left[-\frac{G(v)}{F(v)}\,\drm {\varphi}^{2} + \frac{G(u)}{F(u)}\,\drm {\psi}^{2} + \frac{1}{G(u)} \,\drm u^{2}-\frac{1}{G(v)} \,\drm v^{2} \right],
\end{split}
\end{equation}
where $F({\xi})$ and $G({\xi})$ are 
\begin{equation*}
F({\xi})=1+b{\xi},\quad G({\xi})= (1-{\xi}^{2})(1+c{\xi}),
\end{equation*}
and the parameters vary in the ranges
\begin{equation*}
0<c\leq b<1.
\end{equation*}
The parameter ${\kappa}$ has the dimension of length and for thin rings it is roughly the radius of the ring circle. The constant $C$ is given in terms of $b$ and $c$ by 
\begin{equation*}
C=\sqrt{2b(b-c)\frac{1+b}{1-b}},
\end{equation*}
and the coordinate ranges for $u$ and $v$ are
\begin{equation*}
-1\leq u\leq 1, \quad -\infty \leq v\leq -1
\end{equation*}
with asymptotic infinity recovered as $u\to v\to -1$. For the ${\varphi}$-coordinate the axis of rotation is $v=-1$, and for the ${\psi}$-direction the axis is divided in two components. First $u=1$ which is the disc bounded by the ring, and second $u=-1$ which is the outside of the ring, that is up to infinity. The horizon is located at $v=-\frac{1}{c}$ and outside of it at $v=-\frac{1}{b}$ lies an ergosurface. As argued in \cite[Sec.~5.1.1]{Emparan:2008aa} three independent parameters $b$, $c$, ${\kappa}$ is one too many, since for a ring with a certain mass and angular momentum we expect its radius to be dynamically fixed by the balance between centrifugal and tensional forces. This is here the case as well, because in general there are conical singularities\footnote{A brief note on the nature of those singularities and why they are called conical can be found in Appendix~\ref{app:consing}.} on the plane containing the ring, $u=\pm 1$. In order to cure them ${\varphi}$ and ${\psi}$ have to be identified with periodicity 
\begin{equation*}
{\Delta}{\varphi}={\Delta}{\psi}=4{\pi} \frac{\sqrt{F(-1)}}{|G'(-1)|}=2{\pi}\frac{\sqrt{1-b}}{1-c},
\end{equation*}
and the two parameters have to satisfy
\begin{equation} \label{eq:BRb}
b=\frac{2c}{1+c^{2}}. 
\end{equation}
This leaves effectively a two-parameter family of solutions as expected with the Killing vector fields $X_{0}=∂_{t}$, $X_{1}=∂_{\varphi}$ and $X_{2}=∂_{\psi}$. Henceforth, however, we will keep the conical singularity in and regard the parameter $b$ as free. By \eqref{eq:BRb} it can be replaced at any time.

A straight forward calculation shows
\begin{equation*}
\det J=\frac{4{\kappa}^{4}}{(u-v)^{4}}G(u)G(v), 
\end{equation*}
hence we define
\begin{equation*}
r=\frac{2{\kappa}^{2}}{(u-v)^{2}}\sqrt{-G(u)G(v)}. 
\end{equation*}
The harmonic conjugate can be calculated in the same way as for the Myers-Perry solution (for details see \cite[App.~H]{Harmark:2004rm}) and one obtains
\begin{equation*}
z=\frac{{\kappa}^{2}(1-uv)(2+cu+cv)}{(u-v)^{2}}.
\end{equation*}
Using expressions for $u$, $v$ in terms of $r$, $z$ (see \cite[App.~H]{Harmark:2004rm})
\begin{align*}
u & = \frac{(1-c)R_{1}-(1+c)R_{2}-2R_{3}+2(1-c^{2}){\kappa}^{2}}{(1-c)R_{1}+(1+c)R_{2}+2cR_{3}} \\
v & = \frac{(1-c)R_{1}-(1+c)R_{2}-2R_{3}-2(1-c^{2}){\kappa}^{2}}{(1-c)R_{1}+(1+c)R_{2}+2cR_{3}},
\end{align*}
where
\begin{equation*}
R_{1}=\sqrt{r^{2}+(z+c{\kappa}^{2})^{2}},\ R_{2}=\sqrt{r^{2}+(z-c{\kappa}^{2})^{2}},\ R_{3}=\sqrt{r^{2}+(z-{\kappa}^{2})^{2}},
\end{equation*}
the $J$-matrix can be computed as 
\begin{align*}
J_{00} & = -\frac{(1+b)(1-c)R_{1}+(1-b)(1+c)R_{2}-2(b-c)R_{3}-2b(1-c^{2}){\kappa}^{2}}{(1+b)(1-c)R_{1}+(1-b)(1+c)R_{2}-2(b-c)R_{3}+2b(1-c^{2}){\kappa}^{2}}, \\
J_{01} & = -\frac{2C{\kappa}(1-c)[R_{3}-R_{1}+(1+c){\kappa}^{2}}{(1+b)(1-c)R_{1}+(1-b)(1+c)R_{2}-2(b-c)R_{3}+2b(1-c^{2}){\kappa}^{2}}, \\
J_{22} & = \frac{(R_{3}+z-{\kappa}^{2})(R_{2}-z+c{\kappa}^{2})}{R_{1}-z-c{\kappa}^{2}}, \\
J_{11} & = -\frac{r^{2}}{J_{00}J_{22}}+\frac{J_{01}^{2}}{J_{00}},
\end{align*}
with the remaining components vanishing, and
\begin{align*}
{\erm}^{2ν} & = \left[(1+b)(1-c)R_{1}+(1-b)(1+c)R_{2}+2(c-b)R_{3}+2b(1-c^{2}){\kappa}^{2}\right]\\
& \hspace{0.4cm} ×\frac{(1-c)R_{1}+(1+c)R_{2}+2cR_{3}}{8(1-c^{2})^{2}R_{1}R_{2}R_{3}}.
\end{align*}

The rod structure consists of four components $(-\infty ,-c{\kappa}^{2})$, $(-c{\kappa}^{2},c{\kappa}^{2})$, $(c{\kappa}^{2},{\kappa}^{2})$, $({\kappa}^{2},\infty )$. 
\begin{enumerate}[(1)]
\item For $r=0$ and $z\in ({\kappa}^{2},\infty )$ we have $R_{3}-R_{1}+(1+c){\kappa}^{2}=0$ which implies $J_{01}=J_{11}=0$. Hence, the interval $({\kappa}^{2},\infty )$ is a semi-infinite spacelike rod in direction $∂_{\varphi}$.
\item For $r=0$ and $z\in (c{\kappa}^{2},{\kappa}^{2})$ we have $R_{2}+R_{3}-(1-c){\kappa}^{2}=0$ which implies $J_{22}=0$. Hence, the interval $(c{\kappa}^{2},{\kappa}^{2})$ is a finite spacelike rod in direction $∂_{\psi}$.
\item For $r=0$ and $z\in (-c{\kappa}^{2},c{\kappa}^{2})$ we have $R_{1}+R_{2}-2c{\kappa}^{2}=0$ which implies that the kernel of $J$ in this range is spanned by the vector
\begin{equation*}
\left(\begin{array}{ccc}1 & {\Gamma} & 0\end{array}\right)^{\mathrm{t}}, \quad \text{ where } {\Gamma}=\frac{b-c}{(1-c)Cκ}
\end{equation*}
is again the angular velocity. Thus, $(-c{\kappa}^{2},c{\kappa}^{2})$ is a finite timelike rod and it can be shown that it corresponds to an event horizon with topology $S^{2}\times S^{1}$ (a brief reasoning can be found in \cite[proof of Prop.~2 in Sec.~3]{Hollands:2008fp}).
\item For $r=0$ and $z\in (-\infty ,-c{\kappa}^{2})$ we have $R_{1}-R_{3}+(1+c){\kappa}^{2}=0$ which implies $J_{22}=0$. Hence, the interval $(-\infty ,-c{\kappa}^{2})$ is a semi-infinite spacelike rod in direction $∂_{\psi}$.
\end{enumerate}

We see that on the top rod one twist 1-form vanishes and for the other one we obtain as before
\begin{equation*}
\partial _{r} {\chi} = {\epsilon} \frac{F^{3}H}{G} \partial _{z} {\omega},\quad \partial _{z} {\chi} = - {\epsilon} \frac{F^{3}H}{G} \partial _{r} {\omega},
\end{equation*}
where
\begin{equation*}
\drm {\chi} =  \partial _{r}{\chi}\, \drm r + \partial _{z} {\chi}\, \drm z = * (T\wedge Ψ\wedge \drm T).
\end{equation*}
Note that
\begin{equation*}
{\omega} = \frac{J_{01}}{J_{00}}, \quad G^{2} = - \frac{r^{2}}{J_{00}J_{22}},
\end{equation*}
as $-r^{2} = \det J = (J_{00}^{\vphantom{1}}J_{11}^{\vphantom{1}}-J_{01}^{2})J_{22}^{\vphantom{1}}$. On $r=0$ we also see that
\begin{equation*}
R_{1} = |z+c{\kappa}^{2}|, \quad R_{2} = |z-c{\kappa}^{2}|, \quad R_{3} = |z-{\kappa}^{2}|,
\end{equation*}
and for ${\kappa}^{2}<z<\infty $ the moduli signs can be dropped. Then the metric coefficients behave as 
\begin{equation*}
J_{00} = \Ocal(1),\quad J_{01} = \Ocal(r^{2}),\quad J_{22}=\Ocal(1), \quad {\omega}^{2}= \Ocal(r^{2}),
\end{equation*}
so that we obtain
\begin{align*}
\partial _{z} {\chi} & = - {\epsilon} \frac{(-J_{00})^{\frac{3}{2}}(J_{22})^{\frac{1}{2}}}{r} (-J_{00})^{\frac{1}{2}} (J_{22})^{\frac{1}{2}}\, \partial _{r} \left(\frac{J_{01}}{J_{00}}\right) \\
& = -{\epsilon} \frac{J_{00}^{2}J_{22}^{\vphantom{1}}}{r}\, \partial _{r}\left(\frac{J_{01}^{\vphantom{1}}}{J_{00}^{\vphantom{1}}}\right).
\end{align*}
Now, if $J_{01}=r^{2}B(z)+\Ocal(r^{4})$, then
\begin{equation} \label{eq:BRtwistpot1}
\lim_{r\to 0} \partial _{z}{\chi} = -{\epsilon} \lim_{r\to 0} 2J_{00}J_{22}B(z).
\end{equation}
In order to determine $B(z)$ we do some auxiliary calculations. Denote $α=cκ^{2}$, $β=κ^{2}$. Then with $z>β$ and to leading order in $r$ it is
\begin{align*}
R_{1} & = (z+α)\left(1+\frac{r^{2}}{2(z+α)^{2}}\right),\ R_{2} = (z-α)\left(1+\frac{r^{2}}{2(z-α)^{2}}\right) \\
R_{3} & = (z-β)\left(1+\frac{r^{2}}{2(z-β)^{2}}\right),
\end{align*}
whence
\begin{equation*}
J_{22} = 2 (z-β)\frac{2(z+α)}{r^{2}}\frac{r^{2}}{2(z-α)} = \frac{2(z-β)(z+α)}{z-α}. 
\end{equation*}
Second we compute
\begin{equation} \label{eq:BRlambda}
J_{00} = -\frac{z-α}{z+λ}, \quad \text{where } λ= κ^{2}⋅\frac{2b-bc-c}{1-b}.
\end{equation}
Last, we obtain
\begin{equation*}
J_{01} = -\frac{C(1-c)\kappa^3}{2(1-b)}\frac{1}{(z-β)(z+α)(z+\lambda)}⋅r^{2}. 
\end{equation*}
Using these results \eqref{eq:BRtwistpot1} can be integrated to
\begin{equation*}
\left.{\chi}\right|_{r=0} = \frac{2\nu}{z+\lambda},\quad \nu=\frac{\epsilon C(1-c)\kappa^3}{1-b}.
\end{equation*}
Note that this agrees up to a constant with \cite[Eq.~(25)]{Tomizawa:2004aa}. Now we can compute the quantities which go in the patching matrix. The restriction $r=0$ is not explicitly mentioned, but still assumed in the following.
\begin{align*}
g{\chi} & = \frac{{\chi}}{J_{00}J_{22}} = -\frac{ν}{(z-β)(z+α)},\\
g & = \frac{1}{J_{00}J_{22}} = -\frac{z+λ}{2(z+α)(z-β)},
\end{align*}
For the last matrix entry we first calculate some auxiliaries. From \eqref{eq:BRlambda} we obtain
\begin{equation*}
b = \frac{λ+α}{λ+2β-α}, 
\end{equation*}
hence
\begin{equation*}
b-c=\frac{(β-α)(λ-α)}{β(λ+2β-α)}, \quad 1+b = \frac{2(λ+β)}{λ+2β-α}, \quad 1-b = \frac{2(β-α)}{λ+2β-α}.
\end{equation*}
This yields
\begin{equation*}
2ν^{2} = \frac{4b(b-c)(1+b)(1-c)^{2}κ^{6}}{(1-b)^{3}} = (λ+α)(λ-α)(λ+β),
\end{equation*}
which in turn justifies the following factorization
\begin{equation*}
(z-α)(z+α)(z-β)+2ν^{2} = \big(z+λ\big)\big(z^2-(β+λ)z -α^2+β λ+λ^2\big).
\end{equation*}
and eventually
\begin{align*}
J_{00}+g{\chi}^{2} & = -\frac{z-α}{z+λ} - \frac{2ν^{2}}{(z+λ)(z+α)(z-β)} \\
& = -\frac{z^2-(β+λ)z -α^2+β λ+λ^2}{(z+α)(z-β)}.
\end{align*}
The patching matrix for $z\in (β,\infty )$ and $r=0$ now is 
\begin{equation} \label{eq:BRpatmat}
\renewcommand\arraystretch{2.5}
P_{1}=\left(\begin{array}{ccc}-\dfrac{z+{\lambda}}{2(z+α)(z-β)} & \dfrac{{\nu}}{(z+α)(z-β)} & 0 \\
\cdot  & -\dfrac{z^{2}-{\gamma}z+{\delta}}{(z+α)(z-β)} & 0 \\
0 & 0 & \dfrac{2(z+α)(z-β)}{z-α}\end{array}\right),
\end{equation}
where the index again only indicates that it is adapted to the part of the axis which extends to $+\infty $ and where
\begin{equation} \label{eq:BRpar}
\renewcommand\arraystretch{2.5}
\begin{array}{ccc} α=c{\kappa}^2, & β={\kappa}^2, & {\lambda}=κ^{2}⋅\dfrac{2b-bc-c}{1-b}, \\
{\nu}=\dfrac{\epsilon C(1-c)\kappa^3}{1-b}, & {\gamma} = κ^{2}+λ, & {\delta}=-c^2κ^{4}+κ^{2} λ+λ^2 \end{array}.
\end{equation}
Note that this is based on the assumption that the periodicity of ${\varphi}$, ${\psi}$ is $2{\pi}$, otherwise it has to be modified according to \cite[Eq.~(4.17)]{Harmark:2004rm}.

From \eqref{eq:asymptPbot} we read off the conserved Komar quantities as
\begin{equation*} 
M = \frac{3{\pi}}{4}(λ+κ^{2}), \quad L_{1} = \dfrac{π C(1-c)\kappa^3}{1-b}, \quad L_{2}=0.
\end{equation*}
\chapter{The Converse} \label{ch:converse}

First let us recall the twistor construction for five-dimensional space-times as obtained in Chapter~\ref{ch:bundles} and \ref{ch:fivedim}. It can be summarized as follows.

\begin{summ}
There exists a one-to-one correspondence between five-dimensional stationary and axisymmetric space-times and rank-3 bundles $E→\Rcal$ over the reduced twistor space $\Rcal$, which consists of two Riemann spheres identified over a certain region. 

For $J$ being the matrix of inner products of Killing vectors, we define the Ernst potential as (see Definition~\ref{def:ernstpot})
\begin{equation*} 
\renewcommand{\arraystretch}{1.4}
J'=\frac{1}{\det \skew{7}{\tilde}{A}} \left(\begin{array}{cc}\hphantom{-}1 & -χ^{\mathrm{t}} \\-χ & \det \skew{7}{\tilde}{A} \cdot \skew{7}{\tilde}{A} + χχ^{\mathrm{t}}\end{array}\right),
\end{equation*}
where $\skew{7}{\tilde}{A}$ is obtained from $J$ by cancelling an appropriate row and column, and $χ=(χ_{1},χ_{2})$ are the twist potentials.

If $P$ is the patching matrix of $E→\Rcal$, then $J'(z)=P(z)$ where both are non-singular for $r→0$. 
\end{summ}

Remember that the bundle $E→\Rcal$ was characterized by the so-called twistor data consisting of the patching matrix and three integers. For the bundle corresponding to $J$ itself these integers are $p_{0}=1$, $p_{1}=p_{2}=0$ and for the bundle corresponding to the Ernst potential $J'$ these are $p_{0}=p_{1}=p_{2}=0$. Hence it comes down to determining $P$ when parameterizing the bundle.

\begin{cor}
The patching matrix $P$ (adapted to any portion of the axis $r=0$) determines the metric and conversely. 
\end{cor}
\begin{proof}[Sketch of Proof]
$J'(r,z)$ is obtained from $P(w)$ by the splitting procedure (see Section~\ref{sec:twistconstr}) and conversely $P(w)$ is the analytic continuation of $J'(r=0,z)$.
\end{proof}

As seen in Chapter~\ref{ch:bh} the classification of black holes in four dimensions (see Theorem~\ref{thm:carter}) does not straight-forwardly generalize to five dimensions. The Myers-Perry solution and the black ring are space-times whose range of parameters (mass and angular momenta) do have a non-empty intersection, but their horizon topology is different, which means they cannot be isometric. In order to address this issue the rod structure (see Definition~\ref{def:rodstr}) is introduced to supplement the set of parameters. 

Using this extended set of parameters the following theorem is a first step towards a classification.

\begin{thm}[\citet{Hollands:2008fp}] \label{thm:holluniqueness}
Two five-dimensional, asymptotically flat vacuum space-times with connected horizon where each of the space-times admits three commuting Killing vector fields, one time translation and two axial Killing vector fields, are isometric if they have the same mass and two angular momenta, and their rod structures coincide.
\end{thm}

Note, however, that \cite[Prop.~3.1]{Chrusciel:2011eu} suggests that by adding the rod structure to the list of parameters the mass becomes redundant, at least for connected horizon.

Theorem~\ref{thm:holluniqueness} answers the question about uniqueness of five-dimensional black holes, but not existence. In other words, we do not yet know whether rod structure and angular momenta determine the twistor data, that is essentially $P$, and thereby the metric.

\begin{conj}
Rod structure and angular momenta determine $P$ (even for a disconnected horizon). 
\end{conj}

\section[From Rod Structure to Patching Matrix]{From Rod Structure to Patching Matrix --- an Ansatz}

In the following we will present an ansatz for this reconstruction of the patching matrix from the given data, exemplified in cases where the rod structure has up to three nuts. 

Given a rod structure with nuts at $\{a_{i} | a_{i}∈ℝ\}_{1≤i≤N}$ we know that $P$ can at most have single poles at these nuts (see Corollary~\ref{cor:singofP} and Proposition~\ref{prop:simplepoles}). We shall see that this can also be derived from the switching procedure (Theorem~\ref{thm:switching}) and thus we make the ansatz
\begin{equation*}
P(z) = \frac{1}{{\Delta}} P'(z), 
\end{equation*}
where ${\Delta}=\prod _{i=1}^{N}(z-a_{i})$ and the entries of $P'(z)$ are holomorphic in $z$. If we now, moreover, choose $P$ to be adapted to the top outermost rod $(a_{\scriptscriptstyle N},\infty )$, then Section~\ref{sec:exasympt} tells us its asymptotic behaviour as $z → ∞$, that is $P$ asymptotes $P_{+}$ given in \eqref{eq:asymptPtop}. This implies that the entries of $P'(z)$ are in fact polynomials,
\begin{equation*}
\renewcommand{\arraystretch}{1.5}
P'(z) = \left(\begin{array}{ccc} q_{\scriptscriptstyle N-1}(z) & q_{\scriptscriptstyle N-2}(z) & q_{\scriptscriptstyle N-2}(z) \\ \cdot  & q_{\scriptscriptstyle N}(z) & q_{\scriptscriptstyle N-1}(z) \\ \cdot  & \cdot  & q_{\scriptscriptstyle N+1}(z)\end{array}\right),
\end{equation*}
where $q_{k}$ is a polynomial of degree $k$.\footnote{The notation shall just indicate the degree of the polynomials, that is $q_{\scriptscriptstyle N-1}$ and $q_{\scriptscriptstyle N-2}$ in different entries of the matrix can still be different polynomials, and if $N-2<0$ then it shall be the zero-polynomial.} In fact, from \eqref{eq:asymptPtop} we can not only deduce the degree of the polynomials but also their leading coefficients. The diagonal entries will have leading coefficient $-\frac{1}{2}$, $-1$, and 2, respectively, and the leading coefficients on the superdiagonal will be proportional to the angular momenta. Similarly, one can use \eqref{eq:asymptPbot} for $P$ adapted to the bottom outermost rod $(-∞, a_{1})$. Note that this does not impose any further restrictions on the coefficients of the space-time metric apart from being analytic.

The number of free parameters in $P$ equals the number of independent coefficients in the polynomials. Our aim must be to tie our space-time metric by expressing all those parameters in terms of not more than the $a_{i}$ and the angular momenta $L_{1}$, $L_{2}$. Any free parameter left in $P$ means another free parameter in our (family of) solutions. 

\begin{ex}[One-Nut Rod Structure]\mbox{}\\
Consider the case where the rod structure has one nut, which is without loss of generality at the origin (remember that a shifted rod structure corresponds to a diffeomorphic space-time), see Figure~\ref{fig:onenutrodstr}. We do not make assumptions about the angular momenta $L_{1}$, $L_{2}$. 

\begin{figure}[htbp]
\begin{center}
     \scalebox{0.8}{\input{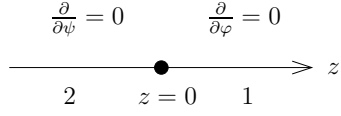}}
     \caption{Rod Structure with one nut at the origin. The numbers are just labelling the parts of the axis.} 
     \label{fig:onenutrodstr}
\end{center}
\end{figure}
According to our ansatz we have for the patching matrix on the top part of the axis
\begin{equation*}
\renewcommand{\arraystretch}{1.5}
P_{1} = \frac{1}{z} 
\left(\begin{array}{ccc}
-\dfrac{1}{2} + c_{1} z & 0 & 0 \\
0 & -z + c_{2} & c_{3} \\
0 & c_{3} & 2z^{2} + c_{4}z+ c_{5}
\end{array}\right),
\end{equation*}
which implies $L_{1}=ζ=0$. On the other hand for the bottom part it is
\begin{equation*}
\renewcommand{\arraystretch}{1.5}
P_{2} = \frac{1}{z} 
\left(\begin{array}{ccc}
-\dfrac{1}{2} + \tilde c_{1} z & 0 & 0 \\
0 & -z + \tilde c_{2} & \tilde c_{3} \\
0 & \tilde c_{3} & 2z^{2} + \tilde c_{4}z+ \tilde c_{5}
\end{array}\right),
\end{equation*}
and therefore necessarily $L_{2}=ζ=0$, thus $c_{3}=\tilde c_{3}=0$. This forces the patching matrix to be diagonal and since it has to have unit determinant,
\begin{equation*}
\det P_{1} = \frac{1}{z^{3}} \left(-\frac{1}{2} + c_{1} z\right)\left(-z + c_{2}\right)\left(2z^{2} + c_{4}z+ c_{5}\right) = 1,
\end{equation*}
we obtain $c_{1}=c_{2}=c_{4}=c_{5}=0$. But this is the patching matrix for flat space (see Section~\ref{sec:exMinkrodstr}). 

Hence we have shown that for a rod structure with one nut not all values for the conserved quantities are allowed, in fact they all (including mass) have to vanish, which in turn uniquely determines the space-time as Minkowski space.\\ \mbox{} \hfill $\blacksquare$
\end{ex}

Attempting the same for a rod structure with two nuts one will quickly notice that more tools are necessary in order to fix all the parameters. 

\begin{thm} \label{thm:invpmatrix}
If $P_{+}$ is the patching matrix adapted to $(a_{\scriptscriptstyle N}, ∞)$, then $P_{-}^{\vphantom{1}}=MP_{+}^{-1}M$ with $M=\left(\begin{smallmatrix}0&0&1\\0&1&0\\1&0&0\end{smallmatrix}\right)$ is the patching matrix adapted to $(-∞,a_{1})$.
\end{thm}
\begin{proof}
The proof is based on the twistor construction as we have seen it in Chapter~\ref{ch:bundles}. For ease of reading we will review some of the elements involved. 

The considerations in Section~\ref{sec:twistconstr} started with a map from the correspondence space to the reduced twistor space $p:\Fcal_{r} = Σ×\Xcal→\Rcal$, where $Σ$ is a two-dimensional conformal manifold, $r$ a solution of the Laplace equation on $Σ$, $z$ its harmonic conjugate and $\Xcal$ a Riemann sphere. For a fixed $σ∈Σ$ this map can be restricted to
\begin{equation*}
π: \Xcal → \Rcal, \quad ζ \mapsto w= \frac{1}{2} r(σ) \left(ζ^{-1}-ζ\right)+z(σ). 
\end{equation*}
Given a bundle $E→\Rcal$ one had to assume that $π^{*}(E)$, the pullback of $E$ to $\Xcal$, is trivial in order to construct $J$. Then the splitting procedure will provide a $J(r,z)$, which depends smoothly on $r$ and $z$. 

$\Rcal$ consists of two Riemann spheres that are identified over a certain set. There will always be points, like $w=∞$, where the spheres cannot be glued. These points are for obvious reasons called double points and the case that is interesting for us is when these points are real and there is only a finite number of them. In other words the two spheres are glued together up to the set $\{∞,a_{1},…,a_{\scriptscriptstyle N}\}$. 

Therefore, $π$ is only well-defined if one specifies the assignment of these poles to the spheres. The roots of the double points $a_{i}$ satisfy
\begin{equation*}
rζ_{i}^{2} + 2(a_{i}^{\vphantom{2}}-z)ζ_{i}^{\vphantom{2}} -r =0,
\end{equation*}
so they are
\begin{equation} \label{eq:rootsdoublepoints}
ζ_{i}^{±} = \frac{1}{r} \left((a_{i}^{\vphantom{2}}-z)±\sqrt{(z-a_{i}^{\vphantom{2}})^{2}+r^{2}}\right).
\end{equation}
We note two things. First, the spheres are labelled by saying that the roots of $w=∞$, namely $ζ=0$ and $ζ=∞$, are mapped to $π(0)=∞_{0}∈S_{0}$ and $π(∞)=∞_{1}∈S_{1}$. Second, $r$ and $z$ are chosen as parameters in the very beginning, but the obtained expressions depend smoothly on $r$ and $z$ so that we can vary them and follow the consequences. One observation of this kind is that for $r→0$ one of the roots in \eqref{eq:rootsdoublepoints} tends to zero and one to infinity. 

Hence, given a solution $J$ of Yang's equation~\eqref{eq:redyang} and the corresponding bundle $E→\Rcal_{U}$, $U=\Cbb\Pbb^{1}\backslash \{∞,a_{1},…,a_{\scriptscriptstyle N}\}$ the region where the spheres are identified and $U$ not simply connected, then the description of the twistor space as $S_{0}^{\vphantom{2}}∪S_{1^{\vphantom{2}}}$ and the patching matrix $P$ are adapted to the component $\Ccal$ of the real axis if those $ζ_{i}^{±}$ that tend to zero for $r→0$ on $\Ccal$ are assigned to $S_{0}^{\vphantom{2}}$ and those that tend to infinity are assigned to $S_{1}^{\vphantom{2}}$; see also \cite[Prop.~3.2]{Fletcher:1990aa}. This is merely a requirement of consistent behaviour under the variation of $r$ and $z$, since $π(0)∈S_{0}^{\vphantom{2}}$ and $π(∞)∈S_{1}^{\vphantom{2}}$. Note that in this case on $\Ccal$ it is $P(z)=J'(z)$. 

More explicitly this can be stated as
\begin{equation*}
\renewcommand{\arraystretch}{1.3}
ζ_{i}^{+} → \left\{ \begin{array}{ll}
0, & i≤k \\
∞, & i>k
\end{array}\right. \quad \text{and} \quad
ζ_{i}^{-} → \left\{ \begin{array}{ll}
∞, & i≤k \\
0, & i>k
\end{array}\right. ,
\end{equation*}
for an adaptation to $\Ccal=(a_{k},a_{k+1})$ and for $r→0$ on $z∈(a_{k},a_{k+1})$.

This allows to draw the conclusion that for a given a bundle a change of adaptation from $(a_{k},a_{k+1})$ to $(a_{k-1},a_{k})$ is achieved by swapping the assignment of $π(ζ_{k}^{±})$ to the spheres; see \cite[Sec.~3.2]{Fletcher:1990aa}. Following this idea, one then obtains the adaptation to $(-∞,a_{1})$  from an adaptation to $(a_{\scriptscriptstyle N}, ∞)$ by swapping all double points $π(ζ_{k}^{±})$, $1≤k≤N$. However, the latter is the same as swapping the double point at infinity which means relabelling the spheres by saying $π(0)∈S_{1}^{\vphantom{2}}$ and $π(∞)∈S_{0}^{\vphantom{2}}$. \\

Now let us step back from this line of ideas and have a look from another side. Note that if $J$ is a solution of Yang's equation~\eqref{eq:redyang} with $\det J =1$ and $J$ is symmetric, $J=J^{\mathrm{t}}$, then $J^{-1}$ is a solution of Yang's equation as well with $\det J^{-1}=1$. On the other hand, just by inspection of the splitting procedure one notices that $J$ is defined as a linear map $J:E_{∞_{0}}→E_{∞_{1}}$, where $E_{w}$ is the fiber of $E→\Rcal$ over $w∈\Rcal$, and that $J^{-1}$ is the solution of Yang's equation generated by the bundle where the spheres are swapped, $S_{0}↔S_{1}$ (see also property (3) in \cite[Sec.~4]{Woodhouse:1988ek}). But this is precisely what we have done above.\\

The last point to note is that even though we have shown that $P^{-1}$ is adapted to $(-∞,a_{1})$ it does not necessarily have to be in our standard form due to the gauge freedom in the splitting procedure. Using the asymptotic form of $P$ in \eqref{eq:asymptPtop} and comparing the fall-off of $P^{-1}$ with \eqref{eq:asymptPbot} one sees that a flip of the first and third row and column brings $P^{-1}$ into the desired standard form. This is implemented by the conjugation with $M$, which completes the proof.
\end{proof}

\begin{cor} \label{cor:invpmatrix}
If $P_{+}$ is the patching matrix adapted to $(a_{\scriptscriptstyle N}, ∞)$, then $Δ$ divides all $2×2$-minors of $Δ⋅P_{+}^{\vphantom{-1}}=P'_{+}$.
\end{cor}
\begin{proof}
From Theorem~\ref{thm:invpmatrix} and Proposition~\ref{prop:simplepoles} it follows that $P_{+}^{-1}$ has at most simple poles at the nuts. But by the general formula for the inverse of a matrix the entries of $P_{+}^{-1}$ are (up to a sign) $P_{+}^{i,j}/Δ^{2}$, where $P_{+}^{i,j}$ is the $2×2$-minor of $Δ⋅P_{+}$ obtained by cancelling the $i^{\mathrm{th}}$ row and $j^{\mathrm{th}}$ column. Hence one factor of $Δ$ has to cancel.
\end{proof}

This turns out to be a powerful tool as seen in the next example. But before we would like to point out an implication of the last corollary. 
\begin{rem}
Taking the Ernst potential \eqref{eq:highernst} in five dimensions and writing it in the following way
\begin{equation*}
\renewcommand{\arraystretch}{1.4}
J'_{+}(z)=\left(\begin{array}{cc} \hphantom{-}g & -gχ^{\mathrm{t}} \\-gχ & \skew{7}{\tilde}{A} + gχχ^{\mathrm{t}}\end{array}\right)
 = \frac{1}{Δ}\left(\begin{array}{cc}p_{0} & \vec{p}^{\,\mathrm{t}} \\ \vec{p} & \Pbb \end{array}\right),
\end{equation*}
the matrix of metric coefficients $\skew{7}{\tilde}{A}$ as a function of $z$ is obtained as
\begin{equation}\label{eq:metricfromP}
\skew{7}{\tilde}{A} = \frac{1}{Δ} \Pbb - \frac{1}{Δp_{0}}\vec{p}\cdot \vec{p}^{\,\mathrm{t}} = \frac{1}{Δp_{0}}\left(p_{0}\Pbb-\vec{p}\cdot \vec{p}^{\,\mathrm{t}}\right).
\end{equation}
All entries of $p_{0}\Pbb-\vec{p}\cdot \vec{p}^{\,\mathrm{t}}$ are $2×2$-minors of $Δ⋅J'$, hence $Δ$ divides them. Thus $\skew{7}{\tilde}{A}= \tilde \Pbb / p_{0}$ where the entries of $\tilde \Pbb$ are polynomials in $z$. We remember from Section~\ref{sec:adapt} that $p_{0}/Δ$ blows up when we approach $a_{N}$, that means $p_{0}$ cannot have a factor $(z-a_{N})$. So, the entries of $\skew{7}{\tilde}{A}$ are bounded as $z↓a_{N}$, a feature which is consistent with our picture of space-time.

Note, however, that we cannot extend that to other nuts without changing the adaptation, that is to say, the expression for a metric coefficient $J_{ij}(z,r=0)$, $z>a_{N}$, might contain poles for $z<a_{N}$ as the $J_{22}$ for the black ring shows
\begin{align*}
\renewcommand{\arraystretch}{2.5}
J_{22}(z) = \left\{\begin{array}{cl}\dfrac{2(z-{\kappa}^{2})(z+c{\kappa}^{2})}{z-c{\kappa}^{2}} & ,\ z>{\kappa}^{2}, \\0 & ,\ c{\kappa}^{2}<z<{\kappa}^{2}, \\ \dfrac{2(z-c{\kappa}^{2})(z+c{\kappa}^{2})}{z-{\kappa}^{2}} & ,\ -c{\kappa}^{2}<z<c{\kappa}^{2}, \\0 & ,\ z<-c{\kappa}^{2}.\end{array}\right.
\end{align*}
The terms in the denominators vanish for certain values of $z$, but these are not singularities of the metric since they are not in the domain of the respective expression. The reason for pointing out this is that when we try to fix the free parameters in our ansatz one cannot take \eqref{eq:metricfromP} and say because the metric is regular the denominator has to divide the numerator up to a constant. 

Note further that even though Theorem~\ref{thm:invpmatrix} and Corollary~\ref{cor:invpmatrix} generalize directly to $n$ dimensions, the conclusion for the metric coefficients in $\skew{7}{\tilde}{A}$ does not. This is because in higher dimensions the entries of $\skew{7}{\tilde}{A}$ will still consist of certain $2×2$-minors of $P_{+}^{\vphantom{1}}$ as in \eqref{eq:metricfromP}, whereas $P_{+}^{-1}$ being a patching matrix requires $Δ^{n-4}$ to divide the $(n-3)×(n-3)$-minors of $P_{+}^{\vphantom{1}}$. This coincides in five dimensions, but is not implied automatically for dimensions greater than five. Yet, the boundedness of the metric coefficients ought to hold always, so that it at most gives extra conditions on the free parameters.\\ \mbox{} \hfill $\blacksquare$
\end{rem}
Now lets turn to the example promised earlier.
\begin{ex}[Two-Nut Rod Structure] \mbox{}\\
Consider the rod structure as in Figure~\ref{fig:twonutrodstr}. 
\begin{figure}[htbp]
\begin{center}
     \scalebox{0.8}{\input{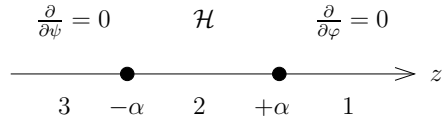}}
     \caption{Rod structure with two nuts.} 
     \label{fig:twonutrodstr}
\end{center}
\end{figure}

In line with the above ansatz we start off from
\begin{equation} \label{eq:MPpatmatansatz}
\renewcommand{\arraystretch}{2.5}
P=\frac{1}{z^{2}-α^{2}} \left(\begin{array}{ccc}-\dfrac{1}{2}z+c_{1} & \dfrac{L_{1}}{π} & c_{2}\\ \cdot  & -z^{2}+c_{3}z+c_{4} & -\dfrac{2L_{2}}{π}z+c_{5} \\ \cdot  & \cdot  & 2z^{3}+c_{6}z^{2}+c_{7}z+c_{8}\end{array}\right),
\end{equation}
which we assume to be adapted to $(a_{N},∞)$ and where the orientation of the basis is without loss of generality chosen such that $\epsilon=1$ in \eqref{eq:asymptPtop}.

We now make use of Corollary~\ref{cor:invpmatrix} which for the minor obtained by cancelling the third row and first column yields
\begin{equation*}
c_{2}z^{2} + \left(-c_{2}c_{3} - \frac{2L_{1}L_{2}}{π^{2}}\right)z-c_{2}c_{4}+\frac{L_{1}c_{5}}{π} \sim z^{2}-α^{2}. 
\end{equation*}
Comparing the (ratio of) coefficients returns
\begin{align}
c_{2}c_{3} & = -\frac{2L_{1}L_{2}}{π^{2}}, \label{eq:MPcoeff1}\\
c_{4} & = α^{2} +\frac{L_{1}c_{5}}{πc_{2}}.\label{eq:MPcoeff2}
\end{align}

Choosing the minor obtained from cancelling second row and third column we repeat this and get
\begin{equation*}
\frac{L_{2}}{π} z^{2} - \left(\frac{1}{2}c_{5}+\frac{2L_{2}c_{1}}{π}\right)z + c_{1}c_{5} - \frac{L_{1}c_{2}}{π} \sim z^{2}-α^{2},
\end{equation*}
thus
\begin{align}
c_{5} & = - \frac{4L_{2}}{π} c_{1}, \label{eq:MPcoeff3}\\
4 c_{1}^{2} & = α^{2}-\frac{L_{1}}{L_{2}}c_{2}^{\vphantom{1}}. \label{eq:MPcoeff4}
\end{align}
These four equations allow us to express $c_{1}$, $c_{2}$, $c_{4}$ and $c_{5}$ in terms of $c_{3}$. 

The coefficients $c_{7}$ and $c_{8}$ can be fixed by the minor which results from cancelling the second row and the first column
\begin{equation*}
-\frac{2L_{1}}{π} z^{3} + \frac{c_{6}L_{1}}{π}z^{2} + \left(\frac{c_{7}L_{1}}{π}+\frac{2c_{2}L_{2}}{π}\right) z + \frac{L_{1}c_{8}}{π}-c_{2}c_{5}  \sim z^{3} + bz^{2} - α^{2} z - b α^{2},
\end{equation*}
where $b$ is some constant. Again the ratios of the coefficients for the linear over the cubic and the constant over the quadratic term give
\begin{align}
c_{7} & = - 2 α^{2} - \frac{2L_{2}}{L_{1}}c_{2}, \label{eq:MPcoeff5} \\
c_{8} & = -α^{2} c_{6} + \frac{π}{L_{1}} c_{2}c_{5}. \label{eq:MPcoeff6}
\end{align}
The last coefficient that remains undetermined is $c_{6}$, but the determinant is going to help us for this. The requirement $\det P = 1$ implies
\begin{align*}
\left(z^{2} -α^{2}\right)^{3} & = z^{6} + \left(\frac{1}{2}c_{6} - 2c_{1}-c_{3}\right) z^{5} \\
							& \hspace{0.4cm} + \left(2c_{1}c_{3}-c_{1}c_{6}-c_{4}-\frac{1}{2}c_{3}c_{6}+\frac{1}{2}c_{7}\right)z^{4} + …
\end{align*}
The quintic term immediately gives the desired expression
\begin{equation}\label{eq:MPcoeff7}
c_{6} = 4c_{1}+2c_{3}. 
\end{equation}
Exploiting furthermore the quartic term we get 
\begin{equation*}
-3α^{2} = 2 c_{1}c_{3}- c_{1}c_{6} - c_{4} - \frac{1}{2}c_{3}c_{6}+\frac{1}{2}c_{7},
\end{equation*}
which, by using the above obtained relations, is equivalent to
\begin{equation} \label{eq:MPcoeff9}
α^{2} = 4c_{1}^{2} + 4 c_{1}^{\vphantom{1}}c_{3}^{\vphantom{1}} + c_{3}^{2} + \frac{L_{1}}{L_{2}}c_{2}^{\vphantom{1}}.
\end{equation}
Let us relabel the parameters in accordance with \cite{Harmark:2004rm} as follows
\begin{equation*}
c_{3}^{\vphantom{1}} = \frac{1}{2} ρ_{0}^{2}, \quad L_{1}^{\vphantom{1}} = \frac{π}{4} a_{1}^{\vphantom{2}}ρ_{0}^{2},  \quad L_{2}^{\vphantom{1}} = \frac{π}{4} a_{2}^{\vphantom{2}}ρ_{0}^{2}.
\end{equation*}
Note that from the asymptotic patching matrix we see that $c_{3}$ is proportional to the mass which justifies the implicit assumption about its positiveness in the above definition. However, the parameters $(ρ_{0},a_{1},a_{2})$ are not unconstrained as we will see soon.

By \eqref{eq:MPcoeff1} it is
\begin{equation} \label{eq:MPcoeff10}
c_{2}^{\vphantom{1}} = - \frac{1}{4} a_{1}^{\vphantom{1}}a_{2}^{\vphantom{1}}ρ_{0}^{2}.
\end{equation}
Equations~\eqref{eq:MPcoeff4}, \eqref{eq:MPcoeff9}, \eqref{eq:MPcoeff10} imply
\begin{equation*}
\frac{L_{1}}{L_{2}} c_{2}^{\vphantom{1}} = 4 c_{1}^{\vphantom{1}} c_{3}^{\vphantom{1}} + c_{3}^{2} + \frac{L_{2}}{L_{1}} c_{2}^{\vphantom{1}} \quad ⇒ \quad c_{1}^{\vphantom{1}} = - \frac{1}{8} \left(ρ_{0}^{2}+a_{1}^{2}-a_{2}^{2}\right).
\end{equation*}
Moreover, from \eqref{eq:MPcoeff3} and \eqref{eq:MPcoeff7} we obtain
\begin{equation*}
c_{5}^{\vphantom{1}} = \frac{1}{8} a_{2}^{\vphantom{1}} ρ_{0}^{2} \left(ρ_{0}^{2}+a_{1}^{2}-a_{2}^{2}\right) \quad \text{and} \quad c_{6}^{\vphantom{1}} = \frac{1}{2} \left(ρ_{0}^{2}-a_{1}^{2}+a_{2}^{2}\right).
\end{equation*}
Continuing with \eqref{eq:MPcoeff2} yields
\begin{equation*}
c_{4}^{\vphantom{1}} = α^{2}-\frac{1}{8} ρ_{0}^{2}\left(ρ_{0}^{2}+a_{1}^{2}-a_{2}^{2}\right),
\end{equation*}
so do \eqref{eq:MPcoeff5} and \eqref{eq:MPcoeff6} give
\begin{align*}
c_{7}^{\vphantom{1}} & = -2α^{2} + \frac{1}{2} a_{2}^{2} ρ_{0}^{2}, \\
c_{8}^{\vphantom{1}} & = \frac{1}{2} α^{2} \left(-ρ_{0}^{2}+a_{1}^{2}-a_{2}^{2}\right) - \frac{1}{8} a_{2}^{2}ρ_{0}^{2}\left(ρ_{0}^{2}+a_{1}^{2}-a_{2}^{2}\right).
\end{align*}
With these parameters being determined and with the help of \eqref{eq:MPcoeff9} we can write $α$ explicitly as 
\begin{equation} \label{eq:MPalpha}
α^{2} = \frac{1}{16}\left(ρ_{0}^{2}-a_{1}^{2}-a_{2}^{2}\right)^{2}-\frac{1}{4}a_{1}^{2}a_{2}^{2}.
\end{equation}
Comparing those expressions with \eqref{eq:MPpatmat} one will find that they coincide. However, note that
\begin{equation*}
16α^{2} = ρ_{0}^{4}-2ρ_{0}^{2}\left(a_{1}^{2}+a_{2}^{2}\right)+\left(a_{1}^{2}-a_{2}^{2}\right)^{2},
\end{equation*}
which implies that for real non-zero $α$ we need the left hand side to be positive and therefore $ρ_{0}^{2}>\left(|a_{1}^{\vphantom{1}}|+|a_{2^{\vphantom{1}}}|\right)^{2}$, a condition on the asymptotic quantities familiar from the discussion of the Myers-Perry solution in \cite{Emparan:2008aa} and \cite{Myers:2011yc}, or we need $0<ρ_{0}^{2}<\left(|a_{1}^{\vphantom{1}}|-|a_{2}^{\vphantom{1}}|\right)^{2}$. This latter possibility is ruled out in \cite{Emparan:2008aa} and \cite{Myers:2011yc} by the additional requirement that $ρ_{0}^{2}>|a_{1}^{\vphantom{1}}|^{2}+|a_{2}^{\vphantom{1}}|^{2}$, and it is also forbidden by \cite[Prop.~3.1]{Chrusciel:2011eu}. Yet, it is allowed by our analysis. So, characterizing our solution by the mass and the two angular momenta we obtain next to the Myers-Perry space-time another branch, which is unphysical.

Mass and angular momenta form a set of three parameters and the position of the nuts can be expressed in terms of these three parameters. This is more than one would have expected by only regarding Theorem~\ref{thm:holluniqueness}. However, we stated already that by \cite[Prop.~3.1]{Chrusciel:2011eu} the mass is redundant in the set of parameters. Here we did not actually eliminate the mass, but the rod length. If one tries vice versa, that is one tries to replace $M$ by $α$ in the set of parameters, then by rearranging \eqref{eq:MPalpha} one is looking for all positive $c_{3}$ which satisfy a $6^{\mathrm{th}}$ order polynomial. As shown in Appendix~\ref{app:MPpar}, with no further conditions on $(α>0, L_{1}^{\vphantom{1}},L_{2}^{\vphantom{1}})$ there are also two positive solutions for $c_{3}$ (unless $L_{1}^{2}=L_{2}^{2}$ when there is only one). Thus, we are facing the same problem as before, and this suggests that further conditions need to be imposed on the parameters in order to rule out the unphysical solutions.

Some of the steps above, when we determined all the parameters in the patching matrix, only work for $L_{1},L_{2}≠0$. Assuming that one of the angular momenta vanishes leads to dichotomies at certain steps when tracing back the parameters. Some of the branches in this tree of possibilities are dead ends others lead to valid solutions such as the Myers-Perry solution with one vanishing angular momentum or an ultrastatic solution, that is where $g_{tt}=1$, $g_{ti}=0$ (but this is not physical as the mass is zero). On the other hand at no point we used the fact that the middle rod is a horizon.

Note also that issues of conicality cannot arise here as the periodicities of $\phi,\psi$ are chosen to be $2\pi$ on the outer parts of 
the axis and no further spatial rods are left. When we turn to higher numbers of nuts more will be required.\\ \mbox{} \hfill $\blacksquare$
\end{ex}
Moving on to the next level, that is to a rod structure with three nuts, we will consider the simpler case where one of the Killing vectors is hypersurface-orthogonal. This simplifies the computation in comparison with the general case as we will. 
\begin{ex}[Three-Nut Rod Structure with one Hypersurface-Orthogonal Killing Vector]\mbox{}\\
We consider the rod structure as in Figure~\ref{fig:threenutrodstr}. 
\begin{figure}[htbp]
\begin{center}
     \scalebox{0.8}{\input{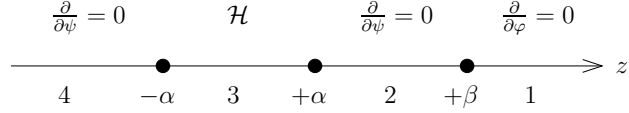}}
     \caption{Rod structure with three nuts, where $α,β>0$, and $S^{2}×S^{1}$ horizon.} 
     \label{fig:threenutrodstr}
\end{center}
\end{figure}
Together with $L_{1}=L\neq 0$, $L_{2}=0$ this comprises our twistor data. In order to simplify the calculations we would like to make assumptions such that the two non-diagonal entries in the third row and column of the patching matrix vanish (when adapted to $(β,∞)$). One therefore needs $g_{t{\psi}}=g_{{\varphi}{\psi}}=0$, however, this cannot be concluded from $L_{2}=0$, as the Black Saturn shows \cite{Elvang:2007rd}. The Black Saturn is non-static, but allows vanishing total angular momentum. We thus make the assumption that $\partial_{\psi}$ is hypersurface-orthogonal, that is $Ψ∧\drm Ψ=0$, and in Appendix~\ref{app:hso} it is shown that this implies $χ_{2}=g_{t{\psi}}=g_{{\varphi}{\psi}}=0$. 

These assumptions turn our ansatz into
\begin{equation*}
\renewcommand{\arraystretch}{1.5}
P = \frac{1}{Δ}
\left( \begin{array}{ccc}
q(z) & l(z) & 0 \\
l(z) & c(z) & 0 \\
0 & 0 & Q(z)
\end{array}\right),
\end{equation*}
where
\begin{equation*}
\renewcommand{\arraystretch}{2}
\begin{array}{rcl}
Δ (z)& = & (z+α)(z-α)(z-β), \\
q(z) & = & \dfrac{1}{2} z^{2} + c_{1}z+c_{2}, \\
l(z) & = &\dfrac{L}{π} z + c_{3}, \\
c(z) & = &-z^{3} + c_{4} z^{2} + c_{5} z +c_{6}, \\
Q(z) & = & 2 z^{4} + c_{7} z^{3} + c_{8} z^{2} + c_{9} z + c_{10}.
\end{array}
\end{equation*}
Theorem~\ref{thm:invpmatrix} gives the following conditions
\begin{equation} \label{eq:BRcoeff1}
\renewcommand{\arraystretch}{1.4}
\begin{array}{rcrrl}
qc-l^{2} & = & \tilde q_{1}Δ,\qquad & \tilde q_{1} & \hspace{-0.2cm} \text{ quadratic},\\
Qq^{\hphantom{2}} & = & \tilde c_{1}Δ,\qquad & \tilde c_{1} & \hspace{-0.2cm} \text{ cubic},\\
Ql^{\hphantom{2}} & = & \tilde q_{2}Δ,\qquad & \tilde q_{2} & \hspace{-0.2cm} \text{ quadratic},\\
Qc^{\hphantom{2}} & = & \tilde Q_{1}Δ,\qquad & \tilde Q_{1} & \hspace{-0.2cm} \text{ quartic}.
\end{array}
\end{equation}
The condition for the patching matrix to have unit determinant then implies
\begin{equation}\label{eq:BRcoeff2}
Δ^{3} = Q (qc-l^{2}) = Q \tilde q_{1} Δ \quad ⇔ \quad Δ^{2} = Q \tilde q_{1}. 
\end{equation}
Now, as $\tilde q_{1}$ is a quadratic, there are six possibilities for it to be a product of $(z+α)$, $(z-α)$ and $(z-β)$. But $∂_{ψ}=0$ on $(α,β)$, thus $\frac{Q}{Δ}→0$ for $z \downarrow β$. To guarantee this $(z-β)^{2}$ has to divide $Q$, which rules out three of those six possibilities. Furthermore, by Theorem~\ref{thm:invpmatrix} we have 
\begin{equation*}
\frac{\tilde q_{1}}{Δ \vphantom{\skew{7}{\tilde}{A}_{4}}}=\frac{1}{\det \skew{7}{\tilde}{A}_{4}} \quad \text{on } (-∞,-α),
\end{equation*}
where $\skew{7}{\tilde}{A}_{4}$ is obtained from $J$ by cancelling the rows and columns containing inner products with $∂_{ψ}$. But from the general theory we know that the entry of $P$ with the inverse determinant contains a simple pole when approaching the nut, that is $z \uparrow -α$, so that $\tilde q_{1}(-α) ≠ 0$. This immediately yields
\begin{equation*}
\tilde q_{1} = \frac{1}{2} (z-α)^{2} \quad \text{and by \eqref{eq:BRcoeff2} also} \quad Q = 2 (z+α)^{2}(z-β)^{2}.
\end{equation*}
Now observe that there is a factor of $(z-α)$ in $Δ$ but not in $Q$, so that by \eqref{eq:BRcoeff1} the monic $(z-α)$ has to divide $l$, $q$ and $c$. We write this as
\begin{equation*}
l = \frac{L}{π} (z-α), \quad q = -\frac{1}{2} (z-α)\, \tilde l_{1}, \quad c = - (z-α)\, \tilde q_{3},
\end{equation*}
where
\begin{equation*}
\tilde l_{1} = z+A, \quad \tilde q_{3} = z^{2} + Bz + C \qquad \text{for } A, B, C = \text{const.}
\end{equation*}
The first equation in \eqref{eq:BRcoeff1} then turns into
\begin{equation*}
\renewcommand{\arraystretch}{2}
\begin{array}{crcl}
& \tilde l_{1} \tilde q_{3} - \dfrac{2L^{2}}{π^{2}} & = & Δ \\
⇔ & z^{3} + (A+B) z^{2} + (C+AB) z + AC-\dfrac{2L^{2}}{π^{2}} & = & z^{3}-βz^{2}-α^{2}z+α^{2}β.
\end{array}
\end{equation*}
Comparing the coefficients one sees
\begin{equation*}
B = -A-β, \qquad C+AB=-α^{2}, \qquad AC-\frac{2L^{2}}{π^{2}} = α^{2}β,
\end{equation*}
and therefore $A$ satisfies
\begin{equation*}
\renewcommand{\arraystretch}{2}
\begin{array}{crcl}
& \dfrac{1}{A} \left(α^{2}β+\dfrac{2L^{2}}{π^{2}}\right) - A(A+β) & = & - α^{2} \\
⇔ & A^{3} + βA^{2}-α^{2}A-α^{2}β - \dfrac{2L^{2}}{π^{2}} & = & 0.
\end{array}
\end{equation*}
Writing $F(a) \coloneqq a^{3} + βa^{2}-α^{2}a-α^{2}β - \dfrac{2L^{2}}{π^{2}}$, we see that since $F(0)<0$, this last polynomial has to have at least one (positive) real root (see Figure~\ref{fig:poly}). 
\begin{figure}[htbp]
\begin{center}
     \scalebox{0.9}{\input{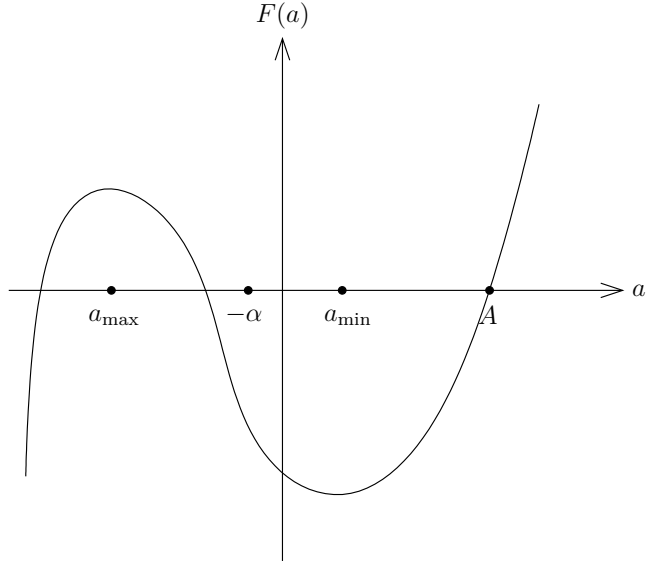}}
     \caption{The cubic $F(a)$.} 
     \label{fig:poly}
\end{center}
\end{figure}
Now from $F'(A)=3A^{2}+2βA-α^{2}$ one concludes that the local maximum of $F$ is at
\begin{equation*}
a_{\mathrm{max}} = -\frac{1}{3} \left(β+\sqrt{β^{2}+3α^{2}}\right). 
\end{equation*}
Furthermore, note that since $α≤β$, it is
\begin{align*}
a_{\mathrm{max}} & ≤ -\frac{1}{3} \left(α+\sqrt{α^{2}+3α^{2}}\right) = -α \quad \text{and} \\
F(-α) & = -\dfrac{2L^{2}}{π^{2}}<0,
\end{align*}
which implies that if $F$ has two more real roots, they will both be smaller than $-α$. On the other hand there is a constraint on $A$ obtained from the asymptotics. In our patching matrix the central entry is
\begin{equation*}
\frac{c}{Δ} = - \frac{(z-α)(z^{2}+Bz+C)}{Δ} = - 1 + (α-β-B) z^{-1} + …
\end{equation*}
Using \eqref{eq:asymptPtop} and the relation between $A$ and $B,$ this gives
\begin{equation*}
A+α = \frac{4M}{3π}.
\end{equation*}
Positivity of $M$ thus implies $A>-α$ and we therefore have shown that there is a unique positive $A∈ℝ$ which satisfies all the constraints.

Consequently, by our ansatz we are able to fix all the parameters in terms of $α$, $β$, $L$, that is in terms of the given data, and the patching matrix is 
\begin{equation*} 
\renewcommand\arraystretch{2.5}
P_{1}=\left(\begin{array}{ccc}-\dfrac{z+A}{2(z+α)(z-β)} & \dfrac{L}{π(z+α)(z-β)} & 0 \\
\cdot  & -\dfrac{z^{2}-\skew{2}{\tilde}{γ} z+\skew{3}{\tilde}{δ}}{(z+α)(z-β)} & 0 \\
0 & 0 & \dfrac{2(z+α)(z-β)}{z-α}\end{array}\right),
\end{equation*}
where
\begin{equation*}
\skew{2}{\tilde}{γ} = β+A, \quad \skew{3}{\tilde}{δ}=-α^{2}+βA +A^{2}.
\end{equation*}
Note that $λ$ and $A$ are zeros of the same polynomial and are restricted by the same inequality involving the mass, so this is the patching matrix for the black ring with the conical singularity not fixed, see \eqref{eq:BRpatmat}.\footnote{I am grateful to Harvey Reall for suggesting this possibility, which turns out to be correct.} For the regular black ring the angular momentum can also be expressed in terms of $α$ and $β$. We will see in Section~\ref{sec:applbr} how this can be done in the twistor picture. \\ \mbox{}\hfill $\blacksquare$
\end{ex}

\section{Local Behaviour of $J$ around a Nut} \label{sec:Jaroundnut}
For the general case of a rod structure with three nuts and especially as the number of nuts gets higher, one will find it increasingly difficult to reduce the number of free parameters to a minimum and would therefore like to obtain more constraints from the inner rods. With this desire in mind it would be useful to have an understanding of how the patching matrices with adaptations to adjacent rods are related to each other. We have seen an example in Theorem~\ref{thm:invpmatrix}, which can be considered as such a switch at the nut at infinity. The proof gives an idea of what is happening when changing the adaptation, yet it will be more difficult for interior nuts, that is nuts for which $|a_{i}|$ is finite. 

A strategy of how to achieve this is describred in \citet[Ch.~3]{Fletcher:1990aa}. There the essence is that ``... redefining the sphere $S_{0}$ and $S_{1}$ by interchanging double points alters the part of the real axis to which the bundle is adapted.'' \cite[Sec.~3.2]{Fletcher:1990aa}. However, as the example in \cite[Sec.~5.1]{Fletcher:1990aa} shows, this comes down to a Riemann-Hilbert problem which will be rather hard and impractical to solve in five or even higher dimensions. Thus we will approach this task in a different way. The idea is that we start off as above on the outermost rods where $|z| \to \infty $, determine as many free parameters as possible by the constraints which we have got on these rods, then take the resulting $P$-matrix (still having free parameters in it which we would like to pin down), calculate its adaptation to the next neighbouring rod and apply analogous constraints there. But before looking at the patching matrix itself let us first study how $J$ behaves locally around a nut.

Consider first a nut where two spatial rods meet, that is like in Figure~\ref{fig:sprods}.\vspace{0.3cm}
\begin{figure}[htbp]
\begin{center}
     \scalebox{0.8}{\input{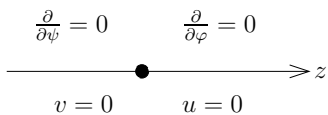}}
     \caption{Two spatial rods with their rod vectors meeting at a nut.} 
     \label{fig:sprods}
\end{center}
\end{figure}

Without loss of generality assume that the nut is at $z=0$. In this case a suitable choice of coordinates are the $(u,v)$-coordinates defined as 
\begin{equation*}
r=uv, \  z=\frac{1}{2}(v^{2}-u^{2}) \quad \Leftrightarrow  \quad u^{2} = -z\pm \sqrt{r^{2}+z^{2}},\ v^{2} = z\pm \sqrt{r^{2}+z^{2}},
\end{equation*}
where the signs on the right-hand side are either both plus or both minus. If we choose both signs to be plus, then the rod $∂_{\varphi}=0$ corresponds to $u=0$ and $∂_{\psi}$ to $v=0$. The metric in the most general case has the form
\begin{equation} \label{eq:GCmetric}
\begin{split}
\drm s^{2} & = X \,\drm t^{2} + 2Y \,\drm t \drm {\varphi} + 2Z \,\drm t \drm {\psi} + U \,\drm {\varphi}^{2} + 2V \,\drm {\varphi} \drm {\psi} + W \,\drm {\psi}^{2}\\
& \hphantom{=} + {\erm}^{2ν} (u^{2}+v^{2})(\drm u^{2}+\,\drm v^{2}),
\end{split}
\end{equation}
or equivalently
\begin{equation*}
\renewcommand{\arraystretch}{1.4}
J (u,v)=\left(\begin{array}{ccc}X & Y & Z \\\cdot  & U & V \\\cdot  & \cdot  & W\end{array}\right).
\end{equation*}
We assume that $\phi,\psi$ have period $2\pi$.

\begin{thm}\label{thm:Jaroundnut}
For a space-time regular on the axis the generic form of $J$ in $(u,v)$-coordinates around a nut, where two spacelike rods meet, is
\begin{equation} \label{eq:Jaroundnut}
\renewcommand{\arraystretch}{1.4}
J= \left(\begin{array}{ccr}
X_{0} & u^{2}Y_{0} & v^{2}Z_{0} \\
\cdot  & u^{2}U_{0} & u^{2}v^{2} V_{0} \\
\cdot  & \cdot  & v^{2}W_{0}\end{array}\right),
\end{equation}
and, furthermore, one needs
\begin{itemize}
\item $\dfrac{U_{0}^{\hphantom{1}}}{v^{2}_{\hphantom{1}}\erm^{2\nu}}=1$ as a function of $v$ on $u=0$,
\item $\dfrac{W_{0}^{\hphantom{1}}}{u^{2}_{\hphantom{1}}\erm^{2\nu}}=1$ as a function of $u$ on $v=0$.
\end{itemize}

If one of the rods is the horizon instead of a spacelike rod corresponding statements hold.
\end{thm}

The second part of the theorem is closely tied to the problem of conicality, which we will investigate shortly.

\begin{proof}
Introduce Cartesian coordinates
\begin{equation} \label{eq:cartaroundnut}
x=u\cos\phi,\quad y=u\sin\phi,\quad z=v\cos\psi,\quad w=v\sin\psi,
\end{equation}
then the metric becomes in these coordinates
\begin{equation} \label{eq:cartesmet}
\begin{split}
\drm s^{2} & = X\, \drm t^2+2\frac{Y}{u^2}\, \drm t(x\, \drm y-y\, \drm x)+2\frac{Z}{v^2}\, \drm t(z\, \drm w-w\, \drm z) +\frac{U}{u^4}(x\, \drm y-y\, \drm x)^2\\
& \hspace{0.4cm}+2\frac{V}{u^2v^2}(x\, \drm y-y\, \drm x)(z\, \drm w-w\, \drm z)+\frac{W}{v^4}(z\, \drm w-w\, \drm z)^2 \\
& \hspace{0.4cm} + \erm^{2\nu}(u^2+v^2)\left(\frac{1}{u^2}(x\,\drm x+y\, \drm y)^2+\frac{1}{v^2}(z\, \drm z+w\, \drm w)^2\right).
\end{split}
\end{equation}
The $x$, $y$, $z$, $w$ are not to be confused with the earlier use of the same symbols. Set $X_{0}=X$. Now as $u→0$ for constant $v$ we immediately see that in order for $g_{ty}$ and $g_{xw}$ to be bounded we need $Y=u^{2}Y_{0}$ and $V=u^{2}V_{1}$ for bounded $Y_{0}$, $V_{1}$.\footnote{Remember that for polar coordinates $(u,\phi)$ the angle $\phi$ is bounded in the limit to but not continuous (or even differentiable) at the origin $u=0$. Thus $\frac{x}{u}$ is not bounded towards the origin.} The remaining singular terms are
\begin{equation*}
\frac{U}{u^4}(x\, \drm y-y\, \drm x)^2+ \erm^{2\nu}(u^2+v^2)\frac{1}{u^2}(x\, \drm x+y\, \drm y)^2.
\end{equation*}
For the fourth-order pole not to be dominant we need $U=u^{2}U_{0}$ for bounded $U_0$; then it is required
\begin{equation} \label{eq:conf1}
\frac{U_{0}}{v^{2}\erm^{2\nu}}=1 \quad \text{as a function of } v \text{ on } u=0
\end{equation}
to remove the remaining second-order pole.

Repeating this for $v→0$ with fixed $u$ yields $Z=v^{2}Z_{0}$, $V_{1}=v^{2}V_{0}$, $W=v^{2}W_{0}$ and 
\begin{equation*}
\frac{W_{0}}{u^{2}\erm^{2\nu}}=1 \quad \text{as a function of } u \text{ on } v=0.
\end{equation*}
This is the minimum that we can demand in terms of regularity of $J$ on the axis and near the nuts.

Assuming now without loss of generality that in Figure~\ref{fig:sprods} the axis segment where $v=0$ is the horizon, we have seen in Chapter~\ref{ch:fivedim} that then the first row and first column degenerate. So, we substitute
\begin{equation*}
z=v\cosh(ωt),\quad w=v\sinh (ωt),
\end{equation*}
where $ω$ is a constant with no further restriction. The coordinates $x$ and $y$ choose as in \eqref{eq:cartaroundnut}. Now the above argument works analogously with all results equivalent, but
\begin{equation*}
\frac{X_{0}}{v^{2}\erm^{2\nu}}=-ω^{2} \quad \text{as a function of } v \text{ on } u=0.
\end{equation*}
\end{proof}

\section{Conicality and the Conformal Factor}
Returning to the case as depicted in Figure~\ref{fig:sprods}, we saw in \eqref{eq:conf1} that regularity at an axis seqment where $\partial_{\phi}$ vanishes forces a relation between $g_{\phi\phi}$ and the conformal factor $\erm^{2\nu}$ of the $(r,z)$-metric. In this section we first establish the following.

\begin{prop} \label{prop:conf1}
On a segment of the axis where $u=0$ we have $\frac{U_0}{v^2\erm^{2\nu}}=\text{constant}$.  
\end{prop}
\begin{proof}
To prove this we need to consider how the conformal factor varies on the axis and this is obtained from \eqref{eq:redeinst3} which we may write as
\begin{equation} \label{eq:conf2}
\partial_{\xi}\left(\log \left(r\erm^{2\nu}\right)\right)=\frac{\irm r}{2} \tr\left(J^{-1}J_{\xi} J^{-1}J_{\xi}\right).
\end{equation}
It will be convenient to work with $\chi=u+\irm v$ where $\xi=z+\irm r=\frac{1}{2}\chi^{2}$ and concentrate on the conformal factor of the $(u,v)$-metric which is $(u^{2}+v^{2})\erm^{2\nu}$ by \eqref{eq:cartesmet}. Then
\begin{equation*} 
\begin{split}
\partial_{\chi}\left(\log\left(\left(u^{2}+v^{2}\right)\erm^{2\nu}\right)\right) & =\partial_{\chi}\left(\log\left(\left(u^{2}+v^{2}\right)(uv)^{-1}r\erm^{2\nu}\right)\right) \\
& =\frac{1}{\chi}-\frac{1}{2u}+\frac{\irm}{2v}+\frac{\irm uv}{2(u+\irm v)}\tr \left(J^{-1}J_{\chi} J^{-1}J_{\chi}\right).
\end{split}
\end{equation*}

Close to the axis segment $u=0$ we substitute from \eqref{eq:Jaroundnut} and expand in powers of $u$ to find
\begin{equation} \label{eq:conf3}
\partial_{\chi}\left(\log\left(\left(u^{2}+v^{2})\erm^{2\nu}\right)\right)\right)=\frac{1}{\chi}-\frac{1}{2u}+\frac{\irm}{2v}+\frac{\irm uv}{2(u+\irm v)}\left(\frac{K_{1}}{u^{2}}+\frac{K_{2}}{u}+\Ocal(1)\right),
\end{equation}
where
\begin{align*}
K_{1} & =\left(U_{0}\left(X_{0}W_{0}-v^{2}Z_{0}^{2}\right)\right)^{2}=1+\frac{1}{v^2}\Ocal\left(u^{2}\right), \\
K_{2} & =\left(U_{0}\left(X_{0}W_{0}-v^{2}Z_{0}^{2}\right)\right)^{2}\frac{\partial_{\chi} U_{0}}{U_{0}}.
\end{align*}
The right hand side of the first equation follows from the determinant
\begin{equation*}
u^{2}v^{2}=\det J = u^{2}v^{2} X_{0}U_{0}W_{0} - u^{2}v^{4}U_{0}Z_{0}^{2} + \Ocal\left(u^{4}\right), 
\end{equation*}
hence
\begin{equation*}
1= U_{0}\left(X_{0}W_{0}-v^{2}Z_{0}^{2}\right)+\frac{1}{v^{2}} \Ocal\left(u^{2}\right).
\end{equation*}
Taking in \eqref{eq:conf3} the limit on to $u=0$ we obtain just
\begin{equation*}
\partial_{v}\left(\log\left(v^{2}\erm^{2\nu}\right)\right)=\partial_{v}\log\left(U_{0}\right),
\end{equation*}
so that
\begin{equation*}
\frac{U_{0}}{v^{2}\erm^{2\nu}}=\text{constant} \quad  \text{on } u=0.
\end{equation*}
\end{proof}

Thus \eqref{eq:conf1} will hold at all points of the axis segment if it holds at one. The following proposition is an analysis similar to \cite[App.~H]{Harmark:2004rm}, but it is simpler and more self-contained to rederive it than translate it.
\begin{prop} \label{prop:conf2}
As a function on the axis $\{u=0\}∪\{v=0\}$, that is as a function of one variable, the factor $\left(u^{2}+v^{2}\right)\erm^{2\nu}$ is continuous at the nut $u=v=0$.
\end{prop}
\begin{proof}
Near the nut introduce polar coordinates
\begin{equation*}
u=R\cos\Theta, \quad v=R\sin\Theta,
\end{equation*}
so that from \eqref{eq:conf2} and with the help of
\begin{equation*}
uv∂_{z}=v∂_{u}-u∂_{v}, \quad uv∂_{r}=u∂_{u}+v∂_{v}, \quad ∂_{ξ}=\frac{1}{2}\left(∂_{z}-\irm ∂_{r}\right)
\end{equation*}
we get
\begin{align*}
\partial_{\Theta}\left(\log\left(\left(u^{2}+v^{2}\right)\erm^{2\nu}\right)\right) & =\left(u\partial_{v}-v\partial_{u}\right)\left(\log\left(\left(u^{2}+v^{2}\right)\erm^{2\nu}\right)\right) \\
& = \left(u\partial_{v}-v\partial_{u}\right) \left(\log \left(\frac{u^{2}+v^{2}}{uv}\right)\right)-r∂_{z} \left(\log\left(r\erm^{2ν}\right)\right) \\
& = -\frac{u}{v}+\frac{v}{u}-\frac{(uv)^{2}}{4}\tr\left(J^{-1}J_{z}J^{-1}J_{r}\right) \\
& =-\frac{u}{v}+\frac{v}{u}-\frac{uv}{4}\tr\left(J^{-1}J_{u} J^{-1}J_{u}-J^{-1}J_{v} J^{-1}J_{v}\right).
\end{align*}
Again we expand this using \eqref{eq:Jaroundnut} to find
\begin{align*}
\partial_{\Theta}\left(\log\left(\left(u^{2}+v^{2}\right)\erm^{2\nu}\right)\right) & =-\frac{u}{v}+\frac{v}{u}-\frac{v}{u}\left(U_{0}\left(X_{0}W_{0}-v^{2}Z_{0}^{2}\right)\right)^{2}\\
& \hspace{0.4cm}+\frac{u}{v}\left(W_{0}\left(X_{0}U_{0}-u^{2}Y_{0}^{2}\right)\right)^2+\Ocal(u)+\Ocal(v) \\
& = \Ocal(u)+\Ocal(v)=\Ocal(R).
\end{align*}
Now the jump in $\log\left(u^{2}+v^{2}\right)\erm^{2\nu}$ round the nut is 
\begin{equation*}
\Delta\left(\log\left(\left(u^{2}+v^{2}\right)\erm^{2\nu}\right)\right)=\lim_{R\rightarrow 0}\int_{0}^{\frac{π}{2}}\partial_{\Theta}\left(\log\left(\left(u^{2}+v^{2}\right)\erm^{2\nu}\right)\right)\drm \Theta=0,
\end{equation*}
and $\left(u^{2}+v^{2}\right)\erm^{2\nu}$ does not jump either.

On $u=0$, $U_{0}$ is continuous and by Proposition~\ref{prop:conf1} $\frac{U_0}{v^2\erm^{2\nu}}=\text{constant}$, so $v^2\erm^{2\nu}$ must be bounded there. Similarly, on $v=0$ for $W_{0}$ and $\frac{W_{0}}{u^{2}\erm^{2\nu}}$. Thus, $\left(u^{2}+v^{2}\right)\erm^{2\nu}$ is continuous on the  two rods and has no jump across the nut, so it is continuous on the axis.
\end{proof}

The strategy for removing conical sigularities is now clear: We start by assuming that $\phi$ and $\psi$ both have period $2\pi$. On the part of the axis extending to $z=+\infty$, where the Killing vector $\partial_{\phi}$ vanishes, we have $ \frac{U_{0}}{v^{2}\erm^{2\nu}}=\text{constant}$ by Proposition~\ref{prop:conf1} and the asymptotic conditions we are imposing make this constant one. The corresponding statement holds on the part of the axis extending to $z=-\infty$ for the same reason. When passing by a nut between two spacelike rods we may suppose, by choosing the basis of Killing vectors appropriately, that $\partial_{\phi}$ vanishes above the nut and $\partial_{\psi}$ below and we know by Proposition~\ref{prop:conf2} that $\left(u^2+v^2\right)e^{2\nu}$ is continuous at the nut. If there is no conical singularity above the nut we have $ \frac{U_{0}}{v^{2}\erm^{2\nu}}=1$ there and we want $ \frac{W_{0}}{u^{2}\erm^{2\nu}}=1$ below the nut. Therefore we require the limits of $U_{0}$ from above and $W_{0}$ from below to be equal.
\begin{cor} \label{cor:conic}
With the conventions leading to \eqref{eq:Jaroundnut}, the absence of conical singularities requires
\begin{equation*}
\lim_{v\rightarrow 0}U_{0}=\lim_{u\rightarrow 0}W_0.
\end{equation*}
\end{cor}
This is what we have just shown. At a nut where one rod is the horizon we do not obtain further conditions as we have no reason to favour a particular value of $ω$. To see how this is applied to the case of the black ring we need a better understanding of going past a nut.

\section{Local Behaviour of $P$ around a Nut: Switching}

In this section we establish a prescription for obtaining the matrix $P_{-}$ adapted to the segment of the axis below a nut from the matrix $P_{+}$ adapted to the segment above. We call this process `switching'. Once we have the prescription we can impose the condition of non-conicality found in Corollary~\ref{cor:conic}. We then apply this to the black ring, but it is clear that with this prescription we have an algorithm for working systematically down the axis given any rod structure so that we obtain all the matrices $P_{i}$ adapted to the different rods labelled by $i$. The result is the following.
\begin{thm} \label{thm:switching}
Let at $z=a$ be a nut where two spacelike rods meet, as in Figure~\ref{fig:sprods}, and assume that we have chosen a gauge where the twist potentials vanish when approaching the nut. Then 
\begin{equation*}
\renewcommand{\arraystretch}{1.5}
P_{-}^{\vphantom{\frac{1}{2}}}=\left(\begin{array}{ccc}0 & 0 & \dfrac{1}{2(z-a)} \\0 & 1 & 0 \\ 2(z-a) & 0 & 0\end{array}\right)P_{+}^{\vphantom{\frac{1}{2}}}\left(\begin{array}{ccc}0 & 0 & 2(z-a) \\0 & 1 & 0 \\\dfrac{1}{2(z-a)} & 0 & 0\end{array}\right),
\end{equation*}
where $P_{+}$ is adapted to $u=0$ and $P_{-}$ is adapted to $v=0$.
\end{thm}

We begin by motivating this prescription from a consideration of \eqref{eq:GCmetric}. First calculate the twist potentials in the same way as in Chapter~\ref{ch:Pexamples}. The metric \eqref{eq:GCmetric} can be rearranged in orthonormal form
\begin{equation*}
\begin{split}
\drm s^{2} & = X (\drm t +{\omega}_{1}\, \drm {\varphi} +{\omega}_{2}\,\drm {\psi})^{2} + \skew{2}{\tilde} U (\drm {\varphi}+ {\Omega} \, \drm {\psi})^{2} \\
& \hphantom{=} + \tilde W\, \drm ψ^{2} - \erm^{2{\nu}} (\drm r^{2}+\drm z^{2}).
\end{split}
\end{equation*}
The orthonormal frame is again
\begin{equation*}
\renewcommand{\arraystretch}{2}
\begin{array}{ll}{\theta}^{0} = X^{\frac{1}{2}}(\drm t + {\omega}_{1}\, \drm {\varphi} + {\omega}_{2} \, \drm {\psi}), & {\theta}^{1} =\skew{2}{\tilde} U ^{\frac{1}{2}} (\drm {\varphi}+{\Omega} \, \drm {\psi}), \\{\theta}^{2} = \tilde W^{\frac{1}{2}} \, \drm {\psi}, \hspace{0.6cm} {\theta}^{3} = \erm^{{\nu}} \, \drm r, & {\theta}^{4} = \erm^{{\nu}} \, \drm z,\end{array}
\end{equation*}
so
\begin{equation*}
\renewcommand{\arraystretch}{2}
\begin{array}{lr}X{\omega}_{1} = Y, \qquad X{\omega}_{2} = Z &  X{\omega}_{1}{\omega}_{2}+\skew{2}{\tilde}U {\Omega} = V, \\ \skew{2}{\tilde}U +X {\omega}_{1}^{2} = U, & \tilde W + \skew{2}{\tilde}U {\Omega}^{2}+X{\omega}_{2}^{2}= W.\end{array}
\end{equation*}
Adapted to $\partial_{\psi}=0$, then for small $r$ it is $Z,V,W\in \Ocal(r^{2})$, hence ${\omega}_{2}, {\Omega}, \tilde W \in \Ocal(r^{2})$,\footnote{In order to see that ${\Omega}\in \Ocal(r^{2})$, derive from $X\in \Ocal(1)$ and $\skew{2}{\tilde}U X = UX-Y^{2}$ that $\skew{2}{\tilde}U\in \Ocal(1)$.} and the other terms $\Ocal(1)$. This implies $\frac{\tilde W}{W} \to  1$ as $r→0$. 

Now the 1-forms are
\begin{equation*}
\partial _{t} \to  T = X^{\frac{1}{2}}{\theta}^{0}, \quad \partial _{{\varphi}} \to  {\Phi} = {\omega}_{1}X^{\frac{1}{2}}{\theta}^{0}+\skew{2}{\tilde}U^{\frac{1}{2}}{\theta}^{1},
\end{equation*}
hence
\begin{align*}
\drm {\chi}_{1} & = * (T\wedge {\Phi}\wedge \drm T) = *(X^{\frac{1}{2}}{\theta}^{0}\wedge \skew{2}{\tilde}U^{\frac{1}{2}}{\theta}^{1}\wedge X\, \drm {\omega}_{2}\wedge \drm {\psi}) \\
\drm {\chi}_{2} & = *(T\wedge {\Phi}\wedge  \drm {\Phi}) = *(X^{\frac{1}{2}}{\theta}^{0}\wedge \skew{2}{\tilde}U^{\frac{1}{2}}{\theta}^{1}\wedge ({\omega}_{1}X\, \drm {\omega}_{2}+\skew{2}{\tilde}U\, \drm {\Omega})\wedge \drm {\psi}),
\end{align*}
which with $*({\theta}^{0}\wedge {\theta}^{1}\wedge {\theta}^{2}\wedge {\theta}^{3})={\epsilon}{\theta}^{4}$ leads to
\begin{align*}
\partial _{z} {\chi}_{1} & = {\epsilon} (X\skew{2}{\tilde}U)^{\frac{1}{2}}X \lim_{r\to 0}\left(\frac{\partial _{r}{\omega}_{2}}{\tilde W ^{\scriptscriptstyle{1/2}}_{\vphantom{0}}}\right) \\
\partial _{z} {\chi}_{2} & = {\epsilon} (X\skew{2}{\tilde}U)^{\frac{1}{2}} \lim_{r\to 0}\left(\frac{X{\omega}_{1}\partial _{r}{\omega}_{2}+\skew{2}{\tilde}U\partial _{r}{\Omega}}{W^{\scriptscriptstyle{1/2}}_{\vphantom{0}}}\right).
\end{align*}
We switch again to $(u,v)$-coordinates by using
\begin{equation*}
\frac{\partial }{\partial u} = u\frac{\partial }{\partial z} + v \frac{\partial }{\partial r}, \quad \frac{\partial }{\partial v} = -v\frac{\partial }{\partial z} + u \frac{\partial }{\partial r}.
\end{equation*}
On $v=0$ this yields
\begin{equation*}
\frac{\partial {\chi}_{1}}{\partial u} = u \frac{\partial {\psi}_{1}}{\partial z} = u {\epsilon} (X\skew{2}{\tilde}U)^{\frac{1}{2}} X \lim_{v\to 0} \left(\frac{1}{u\tilde W^{\scriptscriptstyle{1/2}}_{\vphantom{0}}}\frac{\partial {\omega}_{2}}{\partial v}\right).
\end{equation*}
Now use\footnote{NB: If we did not have $Z=v^{2}Z_{0}$ already, then we could impose it here without loss of generality, since a smooth metric invokes smooth twist 1-forms, but ${\chi}_{1}$ becomes singular for $v\to 0$ if only $Z=vZ_{0}$.}
\begin{equation*}
{\omega}_{2} = \frac{Z}{X} = \frac{v^{2}Z_{0}}{X_{0}} \quad \text{and} \quad \tilde W = v^{2}W_{0} + \Ocal(v^{4}), 
\end{equation*}
to obtain
\begin{align*}
\frac{\partial {\chi}_{1}}{\partial u} & = {\epsilon} \left(X_{0} \skew{2}{\tilde}U\right)^{\frac{1}{2}} \frac{2vZ_{0}}{vW_{0}^{\scriptscriptstyle{1/2}}} = 2{\epsilon} \frac{\left(X_{0}\skew{2}{\tilde}U\right)^{\frac{1}{2}}Z_{0}}{W_{0}^{\scriptscriptstyle{1/2}}}  = 2{\epsilon} \frac{Z_{0}^{\hphantom{\scriptscriptstyle{1/2}}}}{W_{0}^{\scriptscriptstyle{1/2}}} \left(u^{2}U_{0}^{\vphantom{2}}X_{0}^{\vphantom{2}}-u^{4}Y_{0}^{2}\right)^{\frac{1}{2}}\\
& = 2{\epsilon}Z_{0}\left(\frac{U_{0}X_{0}}{W_{0}}\right)^{\frac{1}{2}}u+\Ocal(u^{2})\\
\Rightarrow  {\chi}_{1}^{\vphantom{1}} & = {\chi}_{1}^{0} + {\chi}_{1}^{1} u^{2} + \text{h.o.} \vphantom{\left(\frac{U_{0}X_{0}}{W_{0}}\right)^{\frac{1}{2}}}
\end{align*}
Analogous steps lead to 
\begin{align*}
\frac{\partial {\chi}_{2}}{\partial u} & = u \frac{\partial {\chi}_{2}}{\partial z} = u{\epsilon} (X\skew{2}{\tilde}U)^{\frac{1}{2}} \lim_{v\to 0} \left(\frac{u^{2}Y_{0}\frac{1}{u}\frac{\partial {\omega}_{2}}{\partial v}+\left(U-\frac{Y^{2}}{X}\right)\frac{1}{u}\partial _{v}{\Omega}}{vW_{0}^{\scriptscriptstyle{1/2}}}\right) \\
& = 2u^{2}(X\skew{2}{\tilde}U)^{\frac{1}{2}} \left(\frac{Y_{0}Z_{0}}{X_{0}W_{0}^{\scriptscriptstyle{1/2}}}+\frac{\left(U_{0}-\frac{Y_{0}^{2}}{X_{0}}\right)\frac{1}{u}{\Omega}_{0,v}}{vW_{0}^{\scriptscriptstyle{1/2}}}\right) \\
& = 2{\epsilon} V_{0} \left(\frac{U_{0}X_{0}}{W_{0}}\right)^{\frac{1}{2}}u^{3}+ \text{h.o.} \\
\Rightarrow  {\chi}_{2}^{\vphantom{1}} & = {\chi}_{2}^{0} + {\chi}_{2}^{1} u^{4} + \text{h.o.} \vphantom{\left(\frac{U_{0}X_{0}}{W_{0}}\right)^{\frac{1}{2}}}
\end{align*}
For $u=0$ we only have to swap $Y\leftrightarrow Z$, $U\leftrightarrow W$. With $u^{2} \sim 2z$ the above can be summarized as 
\begin{equation*}
\renewcommand{\arraystretch}{2}
P_{-}^{\vphantom{\frac{1}{2}}}(r=0,z) =\left(\begin{array}{ccr}\dfrac{g_{0}}{2z} + \Ocal(1) & -g_{0}^{\vphantom{1}} {\chi}^{1}_{1}+\Ocal(z) & -g_{0}^{\vphantom{1}} {\chi}^{1}_{2}z+\Ocal(z^{2}) \\
\cdot  & \hspace{0.5cm}X_{0}+\Ocal(z) & 2zY_{0}+\Ocal(z^{3}) \\
\cdot  & \cdot  & 2zU_{0}+\Ocal(z^{4})\end{array}\right),
\end{equation*}
where $g_{0}^{\vphantom{2}}=(X_{0}^{\vphantom{2}}U_{0}^{\vphantom{2}}-2zY_{0}^{2})^{-1}$. Note that here we dropped without loss of generality the constant terms of the twist potentials ${\chi}_{i}^{0}$. This can be done just by a gauge transformation to $P$ of the form $P\to APB$ with constant matrices $A$ and $B$, namely
\begin{equation*}
\renewcommand{\arraystretch}{1.4}
P \to  \left(\begin{array}{rrr}1 & \hphantom{-}0 & \hphantom{-}0 \\-c_{1} & 1 & 0 \\-c_{2} & 0 & 1\end{array}\right)\, P\, \left(\begin{array}{ccc}1 & -c_{1} & -c_{2} \\ 0 & \hphantom{-}1 & \hphantom{-}0 \\0 & \hphantom{-}0 & \hphantom{-}1\end{array}\right).
\end{equation*}
For $P$ in standard form this results in ${\chi}_{i}\to {\chi}_{i}+c_{i}$. Removing the constant term in the twist potentials allows us to assume that without loss of generality the entries which become zero or blow up towards a nut are only on the diagonal. The off-diagonal entries are bounded towards the nut.

Without loss of generality assume that the nut is at $a=0$. Then the calculations above show that to leading order in $z$ the patching matrices below and above the nut are\footnote{Chosen the right orientation for the basis such that the signs which are recorded by $\epsilon$ work out.}
\begin{equation} \label{eq:Pplus}
\renewcommand{\arraystretch}{1.5}
P_{-}^{\vphantom{\frac{1}{2}}}=
\left(\begin{array}{ccc}
\hphantom{-}\dfrac{1}{2zX_{0}U_{0}} & -\dfrac{Z_{0}}{(U_{0}X_{0}W_{0})^{\scriptscriptstyle{1/2}}_{\vphantom{0}}} & -\dfrac{V_{0} z}{(U_{0}X_{0}W_{0})^{\scriptscriptstyle{1/2}}_{\vphantom{0}}} \\
\cdot & X_{0} & 2zY_{0} \\
\cdot & \cdot & 2zU_{0}
\end{array} \right),
\end{equation}
\begin{equation} \label{eq:Pminus}
\renewcommand{\arraystretch}{1.5}
P_{+}^{\vphantom{\frac{1}{2}}}=
\left(\begin{array}{ccc}
-\dfrac{1}{2zX_{0}W_{0}} & \hphantom{-}\dfrac{Y_{0}}{(U_{0}X_{0}W_{0})^{\scriptscriptstyle{1/2}}_{\vphantom{0}}} & -\dfrac{V_{0} z}{(U_{0}X_{0}W_{0})^{\scriptscriptstyle{1/2}}_{\vphantom{0}}} \\
\cdot & X_{0} & -2zZ_{0} \\
\cdot & \cdot & -2zW_{0}
\end{array} \right).
\end{equation}
Using that $\det J = -u^{2}v^{2}$ in \eqref{eq:Jaroundnut} and thus $X_{0}U_{0}W_{0}=-1$ to leading order in $z$ we see that the switching is correct to leading order in $z$. (This is consistent with the different adaptations we calculated for example for the Schwarzschild space-time or flat space, see Section~\ref{sec:PMP}.)

\begin{proof}[Proof of Theorem~\ref{thm:switching}]
To prove Theorem~\ref{thm:switching} the strategy is to follow the splitting procedure outlined in Section~\ref{sec:Pnearaxis}.

We first observe that splitting $P_{+}$ as in \eqref{eq:Pplus} will lead not to $J(r,z)$ as desired, but to $J(r,z)$ with its rows and columns permuted. This can be seen by looking at the diagonal case. To obtain $J(r,z)$ with the rows and columns in the order $(t,\phi,\psi)$ we need to permute
\begin{equation*}
P_{+}\rightarrow \widetilde{P_{+}}=E_{1}P_{+}E_{1} \quad \text{with }
E_{1} = \left(\begin{array}{ccc}
0 & 1 & 0 \\
1 & 0 & 0 \\
0 & 0 & 1
\end{array} \right).
\end{equation*}
Similarly for $P_{-}$ by \eqref{eq:Pminus}, we permute
\begin{equation*}
 P_{-}\rightarrow\widetilde{P_{-}}=E_{2}P_{-}E_{2}^{\mathrm{t}} \quad \text{with }
E_{2} = \left(\begin{array}{ccc}
0 & 1 & 0 \\
0 & 0 & 1 \\
1 & 0 & 0
\end{array} \right).
\end{equation*}
 Note that now the prescription in Theorem~\ref{thm:switching} translates to
\begin{equation} \label{eq:tPplus}
\renewcommand{\arraystretch}{1.5}
\widetilde{P_{-}}=D\widetilde{P_{+}}D \quad \text{with }
D = \left(\begin{array}{ccc}
1 & 0 & 0 \\
0 & 2z & 0 \\
0 & 0 & \dfrac{1}{2z}
\end{array} \right).
\end{equation}
Recall that we have set $a=0$. Following Chapter~\ref{ch:fivedim}, to obtain $J$ we split the matrices 
\begin{equation} \label{eq:Pswitch1}
\renewcommand{\arraystretch}{1.5}
\begin{split}
\widehat{P_{+}} = 
\left(\begin{array}{ccc}
1 & 0 & 0 \\
0 & \dfrac{r}{ζ} & 0 \\
0 & 0 & 1
\end{array} \right)
\widetilde{P_{+}}
\left(\begin{array}{ccc}
1 & 0 & 0 \\
0 & -rζ & 0 \\
0 & 0 & 1
\end{array} \right),\\
\renewcommand{\arraystretch}{1.5}
\widehat{P_{-}} = 
\left(\begin{array}{ccc}
1 & 0 & 0 \\
0 & 1 & 0 \\
0 & 0 & \dfrac{r}{ζ}
\end{array} \right)
\widetilde{P_{-}}
\left(\begin{array}{ccl}
1 & 0 & \hphantom{l}0 \\
0 & 1 & \hphantom{l}0 \\
0 & 0 & -rζ
\end{array} \right).
\end{split}
\end{equation}
The location of the diagonal entries which are not one is dictated by the position of the Killing vector which vanishes on the section of axis under consideration within the basis of Killing vectors $(\partial_{t},\partial_{\phi},\partial_{\psi})$ . In the language of Section~\ref{sec:Pnearaxis}, the integers $(p_{0},p_{1},p_{2})$ are, as we know, a permutation of $(0,0,1)$ and the location of the 1 is determined by the prescription just given.

Assembling \eqref{eq:tPplus} and \eqref{eq:Pswitch1} to
\begin{equation*}
\widehat{P_{-}}=A\widehat{P_{+}}B,
\end{equation*}
where 
\begin{equation*}
\renewcommand{\arraystretch}{1.5}
A = 
\left(\begin{array}{ccc}
1 & 0 & 0 \\
0 & \dfrac{2zζ}{r} & 0 \\
0 & 0 & \dfrac{r}{2zζ}
\end{array} \right), 
\quad
B = 
\left(\begin{array}{ccc}
1 & 0 & 0 \\
0 & -\dfrac{2z}{rζ} & 0 \\
0 & 0 & -\dfrac{rζ}{2z}
\end{array} \right),
\end{equation*}
all that is needed for completing the proof is to show that splitting the left and right hand side of this last equation yield the same $J$-matrix. To perform the splitting, we replace all appearances of $z$ by $w$ and make the substitution \eqref{eq:quadw2}. Note that
\begin{equation*}
w=z+\frac{r}{2}\left(\zeta^{-1}-\zeta\right)=\frac{1}{2}\left(u^{2}-v^{2}+uv\left(\zeta^{-1}-\zeta\right)\right)=\frac{1}{2\zeta}\left(u\zeta+v\right)\left(u-v\zeta\right),
\end{equation*}
so that
\begin{align*}
\frac{2w\zeta}{r} & = \frac{(u\zeta+v)(u-v\zeta)}{uv}=1+O\left(\zeta\right), \\
-\frac{2w}{r\zeta} & = -\frac{1}{\zeta^2}\frac{(u\zeta+v)(u-v\zeta)}{uv}=1+O\left(\zeta^{-1}\right).
\end{align*}
Thus $A(z,r,\zeta)$ is holomorphic and nonsingular in the neighbourhood of $\zeta=0$ with $A(z,r,0)=\id$, and $B(z,r,\zeta^{-1})$ is holomorphic and nonsingular in the neighbourhood of $\zeta^{-1}=0$ with $B(z,r,0)=\id$. Consequently, if $\widehat{P_{+}^{\vphantom{\infty}}}$ splits as
\begin{equation*}
\widehat{P_{+}^{\vphantom{\infty}}}=K_{+}^{0}\left(r,z,\zeta\right)\left(K_{+}^{\infty}\left(r,z,\zeta^{-1}\right)\right)^{-1},
\end{equation*}
with $K_{+}^{0}$ holomorphic and nonsingular in the neighbourhood of $\zeta=0$ and $K_{+}^{\infty}$ holomorphic and nonsingular in the neighbourhood of $\zeta^{-1}=0$, then a splitting of $\widehat{P_{-}^{\vphantom{\infty}}}$ 
is given by taking
\begin{equation*}
\widehat{P_{-}^{\vphantom{\infty}}}=K_{-}^{0}\left(K_{-}^{\infty}\right)^{-1} \quad \text{with } K_{-}^{0}=A K_{+}^{0},\ K_{-}^{\infty}=B^{-1}K_{+}^{\infty}.
\end{equation*}
The corresponding expressions for $J$ are
\begin{equation*}
J=J_{+}^{\vphantom{\infty}}(r,z)=K_{+}^{0}(0)\left(K_+^{\infty}(0)\right)^{-1}
\end{equation*}
and
\begin{equation*}
J=J_{-}^{\vphantom{\infty}}(r,z)=K_{-}^{0}(0)\left(K_{-}^{\infty}(0)\right)^{-1}=A(r,z,0) J_{+}^{\vphantom{\infty}}(r,z) B(r,z,0)=J_{+}^{\vphantom{\infty}}(r,z).
\end{equation*}
These are the same.
\end{proof}

\section{Application to the Black Ring} \label{sec:applbr}
Now we see how to apply the prescription for switching and the discussion of conicality to $P(z)$ for the black ring as in \eqref{eq:BRpatmat}. We are interested in the nut with largest $z$-value, which is the one at $z=\beta$. The first step is to make an additive shift to the twist potential $\chi$ to ensure that the term $P_{12}$ in \eqref{eq:BRpatmat} is finite at $z=\beta$. This needs
\begin{equation*}
\chi\rightarrow\chi+C, \quad C=-\frac{2\nu}{\beta+\lambda},
\end{equation*}
when
\begin{equation*}
P_{12}\rightarrow P_{12}-CP_{11}=P_{12}-\frac{\nu(z+\lambda)}{(\beta+\lambda)(z+\alpha)(z-\beta)}=\frac{\nu}{(z+\alpha)(\beta+\lambda)},
\end{equation*}
which is indeed finite at $z=\beta$, and
\begin{equation*}
P_{22}\rightarrow P_{22}-2CP_{12}+C^2P_{11}=-\frac{(z+\mu)}{(z+\alpha)}, \quad \text{where }\mu= \frac{\kappa^2(2b-c+bc)}{(1+b)},
\end{equation*}
which is also finite at $z=\beta$. We are in position to make the switch as in Theorem~\ref{thm:switching} with $\beta$ in place of $a$ and the result is
\begin{equation*}
\renewcommand{\arraystretch}{2.5}
P_{2}=\left(\begin{array}{ccc}
\dfrac{(z+\alpha)}{2(z-\alpha)(z-\beta)} & 0 & 0 \\
\cdot & -\dfrac{(z+\mu)}{(z+\alpha)} & \dfrac{2\nu(z-\beta)}{\gamma(z+\alpha)} \\
\cdot & \cdot & -2\dfrac{(z+\lambda)(z-\beta)}{(z+\alpha)}
\end{array} \right).
\end{equation*}
We have completed the switching and obtained $P_2(z)$, the transition matrix adapted to the section of axis $\alpha<z<\beta$. We could continue to find the transition matrix adapted to the other segments but that is straightforward and we do not need it. Instead we shall return to the question of conicality addressed in Corollary~\ref{cor:conic}. Compare with Theorem~\ref{thm:Jaroundnut} to find from $P_1$ that
\begin{equation*}
v^{2}W_{0}= \frac{2(z+\alpha)(z-\beta)}{(z-\alpha)}
\end{equation*}
where now $v^{2}=-2(z-\beta)$ and from $P_2$ that
\begin{equation*}
u^{2}U_{0}=-\frac{2(z+\lambda)(z-\beta)}{(z+\alpha)}
\end{equation*}
where now $u^{2}=2(z-\beta)$. Corollary~\ref{cor:conic} implies that there is no conical singularity on the axis section $\alpha<z<\beta$ provided
\begin{equation*}
\lim_{u\rightarrow 0}W_{0}=\lim_{v\rightarrow 0}U_{0},
\end{equation*}
which here requires
\begin{equation*}
\frac{\beta+\lambda}{\beta+\alpha}=\frac{\beta+\alpha}{\beta-\alpha}.
\end{equation*}
Using \eqref{eq:BRpar} this condition can be solved for $b$ as
\begin{equation*}
b=\frac{2c}{1+c^2}
\end{equation*}
which is known \cite{Emparan:2002aa} (or \cite[Eq.~(6.20)]{Harmark:2004rm}) to be the right condition. 
\chapter{Summary and Outlook}

In this work we have presented a possible way for the reconstruction of five- or higher-dimensional black hole space-times from what are at the moment believed to be the classifying parameters, rod structure and angular momenta. The method is based on a twistor construction which in turn relies on the Penrose-Ward transform. 

Our idea assigns a patching matrix to every rod structure where, apart from the possible poles at the nuts, the entries of the patching matrix have to be rational functions with the same denominator $Δ$ --- Chapter~\ref{ch:converse}. By imposing boundary conditions the aim is to determine all the coefficients of the polynomials in the numerator of these rational functions in terms of the nuts and angular momenta.

However, with an increasing number of nuts one needs increasingly sophisticated tools and it is of particular importance to gain a detailed understanding of how the patching matrices, adapted to two neighbouring rods, are related. In Theorem~\ref{thm:switching} we show how to do this and Theorem~\ref{thm:invpmatrix} provides this statement for the nut at infinity, that is, it relates the patching matrices which are adapted to the outer rods. By means of that we are able to reconstruct the patching matrix for a general two-nut rod structure (up to one restriction on the parameters for the physical Myers-Perry solution) and we can show that a three-nut rod structure with one Killing vector hypersurface-orthogonal fixes, together with a given angular momentum, the space-time to be the black ring. 

Also in Chapter~\ref{ch:converse} we discuss conical singularities on the axis and show how to obtain necessary and sufficient conditions for their removal. Applying this to the black ring we obtain the known relation between the parameters. In particular, this implies a relation between the rod structure and the asymptotic quantities for a non-singular solution known to exist.\\

\noindent Further questions which are interesting to pursue in this context are for example: 

Which rod structures are admissible? In other words, are there any restrictions on the rod structures which one allows in the set of parameters?

Can we construct a Lens space-time this way, that is a space-time whose horizon is connected and has the topology of a Lens space \cite[Prop.~2]{Hollands:2008fp}? We know what the corresponding rod structure looks like, but are we able to fix enough parameters and can we see whether the resulting patching matrix does give rise to a space-time without singularities? The latter question seems to be difficult to address as by the analytic continuation one can guarantee the existence of the solution with all its nice regularity properties only in a neighbourhood of the axis, but further away from the axis there might be so-called ``jumping lines'', where the mentioned triviality assumption of the bundle does not hold. 

How many dimensions does the moduli space for an $n$\,-nut rod structure have? Can we find upper and lower bounds on that depending on the imposed boundary conditions? This also does not seem to be an easy questions as most of the conditions we impose on the patching matrix are highly non-linear, for example the determinant condition. 

Which parts of the theory extend to higher dimensions? We have already pointed out along the way that some statements straight-forwardly generalize to more than five dimensions as well, but some others do not. A closer look at those points is certainly interesting. 

Also stepping a dimension down leads to a question for which this set of tools might be appropriate. Are we able to disprove the existence of a double-Kerr solution in four dimensions? It is conceivable that for example the imposed compatibility requirements as one switches at the nuts lead finally to an overdetermined system of conditions and thereby provoke a contradiction.
\hypersetup{bookmarksdepth=0}
\setcounter{section}{0}
\renewcommand{\thesection}{\Alph{section}}
\addtocontents{toc}{\protect\setcounter{tocdepth}{0}}
\chapter*{Appendix}
\section{Conformal Metrics and Null Separated Points} \label{app:confmetric}

We show that two symmetric (0,2)-tensors are conformally equivalent if and only if their sets of null vectors are identical (assuming that we have a bijection or a diffeomorphism between the underlying manifolds). If the (pseudo-) metrics are conformally equivalent, then the null vectors are obviously identical. Conversely, if we know that the (pseudo-) metrics $g_{1}$, $g_{2}$ leave the null cones invariant, then timelike or spacelike vectors, respectively, in $g_{1}$ correspond to timelike or spacelike vectors, respectively, in $g_{2}$. We are interested in the Lorentzian case, and choose an orthonormal tetrad $e_{0}$, $e_{1}$, $e_{2}$, $e_{3}$ for $g_{1}$ such that $e_{0}$ is timelike and $e_{1}$, $e_{2}$, $e_{3}$ are spacelike. By the assumption that null cones are preserved this is also orthogonal tetrad for $g_{2}$, and
\begin{align*}
g_{1}(e_{0},e_{0})=1 \quad & \Rightarrow  \quad g_{2}(e_{0},e_{0})={\lambda}>0,\\
g_{1}(e_{i},e_{i})=1 \quad & \Rightarrow  \quad g_{2}(e_{i},e_{i})={\mu}_{i}<0,\ i=1,2,3.
\end{align*}
Moreover, $e_{0}+e_{i}$ is null for $g_{1}$ hence for $g_{2}$
\begin{align*}
0 & = g_{2}(e_{0}+e_{i},e_{0}+e_{i})=g_{2}(e_{0},e_{0})+2g_{2}(e_{0},e_{i})+g_{2}(e_{i},e_{i})\\
& = g_{2}(e_{0},e_{0})+g_{2}(e_{i},e_{i})= {\lambda}-{\mu}_{i}.
\end{align*}
This shows that all nonzero coefficients for $g_{2}$ are the same, that is $g_{1}$ and $g_{2}$ are conformally equivalent.

\section{Characterization of Simple Bivectors} \label{app:simplebivectors}

A bivector $x^{{\alpha}{\beta}}$ is simple $\Leftrightarrow $ ${\varepsilon}_{{\alpha}{\beta}{\gamma}{\delta}}x^{{\alpha}{\beta}}x^{{\gamma}{\delta}}=0$ $\Leftrightarrow $ $*x_{{\alpha}{\beta}}x^{{\alpha}{\beta}}=0$.

\begin{proof}

$x^{{\alpha}{\beta}}$ is simple $\Rightarrow $ ${\varepsilon}_{{\alpha}{\beta}{\gamma}{\delta}}x^{{\alpha}{\beta}}x^{{\gamma}{\delta}}=0$:
\begin{align*}
{\varepsilon}_{{\alpha}{\beta}{\gamma}{\delta}}x^{{\alpha}{\beta}}x^{{\gamma}{\delta}} & = {\varepsilon}_{{\alpha}{\beta}{\gamma}{\delta}}(Z^{\alpha}{\tilde Z}^{\beta}-Z^{\beta}{\tilde Z}^{\alpha})(Z^{\gamma}{\tilde Z}^{\delta}-Z^{\delta}{\tilde Z}^{\gamma})\\
					& = {\varepsilon}_{{\alpha}{\beta}{\gamma}{\delta}}(Z^{\alpha}Z^{\gamma}{\tilde Z}^{\beta}{\tilde Z}^{\delta}-Z^{\alpha}Z^{\delta}{\tilde Z}^{\beta}{\tilde Z}^{\gamma}\\
					& \hspace{1.5cm} -Z^{\beta}Z^{\gamma}{\tilde Z}^{\alpha}{\tilde Z}^{\delta}+Z^{\beta}Z^{\delta}{\tilde Z}^{\alpha}{\tilde Z}^{\gamma})\\
												& = 4 {\epsilon}_{{\alpha}{\beta}{\gamma}{\delta}}Z^{\alpha}Z^{\gamma}{\tilde Z}^{\beta}{\tilde Z}^{\delta}\\
												& = 0.
\end{align*}
\begin{equation*}
{\varepsilon}^{\vphantom{1}}_{{\alpha}{\beta}{\gamma}{\delta}}x^{{\alpha}{\beta}}x^{{\gamma}{\delta}}=0 \quad  \Rightarrow  \quad  *x_{{\alpha}{\beta}}x^{{\alpha}{\beta}}=0:\  *{x}_{{\alpha}{\beta}}x^{{\alpha}{\beta}}=\frac{1}{2}{\Delta}{\epsilon}_{{\alpha}{\beta}^{\vphantom{1}}{\gamma}{\delta}}x^{{\gamma}{\delta}}x^{{\alpha}{\beta}}=0
\end{equation*}
$*x_{{\alpha}{\beta}}x^{{\alpha}{\beta}}=0$ $\Rightarrow $ $x^{{\alpha}{\beta}}$ is simple:  First note that $x^{{\alpha}{\beta}}$ has to have rank 0, 2 or 4, because it is skew-symmetric (its non-degenerate part defines a symplectic bilinear form and symplectic vector spaces are even-dimensional). This implies that if $x$ has an eigenvector with eigenvalue 0, then $x$ has to be of rank 2 or 0. Supposing that $x^{{\alpha}{\beta}}$ is nonzero we can set 2 rows and 2 columns zero. But then $x$ has to be of the form $x^{{\alpha}{\beta}}=Z^{[{\alpha}}\tilde Z^{{\beta}]}$. It remains to show that such an eigenvector exists. Using relations for the alternating symbol one can verify that
\begin{equation*}
y^{{\alpha}{\gamma}}x_{{\beta}{\gamma}}^{\vphantom{1}}-*x^{{\alpha}{\gamma}}*y^{\vphantom{1}}_{{\beta}{\gamma}}=\frac{1}{2}x^{{\delta}{\lambda}}y_{{\delta}{\lambda}}^{\vphantom{1}}\tensor*{{\delta}}{^{\alpha}_{\beta}}
\end{equation*}
for bivectors $x^{{\alpha}{\beta}}$, $y^{{\alpha}{\beta}}$. Setting $y^{{\alpha}{\beta}}=*x^{{\alpha}{\beta}}$ and using $*^{2}=1$, ${\Delta}^{2}=-1$ yields immediately
\begin{equation*}
*x_{{\alpha}{\beta}}x^{{\beta}{\gamma}}=0 \quad \iff  \quad *x_{{\alpha}{\beta}}x^{{\alpha}{\beta}}=0.
\end{equation*}
So, our assumption implies that $x^{{\beta}{\gamma}}q_{{\gamma}}$ is an eigenvector with eigenvalue 0 for an arbitrary $q_{\gamma}$.
\end{proof}

\section{Computation for Reduced Einstein Equations} \label{app:redeinst}

Using \eqref{eq:cov1} and the Ricci identity we obtain
\begin{align*}
R_{ab}^{\vphantom{1}} X_{i}^{a}X_{j}^{b} & = X_{j}^{b}\tensor{R}{^{c}_{acb}}X^{a}_{i} = - X_{j}^{b} {\nabla}^{c}{\nabla}_{c}^{\vphantom{1}}X_{ib}^{\vphantom{1}}\\
					& = -\frac{1}{2} X_{j}^{b}{\nabla}_{c}^{\vphantom{1}}\left(J^{lk}((\partial ^{c}J_{ik}^{\vphantom{1}})X_{lb}^{\vphantom{1}}-(\partial _{b}^{\vphantom{1}}J_{ki}^{\vphantom{1}})X_{l}^{c})\right).
\end{align*}
Expanding by Leibniz rule
\begin{align*}
R_{ab}^{\vphantom{1}} X_{i}^{a}X_{j}^{b}	& = -\frac{1}{2}({\nabla}_{c}^{\vphantom{1}}J^{lk})\cdot (\underbrace{X_{j}^{b}X_{lb}^{\vphantom{1}}}_{=J_{jl}}(\partial ^{c}J_{ik}^{\vphantom{1}})-\underbrace{X_{j}^{b}({\partial}_{b}^{\vphantom{1}}J_{ki}^{\vphantom{1}})}_{=0}X_{l}^{c}) \\
					& \hspace{0.4 cm} -\frac{1}{2}\underbrace{X_{j}^{b}X_{lb}^{\vphantom{1}}}_{=J_{jl}}J^{lk}{\nabla}_{c}^{\vphantom{1}}\partial ^{c}J_{ik}^{\vphantom{1}}-\frac{1}{2} J^{lk}(\partial ^{c}J_{ik}^{\vphantom{1}})\underbrace{X_{j}^{b}{\nabla}_{c}^{\vphantom{1}}X_{lb}^{\vphantom{1}}}_{\frac{1}{2}\partial _{c}J_{jl}}\\
					& \hspace{0.4cm} +\frac{1}{2} J^{lk}X^{c}_{l}X^{b}_{j}{\nabla}_{c}^{\vphantom{1}}{\partial}_{b}^{\vphantom{1}}J_{ki}^{\vphantom{1}}+\frac{1}{2}J^{kl}\underbrace{X^{b}_{j}({\partial}_{b}^{\vphantom{1}}J_{ki}^{\vphantom{1}})}_{=0}{\nabla}_{c}^{\vphantom{1}}X^{c}_{l},
\end{align*}
and using \eqref{eq:delJredeinst} together with the expression for the Laplace-Beltrami operator
\begin{equation*}
\boxvoid \, u=\nabla _{a}\nabla ^{a}u= \frac{1}{\sqrt{|g|}}\partial _{a}\left(\sqrt{|g|}g^{ab}\partial _{b}u\right),
\end{equation*}
this yields
\begin{align*}
R_{ab}^{\vphantom{1}} X_{i}^{a}X_{j}^{b}	& = -\frac{1}{2}J_{jl}^{\vphantom{1}}(\partial ^{c}J_{ik}^{\vphantom{1}})({\partial}_{c}^{\vphantom{1}}J^{lk})-\frac{1}{2}\boxvoid J_{ij}-\frac{1}{4} J^{lk}(\partial ^{c}J_{ik})(\partial _{c}J_{jl}^{\vphantom{1}})\\
					& \hspace{0.4cm} +\frac{1}{2} J^{lk}X^{c}_{l}\nabla _{c}(\underbrace{X^{b}_{j}\partial _{b}J_{ki}}_{=0})-\frac{1}{2} J^{lk}\underbrace{X^{c}_{l}\nabla _{c}(X^{b}_{j})}_{-\frac{1}{2}\partial _{b}J_{jl}}\partial ^{b}J_{ki}.
\end{align*}
Some further substitutions and cancellations lead to
\begin{align*}
R_{ab}^{\vphantom{1}} X_{i}^{a}X_{j}^{b}	& = -\frac{1}{2}J_{jl}^{\vphantom{1}}({\partial}^{c}J_{ik}^{\vphantom{1}})({\partial}_{c}^{\vphantom{1}}J^{lk})-\frac{1}{2}\boxvoid J_{ij}^{\vphantom{1}}\\
					& \hspace{0.4cm}-\frac{1}{4} J^{lk}(\partial ^{c}J_{ik})(\partial _{c}J_{jl})+\frac{1}{4} J^{lk}(\partial ^{c}J_{ik})(\partial _{c}J_{jl}),
\end{align*}
and eventually
\begin{align*}
R_{ab}^{\vphantom{1}} X_{i}^{a}X_{j}^{b}	& = -\frac{1}{2}J_{jl}^{\vphantom{1}}(\partial ^{c}J_{ik}^{\vphantom{1}})({\partial}_{c}^{\vphantom{1}}J^{lk})-\frac{1}{2}\boxvoid J_{ij}^{\vphantom{1}}\\
					& = -\frac{1}{2}J_{ik}J^{kl}\boxvoid J_{lj}-\frac{1}{2}J_{jk}g^{-\frac{1}{2}}(\partial _{a}J^{kl})g^{\frac{1}{2}}g^{ab}(\partial _{b}J_{li})\\
					& = -\frac{1}{2}J_{ik} g^{-\frac{1}{2}}\partial _{a}(g^{\frac{1}{2}}g^{ab}J^{kl}\partial _{b}J_{lj})
\end{align*}
as claimed. Note that the covariant derivative for the functions $J_{ik}$ is the same as the partial derivative.

\section{Isothermal Coordinates} \label{app:isothermal}

\textit{Isothermal coordinates} on a manifold are local coordinates where the metric has the form 
\begin{equation*}
\drm s^{2} = {\erm}^{2ν} (\drm x_{1}^{2}+\dotsc \drm x_{n}^{2})
\end{equation*}
with ${\erm}^{ν}$ a smooth function. In the real case the manifold should be Riemannian so that isothermal coordinates are conformal to the Euclidean metric. 

For a 2-surface isothermal coordinates always exist locally, which can be seen as follows. The coordinates $r$ and $x$ are isothermal if they satisfy
\begin{equation*}
*\drm r = \drm x
\end{equation*}
where $*$ is again the Hodge star operator. Suppose ${\Delta} = \drm {\delta}+ {\delta} \drm={\delta} \drm$ is the Laplace-Beltrami operator on functions,\footnote{Here ${\delta}$ denotes the codifferential according to the differential $\drm$. As the codifferential lowers the grade of a form, we get on functions ${\delta}r=0$. One can proof that $\nabla _{a}\nabla ^{a}$ and $\drm {\delta}+ {\delta} \drm$ are equivalent when acting on scalar functions.} then by standard elliptic theory, we can choose $r$ to be harmonic near a given point, that is ${\Delta}r=0$, with non-vanishing gradient $\drm r$. In our case we were given such an $r$. By the Poincar\'e lemma $*\drm r = \drm x$ has a local solution $x$ if and only if $\drm * \drm r=0$. Because of ${\delta}=* \drm *$, this is equivalent to ${\Delta}r=0$, thus local solutions exists. From the facts that $\drm r$ is non-zero and  $*^{2}=-1$ on 1-forms, we conclude that $\drm r$ and $\drm x$ are necessarily linearly independent, and therefore give local isothermal coordinates.

In our case the overall sign of the metric is not fixed, because we only asked the orthogonal 2-surfaces to be non-null. But the exclusion of null vectors still leaves two possibilities for the signature, $(+,+)$ or $(-,-)$. 

\section{Integrability of Einstein Equations} \label{app:intefe}

The second part of the reduced vacuum Einstein equations~\eqref{eq:redeinst3} are in full
\begin{align*}
2\irm \partial _{{\xi} \vphantom{\bar {\xi}}}(\log r {\erm}^{2ν}) & = r \tr (\partial _{{\xi} \vphantom{\bar {\xi}}}J^{-1}\partial _{{\xi} \vphantom{\bar {\xi}}}J), \\
-2\irm \partial _{\bar {\xi}}(\log r {\erm}^{2ν}) & = r \tr (\partial _{\bar {\xi}}J^{-1}\partial _{\bar {\xi}}J).
\end{align*}
By use of \eqref{eq:redyang} and the constraint $\det J =-r^{2}$, we have to show that they are integrable, that is
\begin{equation*}
\partial _{\bar {\xi}}\partial _{{\xi} \vphantom{\bar {\xi}}}(\log r {\erm}^{2ν}) = \partial _{{\xi} \vphantom{\bar {\xi}}} \partial _{\bar {\xi}}(\log r {\erm}^{2ν}),
\end{equation*}
which is equivalent to
\begin{equation} \label{eq:efeint} \tag{App1}
\partial _{\bar {\xi}}\left(r \tr (\partial _{{\xi} \vphantom{\bar {\xi}}}J^{-1}\partial _{{\xi} \vphantom{\bar {\xi}}}J)\right) + \partial _{{\xi} \vphantom{\bar {\xi}}}\left( r \tr (\partial _{\bar {\xi}}J^{-1}\partial _{\bar {\xi}}J)\right)=0.
\end{equation}
With 
\begin{align*}
\partial _{{\xi} \vphantom{\bar {\xi}}}=\frac{\partial x}{\partial {\xi}}\partial _{x}+\frac{\partial r}{\partial {\xi}}\partial _{r}=\frac{1}{2}\partial _{x}+\frac{1}{2\irm}\partial _{r},\\
\partial _{\bar {\xi}}=\frac{\partial x}{\partial \bar {\xi}}\partial _{x}+\frac{\partial r}{\partial  \bar {\xi}}\partial _{r}=\frac{1}{2}\partial _{x}-\frac{1}{2\irm}\partial _{r},
\end{align*}
we first rewrite \eqref{eq:redyang} in the following way
\begin{align*}
0 & = \partial _{{\xi} \vphantom{\bar {\xi}}}(r J^{-1}\partial _{\bar {\xi}}J)+\partial _{\bar {\xi}}(r J^{-1}\partial _{{\xi} \vphantom{\bar {\xi}}}J) \\
& = \frac{1}{4}(\partial _{x}+\frac{1}{\irm}\partial _{r})(rJ^{-1}(\partial _{x}-i\frac{1}{\irm}\partial _{r})J)+\frac{1}{4}(\partial _{x}-\frac{1}{\irm}\partial _{r})(rJ^{-1}(\partial _{x}+i\frac{1}{\irm}\partial _{r})J)\\
& = \partial _{x}(r J^{-1}\partial _{x}J)+\partial _{r}(r J^{-1}\partial _{r}J)
\end{align*}
This is applied in the following form
\begin{equation} \label{eq:redyang2} \tag{App2}
2rJ^{-1}\partial _{{\xi} \vphantom{\bar {\xi}}}\partial _{\bar {\xi}}J = -r \partial _{{\xi} \vphantom{\bar {\xi}}}J^{-1}\partial _{\bar {\xi}}J-r\partial _{\bar {\xi}}J^{-1}\partial _{{\xi} \vphantom{\bar {\xi}}}J-\frac{1}{2\irm}(J^{-1}\partial _{\bar {\xi}}J-J^{-1}\partial _{{\xi} \vphantom{\bar {\xi}}}J), 
\end{equation}
where the derivatives act only on the variable immediately to its right if no further parenthesis indicate it differently. Note that by $\partial A^{-1}=-A^{-1}(\partial A)A^{-1}$ for a matrix $A$, equation~\eqref{eq:efeint} is equivalent to
\begin{equation*}
\partial _{\bar {\xi}}\left(r \tr ((J^{-1}\partial _{{\xi} \vphantom{\bar {\xi}}}J)^{2})\right) + \partial _{{\xi} \vphantom{\bar {\xi}}}\left( r \tr ((J^{-1}\partial _{\bar {\xi}}J)^{2})\right)=0,
\end{equation*}
which can now be proved with \eqref{eq:redyang2} by expanding the derivatives
\begin{align*}
& \partial _{\bar {\xi}}\left(r \tr ((J^{-1}\partial _{{\xi} \vphantom{\bar {\xi}}}J)^{2})\right) + \partial _{{\xi} \vphantom{\bar {\xi}}}\left( r \tr ((J^{-1}\partial _{\bar {\xi}}J)^{2})\right)\\
& = -\frac{1}{2\irm} \tr \left((J^{-1}\partial _{{\xi} \vphantom{\bar {\xi}}}J)^{2}\right)+2r\tr\left((\partial _{\bar {\xi}}J^{-1}\partial _{{\xi} \vphantom{\bar {\xi}}}J)(J^{-1}\partial _{{\xi} \vphantom{\bar {\xi}}}J)\right) \\
& \hspace{0.4cm} + 2r \tr \left((J^{-1}\partial _{\bar {\xi}}\partial _{{\xi} \vphantom{\bar {\xi}}}J)(J^{-1}\partial _{{\xi} \vphantom{\bar {\xi}}}J)\right)\\
& \hspace{0.4cm} +\frac{1}{2\irm} \tr \left((J^{-1}\partial _{\bar {\xi}}J)^{2}\right)+2r\tr\left((\partial _{{\xi} \vphantom{\bar {\xi}}}J^{-1}\partial _{\bar {\xi}}J)(J^{-1}\partial _{\bar {\xi}}J)\right) \\
& \hspace{0.4cm} + 2r \tr \left((J^{-1}\partial _{{\xi} \vphantom{\bar {\xi}}}\partial _{\bar {\xi}}J)(J^{-1}\partial _{{\xi} \vphantom{\bar {\xi}}}J)\right)\\
& = -\frac{1}{2\irm} \tr \left((J^{-1}\partial _{{\xi} \vphantom{\bar {\xi}}}J)^{2}\right)+2r\tr\left((\partial _{\bar {\xi}}J^{-1}\partial _{{\xi} \vphantom{\bar {\xi}}}J)(J^{-1}\partial _{{\xi} \vphantom{\bar {\xi}}}J)\right) \\
& \hspace{0.4cm} + \tr \left((-r (\partial _{{\xi} \vphantom{\bar {\xi}}}J^{-1})\partial _{\bar {\xi}}J-r(\partial _{\bar {\xi}}J^{-1})\partial _{{\xi} \vphantom{\bar {\xi}}}J-\frac{1}{2\irm}(J^{-1}\partial _{\bar {\xi}}J-J^{-1}\partial _{{\xi} \vphantom{\bar {\xi}}}J))(J^{-1}\partial _{{\xi} \vphantom{\bar {\xi}}}J)\right)\\
& \hspace{0.4cm} +\frac{1}{2\irm} \tr \left((J^{-1}\partial _{\bar {\xi}}J)^{2}\right)+2r\tr\left((\partial _{{\xi} \vphantom{\bar {\xi}}}J^{-1}\partial _{\bar {\xi}}J)(J^{-1}\partial _{\bar {\xi}}J)\right) \\
& \hspace{0.4cm} + \tr \left((-r (\partial _{{\xi} \vphantom{\bar {\xi}}}J^{-1})\partial _{\bar {\xi}}J-r(\partial _{\bar {\xi}}J^{-1})\partial _{{\xi} \vphantom{\bar {\xi}}}J-\frac{1}{2\irm}(J^{-1}\partial _{\bar {\xi}}J-J^{-1}\partial _{{\xi} \vphantom{\bar {\xi}}}J))(J^{-1}\partial _{{\xi} \vphantom{\bar {\xi}}}J)\right)\\
& = 0
\end{align*}
All terms cancel so that \eqref{eq:efeint} is integrable.

\section{Criterion for Orthogonal Transitivity} \label{app:orthtrans}

We use the following theorem.
\begin{thm*}
Let $X_{i}$, $i=1,\dotsc ,n-2$, be $n-2$ commuting Killing vectors in an $n$-dimensional real or complex manifold such that
\begin{enumerate}
\item the tensor $X_{1}^{[a_{1}}X_{2}^{a_{2}\vphantom{]}}\cdots X_{n-2}^{a_{n-2}\vphantom{]}}\nabla ^{b\vphantom{]}}X_{i}^{c]}$ vanishes at least at one point for every $i=1,\dotsc ,n-2$ and \label{item:orthtranscond1}
\item the tensor $X_{i}^{c\vphantom{]}}\tensor{R}{_{c\vphantom{i}}^{[b}}X_{1}^{a_{1}\vphantom{]}}X_{2}^{a_{2}\vphantom{]}}\cdots X_{n-2}^{a_{n-2}]}=0$ for all $i=1,\dotsc ,n-2$, \label{item:orthtranscond2}
\end{enumerate}
then the 2-planes orthogonal to the Killing vectors $X_{i}$, $i=1,\dotsc ,n-2$, are integrable.
\end{thm*}
For dimension four this is proven in \cite[Thm.~7.1.1]{Wald:1984rz} using Frobenius' theorem on integrable submanifolds; the generalization is due to Emparan and Reall \cite{Emparan:2002dn}. In the Einstein vacuum case condition \eqref{item:orthtranscond2} is automatically satisfied. Condition \eqref{item:orthtranscond1} is moreover satisfied if for example one of the Killing vectors corresponds to an axisymmetry, hence it vanishes on its ``rotation axis''. If we suppose both holds, then the obtain the metric in the ${\sigma}$-model form. In the real case we want our space-time also stationary (to get isothermal coordinates on space-time modulo symmetries). 

\section{Conical Singularities} \label{app:consing}

Suppose we parameterize the flat, real, three-dimensional space in conical coordinates. Then the metric takes the form
\begin{equation*}
\drm x^{2}+\drm y^{2}+\drm z^{2}=\drm r^{2} + k^{2} r^{2} \drm {\varphi}^{2} +\drm z^{2},
\end{equation*}
which is singular at $r=0$. Apparently, this is only a coordinate effect. So, conversely given a metric
\begin{equation*}
\drm r^{2} + k^{2} r^{2} \drm {\varphi}^{2} +\drm z^{2}
\end{equation*}
we can define $x=r \cos {\theta}$, $y=r \sin {\theta}$, ${\theta}=k{\varphi}$ so that
\begin{equation*}
\drm r^{2} + k^{2} r^{2} \drm {\varphi}^{2} +\drm z^{2}=\drm x^{2}+\drm y^{2}+\drm z^{2}.
\end{equation*}
However, this requires that the periodicities of ${\varphi}$ and ${\theta}$ are in a ratio of $\frac{2{\pi}}{k}$. If they do not have the correct periodicity, then a singularity at the apex $\{z=0,\, r=0\}$ appears which is of the form ``$\mathrm{Riem}={\delta}(r)$''. That is, there are diffeomorphism invariant quantities which become singular at the origin but their limit towards it is finite. Lastly, a short picture how the above form of the metric relates to the picture of a cone. In cylinder coordinates
\begin{equation*}
\drm R^{2} + R^{2} \drm {\varphi}^{2}+\drm Z^{2}
\end{equation*}
a cone with aperture $2{\alpha}$ is given by $Z=R \cot {\alpha}$, see Figure~\ref{fig:cone}. With $\drm Z^{2}=\drm R^{2} \cot {\alpha}$ the metric on the cone becomes
\begin{equation*}
\sin^{-2}{\alpha}\, \drm R^{2}+R^{2}\drm {\varphi}^{2}=\drm r^{2}+\sin^{2}{\alpha}\, r^{2}\drm {\varphi}^{2}
\end{equation*}
where $R=r\sin {\alpha} $. This is exactly our conical parameterization of the $(x,y)$-plane above with the relation $k^{2}=\sin^{2}{\alpha}$. 
\begin{figure}[htbp]
\begin{center}
     \scalebox{0.8}{\input{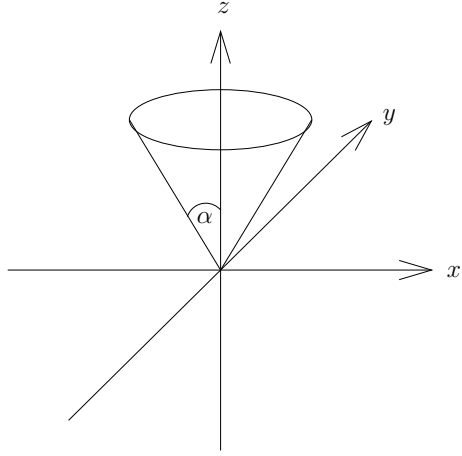}}
     \caption{Cone with aperture $2{\alpha}$.} 
     \label{fig:cone}
\end{center}
\end{figure}

\section{Rod Structure of Schwarzschild Solution} \label{app:schwarz}
With the substitution
\begin{equation*}
z=(R-m)\cos {\Theta}, \quad r=(R^{2}-2mR)^{\frac{1}{2}}\sin {\Theta},\quad t=T, \quad {\theta}= {\Phi}
\end{equation*}
for the Schwarzschild metric
\begin{equation*}
\drm s^{2} = \left(1-\frac{2m}{R}\right) \drm T^{2}-\left(1-\frac{2m}{R}\right)^{-1} \drm R^{2} - R^{2}(\drm {\Theta}^{2}+ \sin^{2}{\Theta}\, \drm {\Phi}^{2})
\end{equation*}
we calculate
\begin{align*}
(z+m)^{2} & = \left(R \cos {\Theta} + m(1-\cos {\Theta})\right)^{2}\\
	& = R^{2}\cos^{2}{\Theta}+2mR\cos{\Theta}(1-\cos{\Theta})+m^{2}(1-\cos{\Theta})^{2},\\
(z-m)^{2} & = \left(R \cos {\Theta} - m(1+\cos {\Theta})\right)^{2}\\
	& = R^{2}\cos^{2}{\Theta}-2mR\cos{\Theta}(1+\cos{\Theta})+m^{2}(1+\cos{\Theta})^{2},\\
r^{2} & = R^{2} \sin^{2}{\Theta} - 2m R \sin^{2}{\Theta}.
\end{align*}
This yields
\begin{align*}
r^{2}+(z+m)^{2} & = R^{2}+2mR(\cos{\Theta}-1)+m^{2}(1-\cos{\Theta})^{2}\\
		& = \left(R+m(\cos{\Theta}-1)\right)^{2},\\
r^{2}+(z-m)^{2} & = R^{2}-2mR(\cos{\Theta}+1)+m^{2}(1+\cos{\Theta})^{2}\\
		& = \left(R-m(\cos{\Theta}+1)\right)^{2},
\end{align*}
and therefore
\begin{equation*}
r_{+}+r_{-}=2R-2m\quad \text{and}\quad f=\frac{2R-4m}{2R}=1-\frac{2m}{R}. 
\end{equation*}
But from the metric we see 
\begin{equation*}
g(Y,Y)=1-\frac{2m}{R},
\end{equation*}
and 
\begin{align*}
g(X,X) & =-R^{2}\sin^{2}{\Theta}=-\frac{(R^{2}-2mR)R^{2}\sin^{2}{\Theta}}{R^{2}-2mR}=-\frac{(R^{2}-2mR)\sin^{2}{\Theta}}{1-\frac{2m}{R}}\\
	& =-\frac{r^{2}}{f}.
\end{align*}
Then $J$ takes the desired form.

\section{General B\"acklund Transformation} \label{app:btrfcalc}
Starting from the given decomposition we have\footnote{Using $\left(\begin{array}{cc}A^{-1} & 0 \\B^{\hphantom{-1}} & 1\end{array}\right)^{-1}
=\left(\begin{array}{cc}A & 0 \\-BA & 1\end{array}\right)$.}
\begin{equation*}
\renewcommand{\arraystretch}{1.5}
J^{-1}=\left(\begin{array}{cc}A^{-1} & 0 \\B^{\hphantom{-1}} & 1\end{array}\right)^{-1}
\left(\begin{array}{cc}1 & \tilde B^{\hphantom{-1}} \\0 & \skew{7}{\tilde}{A}^{-1}\end{array}\right),
\end{equation*}
and
\begin{equation*}
\renewcommand{\arraystretch}{1.5}
\begin{split}
\partial _{{\alpha}}J & = -\left(\begin{array}{cc}1 & \tilde B^{\hphantom{-1}} \\0 & \skew{7}{\tilde}{A}^{-1}\end{array}\right)^{-1}
\left(\begin{array}{cc} 0& \tilde B_{{\alpha}} \\0 & -\skew{7}{\tilde}{A}^{-1}\skew{7}{\tilde}{A}_{{\alpha}} \skew{7}{\tilde}{A}^{-1}\end{array}\right)
\left(\begin{array}{cc}1 & \tilde B^{\hphantom{-1}} \\0 & \skew{7}{\tilde}{A}^{-1}\end{array}\right)^{-1}
\left(\begin{array}{cc}A^{-1} & 0 \\B^{\hphantom{-1}} & 1\end{array}\right)\\
& \hspace{0.4cm} + \left(\begin{array}{cc}1 & \tilde B^{\hphantom{-1}} \\0 & \skew{7}{\tilde}{A}^{-1}\end{array}\right)^{-1}
\left(\begin{array}{cc} -A^{-1}A_{{\alpha}}A^{-1} & 0 \\B_{{\alpha}} & 0\end{array}\right),
\end{split}
\end{equation*}
where ${\alpha}\in \{\tilde w,\tilde z\}$. This yields
\begin{align*}
J^{-1}&\partial _{{\alpha}}J\\
& = \renewcommand{\arraystretch}{1.5}
-\left(\begin{array}{cc}A^{-1} & 0 \\B^{\hphantom{-1}} & 1\end{array}\right)^{-1}
\left(\begin{array}{cc} 0& \tilde B_{{\alpha}} \\0 & -\skew{7}{\tilde}{A}^{-1}\skew{7}{\tilde}{A}_{{\alpha}} \skew{7}{\tilde}{A}^{-1}\end{array}\right)
\left(\begin{array}{cc}1 & \tilde B^{\hphantom{-1}} \\0 & \skew{7}{\tilde}{A}^{-1}\end{array}\right)^{-1}
\left(\begin{array}{cc}A^{-1} & 0 \\B^{\hphantom{-1}} & 1\end{array}\right)\\
& \hspace{0.4cm} + \renewcommand{\arraystretch}{1.5}
\left(\begin{array}{cc}A^{-1} & 0 \\B^{\hphantom{-1}} & 1\end{array}\right)^{-1}
\left(\begin{array}{cc} -A^{-1}A_{{\alpha}}A^{-1} & 0 \\B_{{\alpha}} & 0\end{array}\right)\\
& = \renewcommand{\arraystretch}{1.5}
-\left(\begin{array}{cc}A & 0 \\-BA & 1\end{array}\right)
\left(\begin{array}{cc} 0& \tilde B_{{\alpha}} \\0 & -\skew{7}{\tilde}{A}^{-1}\skew{7}{\tilde}{A}_{{\alpha}} \skew{7}{\tilde}{A}^{-1}\end{array}\right)
\left(\begin{array}{cc}1 & -\tilde B \skew{7}{\tilde}{A} \\0 & \skew{7}{\tilde}{A}\end{array}\right)
\left(\begin{array}{cc}A^{-1} & 0 \\B^{\hphantom{-1}} & 1\end{array}\right)\\
& \hspace{0.4cm} + \renewcommand{\arraystretch}{1.5}
\left(\begin{array}{cc}A & 0 \\-BA & 1\end{array}\right)
\left(\begin{array}{cc} -A^{-1}A_{{\alpha}}A^{-1} & 0 \\B_{{\alpha}} & 0\end{array}\right)\\
& = \renewcommand{\arraystretch}{2}
\left(\begin{array}{cc}-A\tilde B_{{\alpha}}\skew{7}{\tilde}{A} B-A_{{\alpha}}A^{-1} & -A\tilde B_{{\alpha}}\skew{7}{\tilde}{A} \\BA\tilde B_{{\alpha}}\skew{7}{\tilde}{A}B+\skew{7}{\tilde}{A}^{-1}\skew{7}{\tilde}{A}_{{\alpha}}B+BA_{{\alpha}}A^{-1}+B_{{\alpha}} & BA\tilde B_{{\alpha}}\skew{7}{\tilde}{A}+\skew{7}{\tilde}{A}^{-1}\skew{7}{\tilde}{A}_{{\alpha}}\end{array}\right).
\end{align*}
From this we can read off the form of Yang's equation~\eqref{eq:yang2}. The (12)-entry of the matrix yields
\begin{equation*}
\partial _{z \vphantom{\tilde z}}(A \tilde B_{\tilde z}\skew{7}{\tilde}{A})-\partial _{w \vphantom{\tilde z} }(A \tilde B_{\tilde w}\skew{7}{\tilde}{A})=0,
\end{equation*}
which is the second equation in \eqref{eq:btrf}. Using this we get for the (11)-entry
\begin{equation*}
A\tilde B_{\tilde z}\skew{7}{\tilde}{A} B_{z \vphantom{\tilde z} } + \partial _{z \vphantom{\tilde z} }(A_{\tilde z} A^{-1})-A\tilde B_{\tilde w}\skew{7}{\tilde}{A} B_{w \vphantom{\tilde z} }-\partial _{w \vphantom{\tilde z} }(A_{\tilde w} A^{-1})=0,
\end{equation*}
which corresponds to the last equation in \eqref{eq:btrf}. Again using $A\tilde B_{{\alpha}} \skew{7}{\tilde}{A}=0$ we get for the (22)-entry
\begin{equation*}
B_{z}A\tilde B_{\tilde z}\skew{7}{\tilde}{A}+\partial _{z \vphantom{\tilde z} }(\skew{7}{\tilde}{A}^{-1}\skew{7}{\tilde}{A}_{\tilde z})-B_{w \vphantom{\tilde z} }A\tilde B_{\tilde w}\skew{7}{\tilde}{A}-\partial _{w \vphantom{\tilde z} }(\skew{7}{\tilde}{A}^{-1}\skew{7}{\tilde}{A}_{\tilde w})=0
\end{equation*}
corresponding to the third equation in \eqref{eq:btrf}. Lastly, with the vanishing of the (11)-entry we obtain for the (21)-entry
\begin{align*}
0 & = (BA\tilde B_{\tilde z}\skew{7}{\tilde}{A}+\skew{7}{\tilde}{A}^{-1}\skew{7}{\tilde}{A}_{\tilde z})B_{z \vphantom{\tilde z} }+(BA_{\tilde z}A^{-1})_{z \vphantom{\tilde z} }+(\skew{7}{\tilde}{A}^{-1})_{z \vphantom{\tilde z} }\skew{7}{\tilde}{A} B_{\tilde z}+\skew{7}{\tilde}{A}^{-1}(\skew{7}{\tilde}{A} B_{\tilde z}A)_{z \vphantom{\tilde z} }A^{-1}\\
& \hspace{0.4cm}+B_{\tilde z}A(A^{-1})_{z \vphantom{\tilde z} } - (z,\tilde z \leftrightarrow  w, \tilde w)\\
& = \skew{7}{\tilde}{A}^{-1} \skew{7}{\tilde}{A}_{\tilde z}B_{z \vphantom{\tilde z} }+B_{z \vphantom{\tilde z} }A_{\tilde z}A^{-1}+(\skew{7}{\tilde}{A}^{-1})_{z \vphantom{\tilde z} }\skew{7}{\tilde}{A} B_{\tilde z}+B_{\tilde z}A(A^{-1})_{z \vphantom{\tilde z} }\\
& \hspace{0.4cm} +\skew{7}{\tilde}{A}^{-1}(\skew{7}{\tilde}{A} B_{\tilde z}A)_{z \vphantom{\tilde z} }A^{-1} - (z,\tilde z \leftrightarrow  w, \tilde w)\\
& = \skew{7}{\tilde}{A}^{-1}(\skew{7}{\tilde}{A}_{\tilde z}B_{z \vphantom{\tilde z} }-\skew{7}{\tilde}{A}_{z \vphantom{\tilde z} }B_{\tilde z}) + (B_{z \vphantom{\tilde z} }A_{\tilde z}-B_{\tilde z}A_{z \vphantom{\tilde z} })A^{-1}\\
& \hspace{0.4cm} + \skew{7}{\tilde}{A}^{-1}(\skew{7}{\tilde}{A} B_{\tilde z}A)_{z \vphantom{\tilde z} }A^{-1} - (z,\tilde z \leftrightarrow  w, \tilde w).  
\end{align*}
Multiplying with $\skew{7}{\tilde}{A}$ from left and $A$ from right gives
\begin{align*}
0 & = (\skew{7}{\tilde}{A}_{\tilde z}B_{z \vphantom{\tilde z} }-\skew{7}{\tilde}{A}_{z \vphantom{\tilde z} }B_{\tilde z}) A + \skew{7}{\tilde}{A} (B_{z \vphantom{\tilde z} }A_{\tilde z}-B_{\tilde z}A_{z \vphantom{\tilde z} }) + (\skew{7}{\tilde}{A} B_{\tilde z}A)_{z \vphantom{\tilde z} } - (z,\tilde z \leftrightarrow  w, \tilde w)\\
& = (\skew{7}{\tilde}{A} B_{z \vphantom{\tilde z} }A)_{\tilde z} - \skew{7}{\tilde}{A} B_{z\tilde z}A-(\skew{7}{\tilde}{A} B_{\tilde z} A)_{z \vphantom{\tilde z} }+\skew{7}{\tilde}{A} B_{z\tilde z}A+(\skew{7}{\tilde}{A} B_{\tilde z}A)_{z \vphantom{\tilde z} }- (z,\tilde z \leftrightarrow  w, \tilde w)\\
& = (\skew{7}{\tilde}{A} B_{z \vphantom{\tilde z} }A)_{\tilde z}-(\skew{7}{\tilde}{A} B_{w \vphantom{\tilde z} }A)_{\tilde w}, 
\end{align*}
that is the first equation in \eqref{eq:btrf}.

\section{Reduced B\"acklund Transformation} \label{app:redbtrfcalc}

To obtain the reduced form of the Yang's equation in terms of the B\"acklund decomposition we proceed as above. By taking ${\alpha}\in \{r,x\}$ we immediately get the last three equations and the part with the $x$-derivatives in the first equation. Only the $r$-derivatives in the first equation need a closer look. But there we have in the same way as above
\begin{align*}
& r(BA\tilde B_{r}\skew{7}{\tilde}{A}+\skew{7}{\tilde}{A}^{-1}\skew{7}{\tilde}{A}_{r})B_{r}+(rBA_{r}A^{-1})_{r} + (r \skew{7}{\tilde}{A}^{-1} \skew{7}{\tilde}{A} B_{r}AA^{-1})_{r}\\
& = r\skew{7}{\tilde}{A}^{-1}\skew{7}{\tilde}{A}_{r}B_{r}+rB_{r}A_{r}A^{-1}+r(\skew{7}{\tilde}{A}^{-1})_{r}\skew{7}{\tilde}{A} B_{r}+rB_{r}A (A^{-1})_{r}\\
& \hspace{0.4cm}+\skew{7}{\tilde}{A}^{-1}(r\skew{7}{\tilde}{A} B_{r}A)_{r} A^{-1}\\
& = r \skew{7}{\tilde}{A}^{-1}(\skew{7}{\tilde}{A}_{r}B_{r}-\skew{7}{\tilde}{A}_{r} B_{r})+r(B_{r}A_{r}-B_{r}A_{r})A^{-1}+\skew{7}{\tilde}{A}^{-1}(r\skew{7}{\tilde}{A} B_{r}A)_{r} A^{-1}\\
& = \skew{7}{\tilde}{A}^{-1}(r\skew{7}{\tilde}{A} B_{r}A)_{r} A^{-1},
\end{align*}
which yields the claimed form of the equations.

\subsection{Two-Nut Spacetime: Parameter Configurations} \label{app:MPpar}

Given the set of parameters $(α>0, L_{1},L_{2})$, we would like to see whether the mass, that is $c_{3}$, is uniquely determined. If we set $x=c_{3}$, then from \eqref{eq:MPalpha} we get
\begin{equation*}
4α^{2}x^{4}=\left(x^{3}-\frac{\tilde L^{2}}{2}\right)^{2}-\tilde L_{1}^{2}\tilde L_{2}^{2}
\end{equation*}
with $2L_{i}=π\tilde L_{i}$ and $\tilde L^{2} = \tilde L_{1}^{2}+\tilde L_{2}^{2}$. This is equivalent to finding the positive zeros of
\begin{equation*}
F(x)\coloneqq x^{6}-4α^{2}x^{4}-\tilde L^{2} x^{3} + \frac{1}{4}\left(\tilde L_{1}^{2}-\tilde L_{2}^{2}\right)^{2}.
\end{equation*}
Rewrite this as
\begin{equation*}
F(x) = \underbrace{x^{3} \left(x^{3}-4α^{2}x-\tilde L^{2}\right)}_{G(x)}+K, \quad \text{where } K =\frac{1}{4}\left(\tilde L_{1}^{2}-\tilde L_{2}^{2}\right)^{2} ≥ 0.
\end{equation*}
Now
\begin{equation*}
G'(x) = x^{2}\underbrace{(6x^{3}-16α^{2}x-3\tilde L^{2})}_{g(x)}
\end{equation*}
Since $g(x)$ is a third order polynomial with positive leading coefficient, and negative value and slope for $x=0$, it has to have precisely one zero for $x>0$. Therefore, $G(x)$ has precisely two extremal points for $x<0$, a saddle point at the origin and precisely one minimum for $x>0$. At this minimum it is 
\begin{equation*}
\left. x^{3}=\frac{8}{3}α^{2}x+\frac{1}{2}\tilde L^{2}\right|_{x_{\mathrm{min}}>0},
\end{equation*}
hence
\begin{align*}
\left. G(x)\right|_{x_{\mathrm{min}}>0} = \left(\frac{8}{3}α^{2}x+\frac{1}{2}\tilde L^{2}\right)\left(-\frac{4}{3}α^{2}x-\frac{1}{2}\tilde L^{2}\right) = -\frac{32}{9}α^{4}x^{2}-2α^{2}\tilde L ^{2} x -\frac{1}{4} \tilde L^{4}.
\end{align*}
Finally, we see that
\begin{equation*}
\left. F(x) \right|_{x_{\mathrm{min}}>0}=\left. G(x) \right|_{x_{\mathrm{min}}>0} +K = -\frac{32}{9}α^{4}x^{2}-2α^{2}\tilde L ^{2} x - \tilde L_{1}^{2}\tilde L_{2}^{2}<0,
\end{equation*}
and because $G(0)=0$, we conclude that with no conditions on $(α>0, L_{1},L_{2})$ there are two positive solutions for $x$ (unless $\tilde L_{1}^{2}=\tilde L_{2}^{2}$ when there is only one).

\section{Implications of Hypersurface-Orthogonality} \label{app:hso}

A Killing vector $K$ is called \textit{hypersurface-orthogonal} if $K$ is the normal of a hypersurface. By \cite[Thm.~B.3.2]{Wald:1984rz} this is equivalent to the vanishing $K∧\drm K$, where we denote the Killing 1-form by $K$ as well. To sketch how this implies the vanishing of certain metric coefficients we write in our case
\begin{equation*}
Ψ = g_{tψ}\, \drm t + g_{φψ}\, \drm φ + g_{ψψ}\, \drm ψ, 
\end{equation*}
then
\begin{equation*}
\drm Ψ = g_{tψ,r}\, \drm r ∧ \drm t + g_{φψ,r}\, \drm r ∧ \drm φ + g_{ψψ,r}\, \drm r ∧ \drm ψ + r ↔ z
\end{equation*}
and $Ψ∧\drm Ψ=0$ is equivalent to
\begin{align*}
g_{φψ}\, \drm g_{tψ} & = g_{tψ}\, \drm g_{φψ}, \\
g_{ψψ}\, \drm g_{tψ} & = g_{tψ}\, \drm g_{ψψ}, \\
g_{ψψ}\, \drm g_{φψ} & = g_{φψ}\, \drm g_{ψψ}, 
\end{align*}
where $\drm = ∂_{r}+∂_{z}$. If all three of the metric coefficients are non-zero these equations can be integrated easily, showing that all three metric coefficients have to be proportional. However, this is a contradiction with respect to the asymptotic form of the metric if we regard this set of equations for a region that extends to $\sqrt{r^{2}+z^{2}}→∞$. Thus, some of the metric coefficients have to vanish, and it is not hard to see that the only possibility, which is compatible with the asymptotic form of the metric, is when $g_{t{\psi}}=g_{{\varphi}{\psi}}=0$.

\bibliographystyle{NormanPlainnat}						
\bibliography{/Users/norman/mathematics/Papers/library} 							
\end{document}